\documentclass[aos]{imsart_MB}

\RequirePackage{amsthm,amsmath,amsfonts,amssymb}
\RequirePackage[authoryear]{natbib}
\RequirePackage[colorlinks,citecolor=blue,urlcolor=blue]{hyperref}
\RequirePackage{graphicx}

\usepackage{amsxtra,accents}
\usepackage{graphicx}
\usepackage{float}
\usepackage[utf8]{inputenc}
\usepackage[T1]{fontenc}
\usepackage{fancyhdr}
\usepackage{bbm}
\usepackage{color}
\usepackage[linesnumbered, ruled,vlined]{algorithm2e}
\usepackage{rotating}
\usepackage{colortbl} 
\usepackage{xcolor} 
\usepackage{xfrac}
\usepackage{enumerate}
\RequirePackage[OT1]{fontenc}
\RequirePackage[colorlinks]{hyperref} 
\usepackage{multirow}
\usepackage{amsmath}
\usepackage{mathtools}
\usepackage{mathrsfs}
\usepackage{paralist}
\usepackage{bigints}
\usepackage{lscape}
\usepackage{bm}
\usepackage{booktabs}
\usepackage{epstopdf}
\usepackage{accents} 
\usepackage{scalerel}
\usepackage{stackengine,wasysym}
\usepackage{booktabs}

\usepackage{xr}
\DeclareMathAccent{\wtilde}{\mathord}{largesymbols}{"65}

\startlocaldefs
\theoremstyle{plain}

\newtheorem{theorem}{Theorem}[section]
\newtheorem{lemma}[theorem]{Lemma}
\newtheorem{proposition}[theorem]{Proposition}
\newtheorem{Corollary}[theorem]{Corollary}

\theoremstyle{remark}
\newtheorem{definition}{Definition}[section]

\newtheorem{ass}{Assumption}[section]
\newtheorem{rem}{Remark}[section]
\newtheorem{cond}{Condition}

\def\Y{{\boldsymbol Y}}

\def\h{{\boldsymbol h}}
\def\k{{\boldsymbol k}}
\def\K{{\boldsymbol K}}

\def\rbf{{\boldsymbol r}}
\def\m{{\boldsymbol m}}
\def\p{{\boldsymbol p}}
\def\s{{\boldsymbol s}}
\def\u{{\boldsymbol u}}
\def\x{{\boldsymbol x}}
\def\0{{\boldsymbol 0}}
\def\1{{\boldsymbol 1}}
\def\2{{\boldsymbol 2}}
\def\3{{\boldsymbol 3}}
\newcommand{\beq}{\begin{equation}}
\newcommand{\eeq}{\end{equation}}
\newcommand{\bea}{\begin{eqnarray}}
\newcommand{\eea}{\end{eqnarray}}
\newcommand{\bit}{\begin{itemize}}
\newcommand{\eit}{\end{itemize}}
\newcommand{\ben}{\begin{enumerate}} 
\newcommand{\een}{\end{enumerate}}
\newcommand{\bpm}{\begin{pmatrix}}
\newcommand{\epm}{\end{pmatrix}}
\newcommand{\bbm}{\begin{bmatrix}}
\newcommand{\ebm}{\end{bmatrix}}

\renewcommand{\l}{\left}
\renewcommand{\r}{\right}

\newcommand{\E}[0]{\textrm{E}}

\newcommand{\Cov}[0]{\textrm{Cov}}

\newcommand{\nn}{\nonumber}
\newcommand{\wh}{\widehat}

\newcommand{\e}{\boldsymbol{e}}
\newcommand{\sbf}{\boldsymbol{s}}
\newcommand{\hbf}{\boldsymbol{h}}
\newcommand{\kbf}{\boldsymbol{\kappa}}

\newcommand{\Sgmbf}{\boldsymbol{\Sigma}_{n}^{x}}
\newcommand{\pbf}{\boldsymbol{p}}
\newcommand{\mbf}{\mathbf}

\newcommand{\spn}{\overline{\text{span}}}
\newcommand{\A}{\boldsymbol{A}}
\newcommand{\Pbf}{\boldsymbol{P}}

\newcommand{\bfpsi}{\boldsymbol{\psi}}
\newcommand{\bfxi}{\boldsymbol{\xi}}
\newcommand{\Epsbf}{\boldsymbol{\epsilon}}

\newcommand{\bchi}{\boldsymbol{\chi}}

\newcommand{\bth}{\boldsymbol{\theta}}

\newcommand{\Sigmabf}{\boldsymbol{\Sigma}}

\newcommand{\thbf}{\boldsymbol{\theta}}
\newcommand{\Thbf}{\boldsymbol{\Theta}}
\newcommand{\vsbf}{\boldsymbol{\varsigma}}
\newcommand{\taubf}{\boldsymbol{\tau}}

\newcommand{\n}{^{(n)}}
\def\pr{^{\prime}}

\endlocaldefs

\begin{document}

\begin{frontmatter}
\title{General spatio-temporal factor models  for high-dimensional random fields on a lattice}
\runtitle{General spatio-temporal factor models}

\begin{aug}
\author[A]{\fnms{Matteo}~\snm{Barigozzi}\ead[label=e1]{matteo.barigozzi@unibo.it}},
\author[B]{\fnms{Davide}~\snm{La Vecchia}\ead[label=e2]{Davide.LaVecchia@unige.ch}}
\and
\author[C]{\fnms{Hang}~\snm{Liu}\ead[label=e3]{hliu01@ustc.edu.cn}}
\address[A]{Department of Economics,
University of Bologna\printead[presep={,\ }]{e1}}

\address[B]{
Geneva School of Economics and Management, University of Geneva\printead[presep={,\ }]{e2}}

\address[C]{
International Institute of Finance, School of Management, University of Science and Technology
of China\printead[presep={,\ }]{e3}}
\end{aug}

\begin{abstract}
Motivated by the need for analysing large spatio-temporal panel data, we introduce a novel dimensionality reduction methodology for $n$-dimensional random fields observed across  {a number} $S$ spatial locations and $T$ time periods. We call it General Spatio-Temporal Factor Model  (GSTFM). First, we provide the probabilistic and mathematical underpinning needed for the representation of a random field as the sum of two components: the common component (driven by a small number $q$ of latent factors) and the idiosyncratic component (mildly cross-correlated). We show that the two components are identified as  $n\to\infty$.
Second, we propose an estimator of the common component and derive its statistical guarantees (consistency and rate of convergence) as $\min(n,S,T)\to\infty$. Third, we propose an information criterion to determine the number of factors. Estimation makes use of Fourier analysis in the frequency domain and thus we fully exploit the  information on the  spatio-temporal covariance structure of the whole panel. 
Synthetic data examples illustrate the applicability of GSTFM and its advantages over the extant generalized dynamic factor model that ignores the spatial correlations. 
\end{abstract}

\begin{keyword}[class=MSC]
\kwd{62H25}
\kwd{62M15}
\kwd{62M40}
\kwd{62H20}
\kwd{60G35}
\end{keyword}

\begin{keyword}
\kwd{Factor models} 
\kwd{high-dimensional random fields}
\kwd{spectral analysis}
\kwd{functional/phsyical dependence}
\end{keyword}

\end{frontmatter}

\section{Introduction}
\subsection{Big data on spatio-temporal processes}

Many data analysis problems in economics, finance, medicine, environmental sciences, and other scientific areas need to conduct inference on random phenomena observed over time  and registered at different locations. 

Supervised and unsupervised learning methods for random fields  (henceforth, rf)  are suitable tools for the statistical analysis of this type of data: they provide an understanding of the key spatial and/or temporal dynamics of the studied phenomena.  
For instance, rf are routinely applied in medicine for fMRI data analysis (see e.g. \citealp{L08}, Ch.6), in geostatistics for satellite images analysis (see e.g. \citealp{C15, CW15}), in natural sciences  for modeling complex phenomena (see e.g. \citealp{V10,Cr17} for applications in physics and engineering), in economics for the analysis of spatial panel data (see e.g. \citealp{BaltagiBook}) just to mention few  book-length introductions.  

In this paper we consider  datasets containing records on spatio-temporal  rf over a lattice; see e.g. \cite{C15}. We let $(s_1 \ s_2)\in\mathbb Z\times \mathbb Z=\mathbb Z^2$  denote the spatial position in and $t\in\mathbb Z$ represent the time index---in principle, the dimension of the spatial lattice can be larger than two. For instance, the points $(s_1 \ s_2)$ can be: in geostatistics, geographical regions represented as a network with a given adjacency matrix; in image analysis, the position of pixels in an image.  At each $(s_1 \ s_2)$ and time $t$, the object of interest is the $n$-dimensional rf: $\bm x_n=\{\boldsymbol{x}_{n\vsbf}=(x_{1\vsbf}   \cdots x_{\ell\vsbf} \cdots  x_{n \vsbf})^\top, \vsbf=( s_1\ s_2\ t)^\top\in\mathbb Z^3\}$, for $n\in \mathbb{N}$.  
Typical inference goals for these types of data include e.g. constructing and analyzing generative models, quantifying spatio-temporal dependency, prediction or image restoration.

One key aspect related to the statistical analysis of this data is that 
$n$ is of the order of several hundreds and the number of locations and time points may have the same magnitude. A common approach to analyze the resulting large spatio-temporal rf datasets is to resort on standard time series methods, like e.g. spatio-temporal autoregressive models (see e.g. \citealp{CW15}). Nevertheless, because of the high-dimensionality, standard parametric approaches are not feasible (e.g.  in a vector autoregressive model with one time lag for the time series available at each location $\sbf$, the number of parameters is $n^2$) and  dimensionality reduction techniques are needed.


 To solve the curse of dimensionality, one may look at the literature on high-dimensional time series and think of relying on factor models, which allow for a low-dimensional description of high-dimensional data and a limited number of factors capture the common behaviour of the studied phenomena \citep[see, e.g.,][among many others]{FHLR00,stock2002forecasting,bai2002determining,lam2012factor,fan2013large}.  
 
 Among the existing approaches to factor analysis, the General Dynamic Factor Model (GDFM)  of \citet{FHLR00}  defines the most general, nonparametric, factor model which is based on a decomposition of the observations  into the sum of two mutually orthogonal (at all leads and all lags) components: the common component (driven by a small number $q$ of factors) and the idiosyncratic component
(mildly cross-correlated). This decomposition looks attractive since it is able to capture not only contemporaneous correlations but all leading and lagging  co-movements in time among the $n$ components of the the time series.

In the case of a rf the set of correlations among its $n$ components is much richer. Indeed, an observation at time $t$ and spatial location $(s_1\
s_2)$ might depend on observations at time $t{\pr}$ in the same location, or on observations at the same time but at spatial location $(s_1\pr\ s_2\pr)$, but also on observations at time $t\pr$ and spatial location $(s_1\pr\ s_2\pr)$. Thus, factor models for spatio-temporal rf have to account for this richer correlation structure.

\subsection{Our contributions: the paper in brief}
We introduce the General Spatio-Temporal Factor Model  (GSTFM), a new a class of factor models which allows us to reduce the dimensionality of a high-dimensional spatio-temporal datasets by capturing all relevant correlations, across both time and space.
Our results contribute to different streams of literature on rf theory and inference on high-dimensional data.  

(i) We derive  the decomposition of a  spatio-temporal  rf into a common component, which depends on $q$ unobservable factors,  and an idiosyncratic component, see Theorem \ref{Th. q-DFS}. 
To obtain this result, we need to tackle an important challenge rooted into probability theory:   because of the lack of ordering in $\mathbb{Z}^2$, the GDFM results already available in the literature on high-dimensional time series cannot be applied in our setting. Indeed, the extant results are available for discrete time (regularly spaced) time series indexed by $t\in \mathbb{Z}$ and rely on a generalization of the Wold representation to the case of infinite dimensional stationary processes as derived by \citet{FL01} and \citet{HL13}. Similar concepts are not readily available for a rf indexed in $\mathbb{Z}^3$. As a possible solution, we might specify a notion of spatial past, selecting e.g. the half-plane or the quarter-plane formulation. However, this choice entails the drawback that different versions of the Wold decomposition (see \citealp{MR17}) are available, 
with no obvious indication on which one has to be preferred in our context.  To avoid this issue, we resort on the Fourier analysis in the frequency domain. This methodological approach requires a careful extension to rf of the time series notions of canonical isomorphism, dynamic averaging sequences, aggregation space,  dynamic eigenvalues and eigenvectors, spatio-temporal linear  filters, idiosyncratic variables, and many others. Our theoretical developments would not be justified without these preliminary results. 
Clearly, our results nest as a special case the GDFM results.

(ii) The mentioned decomposition is at the population level: to apply our methodology we need an estimation procedure of the common component. To this end, we derive a complete and operational estimation theory, which contributes to the literature on the statistical analysis of rf. More in detail, we  build on \cite{DPW17} and, making use of a suitable notion of functional dependence measure for spatio-temporal rf, we derive a consistent estimator of a high-dimensional spectral density matrix.  
 We provide its statistical guarantees, proving consistency (with rate) of the proposed estimator. These general results (which are of their own theoretical interest, see Appendix \ref{Spec.Y}) substantially extend the applicability of spectral analysis to non-linear, non-Gaussian, or non-strong mixing rf. 
 Thanks to these novel results,  we derive the rate of converge of our estimator of the common component of spatio-temporal rf.
 The asymptotic regime that we consider  is very flexible: it simply requires that the number of locations and the time diverge, but does not need a specification of the type of asymptotics (in-fill or a long-span); see Theorems \ref{Thmsigmax} and \ref{Thm.hatChin} for the mathematical detail. 

(iii) The above theoretical developments hinge on a central aspect: the selection of the number of latent factors. We take care of this aspect and state a simple and operational criterion, providing its theoretical underpinning in Theorem \ref{Prop.qselect_Sample}. 

(iv) We consider the computational aspects needed to implement our methodology by studying synthetic data examples (see the supplementary material 
for additional numerical exercices), in which the underlying data generating process involves different types of convolutions over the lattice, which in turn imply different levels of spatio-temporal aggregation. 

\subsection{Related work}
In the literature on panel data and time series, dimensionality reduction is often achieved by factor models, which allow for a low-dimensional description of high-dimensional data. Modern factor models essentially originate in four pioneering contributions: \cite{G77}, \cite{SS77}, \cite{C83}, and \cite{CR83}. The reference factor model for this work is the GDFM introduced by \citet{FHLR00} and \citet{FL01}, where few latent factors capture all leading and lagging main comovements among the observed variables. The GDFM was then studied by \citet{hallin2011dynamic} in presence of a block structure in the data (where blocks can be seen are spatial locations), and further developed in a predictive context by \citet{forni2005generalized,FHLZ15,FHLZ17}. A criterion for the number of factors is proposed by \citet{HallinLiska2007}. 
The GDFM has been successfully applied to many macroeconomic and financial time series problems; see, e.g., \citet{altissimo2010new,cristadoro2005core,proietti2021nowcasting,hallin2021forecasting,trucios2022forecasting}.

There are many other influential papers on factor models as, e.g., \citet{stock2002forecasting,bai2002determining,lam2012factor,fan2013large}.




Spatial factor models and related techniques for the analysis of large spatial datasets are also available in the statistical literature. For instance, \cite{Sciamenna2002} introduce a generalized shifted-factor model for purely spatial data; \cite{WW03} study correlations which are caused by a common spatially correlated underlying factors;  \cite{Hetal2018} consider many methods for analyzing large spatial data; { \cite{park2009time,BDLV22} propose a semiparametric  (robust) factor model which is connected to the GDFM and achieves dimensionality reduction of  spatio-temporal data.}

Last, a spatio-temporal dataset can in principle also be modeled as a tensor time series, with some of its modes corresponding to the spatial dimensions. Thus, a spectral approach to the analysis of tensor data represents a possible alternative. Factor models for tensor time series data have recently been studied by many authors, see, e.g., \citet{chen2022factor,chang2023modelling}. However, none of these approaches is dynamic in the sense that it allows for the factors, which might be tensors too, to be loaded by the data with lags. 

%

\subsection{Outline}  
The paper has the following structure. In Section \ref{Sec: MotEx} we provide a motivating example for the necessity of introducing a new class of dynamic spatio-temporal factor models. In Section \ref{Sec: basic notation} we review main concepts for the spectral analysis of rf. In Sections \ref{Sec: repr}-\ref{sec:DPCA} we derive the representation theorem for the GSTFM and we define the spatio-temporal dynamic principal components. In Section \ref{Sec: estim} we present our estimator and its asymptotic properties. In Section \ref{Sec.Selectq} we introduce a criterion to estimate the number of factors. In Section \ref{Sec: sim} we show numerical results on simulated data. In Section \ref{Sec:conc} we conclude. 

In the supplementary material: we prove all theoretical results (Appendices \ref{App1_Def}, \ref{dimostraloseriesci}, \ref{appC}, \ref{APPE}, and \ref{App.PFqselect}), we prove new results for the estimation of a large spectral density of a spatio-temporal rf (Appendix \ref{Spec.Y}), we give an algorithm to estimate the number of factors (Appendix \ref{sec:HLABC}), we apply our methodology to fMRI data (Appendix \ref{Sec: real}), and we provide further numerical results (Appendix \ref{App.sim}).

\subsection{Notation}  

Given a complex matrix $\mbf{D}$, we denote by $\mathbf{{D}}^\dag$ the complex conjugate of the transposed of $\mathbf{D}$, by $\mbf{D}^\top$ its transposed, by
$\bar{\mathbf D}$ its complex conjugate, 
 and for a real matrix $\mbf D$ we have $\bar {\mbf D}={\mbf D}$ and $\mathbf{{D}}^\dag=\mathbf{{D}}^\top$. 
A similar notation holds for complex and real vectors. Given a complex scalar $z$ its complex conjugate is denoted as $z^\dag$.
%
Given two complex row vectors $\bm w=(w_1\cdots w_m)$ and $\bm v=(v_1\cdots v_m)$ we let $\langle\bm w,\bm v\rangle = \bm w{\bm v}^\dag=\sum_{i=1}^m w_i  v_i^\dag$ and $\Vert \bm w\Vert = \sqrt{\langle\bm w,\bm w\rangle}$ is the $L_2$ or Euclidean norm. Real or integer vectors are always column vectors and given two such vectors $\bm w=(w_1\cdots w_m)^\top$ and $\bm v=(v_1\cdots v_m)^\top$ we let $\langle\bm w,\bm v\rangle = \bm w^\top{\bm v}=\sum_{i=1}^m w_i v_i$ and $\Vert \bm w\Vert = \sqrt{\langle\bm w,\bm w\rangle}$. For a complex scalar we have $|z|=\sqrt {zz^\dag}$.
We use the notation $\mathcal L$ for the Lebesgue measure either on $\mathbb R^d$ or on $\mathbb C^d$ or on  $\Thbf=[-\pi,\pi) \times [-\pi,\pi) \times [-\pi,\pi)$.
When no ambiguity can arise, we use the shortcuts
$
\sum_{\hbf} = \sum_{h_1\in\mathbb Z}\sum_{h_2\in\mathbb Z}\sum_{h_3\in\mathbb Z}
$
and
$
\int_{\bm\Theta} \mathrm d\bm\theta = \int_{-\pi}^\pi \mathrm d\theta_1 \int_{-\pi}^\pi \mathrm d\theta_2 \int_{-\pi}^\pi \mathrm d\theta_3$.

\section{Motivating example} \label{Sec: MotEx}


To motivate our investigation, we illustrate via numerical examples, the inadequacy of the classical GDFM by \citet{FHLR00} in the spatio-temporal setting. 
Assume we are given realizations of $n$ random variables in $S_1\times S_2$ spatial locations ({therefore, the total number of locations is $S=S_1 + S_2$}) and $T$ time points. We organize the data into
an $n$-dimensional rf  
$$
\bm x_n=\{{x}_{\ell\vsbf},\ \ell=1,\ldots,n,\ \vsbf=( s_1\ s_2\ t)^\top,\ \s_1=1,\ldots,S_1,\ s_2=1,\ldots,S_2,\ t=1,\ldots,T\}.
$$
Under our GSTFM the $\ell$-th component  of $\bm x_n$ 
is such that $x_{\ell \vsbf}= \chi_{\ell\vsbf} + \xi_{\ell\vsbf}$, for $ \ell = 1,\ldots, n$. The term $\chi_{\ell\vsbf}$ is called common component and it is a linear combination of $q$ latent rf, with $q\ll n$, located at the same spatial location and time period as well as at neighbouring spatial points and at various lags. The term $\xi_{\ell\vsbf}$ is called idiosyncratic component and is assumed to be weakly cross-sectionally correlated.
In Theorem \ref{Th. q-DFS} we show that 
under the considered setting the presence of an eigen-gap in the eigenvalues of the spatio-temporal spectral density matrix (see Section \ref{Subsect: spectral} for a definition) is a key distinctive feature. In particular,  the $q$ largest eigeinvalues of the spatio-temporal spectral density matrix diverge as $n\to\infty$ while the remaining $n-q$ stay bounded if and only if the common component $\chi_{\ell\vsbf}$ is driven by $q$ spatio-temporal factors. As $n\to\infty$, we can then disentangle the common and idiosyncratic components and, consequently, we can identify the GSTFM.  This is the main feature of general factor models, sometimes called blessing of dimensionality, as opposed to the curse of dimensionality typically affecting large dimensional models. 

To verify this phenomenon we simulate the common component of the GSTFM with $q=2,3$ common factors, loaded according to a quite general and commonly encountered   configuration (see Model (a) in (\ref{modelAR}) for details) 
 of the spatio-temporal dependencies. For ease of simulation, the idiosyncratic component is generated from the standard normal distribution. 
Then, for different subsets of dimension $m=1,\ldots, n$ we estimate the $m\times m$ spatio-temporal spectral density matrix of $\bm x_m$ and we compute its $q+1$ largest eigenvalues, averaged across all frequencies. 
In Figure \ref{Sim_Eigen_Modela_GSTFR}, we display these eigenvalues as a function of the cross-sectional dimension $m$: we clearly see that the eigen-gap becomes more and more pronounced as  $m$ increases, a manifestation of the blessing of dimensionality. 

 \begin{figure}[h]
 \begin{center}
 \begin{tabular}{cc}
 \includegraphics[width=0.45\textwidth, height=0.35\textwidth]{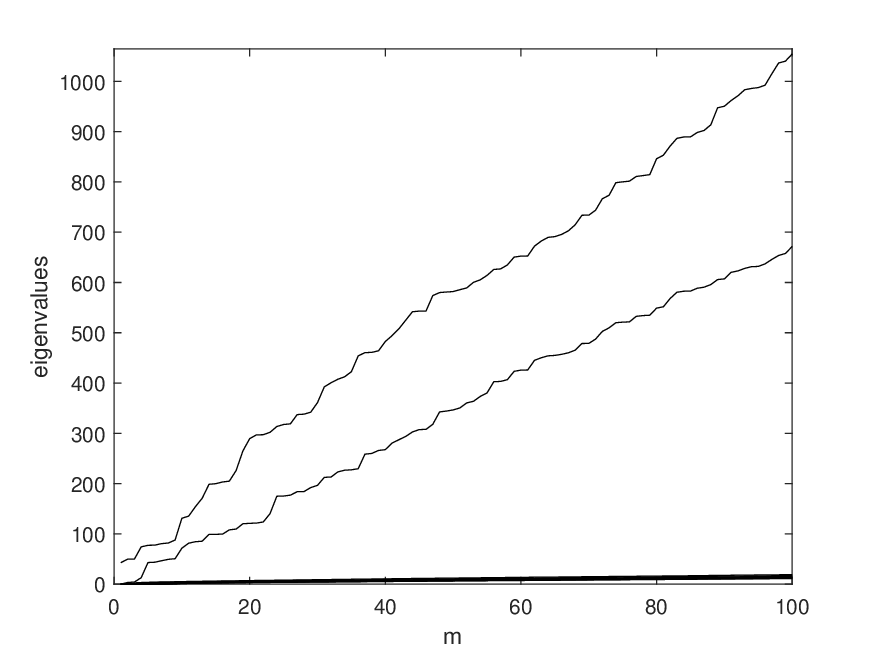}
 &
\includegraphics[width=0.45\textwidth, height=0.35\textwidth]{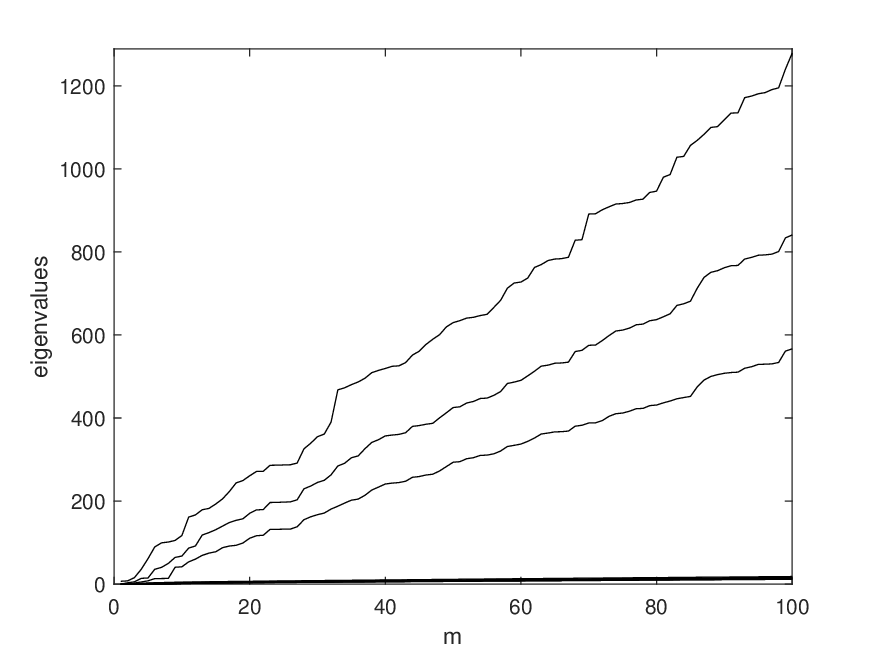}
 \end{tabular}
 \caption{Spatio-temporal dynamic eigenvalues (averaged over all frequencies) for Model (a) in \eqref{modelAR}, with $n= 100$, $(S_1,S_2, T) =(10,10, 100)$, and $q = 2$ (left) or  $q = 3$ (right). The x-axis represents the size of the rf $m=1,\ldots,n$. } \label{Sim_Eigen_Modela_GSTFR} 
 \end{center} 
 \end{figure}

If instead 
we decide to resort on the extant GDFM, then the natural thing to do is to stack at each time point $t$ the data into an $nS_1S_2$-dimensional time series (the order in which the locations and variables are stacked is irrelevant for this discussion):
$$
\mbf x_N=\{x_{it},\ i=1,\ldots, N, \ N=(nS_1S_2),\ t=1,\ldots, T\}.
$$ 
Notice that the rf $\bm x_n$ and the time series $\mbf x_N$ contain the same data points but encoded in different ways. 
Under the GDFM the $i$-th component  of $\mbf x_n$ 
is also such that $x_{it}= \chi_{it}^{\text {\tiny GDFM}} + \xi_{it}^{\text {\tiny GDFM}}$, for $i = 1,\ldots,N$,
where  now $\chi^{\text {\tiny GDFM}}_{it}$ is a linear combination of $r$ latent time series, with $r\ll N$, at the same time period as well as at various lags. Given that the stacking procedure yields a very large dataset,  the asymptotic results in \cite{FHLR00} should apply: the eigenvalues of the estimated spectral density of $\mbf x_n$ should display an eigen-gap, between the $r$-th and $(r+1)$-th eigenvalues, increasing as $N\to\infty$. In fact, if there is no spatial correlation in the data, then we would expect to have $r=q$, as the only correlations left would be cross-sectional and temporal and the GDFM is designed precisely to capture them. But, if there are spatial correlations, then the proposed stacking approach is likely to be flawed:  ignoring spatial correlations 
might generate spurious factors. 
So if the data follows a GSTFM with $q$ factors, but instead we fit a GDFM, at best we might find a number of factors $r>q$.
Indeed, if there are ignored spatial correlations, the GDFM might not even be correctly identified. To show this, we consider again the data simulated from the GSTFM and that yield Figure \ref{Sim_Eigen_Modela_GSTFR}.  
We estimate the spectral density of the stacked vector $\mbf x_m$ for $m=1,\ldots, N$ and, in  Figure 
\ref{Sim_Eigen_Modela_GDFM}, we display, as functions of $m$, the ten largest corresponding eigenvalues averaged over all frequencies. Since now $N\gg n$, we might  expect an eigen-gap even more evident than the one clearly visible in Figure \ref{Sim_Eigen_Modela_GSTFR}. In contrast, in Figure 
\ref{Sim_Eigen_Modela_GDFM} no eigen-gap is detectable at all, 
even for very large cross-sectional dimensions: this means that the  true number $r$ of factors  cannot be recovered and the GDFM is not identifiable in this setting {(for further details on identification of factor models, see Corollary \ref{Mythm3} and the related discussion).}

\begin{figure}[h]
 \begin{center}
 \begin{tabular}{cc}
 \includegraphics[width=0.45\textwidth, height=0.35\textwidth]{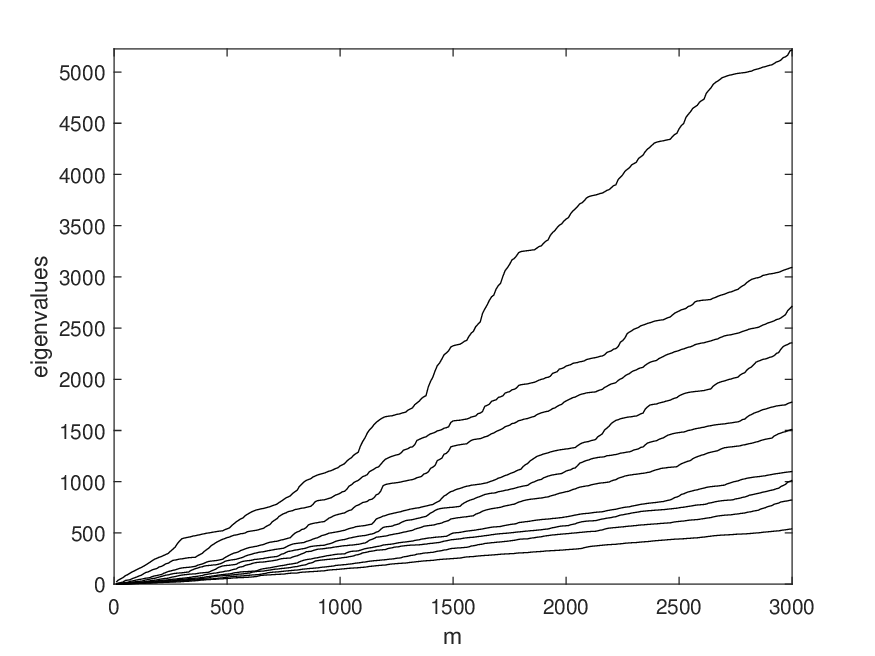}
 &
\includegraphics[width=0.45\textwidth, height=0.35\textwidth]{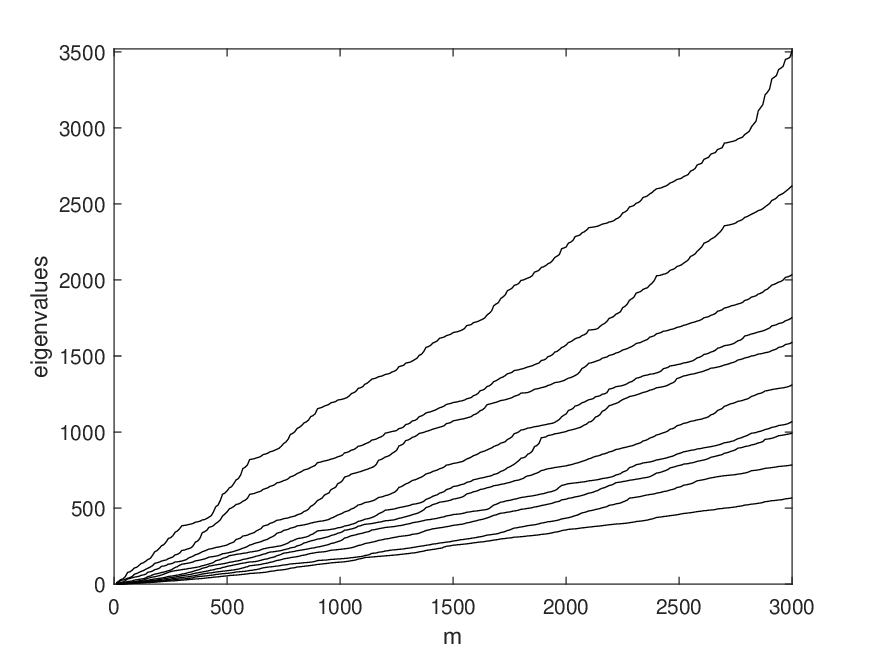}
 \end{tabular}
 \caption{Temporal dynamic eigenvalues (averaged over all frequencies) for Model (a) in \eqref{modelAR}, with $ n= 30$, $(S_1,S_2, T) =(10,10, 100)$, and $q = 2$ (left) or  $q = 3$ (right). The x-axis represents the size of the stacked vector $m=1,\ldots,(nS_1S_2)$.} \label{Sim_Eigen_Modela_GDFM} 
 \end{center} 
 \end{figure}

The above arguments illustrate a two-fold statistical problem related to the development of a novel theory of general factor models for spatio-temporal rf. On the one hand, there is the need for a representation theory for large dimensional rf, which  allows us to capture the common spatio-temporal factors that explain the spatio-temporal co-movements of the process. On the other hand, there is the central need for estimation methods which have statistical guarantees and yield an estimate of the common spatio-temporal component converging to true spatio-temporal common component. In the next sections, we illustrate how to solve this statistical problem. 

\section{Basic theory of linear random fields}\label{Sec: basic notation}

\subsection{Random fields}
Our object of interest is the infinite dimensional random field (hereafter rf)  on a lattice
$\bm x=\{\boldsymbol{x}_{(\sbf\ t)}=(x_{1,(\sbf \ t)} \ x_{2,(\sbf \ t)} \cdots x_{\ell,(\sbf \ t)}\cdots)^\top, \sbf\in\mathbb Z^2, t\in \mathbb Z\}$.
We index space-time points as $\vsbf=(\sbf \ t)^\top=(s_1\ s_2\ t)^\top\in\mathbb Z^3$, with $s_1,s_2$, and $t$ allowed to vary independently. So, for any $\vsbf\in\mathbb Z^3$, we define the infinite dimensional  random vector $\boldsymbol{x}_{\vsbf}=(x_{1\vsbf} \ x_{2,\vsbf} \cdots x_{\ell\vsbf}\cdots)^\top$ and for any $n\in\mathbb N$,we also define the $n$-dimensional column random vector $\boldsymbol{x}_{n\vsbf}=(x_{1\vsbf} \cdots x_{n \vsbf})^\top$, which is an element of the $n$-dimensional rf $\boldsymbol{x}_n=\{\boldsymbol x_{n\vsbf}, \vsbf \in \mathbb{Z}^3\}$. Clearly, $\boldsymbol{x}_n$ is a sub-process of $\boldsymbol{x}$.

Throughout, we let $\mathcal{P}=(\Omega,\mathcal{F},P)$ be a probability space and let $L_2(\mathcal{P},\mathbb{C})$ be the linear space
of all complex-valued, square-integrable random variables defined on $\Omega$. Then, if $x_{\ell\vsbf}\in L_2(\mathcal P,\mathbb C)$, for any $\ell\in\mathbb N$,  the process $\{x_{\ell\vsbf}, \vsbf \in\mathbb Z^3\}$ is a complex valued scalar random field with finite variance, and for any fixed $n$ the process $\{\bm x_{n\vsbf}, \vsbf \in\mathbb Z^3\}$ is a complex valued vector rf with all its elements having finite variance, and the process $\boldsymbol{x}=\{\boldsymbol x_{\vsbf}, \vsbf \in \mathbb{Z}^3\}$ is an infinite dimensional complex valued rf with all its elements having finite variance. Notice that the space $L_2(\mathcal{P},\mathbb{C})$ 
is a complex Hilbert space, thus it possesses the usual inner product given by $\Cov(x_{i\vsbf},x_{j\vsbf'})=\E[(x_{i\vsbf}-\E(x_{i\vsbf}))( x_{j\vsbf'}-\E( x_{j\vsbf'}))^\dag]$, where $\E(x_{i\vsbf})=\int_{\mathbb C} u \mathrm d P(u)$ represents the expected value taken w.r.t. the probability $P$. 

Last, we define $\bm{\mathcal X}_{n} = \spn(\boldsymbol{x}_n)$ as the minimum closed linear subspace of $L_2(\mathcal{P},\mathbb{C})$, containing $\boldsymbol{x}_n$, i.e., the set of all $L_2$-convergent linear combinations of $x_{\ell\vsbf}$'s. Therefore, a generic element of $\bm{\mathcal X}_n$ is of the form 
$
\zeta_{n \vsbf} =\sum_{\ell=1}^{n}\sum_{\kappa_1\in\mathbb Z}\sum_{\kappa_2\in\mathbb Z}\sum_{\kappa_3\in\mathbb Z} \alpha_{\ell\kbf}\ x_{\ell, \vsbf+\kbf}$ with $\alpha_{\ell\kbf}\in\mathbb C$ and $\kbf=(\kappa_1 \ \kappa_2 \ \kappa_3)^\top\in \mathbb Z^3$. 
 Moreover, define $\bm{\mathcal X}= \spn(\boldsymbol{x})$, which is such that
$\bm{\mathcal X}=\overline{\cup_{n=1}^\infty \bm{\mathcal X}_{n}}$ and it contains  also the limits, as $n\to\infty$, of all $L_2$-convergent sequences thereof. 
Hence, both $\bm{\mathcal X}_{n} $ and $\bm{\mathcal X}$ are Hilbert spaces.

\subsection{Spatio-temporal autocovariance and spectral density matrices}\label{Subsect: spectral}
A spatio-temporal shift between pairs of points $\vsbf=(\sbf \ t)^\top \in \mathbb{Z}^3$ and $\vsbf'=(\boldsymbol{s}' \ t')^\top\in \mathbb{Z}^3$ is defined as a vector 
$\hbf=(\vsbf-\vsbf') \in\mathbb{Z}^3$ such that
$\hbf = ( ({s_1}-{s'_1}) \ ({s_2}-{s'_2}) \  (t-t') )^\top=(h_1\ h_2 \ h_3)^\top$. 
We need to set an origin $\boldsymbol 0 =(0 \ 0 \ 0)^\top$, whose location is arbitrary, but once it has been chosen it remains fix. To make our theory
{insensitive to} origin shifts, we impose space-time homostationarity, which  implies that the first two moments (mean and covariance) of a space-time rf are invariant respect to space-time translation,  i.e. a homostationary spatio-temporal random field features homogeneity in space and stationarity in time.

To formalize this property, for any $n\in\mathbb N$ and any $\vsbf,\vsbf'\in\mathbb Z^3$, we define the $n\times n$ autocovariance matrix: 
$\Cov(\boldsymbol x_{n,\vsbf}, \boldsymbol x_{n,\vsbf'} )= \E[(\boldsymbol x_{n,\vsbf}-\E[\boldsymbol x_{n,\vsbf}])(\boldsymbol x_{n,\vsbf'}-\E[\boldsymbol x_{n,\vsbf'}])^\dag].$
Notice that the covariance matrix $\Cov(\boldsymbol x_{n,\vsbf}, \boldsymbol x_{n,\vsbf} ) $ is non-negative definite; see, e.g., \citet[p.15]{S12}.  Then, we introduce the following assumption of homostationarity.
{\begin{ass} \label{Ass.A} 
For any $n \in \mathbb{N}$ the random field $\boldsymbol{x}_n=\{\bm x_{n\vsbf} = (x_{1\vsbf}\cdots x_{n\vsbf} )^\top, \vsbf \in\mathbb Z^3\}$ is such that  $x_{\ell\vsbf}\in L_2(\mathcal P,\mathbb C)$ for any $\ell\le n$, and, for any $\vsbf \in \mathbb{Z}^3$: 
\begin{inparaenum}
\item [(i)] $\text{\upshape E}( x_{\ell\vsbf})=  0$ and  $\text{\upshape Var}( x_{\ell\vsbf})>0$; 
\item [(ii)] $\text{\upshape Cov}(\bm x_{n\vsbf},\bm x_{n\vsbf+\hbf})= \text{\upshape E}(\bm x_{n\vsbf}\bm x^\dag_{n \vsbf+\hbf})
=\bm\Gamma_n^x(\hbf)$ for any $\hbf \in \mathbb{Z}^3$.
\end{inparaenum}
\end{ass} }

A few comments on Assumption \ref{Ass.A}. First, the zero-mean assumption can be made without any loss of generality.
Second, since, all elements of the rf $\bm x_n$ are in $L_2(\mathcal P,\mathbb C)$ then for any fixed $n$ the covariance matrix $\bm\Gamma_n^x(\bm 0)$ is finite and so all autocovariances $\bm\Gamma_n^x(\hbf)$, $\hbf \ne \bm 0$ are finite too. 
Third, note that letting 
$\bm\Gamma_n^x(\hbf)\equiv \bm\Gamma_n^{x}(h_1,h_2,h_3)$, then the following relations holds
\begin{align}
&\bm\Gamma_n^x(-h_1,-h_2,-h_3)=\bm\Gamma_n^{x\dag}(h_1,h_2,h_3),\qquad \bm\Gamma_n^x(\mp h_1,\pm h_2, \pm h_3)=\bm\Gamma_n^{x\dag}(\pm h_1,\mp h_2, \mp h_3),\nonumber\\
&\bm\Gamma_n^x(\pm h_1,\mp h_2, \pm h_3)=\bm\Gamma_n^{x\dag}(\mp h_1,\pm h_2,\mp h_3),\quad \bm\Gamma_n^x(\pm h_1,\pm h_2,\mp h_3)=\bm\Gamma_n^{x\dag}(\mp h_1,\mp h_2,\pm h_3),\nonumber
\end{align}
which are much weaker requirements than assuming space isotropy---for which we would have
$\text{\upshape Cov}(\bm x_{n\vsbf},\bm x_{n\vsbf+\hbf})= \bm\Gamma_n^x(\Vert (\pm h_1\ \pm h_2)^\top\Vert, h_3)$, thus imposing that the second-order moments are invariant under all rigid axes motions ({\citealp{S12}, p.17}). 
Then, we introduce (see \citealp[p. 45]{MR17})

\begin{definition} [Orthonormal white noise rf] \label{RFWN} For a given finite integer $q$, let $\boldsymbol w=\{\boldsymbol w_{\vsbf}=({w}_{1 \vsbf}\cdots{w}_{q \vsbf})^\top, \vsbf \in \mathbb{Z}^3 \}$  be a $q$-dimensional rf such that ${w}_{\ell \vsbf}\in L_2(\mathcal P,\mathbb C)$ and $\text{\upshape E}({w}_{\ell \vsbf })=0$, for any $\ell=1,\ldots, q$ and $\vsbf\in \mathbb{Z}^3$. Then we call $\boldsymbol w$ an orthonormal white noise rf if, for any $\vsbf, \vsbf' \in \mathbb{Z}^3$, $\text{\upshape{Cov}}\left( \boldsymbol w_{\vsbf},\boldsymbol w_{\vsbf'}  \right)=\mbf{I}_q$  if $\vsbf=\vsbf'$ and it is $\mbf 0$ otherwise.
\end{definition}




To conduct spectral analysis we add the following

\begin{ass} \label{Ass.B} For any $n \in \mathbb{N}$,
the spectral measure of $\boldsymbol{x}_{n}$ is absolutely continuous (with respect to $\mathcal L$ on $\Thbf$), so $\boldsymbol{x}_{n}$ admits an $n\times n$ spectral density matrix given by 
\begin{equation}
\bm\Sigma_n^x(\thbf)
= \sum_{\hbf} \bm\Gamma_n^x(\hbf) e^{-i\langle\hbf, \thbf\rangle}, \nn
\end{equation} 
where  $i=\sqrt {-1}$ 
and $\bm\theta=(\theta_1\ \theta_2\ \theta_3)^\top\in\Thbf$.
\end{ass}


For a rf the spectral density matrix depends on a vector of frequencies $\thbf$ \citep{brillinger1970frequency}, and under Assumptions \ref{Ass.A} and \ref{Ass.B}, $\Sgmbf(\thbf)$ is Hermitian and non-negative definite (\citealp[p.13]{leonenko1999limit}) for all $\thbf \in \boldsymbol{\Theta}$ and any $n\in\mathbb N$. 
The lag-$\hbf$ autocovariance matrix is then given by
$
\bm\Gamma_n^x(\hbf)
= \frac{1}{8\pi^3}\int_{\boldsymbol{\Theta}} e^{i\langle\hbf, \thbf\rangle} \Sgmbf(\thbf)  \mathrm d\thbf \nonumber
$, for any $\hbf \in \mathbb{Z}^3 $.

Let $\Sigmabf^{x}(\thbf)$ denote the infinite matrix having the matrix $\Sgmbf(\thbf)$ as its $n \times n$ top-left sub-matrix and notice again that as $n\to\infty$ we allow for the possibility of its eigenvalues to diverge (see below), while all its entries are finite by assumption. Moreover, we notice that a $q$-dimensional orthonormal white noise rf has spectral density $\mbf I_q$. Then we have 
\begin{definition}[Spatio-temporal dynamic eigenvalues] For any $n\in\mathbb N$ and $\ell\le n$, let $\lambda^x_{n\ell}: \Thbf \to \mathbb{R}^+$ 
be defined as the function associating with $\thbf \in \Thbf$ the $\ell$-th eigenvalue in decreasing order of $\Sgmbf(\thbf)$.
 We call $\lambda^x_{n\ell}(\thbf)$ the spatio-temporal dynamic eigenvalues of $\Sgmbf(\thbf)$. 
 \end{definition} 
 
\begin{definition}[Spatio-temporal dynamic eigenvectors] \label{Def. eigenvectors}
For any $n\in\mathbb N$ and $\ell\le n$ let $\pbf^{x}_{n\ell}:\bm\Theta\to \mathbb C^n$ be such that, for any $\thbf \in \Thbf$, the row vector $\pbf_{n\ell}^x(\thbf)$ satisfies
\begin{inparaenum}
\item[(i)]  $\Vert \pbf_{n\ell}^x(\thbf) \Vert=1$;
\item[(ii)]  $\pbf_{n\ell}^x(\thbf) {\pbf}_{nj}^{x\dag}(\thbf) =0$, for $\ell\neq j$; 
\item[(iii)] $\pbf_{n\ell}^x(\thbf)\Sgmbf (\thbf) = \lambda^x_{n\ell}(\thbf)\pbf_{n\ell}^x(\thbf)$.
\end{inparaenum}
Then, the functions $\{\pbf^{x}_{n\ell}, \ell=1,\ldots, n\}$ constitute a set of spatio-temporal dynamic eigenvectors associated with the spatio-temporal eigenvalues  $\lambda^x_{n\ell}$ and the rf $\boldsymbol{x}_n$.
\end{definition}

Hereafter, we also assume strict positive definiteness of $\Sigmabf_n^x(\bm\theta)$
\begin{ass}  \label{Ass.C} 
For any $n \in \mathbb{N}$ and $\ell \leq n$, and any $\thbf \in \Thbf $, $\lambda^{x}_{n\ell}(\thbf) > 0$. 
\end{ass}

\begin{rem}
\label{rem:dyneval} \upshape{
Following \citet[Lemma 3 and 4]{FL01}, one can easily prove that: 
\begin{inparaenum} [(i)]
\item the real functions $\lambda^x_{n\ell}$ are Lebesgue-measurable and integrable in $\Thbf$, for any given $n\in\mathbb N$ and $\ell \leq n$; 
\item $\lambda^x_{n\ell}(\thbf)$ is a non-decreasing function of $n$, for any  $\thbf \in \Thbf$.
\end{inparaenum}
In particular, from (ii) it follows that, for all $\thbf\in\Thbf$,
$\lim_{n\to\infty}\lambda^x_{n\ell}(\thbf)=\sup_{n\in\mathbb N}\lambda^x_{n\ell}(\thbf),
$
and it is well defined for any $\ell\le n$.
}
\end{rem}


\subsection{Spatio-temporal linear filters}

First, for any $n\in\mathbb N$ and $\ell\le n$, let us consider the three linear operators, $L_j:\bm{\mathcal X}_n \to \bm{\mathcal X}_n$, $j=1,2,3$, such that, for any $\vsbf \in\mathbb Z^3$,
\begin{align}
L_1 \, x_{\ell \vsbf } = x_{\ell({s_1}-1 \ {s_2} \ t)},  \quad L_2 \, x_{\ell \vsbf } = x_{\ell ({s_1} \ {s_2}-1 \ t) },  \quad
L_3 \, x_{\ell \vsbf } = x_{\ell ({s_1} \ {s_2} \ t-1) },  \label{Eq: operators} 
\end{align}
so that when $L_j$ is applied to the vector $\bm x_{n\vsbf}$ it shift all its $n$ components along the space or time dimension.  $L_1$
and $L_2$  act on the (spatial) dimensions of the lattice (see \citet{W54}), 
while $L_3$ is the usual time lag operator. 
We also set $L\equiv L_1  L_2  L_3$ and $L^{\kbf}\equiv L_1^{\kappa_1} L_2^{\kappa_2} L_3^{\kappa_3}$. The operators are commutative, e.g. $L_1  L_2  L_3\, x_{\ell \vsbf  }= L_3  L_1  L_2\, x_{\ell \vsbf }$.
In Lemma \ref{Lemma1} in Appendix \ref{App1_Def}, we show that $L_j$ are unitary operators that can be extended to $\bm{\mathcal X}$. 

Second, consider a generic $n$-dimensional row vector of functions $\boldsymbol{f}_n=(f_1\   \cdots f_n)$ with $f_\ell:\Thbf\to\mathbb C$ being measurable for any $\ell\le n$ and such that the following conditions hold: \begin{inparaenum}
\item [(i)]  $ \Vert\boldsymbol{f}_n\Vert^2_{\Sigmabf^{x}_n} = 
\frac 1{8\pi^3}\int_{\Thbf} \boldsymbol f_n(\thbf)\Sigmabf^{x}_n (\thbf) {\boldsymbol f}_n^\dag(\thbf) \mathrm d\thbf <\infty$, and 
\item [(ii)] 
$ \Vert\boldsymbol{f}_n\Vert^2 =  \Vert\boldsymbol{f}_n\Vert^2_{\mbf I_n}
<\infty$, respectively. 
\end{inparaenum}
The space of such functions is a complex Hilbert space denoted as $L_2^{n}(\bm\Theta,\mathbb C,\bm\Sigma_n^x, \mbf I_n)$, obtained from the intersection of two Hilbert spaces each endowed with  inner products derived from one of the two norms defined above.

Third, consider the map $\mathcal{J}:L_2^{n}(\bm\Theta,\mathbb C,\bm\Sigma_n^x, \mbf I_n)\to \bm{\mathcal X}_n$, such that, for any $n\in\mathbb N$ and $\ell\le n$,  
\begin{equation}
\mathcal{J} \left[  (\delta_{\ell1}  \cdots \delta_{\ell k} \cdots \delta_{\ell n} ) e^{i \langle \vsbf , \cdot\rangle}  \right] = x_{\ell \vsbf}, \quad \text{for any} \ \vsbf\in\mathbb Z^3,
\label{JJ}
\end{equation}
where $\delta_{\ell k}=1$ if $k=\ell$ and $\delta_{\ell k}=0$ if $k\ne \ell$ and $e^{i \langle \vsbf , \cdot\rangle}$ indicates the map from $\bm\Theta$ to $\mathbb C$ such that $\bm\theta\mapsto e^{i \langle \vsbf , \thbf \rangle}$. Thus, we have
\beq\label{filter1}
L\boldsymbol x_{n\vsbf} = 
\mathcal J\l[\bm\iota_n e^{i\langle(s_1-1\ s_2-1 \ t-1)^\top,\cdot\rangle} \r]
=\mathcal J\l[\bm\iota_n e^{-i\langle(1\ 1\ 1)^\top,\cdot\rangle}e^{i\langle\vsbf,\cdot\rangle} \r],
\eeq
where $\bm\iota_n$ is an $n$-dimensional vector of ones. In Lemma \ref{Lemma2} in Appendix \ref{App1_Def} we prove that $\mathcal{J}$ is an isomorphism, also called {canonical isomorphism}, and it can be extended to the Hilbert space of infinite dimensional functions $\bm f$, with norms $\Vert \bm f\Vert_{\Sigmabf^x}^2=\lim_{n\to\infty}\Vert \bm f_n\Vert_{\Sigmabf^x_n}^2<\infty$ and $\Vert \bm f\Vert^2=\lim_{n\to\infty}\Vert \bm f_n\Vert^2<\infty$. Notice that $\mathcal J$ is an isomorphism between the measure spaces $(\bm\Theta, \mathcal{B}(\bm\Theta), \mathcal L)$ and ($\bm{\mathcal X}_n,\mathcal{B}(\mathbb C^{n}), \mathcal L)$, where $\mathcal{B}(\bm\Theta)$ and $\mathcal{B}(\mathbb C^{n})$ are the Borel $\sigma$-fields on $\bm\Theta$ and $\mathbb C^{n}$, respectively. The definition of $\mathcal J$ extends to our setting the classical isomorphism 
typically applied  in  time series analysis (see e.g. \citealp{BD06}, Section 4.8).

%


Now, for any $\bm f_n\in L_2^{n}(\bm\Theta,\mathbb C,\bm\Sigma_n^x, \mbf I_n)$ consider the Fourier expansion 
\begin{align} 
\boldsymbol{{f}}_n(\thbf) &=  \sum_{\kbf} \mathrm{\bf  f}_{n\kbf} e^{-i\langle\kbf,\thbf\rangle}, \label{Eq. FS}\\
\mathrm{\bf  f}_{n\kbf}  &= \frac{1}{8\pi^3} \int_{\Thbf} e^{i\langle\kbf,\thbf\rangle}\boldsymbol f_n(\thbf) \mathrm d\thbf. \label{Eq. FC}
\end{align} 
where the equality in \eqref{Eq. FS} holds in the $L_2$-norm, and $\{\mathrm{\bf  f}_{n\kbf},\kbf \in \mathbb Z^3\}$ 
are the Fourier coefficients.

 Since also $(\mathrm {\mbf {f}}_{n\kbf} e^{-i\langle\kbf,\cdot\rangle})\in L_2^{n}(\bm\Theta,\mathbb C,\bm\Sigma_n^x, \mbf I_n)$, then we  apply to it
 the canonical isomorphism $\mathcal J$ to map it into elements $\bm{\mathcal X}_n$. This defines the filtered processes associated to $\boldsymbol{{f}}_n$ 
\begin{equation}
\boldsymbol{\underline{f}}_n(L)\boldsymbol{x}_{n\vsbf} = \mathcal J\l[ \boldsymbol{{f}}_n e^{i\langle \vsbf,\cdot\rangle}\r]. \label{fbar}
\end{equation}
Therefore, from \eqref{Eq. FS}, \eqref{Eq. FC}, and \eqref{fbar}, and by linearity of the canonical isomorphism, we have
\begin{equation}
\boldsymbol{\underline{f}}_n(L)\boldsymbol{x}_{n\vsbf} =  \sum_{\kbf} \mathrm{\bf  f}_{n\kbf} L^{\kbf} \boldsymbol{x}_{n\vsbf}=\left\{ \sum_{\kbf} \left[\frac{1}{8\pi^3} \int_{\Thbf} e^{i\langle\kbf,\thbf\rangle}\boldsymbol f_n(\thbf) \mathrm d\thbf\right] L^{\kbf}\right\} \boldsymbol{x}_{n\vsbf},\label{fFF}
\end{equation}
which defines an $n$-dimensional linear spatio-temporal filter. Note that $\boldsymbol{\underline{f}}_n(L)\boldsymbol{x}_{n\vsbf}\in\bm{\mathcal X}_n$ is the isomorphic map of $(\boldsymbol{{f}}_n e^{i\langle \vsbf,\cdot\rangle})\in  L^{n}_2(\Thbf,\mathbb{C},\Sigmabf^x,\mbf I_n) $, that is multiplications become convolutions via the isomorphism $\mathcal J$ and viceversa via $\mathcal J^{-1}$.  Hereafter, the composition of two linear filters, is denoted as
\beq\label{pachistrani}
\boldsymbol{\underline{g}}_n(L)\star \boldsymbol{\underline{f}}_n(L)\boldsymbol{x}_{n\vsbf} = \left\{ \sum_{\kbf} \left[\frac{1}{8\pi^3} \int_{\Thbf} e^{i\langle\kbf,\thbf\rangle}
\boldsymbol g_n(\thbf)\boldsymbol f_n(\thbf) \mathrm d\thbf\right] L^{\kbf}\right\} \boldsymbol{x}_{n\vsbf}.
\eeq



\section{General Spatio-Temporal Factor Model} \label{Sec: repr}

We show that any $n$ dimensional rf $\bm x_{n }$, satisfying Assumptions \ref{Ass.A}-\ref{Ass.C}, 
 can be summarized
by its projection, $\bm\chi_{n}$, on a $q$-dimensional sub-space generated by $q$ cross-sectional and spatio-temporal aggregation of the components of $\bm x_{n }$ and where $q$ is a given finite positive integer  independent of $n$. 
The rf $\bm\chi_{n}$ is such that as $n\to\infty$ it survives under cross-sectional and spatio-temporal aggregation, i.e., it converges in mean-square to a finite variance rf. The residual rf $\bm\xi_{n}=\bm x_{n }-\bm\chi_{n}$ instead vanishes under cross-sectional and space-time aggregation as $n\to\infty$.
Intuitively, the distinct asymptotic behavior of the two components under aggregation means that if any pervasive signal is present in the rf $\bm x_{n}$ an aggregation operation should help recovering it in the limit $n\to\infty$ and the signal will appear in the elements of $\bm \chi_{n}$.
To make this argument formal we start by introducing 



\begin{definition}[Spatio-temporal aggregation of rf]\label{def:STDAS}
For any $n\in\mathbb N$, consider an $n$-dimensional row vector of functions $\bm a_n\in L_2^n(\bm\Theta,\mathbb C, \Sigmabf_n,\mbf I_n)$. The sequence $\{ \boldsymbol{a}_n, n \in \mathbb{N} \}$ is a spatio-temporal dynamic averaging 
sequence (STDAS) if
\[
\lim_{n\to\infty} \Vert \boldsymbol{a}_n \Vert=\lim_{n\to\infty} \l(\frac 1{8\pi^3} \int_{\Thbf}\boldsymbol{a}_n(\thbf)\boldsymbol{a}_n^\dag(\thbf)\mathrm d\thbf\r)^{1/2}
\!\!\! = 0.
\]
Moreover, we say that $y$ is an aggregate if for any $\vsbf \in\mathbb Z^3$  there exists a STDAS $\{ \boldsymbol{a}_n, n \in 
\mathbb{N} \}$ such that $\lim_{n\to\infty} \underline{\boldsymbol{a}}_n (L) \boldsymbol{x}_{n\vsbf} = y_{\vsbf}$ in mean-square and $y_{\vsbf} \in \bm{\mathcal X}$. 
We denote the set of all aggregates by
$\mathcal{G}(\boldsymbol{x})$ and we refer to it as the aggregation space of $\bm{\mathcal X}$.
\end{definition}

Intuitively, the aggregation via a STDAS corresponds to averaging an infinite dimensional rf both in the cross-section and in the space-time dimensions, simultaneously. 
Notice that, because of the definition of $\bm{\mathcal X}$, any aggregate, i.e., any element of $\mathcal{G}(\boldsymbol{x})$, has variance either finite strictly positive or zero.
By generalizing to rf the definitions given by \citet{FL01} and \citet{HL13}, we have

\begin{definition}[Idiosyncratic and common components] \label{Def. Idiosy}
We say that an infinite dimensional rf $\bm w$ with elements $w_{\ell\vsbf}\in\bm{\mathcal X}$ for any $\vsbf\in\mathbb Z^3$ and $\ell\in\mathbb N$
\begin{inparaenum}
\item [(i)]  is idiosyncratic if for any $\vsbf\in\mathbb Z^3$ and any STDAS $\{\boldsymbol{a}_n, n\in \mathbb{N}\}$, $\lim_{n\to\infty} \underline{\boldsymbol{a}}_n (L)\bm w_{n\vsbf}=0$ in mean-square;
\item [(ii)] is common if it is not idiosyncratic, i.e., if for any $\vsbf\in\mathbb Z^3$ and any STDAS $\{\boldsymbol{a}_n, n\in \mathbb{N}\}$, $\lim_{n\to\infty} \underline{\boldsymbol{a}}_n (L)\bm w_{n\vsbf}=y_{\vsbf}^a$ in mean-square such that $y^a_{\vsbf}\in\bm{\mathcal X}$ and $0<\text{\upshape Var}(y^a_{\vsbf})<\infty$.
\end{inparaenum}
\end{definition}

Hereafter, we also refer to the components $w_\ell$ of $\bm w$ as idiosyncratic or common if $\bm w$ is idiosyncratic or common, respectively.
Note that if $\boldsymbol w$ is idiosyncratic then $\mathcal G(\boldsymbol w)=\{0\}$, that is it contains only the zero element. Moreover, $\boldsymbol w$ is idiosyncratic if and only if its largest dynamic spatio-temporal eigenvalue is an essentially bounded function (see Proposition \ref{Theorem 1} in Appendix \ref{app:suffFL}). 
In contrast, if $\bm w$ is common, by aggregating it we get a rf with finite and strictly positive variance, in other words, $\mathcal G(\boldsymbol w)$ contains only non-degenerate rf. This yields


\begin{definition}
[Common factors] \label{Def. Common}
Given an $n$-dimensional rf $\bm x_n$ with elements $x_{\ell\vsbf}\in\bm{\mathcal X}$ for any $\vsbf\in\mathbb Z^3$ and $\ell\le n$, we say that a scalar rf $w$
is a common factor if 
there exists a STDAS $\{ \boldsymbol{a}_n, n \in 
\mathbb{N} \}$ such that $ w_{\vsbf}=\lim_{n\to\infty} \underline{\boldsymbol{a}}_n (L) \boldsymbol{x}_{n\vsbf} $ in mean-square, $w_{\vsbf}\in\bm{\mathcal X}$, and $0<\text{\upshape Var}(w_{\vsbf})<\infty$.
%
\end{definition}

Clearly, by comparing Definition \ref{Def. Common} with Definition \ref{Def. Idiosy}(ii) we see that the common factors are elements of the aggregation space of the common components.

Denote the  sub-space of all components of an idiosyncratic rf (which are scalars) as $\bm{\mathcal X}^{\text{\tiny idio}}\subseteq \bm{\mathcal X}$ and the sub-space of all components of a common rf as $\bm{\mathcal X}^{\text{\tiny com}}\subseteq \bm{\mathcal X}$.
Given the above definition we have the decomposition 
\beq
\bm{\mathcal X}=\bm{\mathcal X}^{\text{\tiny com}} \oplus\bm{\mathcal X}^{\text{\tiny idio}}.\label{Eq. CanonicalX}
\eeq
Moreover, since the set $\mathcal{G}(\boldsymbol{x})$ is a closed subspace of $\bm{\mathcal X}$, we can also define 

\begin{definition}[Canonical decomposition] \label{cane}
For any $\ell\in\mathbb N$ and any $\vsbf\in\mathbb Z^3$, the orthogonal projection equation:
\begin{equation}
x_{\ell\vsbf} = \text{\upshape proj} (x_{\ell\vsbf}\vert \mathcal{G}(\boldsymbol{x})) + \delta_{\ell \vsbf}\label{Eq. Canonical}
\end{equation}
is called the canonical decomposition of the rf $x_{\ell\vsbf}$. 
\end{definition}

We show that  the decomposition \eqref{Eq. CanonicalX} and the canonical decomposition \eqref{Eq. Canonical} are equivalent. In particular, we will show that 
%
there exists a $q$-dimensional orthonormal white noise rf $\bm u$ with $q\ge 0$ and independent of $n$, such that: 
(i)  $ \text{\rm $\spn$}(\boldsymbol{u})=\mathcal{G}(\boldsymbol{x})$, hence, according to Definition \ref{Def. Common}, $\bm u$ is a vector of common factors;
(ii) $\gamma_{\ell\vsbf}=\text{proj}(x_{\ell \vsbf}\vert \mathcal{G}(\boldsymbol{x}))$ is common and $\bm\gamma=\{\gamma_{\ell\vsbf}, \ell\in\mathbb N, \vsbf \in\mathbb Z^3\}$ has a spectral density of rank $q$; 
%
(iii) $\bm\delta=\{\delta_{\ell\vsbf}, \ell\in\mathbb N, \vsbf \in\mathbb Z^3\}$ is idiosyncratic.

Summing up, given an observed $n$-dimensional rf $\bm x_n$, common factors are obtained as aggregates of $\bm x_n$ and the common component is obtained by projecting $\bm x_n$ onto such factors. Indeed, projecting onto the aggregation space of $\bm x$ or onto the aggregation space of the common component is equivalent, since the aggregation space of the idiosyncratic component contains only the zero element.
 To formalize the above projection argument, we first state 

\begin{definition}[$q$-General Spatio-Temporal Factor Model] \label{Def. q_DFS}
Let $q$ be a non-negative integer. We say that the rf $\boldsymbol{x}=\{x_{\ell \vsbf}, \ell\in\mathbb N, \vsbf \in\mathbb Z^3\}$ with $x_{\ell \vsbf}\in L_2(\mathcal{P},\mathbb{C})$ follows a $q$-General  Spatio-Temporal Factor Model ($q$-GSTFM) if $L_2(\mathcal{P},\mathbb{C})$ contains:
\begin{inparaenum} 
\item [(a)] an orthonormal $q$-dimensional white noise rf $\boldsymbol{u}=\{\boldsymbol{u}_{\vsbf}= (u_{1\vsbf} \ \cdots \ u_{q\vsbf})^\top, \vsbf  \in \mathbb{Z}^3\}$;
\item [(b)] an infinite dimensional rf $\boldsymbol\xi=\{ \xi_{\ell \vsbf}, \ell\in \mathbb{N}, \vsbf \in \mathbb{Z}^3 \}$; 
\end{inparaenum}
both fulfilling Assumptions \ref{Ass.A} and \ref{Ass.B} and such that:
\begin{compactenum}
\item[(i)]  for any $\ell \in \mathbb{N}$ and any $\vsbf\in\mathbb Z^3$
\begin{align}
x_{\ell \vsbf}&= \chi_{\ell\vsbf} + \xi_{\ell\vsbf}, \label{Eq. x-decomp}\\
\chi_{\ell \vsbf} &=  \underline{\boldsymbol b}_{\ell}(L) \boldsymbol{u}_{\vsbf}=\sum_{\kbf}\sum_{j=1}^q  {\mathrm b}_{\ell j,\kbf} {u}_{j,\vsbf-\kbf},\label{Eq. x-decomp2}
\end{align}
thus defining an infinite dimensional rf $\boldsymbol\chi=\{ \chi_{\ell \vsbf}, \ell\in \mathbb{N}, \vsbf \in \mathbb{Z}^3 \}$;

\item [(ii)]  letting $b_{\ell j}(\bm\theta)=\sum_{\kbf}\mathrm b_{\ell j,\kbf}e^{-i\langle \kbf,\thbf\rangle}$, $j=1,\ldots, q$, $\bm\theta\in\bm\Theta$, and  $\bm b_\ell(\bm\theta)=(b_{\ell 1}(\bm\theta)\cdots b_{\ell q}(\bm\theta))$, it holds that  
$
\Vert\boldsymbol b_{\ell} \Vert^2 =\frac 1 {8\pi^3}\int_{\Thbf}\bm b_\ell(\bm\theta)\bm b_\ell^\dag(\bm\theta) \mathrm d\thbf <\infty;
$
\item[(iii)] for any $\ell \in \mathbb{N}$, $j=1,\ldots,q$, and  $\vsbf,\vsbf^\prime\in\mathbb Z^3$ such that $\vsbf\ne \vsbf^\prime$,  it holds that $\text{\upshape E}[\xi_{\ell \vsbf} u_{j \vsbf^\prime}]=0$.
\end{compactenum}
Furthermore, for any $n\in\mathbb N$ consider the $n$-dimensional sub-processes 
$
\boldsymbol{\chi}_n=\{ \bm\chi_{n \vsbf}=(\chi_{1\vsbf}\cdots \chi_{n\vsbf})^\top, \vsbf \in \mathbb{Z}^3 \}$ and
$\boldsymbol{\xi}_n=\{ \bm\xi_{n \vsbf}=(\xi_{1\vsbf}\cdots \xi_{n\vsbf})^\top, \vsbf \in \mathbb{Z}^3 \},$
with $j$-th largest dynamic spatio-temporal eigenvalues $\lambda_{nj}^\chi(\bm\theta)$ and $\lambda_{nj}^\xi(\bm\theta)$, respectively, then
\begin{compactenum}
\item[(iv)] 
$\inf\{M:\mathcal L [\bm\theta: \lim_{n\to\infty} \lambda^{\xi}_{n1}(\bm\theta)>M]=0\}<\infty$;
%
\item[(v)]  $\lim_{n\to\infty}\lambda^{\chi}_{nq}(\thbf)= \infty$, $\mathcal L$-a.e. in $\Thbf$. 
\end{compactenum}
\end{definition}


We refer to the infinite dimensional rf $\boldsymbol{\chi}$ and $\boldsymbol{\xi}$ as the common and idiosyncratic components of the representation  \eqref{Eq. x-decomp}. Indeed, part (iv) implies that $\bm \xi$ is idiosyncratic, since, as proved in Proposition \ref{Theorem 1} in Appendix \ref{app:suffFL}, a rf is idiosyncratic if and only if its dynamic spatio-temporal eigenvalues are essentially bounded functions. Moreover, since by part (iii) $\bm \xi$ and $\bm\chi$ have orthogonal elements, then $\bm\chi$ cannot be idiosyncratic and must be common. 

The $q$-GSTFM has two main features. First, differently from the GDFM, 
the common component in (\ref{Eq. x-decomp2}) accounts for the  spatio-temporal dependence. The $q$-dimensional rf of factors $\bm u$ is loaded by each element of $\bm x$ dynamically in time (possibly in a causal way, see Remark \ref{rem:1side}) and in space, since the filters depend on both dimensions. This means that, being a rf,  a common shock to $\bm x$ can impact different points in space heterogeneously at the same time and at different points in time and 
 it can impact also the variables observed in a given point in space at different times. Second, we do not impose any specific structure on the second moment of the of vector of idiosyncratic components $\bm\xi$ whose elements can be both cross-sectionally and spatio-temporally cross-auto-correlated, as long as part (iv) is satisfied. By allowing for spatial dependencies we then generalize to rf the representation derived for pure time series by \citet{FL01}, which in turn extended the approximate static factor model by \citet{C83} and \citet{CR83} and the exact dynamic factor model by \citet{G77} and \citet{SS77}, as well as the standard classical exact static factor model for cross-sectional data \citep[see ][]{lawley1971factor}.

\begin{rem}\label{ss1}
\upshape{A sufficient condition for part (ii) to hold is to ask for square summability of the coefficients of the linear filter $\underline {\bm b}_\ell(L)$, i.e., to assume $\sum_{\kbf}\vert {\mathrm b}_{\ell j,\kbf}\vert^2\le C$ for some finite $C>0$ independent of $j$. Indeed, by definition
\begin{align}\nn
\Vert\boldsymbol b_{\ell} \Vert^2
&=\frac 1 {8\pi^3}\int_{\Thbf} \sum_{j=1}^q
\sum_{\kbf}\l\lvert{\mathrm b}_{\ell j,\kbf}  \r\vert^2\mathrm d\thbf\le \max_{j=1,\ldots, q} \sum_{\kbf}\l\lvert{\mathrm b}_{\ell j,\kbf}  \r\vert^2\le C.
\end{align}
}
\end{rem}

\begin{rem}\label{rem:alt_idio}
\upshape{
Parts (iv) and (v) require some further clarifications. Because of Remark \ref{rem:dyneval}, the function $\lim_{n\to\infty} \lambda^{\chi}_{nq}$ is the $q$-largest dynamic spatio-temporal eigenvalue of the infinite dimensional rf $\bm\chi$, and, similarly the function $\lim_{n\to\infty} \lambda^{\xi}_{n1}$ is the largest dynamic spatio-temporal eigenvalue of the infinite dimensional rf $\bm\xi$. Now,  by (v) the former is to be intended as an extended function in the sense that its value is infinite but measurable (\citealp[p. 55]{royden1988real}), 
while, by (iv) the latter is instead an essentially bounded function \citep[p. 66]{rudin1987realcomplex}. 
}
\end{rem}

\begin{rem}\label{rem:eval_idio}
\upshape{
From  (iv) in Definition \ref{Def. q_DFS} and Remark \ref{rem:alt_idio}, it follows that there exists a finite $C>0$ independent of $\bm\theta$ such that:
$\lim_{n\to\infty}\lambda^\xi_{n1}(\thbf)\le C$, $\mathcal L\text{-a.e. in } \bm\Theta$.
And by the monotone convergence theorem, which holds because of Remark \ref{rem:dyneval}, we have 
\[
\lim_{n\to\infty} \int_{\bm\Theta}\lambda^{\xi}_{n1}(\bm\theta) \mathrm d\bm\theta=\int_{\bm\Theta}\lim_{n\to\infty}\lambda^{\xi}_{n1}(\bm\theta) \mathrm d\bm\theta 
 \le C.
\]
This, in turn implies that the idiosyncratic covariance matrix $\bm\Gamma_n^\xi(\bm 0)=\text{\upshape E}(\bm \xi_{n\vsbf}\bm \xi^\dag_{n \vsbf})$ has largest eigenvalue $\mu_{n1}^\xi$ such that 
\begin{align}
\lim_{n\to\infty}\mu_{n1}^\xi &
=\lim_{n\to\infty}  \max_{\bm b :\bm b^\top\bm b=1}\bm b^\top\l(\int_{\bm\Theta}\bm\Sigma_n^{\xi}(\bm\theta) \mathrm d\bm\theta\r)\bm b\le 
 \lim_{n\to\infty} \int_{\bm\Theta} \lambda_{n1}^{\xi}(\bm\theta)\mathrm d\bm\theta\le C.\nn
\end{align}
The latter condition is the usual assumption made in the vector static factor model literature to characterize an idiosyncratic component (see, e.g., \citealp{FGLR09}). Notice, however, that (v) in Definition \ref{Def. q_DFS} in general does not imply that the  common covariance matrix $\bm\Gamma_n^\chi(\bm 0)=\text{\upshape E}(\bm \chi_{n\vsbf}\bm \chi^\dag_{n \vsbf})$ has  eigenvalues diverging as $n\to\infty$, for the effect of common factors might be just lagged and not contemporaneous, in which case only the products $\bm\Gamma_n^\chi(\bm \kappa)\bm\Gamma_n^\chi(\bm \kappa)$ for $\bm \kappa\ne \bm 0$ will display diverging eigenvalues. This case has been studied by \citet{lam2012factor} in the vector case.
}
\end{rem}


There are essentially two ways to obtain a $q$-GSTFM. On the one hand,  one may assume that the rf spatio-temporal dynamics can be modeled as in (\ref{Eq. x-decomp})-(\ref{Eq. x-decomp2}), mimicking the approach in \cite{FHLR00}. On the other hand, one may find a set of very mild assumptions such that a spatio-temporal rf can be represented as in (\ref{Eq. x-decomp})-(\ref{Eq. x-decomp2}), extending to the rf setting the results of \cite{FL01}. In what follows, we consider the latter approach, which is more general and powerful than the former one: indeed, (\ref{Eq. x-decomp})-(\ref{Eq. x-decomp2}) is a representation which holds under Assumptions \ref{Ass.A}-\ref{Ass.C} and it is not a model imposed by the statistician. 
Our main result of this section is the following 

\begin{theorem}  Under Assumptions \ref{Ass.A}-\ref{Ass.C}, the rf
$\boldsymbol{x}$ follows a $q$-GSTFM if and only if
\begin{inparaenum}
\item[(i)] $\inf\{M:\mathcal L [\bm\theta: \lim_{n\to\infty}\lambda^{x}_{nq+1}(\bm\theta)>M]=0\}<\infty$;
\item[(ii)] $\lim_{n\to\infty}\lambda_{nq}^x(\thbf)=\infty$, $\mathcal L$-a.e. in $\Thbf$.
\end{inparaenum}
\label{Th. q-DFS}
\end{theorem}

Theorem \ref{Th. q-DFS}  characterizes the class of rf which admit the $q$-GSTFM in Definition \ref{Def. q_DFS}.  First, notice that the same comments of Remark \ref{rem:alt_idio} apply also to the functions in parts (i) and (ii). It follows that the presence of an eigen-gap in the dynamic spatio-temporal eigenvalues of the infinite dimensional rf $\bm x$  is a necessary and sufficient condition for the $q$-GSTFM to hold. To this end no assumption is needed other than Assumptions \ref{Ass.A}-\ref{Ass.C}, which are very mild. Notice also that the case  $q=0$ is possible, in which case $\bm x$ has no common factor and it is purely idiosyncratic. In practice, if for an observed rf $\bm x_n$, we see evidence of an eigen-gap in the eigenvalues of its spectral density (as e.g. in Figure \ref{Sim_Eigen_Modela_GSTFR}), then  $\bm x_n$ admits the $q$-GSTFM.



The proof of the theorem is given in Appendix \ref{dimostraloseriesci} and it is rather technical and lengthy. Here, we present only the key aspects of the whole derivation. 
The necessary condition part (``only if'') is is easy to prove (see Appendix \ref{app:necFL}). Indeed, by Weyl's inequality (see Appendix \ref{App:Weyl}), it is straightforward to see that if (iv) and (v) in Definition \ref{Def. q_DFS} hold then (i) and (ii) in Theorem \ref{Th. q-DFS} hold.
The sufficient condition part (``if'') is more difficult to prove and it based on a series of intermediate results.  
In a nutshell, in the proof we proceed by first constructing a $q$-dimensional orthonormal white noise vector rf, $\boldsymbol{z}$, say (see Proposition \ref{MyLemma11}).
Then, we show $\spn(\boldsymbol{z})=\mathcal{G}(\boldsymbol{x})$ (see Proposition \ref{MySection4.5}). 
 It follows that the canonical projection
$x_{\ell \vsbf} = \text{proj}(x_{\ell \vsbf}\vert \mathcal{G}(\boldsymbol{x})) + \delta_{\ell \vsbf}$, is such that $\delta_{\ell\vsbf}$ is idiosyncratic (Propositions \ref{Theorem 1} and \ref{deltaidio}), hence
$\text{proj}(x_{\ell \vsbf}\vert \mathcal{G}(\boldsymbol{x}))$, being orthogonal to $\delta_{\ell \vsbf}$, is common. The proof is completed by means of the arguments in the following remark on the identifiability of the white noise.

\begin{rem}\label{Rem_Indet}
\upshape{It must be pointed out that, in general, neither the $q$-dimensional orthonormal white noise rf $\boldsymbol {u}$ nor the filters
$\underline{\boldsymbol b}_{\ell}(L)$ in (\ref{Eq. x-decomp2}) are identified. Indeed, if (\ref{Eq. x-decomp}) and (\ref{Eq. x-decomp2}) hold, then infinitely many other equivalent representations of $\chi_{\ell \vsbf}$ are obtained by setting $\chi_{\ell \vsbf} = \underline{\boldsymbol m}_{\ell}(L) \boldsymbol{z}_{\ell \vsbf}$ for a $q$-dimensional rf $\bm z$ such that $\boldsymbol{z}_{\ell \vsbf} = \underline{\boldsymbol D}(L) \boldsymbol{u}_{\ell \vsbf}$, $\boldsymbol {m}_\ell(\bm\theta) = \boldsymbol {b}_\ell(\bm\theta)  {\boldsymbol D}^\dag (\bm\theta) $, with $\bm D(\thbf)$ which is $q\times q$ and such that 
$\Vert \bm D\Vert^2=\frac 1{8\pi^3}\int_{\bm\Theta} \bm D(\bm\theta)\bm D^\dag(\bm\theta)\mathrm d \bm\theta <\infty$ and $ {\boldsymbol D}^\dag(\bm\theta)  {\boldsymbol D}(\bm\theta) = \mbf {I}_q$ for all $\bm\theta\in\bm\Theta$.  It follows that $\bm z$ is also a $q$-dimensional orthonormal white noise rf. In fact, if ${\boldsymbol D}^\dag(\bm\theta)$ were not orthogonal, as assumed, but just invertible, we could still find equivalent representations of $\chi_{\ell \vsbf}$ where, however, $\bm z$ would no more be a white noise rf, but it is a rf autocorrelated in both the spatial and time dimension. 
 }
\end{rem}

Uniqueness of the $q$-GSTFM follows:

\begin{Corollary}\label{Mythm3}
%
If $\bm x$ follows a $q$-GSTFM as in Definition (\ref{Def. q_DFS}), then 
$\text{\rm $\spn$}(\boldsymbol{\chi}) = \text{\rm $\spn$}(\boldsymbol{u}) = \mathcal{G}(\boldsymbol{x})$ and
$\chi_{\ell \vsbf} = \text{\rm proj}(x_{\ell \vsbf}|\mathcal{G}(\boldsymbol{x})).$ Moreover,
the number of factors $q$, the common component $\bm \chi$, and the idiosyncratic component $\bm \xi$, are uniquely identified.
%
%
\end{Corollary}

Notice that this result implies: (i) $\spn(\boldsymbol{z})=\spn(\boldsymbol{u})$ for any $q$-dimensional white noise rf $\bm z$ obtained from $\bm u$ as in Remark \ref{Rem_Indet}, and (ii) $\chi_{\ell\vsbf},\xi_{\ell\vsbf}\in\bm{\mathcal X}$ for any $\ell\in\mathbb N$ and $\vsbf\in\mathbb Z^3$. Moreover, no representation with a smaller or larger number of factors fulfilling Definition \ref{Def. q_DFS} is possible. In other words the $q$-GSTFM is identified. 
It has to be stressed though that, since the definition of common and idiosyncratic components are only asymptotic ones, i.e., holding in the limit $n\to\infty$ (see Defintion \ref{Def. Idiosy}), identification is achieved only asymptotically.  Indeed, as shown later, if $n$ is fixed no consistency result can be derived when we estimate the model. This is again an instance of the blessing of dimensionality and it emphasizes that factor analysis is effective in high-dimensions.

\section{Recovering the common component - Population results}\label{sec:DPCA}

In this section we prove that among all possible $q$ dimensional aggregates which we can project $\bm x_{n}$ on, the first $q$-dynamic spatio-temporal principal components of $\Sgmbf(\thbf)$ are the optimal ones in the sense that they are those with largest variance, and in Theorem \ref{Mythm5} we prove that by projecting $\bm x_{n}$ onto such aggregates we can recover the common component $\bm\chi_n$ in the limit $n\to\infty$.

The canonical decomposition in Definition \ref{cane} 	is optimal in the sense that, by definition of linear projection, it minimizes the variance of the residual idiosyncratic term. However, to achieve such decomposition in practice we need to define a basis for the space of aggregates $\mathcal G(\bm x)$ to project an the elements of $\bm x$ onto. Therefore, given a $q$-GSTFM in Definition \ref{Def. q_DFS} all we need to do is to find a $q$-dimensional rf common factors, which, because of Definition \ref{Def. Common} belong to $\mathcal G(\bm x)$, thus have finite and strictly positive variance. Moreover, we shall require this $q$-dimensional rf of factors to be an orthonormal white noise rf.


The definition of common factors holds asymptotically, but in practice we deal with a given fixed $n$, then, for such given $n$ and any $j=1,\ldots, q$, we should look 
for those weights $\{\bm \alpha_{nj\kbf}, \kbf\in\mathbb Z^3\}$, such that $\underline{\bm\alpha}_{nj}(L) \bm x_{n\vsbf}=\sum_{\bm\kappa}\bm \alpha_{nj\kbf} \bm x_{n,\vsbf-\kbf}$, has maximum variance. 
In view of the canonical isomorphism in \eqref{fbar}, we shall then consider the equivalent maximization problem in the frequency domain, i.e., for any $\thbf\in\Thbf$ we shall solve:
\beq
\max_{\bm\alpha_{n j}}\, \bm\alpha_{nj}^{}(\thbf) \bm\Sigma_n^x(\thbf) \bm\alpha_{nj}^\dag (\thbf) \;\text{  s.t. }\; \Vert \bm\alpha_{n j}\Vert =1,\quad  \bm\alpha_{n j}^{}\bm\alpha_{n k}^\dag = 0, 
\quad j,k=1,\ldots, q, \;\; j\ne k.\label{furka}
\eeq
Notice that the objective function is the variance of the discrete Fourier transform of  $\underline{\bm\alpha}_{nj}(L) \bm x_{n\vsbf}$. For any given $\thbf\in\Thbf$, the solution of \eqref{furka} is clearly given by the eigenvector $\pbf^{x}_{nj}(\thbf)$ of the spectral density matrix $\Sgmbf(\thbf)$ corresponding to the $j$-th largest eigenvalue $\lambda_{nj}^x(\thbf)$, see Definition \ref{Def. eigenvectors}. For any $j=1,\ldots, q$, to this solution corresponds a scalar filtered rf $\{\underline{\pbf}_{nj}^x(L)\boldsymbol{x}_{n\vsbf}, \vsbf \in \mathbb{Z}^3\}$ which has spectral density $\lambda^x_{nj}(\thbf)$ and variance $\int_{\bm\Theta}\lambda^x_{nj}(\thbf)\mathrm d\bm\theta$. So the first, $j=1$, dynamic spatio-temporal principal component has largest variance as expected. Moreover, these rf are orthogonal contemporaneously and at any spatio-temporal shift, indeed, for $j \neq j^\prime$,  we have $\E[(\underline{\pbf}_{nj}^x(L)\boldsymbol{x}_{n\vsbf})({\underline{\pbf}_{nj'}^x(L)\boldsymbol{x}_{n\vsbf'}})^\dag]=0$ for all $\vsbf,\vsbf'\in\mathbb Z^3$. Notice that if $\E(\bm x_{n\vsbf})=\bm\mu_{n} \ne \bm 0$ then the filtered process should be defined as $\underline{\pbf}_{nj}^x(L)(\boldsymbol{x}_{n\vsbf}-\bm \mu_n)$, hence, they always have zero-mean.

However, the $q$ filtered processes defined by solving \eqref{furka} cannot be directly used as a basis for $\mathcal G(\bm x)$ for two reasons. First, they are not white noise rf. Second, and most importantly, as $n\to\infty$, their variance is not finite, indeed, under a $q$-GSTFM, we know that $\lim_{n\to\infty}\lambda^x_{nj}(\thbf)=\infty$ for all $j=1,\ldots, q$. Therefore, we need to rescale and whiten those rf. This is accomplished by means of the following (recall the notation in \eqref{pachistrani})

\begin{definition}[Normalized dynamic spatio-temporal principal components]\label{def:DYNEVEC}
For any $n \in \mathbb{N}$ and $\ell\le n$, the filtered rf processes
$$
\psi^n_{\ell \vsbf}=[\underline{\lambda}_{n\ell}^x(L)]^{-1/2}\star \underline{\pbf}_{n\ell}^x(L)\boldsymbol{x}_{n\vsbf}= \l\{\sum_{\kbf }   \l[\frac{1}{8\pi^3} \int_{\Thbf} e^{i\langle\kbf,\thbf\rangle}
[{\lambda}_{n\ell}^x(\thbf)]^{-1/2}
\pbf_{n\ell}^x(\thbf) \mathrm d\thbf
\r]
 L^{\kbf} \r\}\boldsymbol{x}_{n\vsbf}
$$
form a set of normalized dynamic spatio-temporal principal components associated with $\boldsymbol{x}_{n\vsbf}$.
\end{definition}

Notice that this definition makes sense since Assumption \ref{Ass.C} implies that $[{\lambda}_{n\ell}^x(\thbf)]^{-1}$ is finite for any  $n\in\mathbb N$ and $\ell\le n$. Now, for any $n\in\mathbb N$, define the rf ${\bfpsi}^n =\{{\bfpsi}^n_{\vsbf} = \left( \psi^n_{1 \vsbf}\cdots\psi^n_{q \vsbf} \right)^\top,\, \vsbf\in\mathbb Z^3\}$. Then, from Definition \ref{def:DYNEVEC} we have:
\beq\label{fagioli}
{\bfpsi}^n_{\vsbf} = [\underline{\bm\Lambda}_n(L)]^{-1/2}\star\underline{\boldsymbol{P}}_n(L) \bm x_{n\vsbf},
\eeq
where the linear spatio-temporal filters $\underline{\boldsymbol{\Lambda}}_n(L)$ and $\underline{\boldsymbol{P}}_n(L)$  are, respectively, obtained from 
the $q\times q$ diagonal matrix ${\bm\Lambda}_n(\thbf)$ having as entries the dynamic spatio-temporal eigenvalues $\lambda_{nj}^x(\thbf)$, for $j=1,\ldots, q$, and 
the $q\times n$ matrix  ${\boldsymbol{P}}_n(\thbf)=(\pbf_{n1}^{x\top}(\thbf)\cdots \pbf_{nq}^{x\top}(\thbf))^\top$
having as rows the $q$ corresponding dynamic spatio-temporal eigenvectors. Now, let ${\bm \Phi}_n(\thbf)$ be the $n-q\times n-q$ diagonal matrix having as entries the dynamic spatio-temporal eigenvalues $\lambda_{nj}^x(\thbf)$, for $j=q+1,\ldots, n$, and let ${\boldsymbol{Q}}_n(\thbf)$ be the $n-q\times n$ matrix having as rows the $n-q$ corresponding dynamic spatio-temporal eigenvectors. Then, for all $\thbf\in\Thbf$, $\bm\Sigma_n^x(\thbf) = {\boldsymbol{P}}^\dag_n(\thbf){\bm\Lambda}_n(\thbf){\boldsymbol{P}}_n(\thbf)+ {\boldsymbol{Q}}^\dag_n(\thbf){\bm \Phi}_n(\thbf){\boldsymbol{Q}}_n(\thbf)$. Therefore, since $\mbf I_n= {\boldsymbol{P}}^\dag_n(\thbf){\boldsymbol{P}}_n(\thbf)+ {\boldsymbol{Q}}^\dag_n(\thbf) {\boldsymbol{Q}}_n(\thbf)$, from \eqref{fagioli} we immediately see that $\bfpsi^n$ has spectral density $\mbf I_q$, hence  it is a $q$-dimensional orthonormal white noise rf as required. By letting $n\to\infty$, we obtain from ${\bfpsi}^n$ the basis for $\mathcal G(\bm x)$ we are looking for. This is formalized by means of the following

\begin{theorem}\label{Mythm5}
For any $n \in \mathbb{N}$ 
and $\ell\le n$, denote by $\underline{\boldsymbol{\pi}}_{n\ell}(L)$ the $\ell$-th $q$-dimensional row of $\underline{\boldsymbol{P}}_n^\dag(L)$. 
Suppose that (i) and (ii) of Theorem \ref{Th. q-DFS} and Assumptions \ref{Ass.A}-\ref{Ass.C} hold. 
Then, for all $\vsbf\in\mathbb Z^3$,
 $\lim_{n\to\infty}\underline{\boldsymbol{\pi}}_{n\ell}(L)\star \underline{\boldsymbol{\Lambda}}_n^{1/2}(L) {\bfpsi}^n_{\vsbf} = \lim_{n\to\infty}\underline{\boldsymbol{\pi}}_{n\ell}(L)\star \underline{\boldsymbol{P}}_{n}(L) \boldsymbol{x}_{n\vsbf}=\chi_{\ell\vsbf}
$ in mean-square.
\end{theorem} 

This result is the basis for our estimation approach. It implies that if, for a given $n\in\mathbb N$, we knew the spectral density matrix of $\bm x_n$, then, for any $\ell\le n$ and $\vsbf\in\mathbb Z^3$, an estimator of the common component would be: 
\beq\label{chatbot}
\chi_{\ell\vsbf}^{(n)}= \underline{\boldsymbol{\pi}}_{n\ell}(L)\star \underline{\boldsymbol{P}}_{n}(L) \boldsymbol{x}_{n\vsbf} =\underline {\bm K}_{n\ell}^x(L)\boldsymbol{x}_{n\vsbf},\; \mbox{ say.}
\eeq 
This is a consistent estimator since as $n\to\infty$ it converges in mean-square to the unobservable common component $\chi_\ell$. Notice that since we are dealing with projections the rescaling by means of the eigenvalues introduced in Definition \ref{def:DYNEVEC} is actually not needed in practice, as we just need the eigenvectors.



\begin{rem}\label{rem:freddo}
\upshape{
For any $n\in\mathbb N$, let $\bm P_n^\chi(\thbf)$ be the $q\times n$ matrix having as rows the spatio-temporal dynamic eigenvectors of the spectral density matrix of the common component $\bm\Sigma_n^\chi(\thbf)$. Let also $\underline{\bm P}_n^\chi(L)$ the associated linear spatio-temporal filter and for any  $\ell\le n$, denote by $\underline{\boldsymbol{\pi}}^\chi_{n\ell}(L)$ the $\ell$-th $q$-dimensional row of $\underline{\boldsymbol{P}}_n^{\chi\dag}(L)$.  Then, since $\text{rk}(\bm\Sigma_n^\chi(\thbf))=q$ for all $n\in\mathbb N$ and $\mathcal L$-a.e.~in $\Thbf$, we immediately see that, for any $\ell\le n$ and $\vsbf\in\mathbb Z^3$, we can always write:
\beq\label{chatbot2}
\chi_{\ell\vsbf} =  
\underline{\bm \pi}_{n\ell}^{\chi}(L)\star\underline{\bm P}_n^\chi(L)\bm\chi_{n\vsbf}=
\underline{\bm \pi}_{n\ell}^{\chi}(L)\star\underline{\bm P}_n^\chi(L)\bm x_{n\vsbf} = \underline {\bm K}_{n\ell}^\chi(L)\boldsymbol{x}_{n\vsbf},\; \mbox{ say}, 
\eeq
because $\text{Cov}(\bm\chi_{n\vsbf},\bm\xi_{n\vsbf^\prime})=\mbf 0$ for all $\vsbf,\vsbf^\prime\in\mathbb Z^3$. This, together with \eqref{chatbot},
 implies that, as $n\to\infty$, the coefficients of  $\underline {\bm K}_{n\ell}^x(L)$ converge in mean-square to the coefficients of $\underline {\bm K}_{n\ell}^\chi(L)$.
}
\end{rem}

\begin{rem}\upshape{
In general, the dynamic spatio-temporal eigenvectors are complex vectors. However,  for any $n\in\mathbb N$, we know that 
$\mbf I_n= {\boldsymbol{P}}^\top_n(\thbf)\bar {\boldsymbol{P}}_n(\thbf)+ {\boldsymbol{Q}}^\top_n(\thbf) \bar {\boldsymbol{Q}}_n(\thbf)$, and that the spectral density matrix is Hermitian, i.e., $\bar{\bm\Sigma}_n^x(\thbf)={\bm\Sigma}_n^{x\top}(\thbf)={\bm\Sigma}_n^x(-\thbf)$ and $\bm\Lambda_n(\thbf)$ is a real matrix. Therefore, we can always impose $\bm p_{n\ell}^x(-\thbf)=\bar{\bm p}_{n\ell}^x(\thbf)$ for all $\ell\le n$. This implies that $\int_{\Thbf} e^{i\langle\kbf,\thbf\rangle}
[{\lambda}_{n\ell}^x(\thbf)]^{-1/2}
\pbf_{n\ell}^x(\thbf) \mathrm d\thbf$ is always a real number, and, thus, the normalized dynamic spatio-temporal principal components are real rf, see also \citet{hallin2018optimal}.
}
\end{rem}

\section{Recovering the common component - Estimation} \label{Sec: estim}

\subsection{Estimation in practice} \label{Sec: estim_practice}

The population results derived in Section \ref{Sec: repr} show that the spatio-temporal common component $\chi_{\ell}$ can be recovered as $n\to\infty$ from a sequence of projections, see Theorem \ref{Mythm5}. The filters needed to define this projection are given in Definition \ref{def:DYNEVEC} and depend on the dynamic spatio-temporal eigenvalues and eigenvectors of the unknown spectral density matrix. 

Let us assume now to observe a finite $n$-dimensional realization $\bm x_n$ of the infinite dimensional rf $\boldsymbol{x}$ over $S_1\times S_2$ points on a 2-dimensional lattice and over $T$ time periods. 
 In order to proceed we need to fix the origin of the lattice, because of homostationarity this can be chosen arbitrarily in any location of $\mathbb Z^2$. Here we adopt the convention that the point $(s_1\ s_2)=(1\ 1)$ corresponds to the South-West corner of the given lattice. Then, index $s_1$ grows by moving East while the index $s_2$ grows by moving North. With this definition of the spatial coordinates, our observations are collected into the $n\times S_1S_2T$-dimensional matrix: $\{x_{\ell\vsbf}=x_{\ell(s_1\ s_2\ t)},\ \ell=1,\ldots, n,\ s_1=1,\ldots, S_1,\ s_2=1,\ldots, S_2,\ t=1,\ldots, T\}$. 
 
If the spatio-temporal dynamic eigenvalues of $\bm x_n$ satisfy Theorem \ref{Th. q-DFS}, then, according to the $q$-GSTFM, for all $\ell=1,\ldots, n$, $s_1=1,\ldots, S_1$, $s_2=1,\ldots, S_2$, and $t=1,\ldots, T$ we can write $x_{\ell(s_1\ s_2\ t)}=\chi_{\ell(s_1\ s_2\ t)}+\xi_{\ell(s_1\ s_2\ t)}$, where $\chi_{\ell}$ is the common component and $\xi_\ell$ is idiosyncratic. For any given $n$, we denote as $\bm\chi_n$ and $\bm\xi_n$ the $n$-dimensional rf of the common and idiosyncratic components.

Throughout this section we assume that the number of factors, $q$, driving the common component is known (see Section \ref{Sec.Selectq}) and we now describe our estimation strategy.
Let  $\vsbf_1 = (s_{11} \ s_{12} \ t_1)^\top$ and $\vsbf_2 = (s_{21} \ s_{22} \ t_2)^\top$, then an estimator of 
$\Sgmbf(\thbf)$ is
\begin{equation}\label{hatSimgax}
\widehat{\boldsymbol \Sigma}_n^x(\thbf)  = \frac{1}{S_1 S_2 T} \sum_{\vsbf_1,\vsbf_2 = (1\ 1\ 1)^\top}^{ (S_1 \ S_2 \ T)^\top} \!\!\!\! \! \! \! \! \!  \x_{n \vsbf_1} \x_{n\vsbf_2}^\top K_1\left(\frac{s_{11}-s_{21}}{B_{S_1}} \right) K_2\left(\frac{s_{12}-s_{22}}{B_{S_2}} \right) K_3\left(\frac{t_{1}-t_{2}}{B_{T}} \right) e^{-i \left\langle \vsbf_1 - \vsbf_2,  \thbf\right\rangle},
\end{equation}
with $K_1(\cdot), K_2(\cdot)$, and $K_3(\cdot)$ being kernel functions and $B_{S_1}, B_{S_2}$, and $B_{T}$ being bandwidths, whose properties are discussed later.

In agreement with the population results of Theorem \ref{Mythm5}, the common component is estimated by projecting $\bm x_n$ onto the space spanned by linear filters generated by the $q$ leading spatio-temporal dynamic eigenvectors.
For all $\thbf\in\Thbf$   let us denote by $\widehat{\bm P}_n(\thbf)$ the $q\times n$ matrix having as rows the spatio-temporal dynamic eigenvectors of  $\widehat{\boldsymbol \Sigma}_n^x(\thbf)$ and, for any $\ell=1,\ldots, n$, let $\widehat{\bm\pi}_{n\ell}(\thbf)$ the $\ell$-th $q$-dimensional row of $\widehat{\bm P}^\dag_n(\thbf)$, and, in agreement with \eqref{chatbot} define $\widehat{\K}^x_{n\ell}(\thbf) = {\widehat{\bm\pi}}_{n\ell}(\thbf) \widehat{\bm P}_{n}(\thbf)$, generating the linear filter   $\widehat{\underline{\K}}^{x}_{n\ell}(L)$.

%



Now, since $\widehat{\underline{\K}}^{x}_{n\ell}(L)$ is in general infinite and two-sided, but $\bm x_{n\vsbf}$ is not available for $\vsbf < (1 \ 1 \ 1)$ and $\vsbf > (S_1 \ S_2 \ T)$, we consider instead a truncated linear filter, whose definition depends on the space-time location $\vsbf=(s_1\ s_2 \ t)^\top$ in correspondence of which the filter is applied to $\bm x_n$. Namely, we consider
%
%
%
%
%
\begin{equation}\label{def.underlineK}
\begin{aligned}
 \widehat{\underline{\K}}^{x, \vsbf}_{n\ell}(L) = & \frac{1}{8\pi^3} 
\sum_{\kappa_1 = \underline{\kappa}_1(s_1)}^{\overline{\kappa}_1(s_1)} 
\sum_{\kappa_2 = \underline{\kappa}_2(s_2)}^{\overline{\kappa}_2(s_2)} 
\sum_{\kappa_3 = \underline{\kappa}_3(t)}^{\overline{\kappa}_1(t)} 
 \left( \int_{\Thbf} \widehat{\K}^x_{n\ell}(\thbf) e^{i \left\langle (\kappa_1\ \kappa_2\ \kappa_3)^\top, \thbf \right\rangle} {\rm d} \thbf \right)  L_1^{\kappa_1}L_2^{\kappa_2}L_3^{\kappa_3},
\end{aligned} 
\end{equation}
where, for some integers $M_{S_1}<S_1, M_{S_2}<S_2$, and $M_T<T$, we defined 
\begin{align}
&\underline{\kappa}_1(s_1)=\max\{s_1-S_1, -M_{S_1}\},&&\overline{\kappa}_1(s_1)= \min\{s_1 -1, M_{S_1}\},\nn\\
&\underline{\kappa}_2(s_2)=\max\{s_2-S_2, -M_{S_2}\},  &&\overline{\kappa}_2(s_2)= \min\{s_2 -1, M_{S_2}\},\label{troncamento}\\
&\underline{\kappa}_3(t)= \max\{t-T, -M_T\}, &&\overline{\kappa}_3(t)= \min\{t -1, M_T\}.\nn
\end{align}
%
%
%
%
For any given $\ell=1,\ldots, n$ and any $\vsbf=(s_1\ s_2 \ t)^\top$ such that $s_1=,1\ldots, S_1$, $s_2=,1\ldots, S_2$, and $t=1,\ldots, T$, the common component is then estimated as
\beq
\widehat{\chi}_{\ell\vsbf}^{(n)} =  \widehat{\underline{\K}}^{x, \vsbf}_{n\ell}(L)\bm x_{n\vsbf}. \label{finiamo!}
\eeq
%
%

\begin{rem}\label{rem:freq}
\upshape{In practice all estimated quantities in the frequency domain, as $\widehat{\boldsymbol \Sigma}_n^x(\thbf)$, $\widehat{\bm P}_n(\thbf)$, and $\widehat{\K}^x_{n\ell}(\thbf)$, should be computed only for a finite number of frequencies, defined as $\thbf_{\hbf}=(\theta_{1,h_1}\ \theta_{2,h_2}\ \theta_{3,h_3})^\top$, with $\theta_{1,h_1}= \pi h_1/ B_{S_1}$, $\theta_{2,h_2}= \pi h_2/ B_{S_2}$, and $\theta_{3,h_3}= \pi h_3/ B_{T}$, for integers $h_1=-B_{S_1},\ldots, B_{S_1}$, $h_2=-B_{S_2},\ldots, B_{S_2}$, and $h_3=-B_{T},\ldots, B_{T}$. For simplicity, in this and the following  sections we implicitly assume the identities $\int_{\Thbf}\mathrm d\thbf\equiv 
\sum_{|h_1|\le B_{S_1}}\sum_{|h_2|\le B_{S_2}}\sum_{|h_3|\le B_{T}}$, ${8\pi^3}\equiv{(2B_{S_1}+1)(2B_{S_2}+1)(2B_{T}+1)}$, and  
$\sup_{\thbf\in\Thbf}\equiv \max_{|h_1|\le B_{S_1}}\max_{|h_2|\le B_{S_2}}\max_{|h_3|\le B_{T}}$.
}
\end{rem}

%

\subsection{Assumptions}
For estimation we need to add few more assumptions. First, the GSTFM has two-sided filters as defined in \eqref{Eq. x-decomp2}, however, it is desirable to have one-sided filters in the time dimension. This can be obtained by imposing the following

\begin{ass}\label{Ass.MA}
\begin{inparaenum}
For any $\vsbf \in\mathbb Z^3$   and $\ell\in\mathbb N$:
\item[(i)]  $\chi_{\ell \vsbf} =  \underline{\boldsymbol c}_{\ell}(L) \boldsymbol{v}_{\vsbf}= \sum_{\kappa_1,
\kappa_2 \in \mathbb{Z}} \sum_{\kappa_3=0}^{\infty}$ $\sum_{j=1}^q  {\mathrm c}_{\ell j,\kbf} {v}_{j,\vsbf-\kbf}$, where 
$\{\bm v_{\vsbf} =(v_{1\vsbf}\cdots v_{q\vsbf})^\top,\vsbf \in\mathbb Z^3\}$ is an i.i.d.~$q$-dimensional zero-mean orthonormal rf;
\item[(ii)] 
$\xi_{\ell\vsbf} =\underline{\bm \beta}_\ell(L)\bm\varepsilon_{\vsbf}= \sum_{\kappa_1,\kappa_2}  \sum_{\kappa_3 =0}^\infty$ $ \sum_{j =1}^\infty \beta_{\ell j, \k} \varepsilon_{j, \vsbf-\k}$, where 
$\{{\boldsymbol \varepsilon}_{\vsbf} = (\varepsilon_{1\vsbf} \ \varepsilon_{2\vsbf} \ \cdots)^\top, \vsbf \in \mathbb{Z}^3\}$  is an i.i.d.~infinite dimensional zero-mean orthonormal rf;
\item[(iii)] For any $\vsbf^\prime\in\mathbb Z^3$, any $j=1,\ldots, q$, and any $i\in\mathbb N$,
$\text{\upshape Cov}(v_{j\vsbf},\varepsilon_{i\vsbf^\prime})=0$.
\end{inparaenum}
\end{ass}

The existence of one-sided time representations in parts (i) and (ii) is a mild one. 
 For the idiosyncratic component, our requirement is for the Wold representation to exist also for an infinite dimensional process. For the common component, which is singular, the existence of the assumed one-sided representation has been investigated by \citet{FHLZ15}  in the pure time series case (see also Remark \ref{rem:1side} below). Notice also the \citet{HL13} derived an analogous of our Theorem \ref{Th. q-DFS}, where only one-sided filters are used. Such approach, however, does not ensure the existence of a $q$-dimensional white noise rf driving the common component, and its existence is instead assumed.
For the common component the two-sided representation in space is implied by the $q$-GSTRF in Definition \ref{Def. q_DFS}, and for the idiosyncratic component we make an analogous assumption but based on an infinite dimensional white noise rf.  


By means of parts (i) and (ii) we also strengthen the conditions on the rf $\bm v$ and $\bm\varepsilon$ which are now independent along the spatio-temporal dimensions. Note that the independence assumption could be relaxed. For example we could just assume $\bm v$ and $\bm\varepsilon_n$ to be martingale differences in the time dimension so to allow for conditional heteroskedasticity in time (see, e.g., \citealp{barigozzi2023fnets}). 

Part (iii) implies orthgonality of common and idiosyncratic components at all leads and lags consistently with the GSTFM in Definition \ref{Def. q_DFS}.

\begin{rem}\label{rem:1side}
\upshape{
If for any fixed $n\in \mathbb N$ the $n$-dimensional vector of common components has a spectral density matrix $\bm\Sigma_n^\chi(\thbf)$ which is a rational function of $\theta_3$, then, from \citet[Ch.1, Section 10]{Roz67} 
 it follows that, for all $\ell\le n$ and $\vsbf\in\mathbb Z^3$,
\begin{align}
\chi_{\ell\vsbf} &= \sum_{j=1}^q \frac{\underline {a}_{\ell j}(L_1,L_2,L_3)}{ \underline d_{\ell j} (L_3)} v_{j\vsbf}=
\sum_{j=1}^q\sum_{\kappa_1,\kappa_2\in\mathbb Z} \sum_{\kappa_3=0}^{p_1}  \mathrm a_{\ell j,(\kappa_1 \ \kappa_2 \ \kappa_3)}  L_1^{\kappa_1}L_2^{\kappa_2}L_3^{\kappa_3}
   \left[  \sum_{h_3=0}^{p_2} \mathrm d_{\ell j, h_3} L_3^{h_3}\right]^{-1}v_{j\vsbf},\nn
\end{align}
 for some finite positive integers $p_1$ and $p_2$, which, without loss of generality we can assume to be independent of $\ell$. Moreover, $d_{\ell j}(z)\ne 0$ for all $z\in\mathbb C$ such that $\vert z\vert \le 1$, and $a_{\ell j}(z_1,z_2,z_3)\ne 0$ for all $z_3\in\mathbb C$ such that $\vert z_3\vert< 1$.  By defining $f_{\ell j}(\theta_3)=[d_{\ell j}(\theta_3)]^{-1}=\sum_{\kappa_3=0}^\infty \mathrm f_{\ell j,\kappa_3} e^{-i\langle \kappa_3,\theta_3\rangle}$, it follows that 
\begin{align}
\chi_{\ell\vsbf} &=  \sum_{j=1}^q \sum_{\kappa_1,\kappa_2\in\mathbb Z} \sum_{\kappa_3=0}^{\infty} {\sum_{m_3=0}^{p_1}  \mathrm a_{\ell j,(\kappa_1 \ \kappa_2 \ m_3)} \mathrm f_{\ell j, \kappa_3-m_3}}
L_1^{\kappa_1}L_2^{\kappa_2}L_3^{\kappa_3}v_{j\vsbf}
,\nn
\end{align}
which, by setting $\mathrm c_{\ell j, (\kappa_1 \ \kappa_2\ 
\kappa_3)}=\sum_{m_3=0}^{p_1}  \mathrm a_{\ell j,(\kappa_1 \ \kappa_2 \ m_3)} \mathrm f_{\ell j, \kappa_3-m_3}$, coincides with Assumption \ref{Ass.MA}(i). Thus, for all $n\in \mathbb N$ and all $\vsbf\in\mathbb Z$,
$\bm v_{\vsbf} \in \overline{\text{\upshape span}}
 (\boldsymbol{\chi}_{n \vsbf-\kbf}, \kbf=(\kappa_1\ \kappa_2\ \kappa_3)^\top, \kappa_1, \kappa_2 \in\mathbb Z, \kappa_3\ge 0)$, i.e., $\bm v$ is fundamental for $\bm\chi_n$.
 The generalization of this reasoning to the infinite dimensional process $\bm\chi$ is considered in \citet[Lemma 1 and 2]{FHLZ15} in the case of pure time series, where it is shown that, under rationality of the spectral density, then fundamentalness of $\bm v$ is always true for any $n>q$ generically, i.e., for any value of the coefficients $c_{\ell j,\kbf}$ such that Assumption \ref{Ass.MA}(i) holds with the exception of a zero-measure set (see also \citealp{AD2008properties}). 
 }\end{rem}


The coefficients of the representations in Assumption \ref{Ass.MA} are characterized by 

\begin{ass}\label{Ass.coef1} 
For all $\ell\in\mathbb N$, $j=1,\ldots, q$, and $\bm \kappa=(\kappa_1\ \kappa_2 \ \kappa_3)^\top\in \mathbb Z^2\times \mathbb N_0$: 
\begin{inparaenum}
\item [(i)] $|c_{\ell j, \bm\kappa}| \leq A^\chi_{\ell j} {\rho_1^{\chi|\kappa_1|}} {\rho_2^{\chi|\kappa_2|}} {\rho_3^{\chi\kappa_3}}$, for some finite $\rho_1^\chi, \rho_2^\chi, \rho_3^\chi \in (0, 1)$ independent of $\ell$, $j$, and $\bm\kappa$, and some finite $A^\chi_{\ell j} >0$ independent of $\bm\kappa$ and such that 
$\sum_{j=1}^q A_{\ell j}^\chi \leq A^\chi$, for some finite $A^\chi>0$ independent of $\ell$;
\item [(ii)] 
$|\beta_{\ell j, \bm\kappa}| \leq A^\xi_{\ell j} {\rho_1^{\xi|\kappa_1|}} {\rho_2^{\xi|\kappa_2|}} {\rho_3^{\xi\kappa_3}}$, for some finite $\rho_1^\xi, \rho_2^\xi, \rho_3^\xi \in (0, 1)$ independent of $\ell$, $j$, and $\bm\kappa$, and some finite $A^\xi_{\ell j} >0$ independent of $\bm\kappa$ and such that 
$\sum_{j=1}^\infty A^\xi_{\ell j} \leq A^\xi$ and  $\sum_{\ell=1}^\infty A^\xi_{\ell j} \leq A^\xi$, for some finite $A^\xi>0$ independent of $\ell$ and~$j$. 
\end{inparaenum}
\end{ass}

%

This assumption implies square-summability of the coefficients of the filters, which for the common component is a sufficient condition for (ii) in Definition \ref{Def. q_DFS} to hold, see Remark \ref{ss1}. This assumption has two other important implications. First, part (ii) implies that the largest spatio-temporal dynamic eigenvalue of $\bm\xi_n$ satisfies (see Proposition \ref{lukagu} in Appendix \ref{APPE})
\beq
\sup_{\thbf \in \Thbf}\lim_{n\to\infty} \lambda_{n1}^\xi(\bth) \le C,
\eeq
for some finite $C>0$. Hence, according to (i) in Theorem \ref{Th. q-DFS}, $\bm\xi_n$ is effectively an idiosyncratic component.  Second, in part (i) we do not require summability of the coefficients along the rows, so that the spatio-temporal dynamic eigenvalues of $\bm\chi_n$ can be diverging with $n$. Divergence of those eigenvalues is made formal by means of the following assumption which strengthens (ii) in Theorem \ref{Th. q-DFS}:

\begin{ass}\label{Ass.lamChi}
For all $j=1,\ldots, q-1$ there exist continuous functions 
$\bth \mapsto \widetilde{\omega}_j(\bth)$ and $\bth \mapsto \undertilde{\omega}_j(\bth)$
such that for all $\thbf \in\Thbf$
$$
0<
\undertilde{\omega}_{j+1}(\bth)\leq
\lim_{n\to\infty}\frac{\lambda_{n,j+1}^{\chi}(\bth)}n\leq \widetilde{\omega}_{j+1}(\bth)<\undertilde{\omega}_j(\bth) \leq \lim_{n\to\infty}\frac{ \lambda_{nj}^{\chi}(\bth)}n \leq \widetilde{\omega}_j(\bth)<\infty. 
$$
\end{ass}
The requirements of distinct and linearly diverging eigenvalues are standard in the factor model literature. While the former requirement is merely technical, the latter implies that here we are dealing only with  factors which are pervasive for the whole cross-section, which in turn implies that the ordering of the cross-sectional units is irrelevant for estimation. 
Both requirements could, in principle be relaxed. For example, the case of local, or group specific dynamic factors, could be considered along the lines of what done by \citet{hallin2011dynamic} in the purely time series case. We do not make any distributional assumption but we require only the following moment conditions
\begin{ass}\label{Ass.moments}
For all $j=1,\ldots, q$ and $\ell \in \mathbb{N}$,
$\max\l\{{\rm E}\left(\vert v_{h\vsbf}\vert^p\right), {\rm E}\left(\vert \varepsilon_{j\vsbf}\vert^p\right)\r\} \leq \bar{A}$,
for some  $p>4$ and $\bar{A} >0$ independent of $j$ and $\ell$.
\end{ass}

Two technical assumptions are also required.
First, we characterize the kernel functions and bandwidths needed to estimate the spectral density matrix and the truncation levels in \eqref{troncamento}  by the following

\begin{ass}\label{AssSpec2}
\begin{inparaenum}
\item[(i)]
For any $l=1,2,3$, the kernel functions $K_l:[-1,1]\to \mathbb R^+$ are symmetric and bounded, and such that
\begin{inparaenum}
\item[(a)] $K_l(0) = 1$;
\item[(b)] for some $\vartheta_l >0$, $|K_l(u) -1| = O(|u|^{\vartheta_l})$ as $u \rightarrow 0$;
\item[(c)] $\int_{\mathbb{R}} K_l^2(u){\rm d} u < \infty$; 
\item[(d)] $\sum_{h_1 \in \mathbb{Z}} \sup_{|h_1 - h_2| \leq 1} |K_l(h_1 u) - K_l (h_2 u)| = O(1)$ as $u  \rightarrow 0$.
\end{inparaenum}
\item[(ii)] 
The bandwidths are such that
$c_1 S_1^{b_1} < B_{S_1} < c_2 S_1^{b_2}$, $c_1^* S_2^{b_1^*} < B_{S_2} < c_2^* S_2^{b_2^*}$, and $c_1^{**} T^{b_1^{**}} < B_{T} < c_2^{**} T^{b_2^{**}}$, for some
$c_1, c_2, c_1^*, c_2^*, c_1^{**}, c_2^{**} >0$ and $0<b_1<b_2<1$, $0<b_1^*<b_2^*<1$, $0<b_1^{**}<b_2^{**}<1$. 
\item [(iii)] 
$d_1 S_1^{p_1} < M_{S_1}< d_2S_1^{p_2}$, 
$d_1^* S_2^{p_1^*} < M_{S_2}< d_2^*S_2^{p_2^*}$, and
$d_1^{**} T^{p_1^{**}} < M_{T}< d_2^{**}T^{p_2^{**}}$, for some 
$d_1, d_2, d_1^*, d_2^*, d_1^{**}, d_2^{**} >0$ and $0<p_1<p_2<1$, $0<p_1^*<p_2^*<1$, $0<p_1^{**}<p_2^{**}<1$. 
\end{inparaenum}
\end{ass}

Part (i) and (ii) are standard. Part (iii) controls the truncation of the linear filter defined in \eqref{def.underlineK} and \eqref{troncamento}.

Second, we assume that the effect of the linear spatio-temporal filters ${\underline{\K}}^{\chi}_{n\ell}(L)$, as defined in \eqref{chatbot2}, decreases geometrically. 
 
\begin{ass}\label{Ass.ExpK}
%
For any $\ell=1,\ldots, n$,
let ${\underline{\K}}^{\chi}_{n\ell}(L)=\sum_{(\kappa_1\ \kappa_2 \ \kappa_3)^\top\in \mathbb Z^3}
{{\K}}^{\chi}_{n\ell,(\kappa_1\ \kappa_2\ \kappa_3)}$ $L_1^{\kappa_1}L_2^{\kappa_2}L_3^{\kappa_3}$,
then, 
$\Vert {{\K}}^{\chi}_{n\ell,(\kappa_1\ \kappa_2\ \kappa_3)} \Vert  \leq C_0 (1+\varepsilon_1)^{-|\kappa_1|}  (1+\varepsilon_2)^{-|\kappa_2|} (1+\varepsilon_3)^{-|\kappa_3|} \Vert {{\K}}^{\chi}_{n\ell,(0\ 0\ 0)}\Vert,
$
for some finite $C_0, \epsilon_1,\epsilon_2,\epsilon_3>0$ independent of $\ell$.
\end{ass}




\subsection{Asymptotic results}

To study the asymptotic properties of the estimated spectral density matrix, we generalize to the case of spatio-temporal rf the approaches by  \cite{wu2018asymptotic} and \citet{zhang2021convergence} for time series and by  \cite{DPW17}  for purely spatial models, which in turn are all are based on the notion of functional dependence originally proposed by \citet{wu2005nonlinear} in a univariate time series context. 
The resulting estimation theory is available in Appendix \ref{Spec.Y} and represents a novel contribution to the literature on the inference for spatio-temporal rf.  

Letting $\widehat{\sigma}_{ij}^x(\thbf)$ be the $(i,j)$-th entry of the estimator $\wh{\bm\Sigma}_n^x(\thbf)$, defined in \eqref{hatSimgax}, we prove the following


\begin{theorem}\label{Thmsigmax}
Let Assumptions \ref{Ass.A}, \ref{Ass.B},  \ref{Ass.MA},  \ref{Ass.coef1}, \ref{Ass.moments}, and~\ref{AssSpec2} hold. Then, there exists a finite $C>0$ independent of $n, S_1,S_2$ and $T$, such that 
\begin{align} 
 \max_{1\leq i, j \leq n} \sup_{\thbf \in \Thbf} {\rm E} \left| \widehat{\sigma}_{ij}^x(\thbf)  - \sigma_{ij}^x(\thbf) \right|^2  =  
C \max\left\lbrace \frac{(\log B_{S_1} \log B_{S_2}  \log B_{T})^2 B_{S_1} B_{S_2} B_{T}}{S_1 S_2 T},  \frac 1{B_{S_1}^{2\vartheta_1}},\frac 1{B_{S_2}^{2\vartheta_2}},\frac 1{B_T^{2\vartheta_3}} \right\rbrace,\nn
\end{align}
where $\vartheta_1$, $\vartheta_2$, and $\vartheta_3$ are defined in Assumption \ref{AssSpec2}.
\end{theorem}



Our results are nonstandard in the literature on geostatistics: we do not need to choose between in-fill or long-span asymptotic regime and we simply require that both $S_1$ and $S_2$ diverge, so $S\to\infty$. With this regard, we emphasize that our estimator of the spectral density matrix entries as in (\ref{hatSimgax}) bears some similarities with the tapered estimator of the Fourier transform of the covariance matrix of a spatial rf on a lattice proposed by \cite{DK87}. Differently from their method, in our approach  we replace data tapers with kernels. This yields a two-fold advantage: first, it allows to control for the estimation bias of $\sigma_{ij}^x(\thbf)$, taking care of the boundary effects; second, it offers the possibility of using the mentioned flexible asymptotic regime. We refer to \cite{ElMDW13} for a related discussion; see also \cite{DPW17} for similar comments.  

\begin{rem}\label{rem:band}
\upshape{
The rate in Theorem \ref{Thmsigmax} depends on the kernel smoothness $\vartheta_l$, $l=1,2,3$ and the bandwidths $B_{S_1}$, $B_{S_2}$, and $B_T$ (see Assumption \ref{AssSpec2}). Typically the same kernel is used in all dimensions, so we can assume $\vartheta_l= \vartheta_o$ for all $l=1,2,3$.
Consider the case in which $S_1\asymp S_2\asymp T$, then,  up to logarithmic terms, the optimal spatial bandwidths are such that 
$B_{S_1} \asymp S_1^{3/(2\vartheta_o+3)}$, $B_{S_2} \asymp S_2^{3/(2\vartheta_o+3)}$, and $B_{T} \asymp T^{3/(2\vartheta_o+3)}$. 
This implies that the optimal rate of consistency for our estimator of the spectral density matrix is $S_1^{3\vartheta_o/(2\vartheta_o+3)}= S_2^{3\vartheta_o/(2\vartheta_o+3)}= T^{3\vartheta_o/(2\vartheta_o+3)}$.  In our applications we used the Epanechnikov kernel for which $\vartheta_o=2$, hence, the rate of consistency is $S_1^{6/7}=S_2^{6/7}=T^{6/7}$.
In a pure time series model the consistency rate is $T^{\vartheta_o/(2\vartheta_o+1)}$  \citep{barigozzi2022algebraic}, which for a Epanechnikov kernel implies a rate $T^{2/5}$, much slower than what achieved using also the spatial information.
}
\end{rem}
We then prove consistency of the common component estimator $\wh{\chi}_{\ell\vsbf}^{(n)}$ defined in \eqref{finiamo!}

\begin{theorem}\label{Thm.hatChin} 
Let Assumptions \ref{Ass.A}, \ref{Ass.B},  \ref{Ass.MA},  \ref{Ass.coef1}, \ref{Ass.lamChi},  \ref{Ass.moments}, \ref{AssSpec2}, and \ref{Ass.ExpK} hold. 
Define
\[
\alpha_{n,S_1,S_2,T}= \max\l\{\frac1{\sqrt n},  (\log B_{S_1} \log B_{S_2}  \log B_{T})\sqrt{\frac{ B_{S_1} B_{S_2} B_{T}}{S_1 S_2 T}},  \frac 1{B_{S_1}^{\vartheta_1}},\frac 1{B_{S_2}^{\vartheta_2}},\frac 1{B_T^{\vartheta_3}} \r\},
\]
where $\vartheta_1, \vartheta_2$, and $\vartheta_3$ are defined in Assumption \ref{AssSpec2}. Then, 
 there exists finite ${C}, C^*,\widetilde{C}  >0$ independent of $n,S_1,S_2$, and $T$, such that, 
\begin{enumerate}
\item[(i)] 
for any $\vsbf=(s_1\ s_2 \ t)^\top$ with $s_1=1,\ldots, S_1$, $s_2=1,\ldots, S_2$, and $t=1,\ldots, T$, and  for all $\varepsilon>0$,
\begin{align*}
\max_{1\leq \ell \leq n}  {\rm P}\left[ \left\vert \widehat{\chi}^{(n)}_{\ell\vsbf} - \chi_{\ell\vsbf} \right\vert \geq \varepsilon \right]   &\leq \frac{C}{\varepsilon}  \alpha_{n,S_1,S_2,T}
 M_{S_1} M_{S_2} M_T\nn\\
&+ \frac{C^*}{\varepsilon} (1+\varepsilon_1)^{-\kappa^*_1(s_1)}  (1+\varepsilon_2)^{-\kappa^*_2(s_2)} (1+\varepsilon_3)^{-\kappa^*_3(t)},
\end{align*}
with $\kappa^*_1(s_1) = \min\{\vert \underline{\kappa}_1(s_1) -1 \vert, \overline{\kappa}_1(s_1)+1\}$, $\kappa^*_2(s_2) = \min\{\vert \underline{\kappa}_2(s_2) -1 \vert, \overline{\kappa}_2(s_2)+1\}$ and $\kappa^*_3(t) = \min\{\vert \underline{\kappa}_3(t) -1 \vert, \overline{\kappa}_3(t)+1\}$, and where
$M_{S_1}$, $M_{S_2}$, and $M_T$ are defined in Assumption \ref{AssSpec2}, 
$\varepsilon_1,\varepsilon_2$, and $\varepsilon_3$ are defined in Assumption \ref{Ass.ExpK}, and 
$\underline{\kappa}_1(s_1), \overline{\kappa}_1(s_1),  \underline{\kappa}_2(s_2), \overline{\kappa}_2(s_2), \underline{\kappa}_3(t)$, and $\overline{\kappa}_3(t)$ are defined in \eqref{troncamento}.

\item[(ii)]
for any $\vsbf=(s_1\ s_2 \ t)^\top$ with $s_1=M_{S_1},\ldots, S_1-M_{S_1}$ and $s_2=M_{S_2},\ldots,S_2-M_{S_2}$ and $t=M_{T},\ldots, T-M_{T}$, and for all $\varepsilon>0$,
%
\begin{align*}
\max_{1\leq \ell \leq n}  {\rm P}\left[ \left\vert \widehat{\chi}^{(n)}_{\ell\vsbf} - \chi_{\ell\vsbf} \right\vert \geq \varepsilon \right]   &\leq\frac{\widetilde{C}}{ \varepsilon}  
 \alpha_{n,S_1,S_2,T}
 M_{S_1} M_{S_2} M_T.
\end{align*}
\end{enumerate}
\end{theorem}

Theorem~\ref{Thm.hatChin} proves that for consistency of $\widehat{\chi}^{(n)}_{\ell\vsbf}$, the number of lags $M_{S_1}$, $M_{S_2}$, and $M_T$ used in \eqref{def.underlineK} should not be too large, while $n$, $S_1$, $S_2$ and $T$ should all diverge to infinity. As it is clear from Theorem \ref{Mythm5}, we need a large $n$ to disentangle the common and the idiosyncratic components, while from Theorem \ref{Thmsigmax} we see that we need large $S_1$, $S_2$, and $T$  to consistently estimate the spectral density matrix of the observed rf. We remark  that part (i)  yields a rate of convergence also when the spatial locations and the time are close to the boundaries: this aspects has been neglected in the literature on factors models.

\begin{rem}
\upshape{The consistency rate depends on the truncation level we choose when applying the two-sided filter in \eqref{def.underlineK}. When considering the same setting as in Remark \ref{rem:band} so that the consistency rate for the estimate spectral density is $T^{6/7}$, and assuming $M_{S_1}=M_{S_2}=M_T=M$, we need $M= o(T^{2/7})$. 
}
\end{rem}

\section{Determining the number of factors}\label{Sec.Selectq}

An essential aspect for the implementation of the GSTFM is the correct identification of the number of factors $q$. Theorem~\ref{Th. q-DFS} provides a rough guideline for this:  intuitively, one should choose  the value of $q$ such that the $q$-th dynamic eigenvalue should be ``sufficiently large" while the $q+1$-th one should not be ``small". To provide a more precise selection procedure, we define an information criterion (IC) that enables us to estimate $q$ consistently. To this end, 
we propose the use of a criterion which is based on the eigenvalues, $\widehat{\lambda}^x_{nj}(\thbf)$, $j=1,\ldots, n$, of  $\widehat{\boldsymbol \Sigma}_n^x(\thbf)$. 

Letting $p(n, S_1,S_2,T)$ denote a penalty depending on both $n$ and on $S_1,S_2$, and $T$, we consider the information criterion
\begin{equation}
\widehat{\rm IC}^{(n)}(k) =\log\left[ \frac{1}{n} \sum_{j=k+1}^n \frac 1{8\pi ^3}\int_{\thbf\in\Thbf}
\widehat{\lambda}^x_{nj}(\thbf)\mathrm d\thbf\right] + k\, p(n, S_1,S_2,T),\nn
\end{equation}
 and we define the estimator of the number of factors
\beq
\widehat{q}^{(n)} = \arg\!\!\!\!\!\!\min_{0\leq k \leq q_{\max}} \widehat{\rm IC}^{(n)}(k),\label{eq.hatICtext}
\eeq
for some a priori chosen maximum number of factors $q_{\max}$.
We assume the following standard divergence rate of the penalty 
\begin{ass}\label{Ass.pnvs}
As $n , S_1, S_2, T \rightarrow \infty$,
$p(n, S_1, S_2, T) \rightarrow 0$ and 
$$
\min\left\lbrace n, \frac 1{\log B_{S_1} \log B_{S_2} \log B_{T} }\sqrt{\frac{S_1 S_2 T}{B_{S_1} B_{S_2} B_T }}, B_{S_1}^{\vartheta_1}, B_{S_2}^{\vartheta_2}, B_T^{\vartheta_3}\right\rbrace p(n, S_1, S_2, T) \rightarrow \infty.
$$
\end{ass}

Finally, we establish consistency of $\widehat{q}^{(n)}$
\begin{theorem}\label{Prop.qselect_Sample}
Let Assumptions \ref{Ass.A}, \ref{Ass.B},  \ref{Ass.MA},  \ref{Ass.coef1}, \ref{Ass.lamChi}, \ref{Ass.moments}, \ref{AssSpec2}, and  \ref{Ass.pnvs} hold. 
Then, as $n, S_1, S_2, T \rightarrow \infty$, ${\rm P}(\widehat{q}^{(n)} = q) \to 1$.
\end{theorem}


\section{Monte Carlo experiments} \label{Sec: sim}

Before delving into numerical studies,  we summarize the estimation procedure in the following \\
\begin{algorithm}[H] \label{Algm1}
\SetAlgoLined
\KwIn{data $\{x_{\ell\vsbf},\ \ell =1,\ldots, n, \vsbf=(s_1\ s_2\ t)^\top,  s_1=1,\ldots, S_1,  s_2=1,\ldots, S_2, t=1,\ldots, T\}$;  
estimated number of factors $\widehat{q}^{(n)}$ (see Algorithm \ref{Algm2} in Appendix \ref{sec:HLABC}); \linebreak
kernel functions $K_1(\cdot)$, $K_2(\cdot)$, and $K_3(\cdot)$;\linebreak 
bandwidths integers $B_{S_1}, B_{S_2}$, and $B_T$; \linebreak 
truncation integers $M_{S_1}$, $M_{S_2}$, and $M_T$.}
\KwOut{$\{\widehat{\chi}_{\ell\vsbf}\n$, $\ell=1,\ldots,n$, $\vsbf=(s_1\ s_2\ t)^\top,  s_1=1,\ldots, S_1,  s_2=1,\ldots, S_2, t=1,\ldots, T\}$.}

Compute $\widehat{\boldsymbol \Sigma}_n^x(\thbf_{\hbf})$ as in \eqref{hatSimgax}, with $\thbf_{\hbf}$ as in Remark \ref{rem:freq}.

Compute the $\widehat q^{(n)}$ eigenvectors $\widehat{\p}_{nj}^x(\thbf_{\hbf}), j = 1, \ldots, \widehat{q}$, of $\widehat{\boldsymbol \Sigma}_n^x(\thbf_{\hbf})$,  with $\thbf_{\hbf}$ as in Remark \ref{rem:freq}.

Compute $\widehat{\K}^x_{n\ell}(\thbf_{\hbf})$ and $\widehat{\underline{\K}}^{x, \vsbf}_{n\ell}(L)$ as in \eqref{def.underlineK}, with $\thbf_{\hbf}$ as in Remark \ref{rem:freq}.

Compute $\widehat{\chi}^{(n)}_{\ell\vsbf} = \widehat{\underline{\K}}^{x, \vsbf}_{n\ell}(L) \x_{n\vsbf}$ as in \eqref{finiamo!}.

\caption{Algorithm for estimating the common component.}
\end{algorithm}

We  illustrate how Algorithm \ref{Algm1} works and we provide evidence of our key theoretical results. In Section~\ref{Sec: MotEx} we already showed the presence of the eigen-gap in finite-samples as predicted by our results in Section \ref{Sec: repr}, further evidence is available in Appendix \ref{App.sim}; in Section~\ref{Sec.Sim.hatChi}, we study the performance of the estimator of the common component proposed in Section~\ref{Sec: estim}, and we provide a comparison of our GSTFM with the extant GDFM; in Section~\ref{Sec.Sim.Selectq}, we explain how to select the number of factors following Section \ref{Sec.Selectq}.

In the whole section we simulate data using  
$x_{\ell \vsbf}= \chi_{\ell\vsbf} + \xi_{\ell\vsbf}$, for $ \ell = 1,\ldots, n$, $\vsbf=(s_1\ s_2\ t)^\top$ with $s_1=1,\ldots, S_1$, $s_2=1,\ldots, S_2$, and $t=1,\ldots, T$. The case of cross- and serially correlated idiosyncratic components is studied in Appendix \ref{App.sim}.
The idiosyncratic component $\xi_{\ell\vsbf}$  is i.i.d. from a standard normal distribution and the common component $\chi_{\ell\vsbf}$ is generated according to two different mechanisms. 
%
\begin{itemize}
\item[Model (a)] is an infinite convolution over the lattice:
\begin{equation}\label{modelAR}
\begin{aligned}
\chi_{\ell \vsbf} = \sum_{\kbf} \sum_{j=1}^q a_{\ell j} b_{\ell j}^{\vert \kappa_1 \vert + \vert \kappa_2 \vert + \vert \kappa_3  \vert} L^{\kbf} {u}_{j,\vsbf}.
\end{aligned}
\end{equation}
\item[Model (b)] is a finite convolution over the lattice:
\begin{equation}\label{modelMA}
\chi_{\ell \vsbf} = \sum_{\kbf = (-1 \ -1 \ 0)^\top}^{(1 \ 1 \ 1)^\top} \sum_{j=1}^q a_{\ell j} 0.5^{\vert \kappa_1 \vert + \vert \kappa_2 \vert + \vert \kappa_3 \vert}  L^{\kbf} {u}_{j,\vsbf}.
\end{equation}
\end{itemize}
We generate 
 $a_{\ell j}$ and $u_{j,\vsbf}$, $j = 1, \ldots, q,$ from i.i.d. standard normal distributions and $b_{\ell j}$ from i.i.d. uniform distributions on $[0.5, 0.8]$.  The Monte Carlo (MC) experiments are repeated $N = 100$ times.

\subsection{The common component}
\label{Sec.Sim.hatChi}
Section \ref{Sec: estim} contains the asymptotics of the proposed estimation methods. A practically relevant question is related to the finite-sample behaviour of the proposed estimators. To investigate this aspect,  we set $q= 2$ and we study numerically how the mean square error (MSE)
$$E_1 = \frac{1}{n S_1 S_2 T} \sum_{\ell = 1}^n \sum_{s_1 = 1}^{S_1}   \sum_{s_2 = 1}^{S_2} \sum_{t = 1}^{T}    (\widehat{\chi}^{(n)}_{\ell\vsbf}- \chi_{\ell\vsbf})^2$$
and the standardised MSE
$$E_2 = \frac{\sum_{\ell = 1}^n \sum_{s_1 = 1}^{S_1}   \sum_{s_2 = 1}^{S_2} \sum_{t = 1}^{T}    (\widehat{\chi}^{(n)}_{\ell\vsbf}- \chi_{\ell\vsbf})^2}{\sum_{\ell = 1}^n \sum_{s_1 = 1}^{S_1}   \sum_{s_2 = 1}^{S_2} \sum_{t = 1}^{T}    \chi_{\ell\vsbf}^2}
$$
change with $n$ and with the spatio-temporal dimensions $S_1, S_2$ and $T$. 

In the top panel of Table~\ref{Tab.GDRF_diffn}, we display the averaged  (over all MC runs) $E_1$  and $E_2$ for $n = 20, 40, 60, 80$ and $(S_1, S_2, T) = (20, 20, 20)$. The table clearly shows that the estimation errors decrease as $n$ increases: this illustrates the blessing of dimensionality for the estimation of the common component.  Interestingly, we remark that already with $n=20$, $E_1$ and $E_2$ have values that are very similar to the ones obtained for larger sample sizes (e.g. $n=60$).

In the bottom panel of Table~\ref{Tab.GDRF_diffn} we report the averaged (over all MC runs) values of $E_1$ and $E_2$ for $n = 40$ and $(S_1, S_2, T) = (10, 10, 10)d,$  with $d = 1, 2, 3, 4$.  In line with the theoretical results, the errors decrease as the spatio-temporal dimensions  increase.



%

\begin{table}
 \caption{$E_1$ and $E_2$ of the GSTFM, $q=2$.}\label{Tab.GDRF_diffn}
\centering
\begin{tabular}{ c c c c c}  \hline \hline
&\multicolumn{4}{c}{$n$}\\
\cmidrule(lr){2-5}
$(S_1,S_2,T)=(20,20,20)$  & $20$ & $40$ & $60$ & $80$ \\ \hline
  Model (a) in~\eqref{modelAR} &    &   &  &    \\
 $E_1$ &  0.389  & 0.346  & 0.339 & 0.331    \\
 $E_2 $ &  0.066 & 0.060  & 0.059  & 0.058   \\ \hline 
Model (b) in~\eqref{modelMA} &    &   &  &    \\
 $E_1$ &  0.251  &  0.193 & 0.175 & 0.164    \\
 $E_2 $ &  0.047 & 0.036   & 0.031  & 0.030   \\ \hline \hline
%
%
%
&\multicolumn{4}{c}{$(S_1,S_2,T)$}\\
\cmidrule(lr){2-5}
$n=40$  & $ (10, 10, 10)$ & $ (20, 20, 20)$ & $ (30, 30, 30)$ & $ (40, 40, 40)$ \\ \hline
  Model~\eqref{modelAR} &    &   &  &    \\
 $E_1$ &    0.372 & 0.289  & 0.302  &  0.301  \\
 $E_2 $ &    0.077 & 0.051   &  0.050 & 0.046  \\ \hline 
Model~\eqref{modelMA} &    &   &  &    \\
 $E_1$ &    0.196 & 0.115  & 0.146 & 0.118  \\
 $E_2 $ &  0.036 & 0.021  &  0.027 & 0.021   \\ \hline \hline
\end{tabular}
\end{table}

To elaborate on the motivating example of Section \ref{Sec: MotEx}, we compare the performance of the GSTFM and the GDFM in terms of estimation accuracy of the common components. We set $n = 30$, $q = 2$ and $(S_1, S_2, T) = (10, 10, 20)$  or  $(S_1, S_2, T) =(20, 10, 20)$.   In Table~\ref{Tab.GDFM_GDRF}, we report the average  (over all MC runs) values of $E_1$ and $E_2$, for the GSTFM and GDFM. The advantage of our approach is evident:  the GSTFM produces smaller estimation errors of the common components than the GDFM. We emphasize  that $E_1$ of the GDFM displays a sharp rise  as $S_1$ increases from $10$ to $20$. This aspect illustrates that there is no blessing of dimensionality for the GDFM if the spatial dependencies are ignored: adding more time series does not yield any accuracy improvement and the results of \cite{FHLR00} do not apply. Indeed,
when $S_1$  increases,  stacking the new observations in a vector, as  in Section \ref{Sec: MotEx}, implies that we are dealing with a larger number of  spatially dependent variables: the GDFM ignores these spatial dependencies and, as a result, it becomes less reliable in the estimation of the common component, entailing larger values of $E_1$---incidentally, this point is not detectable looking at $E_2$ because of its standardisation based on the variance of the true common component.

\begin{table}
 \caption{$E_1$ and $E_2$ of the GSTFM and GDFM, $n=30$, $S_2=10$, $T=20$, and $q=2$.}\label{Tab.GDFM_GDRF}
 \label{Tab:comp}
\centering
\begin{tabular}{ c |c |c | c |c}  \hline \hline
  & \multicolumn{2}{c |}{Model (a) in~\eqref{modelAR}} &   \multicolumn{2}{c}{Model (b) in~\eqref{modelMA}}  \\  \hline
  & GSTFM & GDFM & GSTFM & GDFM \\ \hline
  $S_1 = 10$ &   &  & &  \\
 $E_1$ &   0.606 &  1.632 & 0.566 & 2.773  \\
 $E_2 $ &  0.325 &  0.807 & 0.149 & 0.709  \\ \hline
 $S_1 = 20$ &   &  & &  \\
 $E_1$ &  0.817 & 4.202  & 0.470 & 4.160  \\
 $E_2 $ &  0.150 & 0.747 & 0.085 & 0.731 \\ \hline \hline
\end{tabular}
\end{table}

\subsection{Selection of the number of factors }\label{Sec.Sim.Selectq}

We investigate the finite sample performance of the estimator of $q$ defined in Section~\ref{Sec.Selectq}. However, looking at \eqref{eq.hatICtext}, we remark that, if the estimator $\widehat{q}^{(n)}$ is consistent, then the  estimator $\widehat{q}_{c}^{(n)}$ obtained via the penalty $cp(n, S_1,S_2,T)$, $c> 0$, is consistent as well. Hence, in practice, one needs to choose also 
$c$ to estimate $q$ consistently. The detailed procedure for the automatic selection of the number of factors is summarized in Algorithm~\ref{algorithm.Estq} in Appendix \ref{sec:HLABC}.
To evaluate the estimation accuracy of Algorithm~\ref{algorithm.Estq}, 
we set  $n = 100$, $(S_1, S_2, T) = (25, 25, 25)$, and
$q = 0, 1, 2, 3$ and we run $200$ MC replications. Table~\ref{Tab.selectq_AR_MA} shows the under- and over-identification proportions for $\widehat{q}_{\widehat{c}}^{(n)}$. The results illustrate good finite-sample performance of the selection procedure of $q$: for Model (b) in \eqref{modelMA}, the algorithm identifies $q$ correctly for all replications and for all values of $q$; for Model (a) in \eqref{modelAR}, the over-identification rate is not zero for $q=1, 2, 3$ but it is nevertheless very small. 
  
\begin{table}
\caption{Under- and over-identification rates for $\widehat{q}_{\widehat{c}}^{(n)}$, with $q = 0, 1, 2, 3$.}\label{Tab.selectq_AR_MA}
\centering
\begin{tabular}{ c c c c c }  \hline \hline
  & $q = 0$ & $q = 1$ & $q = 2$ & $q = 3$  \\  \hline
  Model (a) in~\eqref{modelAR} &  &  &   &      \\
 Under-identification & 0   & 0 & 0  &   0   \\
 Over-identification &  0 &  0.10 &    0.08 &0.04   \\ \hline 
Model (b) in~\eqref{modelMA} &   & &   &      \\
 Under-identification &  0 &  0 & 0 & 0   \\
 Over-identification &  0& 0  & 0   &  0   \\ \hline \hline
\end{tabular}
\end{table}
 
\section{Conclusions and further developments} \label{Sec:conc}
{We develop the theory and provide the complete inference toolkit (estimation of the common component and selection of the number of factors)  for the factor analysis of high-dimensional
spatio-temporal rf defined on a  lattice. Our model accounts for all spatio-temporal common correlations among all components of the rf. We give statistical guarantees of the proposed estimation methods. 
Our asymptotic theory extends the one available in \cite{FHLR00}, whose rates of convergence, which are unavailable in the literature on factor models for time series, can be derived as a special case of our rates in Section \ref{Sec: estim}. Monte Carlo studies illustrate the applicability and the good performance 
of our GSTFM under many different settings, commonly encountered in data analysis.} 

We foresee some 
extensions of our results. For instance, one may define estimators of the common component which involve one-sided filters in time, thus allowing for forecasting. We conjecture that this  is possible along the lines of \citep{forni2005generalized,FHLZ17}. Nevertheless, such extensions cannot be directly obtained within the setting of this paper: they require further assumptions and more involved estimation steps, whose statistical guarantees need to be derived. Therefore, we leave them for further research.

\begin{funding}
%
The third author was supported by grants  WK2040000055 and YD2040002016.
\end{funding}



\bibliographystyle{imsart-nameyear} 
\bibliography{spacefactor}       

\clearpage
\begin{supplement}
\appendix
\setcounter{equation}{0}
\renewcommand{\theequation}{\thesection.\arabic{equation}}

\small
\section{Preparatory results on infinite dimensional random fields} \label{App1_Def}

\subsection{Lag operator in space-time}

\begin{lemma} \label{Lemma1} 
Under Assumption \ref{Ass.A}, for any $j=1,2,3$, the operator $L_j$ in \eqref{Eq: operators} can be straightforwardly and uniquely extended to $L_j: \bm{\mathcal X}\to \bm{\mathcal X}$ which is well defined, preserves the inner product and is onto. Thus, $L_j$ is a unitary operator. 
\end{lemma}

\begin{proof}
Let $\{\alpha_{\vsbf}\}$ and $\{ \beta_{\vsbf}\}$ be sequences of complex numbers. First, we show that the operator $L_1$ is well-defined. To this end, for any $\ell\in\mathbb N$, we notice that for each $\vsbf=(s_1\ s_2 \ t)^\top \in \mathbb{Z}^3$, the linearity of the operator yields, for a finite linear combination,
\begin{eqnarray*}
L_1 \left[\sum_{\vsbf} \alpha_{\vsbf} x_{\ell \vsbf} \right] = \sum_{\vsbf} \alpha_{\vsbf} x_{\ell (s_1-1 \ s_2 \ t)}. 
\end{eqnarray*}
Let us recall that $L_2(\mathcal P,\mathbb C)$ is a Hilbert space with inner product $\langle x_{i\vsbf},x_{j\vsbf'}\rangle=\Cov(x_{i\vsbf},x_{j\vsbf'})$ and $\Vert{x_{i\vsbf}}\Vert^2=\langle x_{i\vsbf},x_{i\vsbf}\rangle$.
Now, we remark that if 
$$
\sum_{\vsbf} \alpha_{\vsbf} x_{\ell \vsbf} = \sum_{\vsbf} \beta_{\vsbf} x_{\ell \vsbf},
$$
then we have
\begin{align}
&\left \Vert L_1 \left[\sum_{\vsbf} \alpha_{\vsbf} x_{\ell \vsbf} \right] - L_1 \left[\sum_{\vsbf} \beta_{\vsbf} x_{\ell \vsbf} \right] \right \Vert^2 =  \left \Vert \sum_{\vsbf} \alpha_{\vsbf} x_{\ell (s_1-1 \ s_2 \ t)} - \sum_{\vsbf} \beta_{\vsbf} x_{\ell (s_1-1 \ s_2 \ t)} \right \Vert^2  \nonumber \\
&=\sum_{\vsbf} \sum_{\vsbf'} (\alpha_{\vsbf} - \beta_{\vsbf}) {(\alpha_{\vsbf'}^\dag - \beta_{\vsbf'}^\dag)} \left \langle x_{\ell (s_1-1 \ s_2 \ t)}, x_{\ell (s'_1-1 \ s'_2 \ t')}  \right \rangle \nonumber \\
&= \sum_{\vsbf} \sum_{\vsbf'} (\alpha_{\vsbf} - \beta_{\vsbf}) {(\alpha_{\vsbf'}^\dag - \beta_{\vsbf'}^\dag)} \left \langle x_{\ell \vsbf}, x_{\ell \vsbf'}  \right \rangle \nn\\
&  =  \left \Vert \sum_{\vsbf} \alpha_{\vsbf} x_{\ell \vsbf} -  \sum_{\vsbf'} \beta_{\vsbf} x_{\ell \vsbf} \right \Vert^2 =  0, \nn 
 \end{align}
where, in the third line, we made use of the homostationarity as in Assumption \ref{Ass.A}. Thus, the operator is well defined.
Now, we show that $L_1$ preserves the inner product. For any $\ell\in\mathbb N$, consider
\begin{align}
&\left \langle L_1 \left[\sum_{\vsbf} \alpha_{\vsbf} x_{\ell \vsbf} \right], L_1 \left[\sum_{\vsbf} \beta_{\vsbf} x_{\ell\vsbf} \right] \right \rangle =  \left \langle \sum_{\vsbf} \alpha_{\vsbf} x_{\ell (s_1-1 \ s_2 \ t)}, \sum_{\vsbf} \beta_{\vsbf} x_{\ell (s_1-1 \ s_2 \ t)}  \right \rangle \nn\\
&=  \sum_{\vsbf}   \sum_{\vsbf'} \alpha_{\vsbf} \beta_{\vsbf'}^\dag \left \langle  x_{\ell (s_1-1 \ s_2 \ t)}, x_{\ell (s'_1-1 \ s'_2 \ t')}  \right \rangle 
=  \sum_{\vsbf}   \sum_{\vsbf'} \alpha_{\vsbf} \beta_{\vsbf'}^\dag \left \langle  x_{\ell \vsbf}, x_{\ell \vsbf'}  \right \rangle \nonumber \\
&= \left \langle  \sum_{\vsbf} \alpha_{\vsbf}  x_{\ell \vsbf} , \sum_{\vsbf'}  {\beta_{\vsbf'}}   x_{\ell \vsbf'}  \right \rangle, \nn
\end{align}
where in the fourth equality we made use of homostationarity as in Assumption \ref{Ass.A}. As a consequence $L_1$ is bounded. It is straightforward to show that $L_1$ is onto, therefore it is unitary, {which implies that its inverse $L^{-1}_1$ coincides with the adjoint operator, say  $L_1^\dag$, and $L_1^\dag L_1 = I$, where $I$ is the identity operator.} 

The above properties hold for $L_1$ applied just to one $x_{\ell \vsbf}$, but, since $L_1$ is linear the same hold when applying $L_1$ to elements of $\cup_{n=1}^{\infty} \bm{\mathcal X}_n$. We can then extend $L_1$ to $\bm{\mathcal X}$ as follows. Let $\zeta \in \bm{\mathcal X}$ then there must exist a sequence $\{\zeta_n\}_{n\in\mathbb N}$ with $\zeta_n\in\bm{\mathcal X}_n$
such that
$
\lim_{n\to\infty}\big \Vert\zeta - \zeta_n\big\Vert=0.
$ 
Therefore, $\{\zeta_n\}_{n\in\mathbb N}$ is a Cauchy sequence and since $L_1$ preserves the norm we have that $\{L_1 \zeta_n\}_{n\in\mathbb N}$ is also a Cauchy sequence and must converge to an element $\eta\in\bm{\mathcal X}$, i.e. $
\lim_{n\to\infty}\big \Vert\eta - L_1\zeta_n\big\Vert=0,
$
and it must be that $\eta=L_1\zeta$ in order for $L_1$ to be norm preserving. Using the linearity and boundedness of the operator and continuity of the inner product, it is then easy to show that the operator $L_1$ extended in this way to $\bm{\mathcal X}$ is still linear, well defined and it preserves inner product. We {still} denote by $L_1$ the extended operator. Finally, note that $L_1$ is unitary on $\cup_{n=1}^{\infty} \bm{\mathcal X}_n$, we have { $\zeta_n=L_1^\dag L_1 \zeta_n$} and therefore for any $\zeta \in \bm{\mathcal X}$ there exists an $\eta \in \bm{\mathcal X}$ such that $L_1\eta=\zeta$, which shows that $L_1$  {extended to $\bm{\mathcal X}$} is onto.

%

Moving along the same lines of this proof, one can verify that the same reasoning applies also for the operators $L_2$  and $L_3$. That concludes the proof. 
\end{proof}

\subsection{Canonical isomorphism}\label{app:cane}
Consider generic infinite dimensional row vectors of functions $\boldsymbol{f}=(f_1\ f_2 \  \cdots f_\ell \cdots)$ such that $f_\ell:\Thbf\to\mathbb C$ is measurable for all $\ell\in\mathbb N$ and with $n$ dimensional row sub-vectors $\boldsymbol f_n=(f_1\cdots f_n)$. We define the complex linear spaces
\begin{enumerate} 
\item [(i)] $L^{\infty}_2(\Thbf,\mathbb{C},\Sigmabf^x)$ of all $\bm f$ such that  
$\Vert\bm f\Vert_{\Sigmabf^x}=\lim_{n\to\infty}\Vert \bm f_n\Vert_{\Sigmabf^x}
<\infty$, where $\Vert \boldsymbol{f} \Vert_{\Sigmabf^{x}} = \sqrt{\langle \boldsymbol{f}, \boldsymbol{f}\rangle_{\Sigmabf^{x}}}$ with 
the inner product is given by
$ \langle \boldsymbol{f},\boldsymbol{g}\rangle_{\Sigmabf^{x}} = 
\int_{\Thbf} \boldsymbol f(\thbf)\Sigmabf^{x} (\thbf) {\boldsymbol g}^\dag(\thbf) \mathrm d\thbf /{8\pi^3}$. 
\item [(ii)] $L^{\infty}_2(\Thbf,\mathbb{C})\equiv L^{\infty}_2(\Thbf,\mathbb{C}, \mathbf{I})$
where $\mbf I$ is the infinite dimensional identity matrix (namely, the matrix having $\mathbf I_{n}$ as the $n\times n$ top-left sub-matrix). On $L^{\infty}_2(\Thbf,\mathbb{C})$, the inner product and the norm are indicated by $\langle \boldsymbol{f},\boldsymbol{g} \rangle$ and $\Vert \boldsymbol{f} \Vert$, respectively. 

\item [(iii)] $L^n_2(\Thbf,\mathbb{C}, \Sgmbf)$ and $L^n_2(\Thbf,\mathbb{C})$ as $L^{\infty}_2(\Thbf,\mathbb{C}, \Sigmabf^x) $ and $L^{\infty}_2(\Thbf,\mathbb{C})$, but with the $n$-dimensional vectors $\bm f_n$ instead of $\bm f$.

\item [(iv)] $L_{\infty}^n(\Thbf,\mathbb{C})$ of all $\bm f_n$ such that $g=\Vert{\boldsymbol f_n}\Vert$ is essentially bounded, i.e., $\text{\upshape ess} \sup(g) < \infty $, where $\text{\upshape ess} \sup(g)=\inf\{M:\mathcal L[y:g(y)>M]=0\}$.
\end{enumerate}

Notice that $L^{\infty}_2(\Thbf,\mathbb{C},\Sigmabf^x)$, $L^{n}_2(\Thbf,\mathbb{C},\Sigmabf^x)$, $L^{\infty}_2(\Thbf,\mathbb{C})$, and $L^{n}_2(\Thbf,\mathbb{C})$ are Hilbert spaces.

\begin{lemma} \label{Lemma2} 
Under Assumptions \ref{Ass.A} and  \ref{Ass.B}, the map $\mathcal{J}$ in \eqref{JJ} can be straightforwardly and uniquely extended to a map $\mathcal J:L^{\infty}_2(\Thbf,\mathbb{C},\Sigmabf^x) \to \bm{\mathcal X}$ which is well defined, preserves the inner product and is one-to-one. Thus, $\mathcal J$ is an isomorphism. 
Moreover,
let $L_2^{n*}(\Thbf, \mathbb C,\Sigmabf^x)=\{\bm f\in L_2^\infty(\Thbf, \mathbb C,\Sigmabf^x), f_j=0, j> n\}$. For any $\vsbf\in\mathbb Z^3$ and any $n \in\mathbb N$, and for a given integer $\ell \le n$, 
define the mapping $\mathcal{J}^*: \cup_{n=1}^\infty L_2^{n*}(\Thbf, \mathbb C,\Sigmabf^x) \to \cup_{n=1}^\infty \bm{\mathcal X}_n$ 
as:
\begin{equation}
\mathcal{J}^*  \left[ (\delta_{\ell1} \ \cdots \delta_{\ell k} \cdots \delta_{\ell n} ) e^{i \langle \vsbf , \cdot\rangle}\right] = x_{\ell \vsbf},\nn
\end{equation}
The map $\mathcal{J}^*$ can be extended in a unique way to $\mathcal J:L^{\infty}_2(\Thbf,\mathbb{C},\Sigmabf^x) \to \bm{\mathcal X}$ which is an isomorphism.  
\end{lemma}

\begin{proof}
We show that $\mathcal{J}$ is an isometric
isomorphism since it is  onto,  it preserves the inner product and it is one-to-one. 

Let us set $\ell \in \mathbb{N}$,  
$\vsbf \in\mathbb{Z}^3$ and $\vsbf'  \in\mathbb{Z}^3$. 
From (\ref{JJ}), we have $
\mathcal{J}  \left[(\delta_{\ell1} \  \cdots \ \delta_{\ell n} )  e^{i\langle\vsbf,\cdot\rangle} \right] = x_{\ell \vsbf}$, by linearity and for a sequence of complex numbers $\{\alpha_{\vsbf}\}$, we have  
$$
\mathcal{J}  \left[ \sum_{\vsbf} \alpha_{\vsbf}  (\delta_{\ell1} \  \cdots \ \delta_{\ell n} ) e^{i\langle\vsbf,\cdot\rangle} \right] = \sum_{\vsbf} \alpha_{\vsbf} x_{\ell \vsbf},
$$
which is an onto linear mapping  between two spaces: the collection of finite combinations of $\{(\delta_{\ell1} \  \cdots \ \delta_{\ell n} )  e^{i\langle\vsbf,\cdot\rangle}: \vsbf \in \mathbb{Z}^3 \}$ and all finite linear combinations of $\{x_{\ell \vsbf}: \vsbf \in \mathbb{Z}^3 \}$.  We remark that these two spaces are dense manifolds of $L^{\infty}_2(\Thbf,\mathbb{C},
\Sigmabf^x)$ and $\boldsymbol{\mathcal X}$, respectively.

To show that $\mathcal{J}$ preserves the inner product, let us consider the complex sequences $\{\alpha_{\vsbf} \}$
and $\{\beta_{\vsbf'} \}$ and define the $n$ dimensional vector $\boldsymbol{A}_{n\vsbf}=(0 \  \ \cdots \ \alpha_{\vsbf} \  \cdots \ 0)^\top$, which has all entries equal to zero, but the $\ell$-th entry which is equal to  $\alpha_{\vsbf}$. Similarly, we define the $n$ dimensional vector $\boldsymbol{B}_{n\vsbf'}=(0 \  \ \cdots \ \beta_{\vsbf'} \  \cdots \ 0)^\top$, which has all entries equal to zero, but the $\ell$-th entry which is equal to  $\beta_{\vsbf'}$. Then, we have
\begin{eqnarray}
& & \left \langle \mathcal{J} \left[\sum_{\vsbf} \alpha_{\vsbf} (\delta_{\ell1} \  \cdots \ \delta_{\ell n} ) e^{i\langle\vsbf,\cdot\rangle} \right], \mathcal{J} \left[\sum_{\vsbf} \beta_{\vsbf} (\delta_{\ell1} \  \cdots \ \delta_{\ell n} ) e^{i\langle\vsbf,\cdot\rangle} \right] \right \rangle 
\nonumber \\
&=&  \left \langle \sum_{\vsbf} \alpha_{\vsbf} {x}_{\ell \vsbf}, \sum_{\vsbf} \beta_{\vsbf} {x}_{\ell \vsbf}  \right \rangle 
=  \sum_{\vsbf}   \sum_{\vsbf'} \alpha_{\vsbf} {\beta^{\dag}_{\vsbf'}} \left \langle  {x}_{\ell \vsbf}, {x}_{\ell \vsbf'}  \right \rangle  \nonumber \\
&=&  \sum_{\vsbf}   \sum_{\vsbf'} \boldsymbol{A}_{n\vsbf} \bm\Gamma^{x}_n (\vsbf-\vsbf') {\boldsymbol{B}_{n \vsbf'}^{\dag}}  \nonumber \\
&=&   \frac{1}{8\pi^3} \sum_{\vsbf}   \sum_{\vsbf'} \boldsymbol{A}_{n\vsbf} \left[ \int_{\Thbf}  
e^{i\langle\vsbf-\vsbf',\thbf\rangle}\Sgmbf(\thbf)   d\thbf \right] {\boldsymbol{B}_{n \vsbf'}^{\dag}} \nonumber \\
&=& \left \langle  \sum_{\vsbf} \alpha_{\vsbf}   e^{i \langle\vsbf , \thbf\rangle} , \sum_{\vsbf}  {\beta_{\vsbf}}   e^{i \langle\vsbf , \thbf\rangle}  \right \rangle_{\Sgmbf(\thbf)}, 
\label{Eq. innerprod_iso} 
\end{eqnarray}
which implies that $\mathcal{J}$ is a one-to-one mapping and it preserves the inner product. 
The above properties hold for $\mathcal{J}$ applied just to  $(\delta_{\ell1} \  \cdots \ \delta_{\ell n} ) e^{i \langle\vsbf , \cdot\rangle}$, which yields $x_{\ell \vsbf}$. Nevertheless,  $\mathcal{J}$ is linear thus the same properties hold when applying it to get all the elements of $\bm{\mathcal X}$. 

We have shown that $\mathcal{J}$ is an isometric isomorphism between two dense linear manifolds. Lemma 4.1 in \citet[p.14]{Roz67} implies that one can always extend the isomorphism to the closed linear manifolds generated by 
these manifolds. With a slight abuse of notation, we call $\mathcal{J}$ the extended isomorphism such that $\mathcal{J}: L^{\infty}_2(\Thbf,\mathbb{C},
\Sigmabf^x) \to \boldsymbol{\mathcal X}$. 
That concludes the proof. 
\end{proof}

\begin{rem}\label{rem:Jinv}
\upshape{
A consequence of Lemma \ref{Lemma2} is that 
also the inverse mapping of $\mathcal{J}$, let us call it  $\mathcal{J}^{-1}: \bm{\mathcal X} \to L^{\infty}_2(\Thbf,\mathbb{C},\Sigmabf^x)$, is an  isomorphism, such that for $\vsbf\in\mathbb Z^3$, any $n\in\mathbb N$, and any $\ell\le n$,
\beq
\mathcal{J}^{-1}\left[x_{\ell\vsbf}\right]= (\delta_{\ell1} \ \cdots \delta_{\ell k} \cdots \delta_{\ell n})e^{i\langle \vsbf,\cdot\rangle}.
\eeq
Moreover, Lemma \ref{Lemma2}  implies that for any $n \in \mathbb{N}$, the process $\boldsymbol{x}_n$ is harmonizable, namely 
for any $\vsbf\in\mathbb Z^3$, we can write
\begin{equation}
\boldsymbol{x}_{n \vsbf} =\int_{\Thbf} e^{i\langle \vsbf,\thbf\rangle} \boldsymbol{\mathcal{M}}_{n}(\thbf) \mathrm d\thbf,\qquad \text{w.p. 1}, \label{Eq. SpecRapp}
\end{equation}
{where $\boldsymbol{\mathcal{M}}_{n}$ is a complex random measure on (the Borel $\sigma$-field of) $\Thbf$, which is such that  
$\boldsymbol{\mathcal{M}}_{n}(\boldsymbol{\Delta_1} \cup \boldsymbol{\Delta_2})=\boldsymbol{\mathcal{M}}_{n}(\boldsymbol{\Delta_1})+\boldsymbol{\mathcal{M}}_{n}(\boldsymbol{\Delta_2})$, for disjoint Borel sets $\boldsymbol{\Delta_1}$ and $\boldsymbol{\Delta_2}$. In addition, for any Borel set $\boldsymbol{\Delta}$, we have that $\E(\boldsymbol{\mathcal{M}}_{n}(\boldsymbol{\Delta}))=\boldsymbol 0_n$, $\E(\vert\boldsymbol{\mathcal{M}}_{n}(\boldsymbol{\Delta})\vert^2)=\int_{\boldsymbol\Delta}\Sgmbf(d\thbf)/8\pi^3 $ and $\E(\boldsymbol{\mathcal{M}}_{n}(\boldsymbol{\Delta_1})\boldsymbol{\mathcal{M}}^{\dag}_{n}(\boldsymbol{\Delta_2}))=\boldsymbol 0_{n\times n}$, for all disjoint  Borel sets 
$\boldsymbol{\Delta_1}$ and $\boldsymbol{\Delta_2}$; see, e.g., \citet{S12}, p.21 or \citet{CW15} Ch.6, section 6.1.6, for book-length discussions.} The result in (\ref{Eq. SpecRapp}) is analogous to the standard spectral representation of time series, where $\boldsymbol{\mathcal{M}}_{n}$ is the measure related to an orthogonal increment process; see e.g. \citet{BD06}, Ch.4. 
}
\end{rem}

\begin{rem}\label{rem:spfilt}
\upshape{
To understand the importance of linear filter, notice that any scalar rf of the form $y_{\vsbf}=\boldsymbol{\underline{b}}(L)\boldsymbol{x}_{\vsbf}$ is such that $y\in\bm{\mathcal X}$ and therefore is co-homostationary with $x_\ell$ for any $\ell\in\mathbb N$. 
Moreover, $y$ has a scalar spectral density $\sigma^y(\thbf)$. In particular, recalling the definition of inner product in the Hilbert space $L_2(\mathcal P,\mathbb C)$, we have
by Lemma \ref{Lemma2} and \eqref{fFF}, for any $\vsbf, \hbf\in\mathbb Z^3$ 
\begin{align}
\E(y_{\vsbf} y_{\vsbf-\hbf}^\dag)&=\langle y_{\vsbf}, y_{\vsbf-\hbf}\rangle =\langle \mathcal J^{-1}[y_{\vsbf}], \mathcal J^{-1}[y_{\vsbf-\hbf}]\rangle_{\Sigmabf^x} \nn\\
&=\frac 1{8\pi^3}\int_{\Thbf} 
e^{i\langle \vsbf,\thbf\rangle}
\l[
\l(
\sum_{\kbf}  \mbf b_{\kbf}^{} 
e^{-i\langle \kbf,\thbf\rangle}
\r) 
\Sigmabf^x(\thbf)
  \l(
  \sum_{\kbf }  \mbf b_{\kbf}^{} 
  e^{-i\langle \kbf,\thbf\rangle}
  \r)^\dag\r]
  e^{-i\langle \vsbf-\hbf,\thbf\rangle}
  \mathrm d\thbf \nn\\
  &=\frac 1{8\pi^3}\int_{\Thbf}  
\l[\l(\sum_{\kbf }  \mbf b_{\kbf}^{} e^{-i\langle \kbf,\thbf\rangle}\r) 
\Sigmabf^x(\thbf)
  \l(\sum_{\kbf }  \mbf b_{\kbf}^{} e^{-i\langle \kbf,\thbf\rangle}\r)^\dag\r]
  e^{i\langle \hbf,\thbf\rangle}\mathrm d\thbf \nn\\
  &=\frac 1{8\pi^3}\int_{\Thbf} \sigma^y(\thbf)
  e^{i\langle \hbf,\thbf\rangle}\mathrm d\thbf, \nn
\end{align}
which shows that, for any $\thbf\in\Thbf$, the spectral density of $y$ is
\[
\sigma^y(\thbf) =  
\l(\sum_{\kbf }  \mbf b_{\kbf}^{} e^{-i\langle \kbf,\thbf\rangle}\r) 
\Sigmabf^x(\thbf)
  \l(\sum_{\kbf }  \mbf b_{\kbf}^{} e^{-i\langle \kbf,\thbf\rangle}\r)^\dag=
\l(\sum_{\kbf }  \mbf b_{\kbf}^{} e^{-i\langle \kbf,\thbf\rangle}\r) 
\Sigmabf^x(\thbf)
  \l(\sum_{\kbf }  \mbf b_{\kbf}^{} e^{i\langle \kbf,\thbf\rangle}\r).
\]
The same reasoning can be easily generalized to the case of $m$-dimensional linear filters, for any $m\in\mathbb N$, defining filtered rf as $\boldsymbol y_{m\vsbf}=\underline{\boldsymbol B}(L)\boldsymbol x_{\vsbf} = \sum_{\kbf}  \boldsymbol B_{\kbf} L^{\kbf} \boldsymbol x_{\vsbf}$, where $\boldsymbol B_{\kbf}$ is a complex matrix with $m$ rows and infinite columns and such that $\boldsymbol{y}_{m\vsbf}\in\bm{\mathcal X}$ for any $\vsbf\in\mathbb Z^3$. For $j=1,\ldots, m$, each component $ y_{j\vsbf}=\underline{\boldsymbol b}_j(L)\boldsymbol x_{\vsbf}$ satisfies  the above properties. 
}
\end{rem}

\begin{rem}
\upshape{If $\boldsymbol x_n$ is a white noise rf  for any $n\in\mathbb N$, then it has a constant spectral density matrix.}
\end{rem}

\section{Proof of results of Section 4 }\label{dimostraloseriesci} 

\subsection{Weyl's inequality} \label{App:Weyl}

For the sake of completeness, we recall some properties of the eigenvalues of Hermitian nonnegative definite matrices, that go under the name of Weyl's inequality.
\begin{itemize}
\item[(a)] Let $\mathbf{D}$ and $\mathbf{E}$ be $m \times m$ Hermitian nonnegative definite and $\bm{F}=\bm{D}+\bm{E}$. Then
$$
\nu_s (\bm F) \leq \nu_s (\bm D) +\nu_1 (\bm E), \quad \nu_s (\bm F) \leq \nu_1 (\bm D) +\nu_s (\bm E), \quad \nu_s (\bm F) \geq \nu_s (\bm D), \quad \nu_s (\bm F) \geq \nu_s (\bm E)
$$
for any $s=1, \ldots, m$.
\item[(b)] Let $\mathbf{D}$ be as in (a) and let $\mathbf{G}$ be the top-left $(m-1) \times(m-1)$ submatrix of $\bm {D}$. Then $\nu_s (\bm D) \geq \nu_s (\bm G)$ for $s=1, \ldots, m-1$.
\end{itemize}
%


\subsection{Proof of Theorem \ref{Th. q-DFS} - sufficient condition}\label{app:suffFL}

The proof of the sufficient condition is based on a series of intermediate results. In the sequel, for ease of notation, when we write matrix products $\boldsymbol{A}\boldsymbol{B}$ but the number of columns of $\boldsymbol{A}$ is smaller than the number of rows of $\boldsymbol{B}$ we mean that $\boldsymbol{A}$ has been augmented with columns of zeros to match the number of rows of $\boldsymbol{B}$. For example, for $m < n$, $\underline{\boldsymbol{P}}_m(L)\boldsymbol{x}_{n\vsbf}$ means nothing but $\underline{\boldsymbol{P}}_m(L)\boldsymbol{x}_{m\vsbf}$.

\subsubsection{Construction of $q$-dimensional orthonormal white noise rf}
We start proving that there exists a converging sequence of $q$-dimensional orthonormal white noise rf which belongs to $\bm{\mathcal X}$.
For a given integer $q\leq n$ and any $\thbf\in\bm\Theta$, let $\boldsymbol{P}_n(\thbf)=(\pbf_{n1}^{x\top}(\thbf)\cdots\pbf_{nq}^{x\top}(\thbf))^\top$ be the $q\times n$ matrix having as rows the normalized eigenvectors of $\Sigmabf_n^x(\bm\theta)$ corresponding to the $q$ largest eigenvectors. Similarly let $\boldsymbol{Q}_n(\thbf)=(\pbf_{n q+1}^{x\top}(\thbf)\cdots\pbf_{nn}^{x\top}(\thbf))^\top$ which is $(n-q)\times n$. 
Define also $\boldsymbol{\Lambda}_n(\thbf)$ as the $q \times q$ diagonal matrix containing on the diagonal the largest $q$ eigenvalues $\lambda^{x}_{nj}(\thbf)$, $j=1,\ldots,q$ and denote by $\boldsymbol{\Phi}_n(\thbf)$ the $(n-q) \times (n-q)$ diagonal matrix containing on the diagonal the remaining eigenvalues $\lambda^{x}_{nj}(\thbf)$, $j=q+1,\ldots,n$. The spectral decomposition yields 
\begin{eqnarray}
\Sgmbf (\thbf)&=& {\boldsymbol{P}}^\dag_n(\thbf) \boldsymbol{\Lambda}_n(\thbf) \boldsymbol{P}_n(\thbf) + {\boldsymbol{Q}}^\dag_n(\thbf) \boldsymbol{\Phi}_n(\thbf)  \boldsymbol{Q}_n(\thbf), \nonumber \\
\mbf {I}_n &=& {\boldsymbol{P}}^\dag_n(\thbf) {\boldsymbol{P}_n}(\thbf) + {\boldsymbol{Q}}^\dag_n(\thbf) {\boldsymbol{Q}}_n(\thbf) \label{Eq. IPQ},\\
\mbf I_q &=& {\boldsymbol{P}_n}(\thbf){\boldsymbol{P}^\dag_n}(\thbf).\nonumber
\end{eqnarray}

Now, let $ \underline{\boldsymbol{P}}_n(L)$
and
$ \underline{\boldsymbol{\Lambda}}^{-1/2}_n (L)$
be the spatio-temporal linear filters built as in \eqref{fFF} from $\boldsymbol{P}_n(\thbf) $ and $ \boldsymbol{\Lambda}^{-1/2}_n(\thbf)$, respectively. Notice that since $\Sgmbf(\bm\theta)$ is positive definite by construction, then  $\boldsymbol{\Lambda}_n(\thbf)$ is bounded away from zero  for all $n\in\mathbb N$ and all $\bm\theta\in\bm\Theta$ and also bounded for any fixed $n$. So, recall the notation in \eqref{pachistrani}, 
\begin{equation}
{\bfpsi}^n_{\vsbf} = \left( \psi^n_{1 \vsbf}\cdots\psi^n_{q \vsbf} \right)^\top = \underline{\boldsymbol{\Lambda}}^{-1/2}_n (L) \star \underline{\boldsymbol{P}}_n(L)\boldsymbol{x}_{n\vsbf} \label{B10}
\end{equation}
is well defined and it is a $q$-dimensional orthonormal white noise rf since by \eqref{Eq. IPQ} its spectral density is (see Remark \ref{rem:spfilt}): $\bm\Sigma^{\bfpsi^n}(\thbf)=\mbf I_q$.

For $\boldsymbol{M}\subset \Thbf$, let $\boldsymbol{K}_M\subset L^{q \times q}_\infty(\Thbf,\mathbb{C})$ whose elements $\boldsymbol{C}$ are $q\times q$ complex matrices with  elements being  functions defined on $\bm\Theta$ and such that: (a) $\boldsymbol{C}(\thbf) = \boldsymbol{0}$ for $\thbf \notin \boldsymbol{M}$, (b) $\boldsymbol{C}(\thbf) {\boldsymbol{C}}^\dag(\thbf)= \mbf {I}_q$ for $\thbf \in \boldsymbol{M}$. 
For any $m,n\in\mathbb N$, let $\boldsymbol{C} \in \boldsymbol{K}_M$, so that $\underline{\boldsymbol{C}} (L){\bfpsi}^m_{\vsbf} \in \boldsymbol{\mathcal X}$ and it is still a $q$-dimensional orthonormal white noise rf.  Let us consider the orthogonal projection of this new white noise on the space $\spn(\{\psi^n_{j \vsbf}, j = 1, \ldots, q, \vsbf \in \mathbb{Z}^3\})$.
To this end, note that (\ref{Eq. IPQ}) yields $\boldsymbol{x}_{n\vsbf} = \underline{{\boldsymbol{P}}}^\dag_n (L)\star \underline{{\boldsymbol{P}}}_n(L) \boldsymbol{x}_{n\vsbf} + \underline{{\boldsymbol{Q}}}^\dag_n(L)\star \underline{{\boldsymbol{Q}}}_n(L) \boldsymbol{x}_{n\vsbf}
$, thus
\begin{equation}
\boldsymbol{x}_{n\vsbf} 
= \underline{{\boldsymbol{P}}}^\dag_n (L)\star \underline{\boldsymbol{\Lambda}}^{1/2}_n (L) {\bfpsi}^n_{\vsbf} +  \underline{{\boldsymbol{Q}}}^\dag_n(L)\star \underline{{\boldsymbol{Q}}}_n(L) \boldsymbol{x}_{n\vsbf}, \label{Eq.xPQ} 
\end{equation}
where  ${{\boldsymbol{P}}}^\dag_n {\boldsymbol{\Lambda}}^{1/2}_n \in L^{n \times q}_2(\Thbf,\mathbb{C})$  due to integrability of the dynamic spatio-temporal eigenvalues (see Remark \ref{rem:dyneval}).
Moreover, since $\boldsymbol{Q}_n(\thbf) \Sgmbf (\thbf) {{\boldsymbol{P}}}^\dag_n (\thbf) = \boldsymbol{\Phi}_n (\thbf) \boldsymbol{Q}_n(\thbf) {{\boldsymbol{P}}}^\dag_n (\thbf) = \boldsymbol{0}$ for all $\thbf\in\bm\Theta$, the two terms on the right-hand side of (\ref{Eq.xPQ}) are orthogonal at any lead and lag element by element. Therefore, the first term is the projection of $\boldsymbol{x}_{n\vsbf}$ on $\spn(\{\psi^n_{j \vsbf}, j = 1, \ldots, q, \vsbf \in \mathbb{Z}^3\})$ and the second term is the residual. By applying $\underline{\boldsymbol{C}} (L)\star \underline{\boldsymbol{\Lambda}}^{-1/2}_m (L)\star \underline{\boldsymbol{P}}_m(L)$ to both sides of (\ref{Eq.xPQ}) and using that $\underline{\boldsymbol{\Lambda}}^{-1/2}_m (L) \star\underline{\boldsymbol{P}}_m(L) \boldsymbol{x}_{n \vsbf} = \underline{\boldsymbol{\Lambda}}^{-1/2}_m (L) \star\underline{\boldsymbol{P}}_m(L) \boldsymbol{x}_{m \vsbf} = {\bfpsi}^m_{\vsbf}$, we obtain the projection as follows
\begin{equation}
\underline{\boldsymbol{C}} (L)  {\bfpsi}^m_{\vsbf} = \underline{\boldsymbol{D}} (L)  {\bfpsi}^n_{\vsbf} + \underline{\boldsymbol{R}} (L)  \boldsymbol{x}_{n\vsbf}, \label{Eq.Cpsi}
\end{equation}
where
\begin{equation}\label{defDR}
{\boldsymbol{D}}(\thbf) = {\boldsymbol{C}}(\thbf) {\boldsymbol{\Lambda}}^{-1/2}_m(\thbf) {\boldsymbol{P}}_m(\thbf) {\boldsymbol{P}}^\dag_n(\thbf) {\boldsymbol{\Lambda}}^{1/2}_n(\thbf), \quad {\boldsymbol{R}}(\thbf) = {\boldsymbol{C}}(\thbf)  \boldsymbol{\Lambda}_m^{-1/2}(\thbf) \boldsymbol{P}_m(\thbf) {\boldsymbol{Q}}^\dag_n(\thbf) \boldsymbol{Q}_n(\thbf).
\end{equation}
Notice that by taking the spectral density matrices of both sides of (\ref{Eq.Cpsi}) yields
\begin{equation}
\mbf {I}_q = \boldsymbol{D}(\thbf) {\boldsymbol{D}}^\dag (\thbf)  + \boldsymbol{R}(\thbf) \boldsymbol{\Sigma}_n^x (\thbf) {\boldsymbol{R}}^\dag (\thbf).\label{Eq.IDR}
\end{equation}
Now, let us denote by $\mu(\thbf)$ the largest eigenvalue of the spectral density matrix of $\underline{\boldsymbol{R}} (L)  \boldsymbol{x}_{n\sbf t}$. Then, we have 

\begin{lemma}\label{MyLemma7}
Suppose that (i) and (ii) of Theorem \ref{Th. q-DFS}, and Assumptions \ref{Ass.A} and \ref{Ass.B} hold. Then, for $m <n$ and $\boldsymbol{C} \in \boldsymbol{K}_M$,
$\mu(\thbf) \leq \lambda_{n,q+1}^x(\thbf)/\lambda_{mq}^x(\thbf).$
\end{lemma}

\begin{proof} 
Due to (\ref{Eq. IPQ}) both $\mbf {I}_n - {\boldsymbol{Q}}^\dag_n(\bm\theta) {\boldsymbol{Q}}_n(\bm\theta)$ and $\lambda_{n,q+1}^x(\bm\theta) {\boldsymbol{Q}}^\dag_n(\bm\theta) {\boldsymbol{Q}}_n(\bm\theta) - {\boldsymbol{Q}}^\dag_n(\bm\theta) \boldsymbol{\Phi}_n(\bm\theta) {\boldsymbol{Q}}_n(\bm\theta)$ are non-negative definite. Therefore, $\lambda_{n,q+1}^x(\bm\theta) \mbf {I}_n -  {\boldsymbol{Q}}^\dag_n(\bm\theta) \boldsymbol{\Phi}_n(\bm\theta) {\boldsymbol{Q}}_n(\bm\theta)$ is non-negative definite too, which implies that 
$$\boldsymbol{C}(\bm\theta) {\boldsymbol{\Lambda}}^{-1/2}_m(\bm\theta) \boldsymbol{P}_m(\bm\theta) (\lambda_{n,q+1}^x(\bm\theta) \mbf{I}_n -  {\boldsymbol{Q}}^\dag_n(\bm\theta) \boldsymbol{\Phi}_n (\bm\theta){\boldsymbol{Q}}_n(\bm\theta)) {\boldsymbol{P}}^\dag_m(\bm\theta) {\boldsymbol{\Lambda}}^{-1/2}_m(\bm\theta)
{\boldsymbol{C}}^\dag(\bm\theta)$$ is also non-negative definite. Moreover,
\begin{align}
\boldsymbol{C}(\bm\theta) {\boldsymbol{\Lambda}}^{-1/2}_m(\bm\theta) \boldsymbol{P}_m(\bm\theta) (\lambda_{n,q+1}^x(\bm\theta) \mbf{I}_n& -  {\boldsymbol{Q}}^\dag_n(\bm\theta) \boldsymbol{\Phi}_n (\bm\theta){\boldsymbol{Q}}_n(\bm\theta)) {\boldsymbol{P}}^\dag_m(\bm\theta) {\boldsymbol{\Lambda}}^{-1/2}_m(\bm\theta)
{\boldsymbol{C}}^\dag(\bm\theta)\nn\\
& = \lambda_{n,q+1}^x(\bm\theta) \boldsymbol{C}(\bm\theta) {\boldsymbol{\Lambda}}^{-1}_m(\bm\theta) {\boldsymbol{C}}^\dag(\bm\theta) - \boldsymbol{R}(\bm\theta) \boldsymbol{\Sigma}_n^x(\bm\theta) {\boldsymbol{R}}^\dag(\bm\theta),
\end{align}
where we obtain the second term making use of $\boldsymbol{\Phi}_n(\bm\theta) = \boldsymbol{Q}_n(\bm\theta) \boldsymbol{\Sigma}_n^x(\bm\theta) {\boldsymbol{Q}}^\dag_n(\bm\theta)$. 
Letting $\bm A(\thbf)=\lambda_{n,q+1}^x(\bm\theta)\boldsymbol{C}(\bm\theta) {\boldsymbol{\Lambda}}^{-1}_m(\bm\theta) {\boldsymbol{C}}^\dag(\bm\theta)$ and $\bm B(\thbf) = \boldsymbol{R}(\bm\theta) \boldsymbol{\Sigma}_n^x(\bm\theta) {\boldsymbol{R}}^\dag(\bm\theta)$ we have that
\[
 \nu_{\max}\l(\bm A(\thbf)\r) - \nu_{\max}(\bm B(\thbf))\ge \nu_{\min}\l(\bm A(\thbf)\r) - \nu_{\max}(\bm B(\thbf))\ge \nu_{\min}\l(\bm A(\thbf)-\bm B(\thbf)\r)\ge 0,
\]
where the left inequality follows from Weyl's inequality; see Appendix \ref{App:Weyl}. Then,
\[
0\le  \nu_{\max}\l(\bm A(\thbf)\r) - \nu_{\max}(\bm B(\thbf)) = \lambda_{n,q+1}^x(\bm\theta) \nu_{\max}\l(\boldsymbol{C}(\bm\theta) {\boldsymbol{\Lambda}}^{-1}_m(\bm\theta) {\boldsymbol{C}}^\dag(\bm\theta)\r)-\mu(\thbf).
\]
Using  $ \nu_{\max}\l(\boldsymbol{C}(\bm\theta) {\boldsymbol{\Lambda}}^{-1}_m(\bm\theta) {\boldsymbol{C}}^\dag(\bm\theta)\r)= 1/\nu_{\min}({\boldsymbol{\Lambda}}_m(\bm\theta))=1/\lambda_{mq}^x(\thbf)$ concludes the proof.
\end{proof}


With the projection as in (\ref{Eq.Cpsi}), we are now ready to construct our converging sequence. Under (i) and (ii) of Theorem~\ref{Th. q-DFS}, there exists a set $\boldsymbol{\Pi} \subseteq \Thbf$ and a real number $W$ such that $\Thbf \setminus \boldsymbol{\Pi}$ has null measure and (1) $\lambda_{n q +1}^x (\thbf) \leq W$ for any $n \in \mathbb{N}$ and any $\thbf \in \boldsymbol{\Pi}$; (2) $\lim_{n\to\infty}\lambda_{nq}^x (\thbf) = \infty$ for any $\thbf \in \boldsymbol{\Pi}$. Let $\boldsymbol{M}$ be a positive measure subset of $\boldsymbol{\Pi}$ and $\{\alpha_n, n \in \mathbb{N} \}$ a real positive nondecreasing sequence such that $\lim_{n\to\infty} \alpha_n = \infty$ and $\lambda_{n q}^x (\thbf) \geq \alpha_n$ for $\thbf \in \boldsymbol{M}$. 
Then, for $\thbf \in \boldsymbol{M}$ according to Lemma \ref{MyLemma7}, $\mu(\thbf)  \leq \lambda_{n,q+1}^x(\thbf)/\lambda_{mq}^x(\thbf) \leq W/\alpha_m$. Denote by $\Delta_j(\thbf), j = 1,\ldots, q$ the eigenvalues of $\boldsymbol{D}(\thbf) {\boldsymbol{D}}^\dag (\thbf)$ in descending order. By \eqref{Eq.IDR} and Weyl's inequality, we have
\begin{equation}\label{inEq.Delta0}
1 \geq \Delta_q(\thbf) \geq 1 - W/\alpha_m
\end{equation}
for any $\thbf \in \boldsymbol{M}$. Hence, if $m^*$ is such that $W/\alpha_{m^*} < 1$, we have 
\begin{equation}\label{inEq.Delta}
\Delta_q(\thbf) \geq 1 - W/\alpha_{m^*} > 0
\end{equation}
 for any $\thbf \in \boldsymbol{M}$ and $m \geq m^*$.

Assuming $m \geq m^*$, we denote by $\boldsymbol{\Delta}(\thbf)$ the $q \times q$ diagonal matrix with $\Delta_j(\thbf)$, $j=1,\ldots, q$, on the diagonal. Let $\boldsymbol{H}(\thbf)$ be a  matrix that is measurable in $\boldsymbol{M}$ and satisfies that for any $\thbf \in  \boldsymbol{M}$, $\boldsymbol{H}(\thbf) {\boldsymbol{H}}^\dag (\thbf) = \boldsymbol{I}_q$ and $\boldsymbol{H}(\thbf) \boldsymbol{\Delta}(\thbf) {\boldsymbol{H}}^\dag (\thbf) = \boldsymbol{D}(\thbf)  {\boldsymbol{D}}^\dag (\thbf)$. Notice that due to (\ref{inEq.Delta}), $\Delta_j^{-1/2}(\thbf)$ is bounded for all $\thbf\in\boldsymbol{M}$ and all $j = 1, \ldots, q$. Therefore,
\begin{equation}\label{defF}
\boldsymbol{F}(\thbf) =
\begin{cases}
\boldsymbol{H}(\thbf) \boldsymbol{\Delta}^{-1/2}(\thbf) {\boldsymbol{H}}^\dag (\thbf)  \boldsymbol{D}(\thbf) & \text{if $\thbf \in \boldsymbol{M}$} \\
\boldsymbol{0} & \text{if $\thbf \notin \boldsymbol{M}$}
\end{cases}
\end{equation}
is well defined and, clearly, it belongs to $\boldsymbol{K}_M$. We then have the following result:

\begin{lemma}\label{Mylemma8}
Suppose that (i) and (ii) of Theorem \ref{Th. q-DFS}  hold. Then given $\tau$ such that $0<  \tau < 2$, there exists an integer $m_\tau$ such that: (1) $W/\alpha_{m_\tau} < 1$; (2) for $n>m>m_\tau$, the largest eigenvalue of the spectral density matrix of $\underline{\boldsymbol{C}} (L)  {\bfpsi}^m_{\vsbf} - \underline{\boldsymbol{F}} (L)  {\bfpsi}^n_{\vsbf}$ is less than $\tau$ for any $\thbf \in \boldsymbol{\Pi}$.
\end{lemma}

\begin{proof}
Denote by $\boldsymbol{S}(\thbf)$ the spectral density matrix of $\underline{\boldsymbol{C}} (L)  {\bfpsi}^m_{\vsbf} - \underline{\boldsymbol{F}} (L)  {\bfpsi}^n_{\vsbf}$. Note that due to (\ref{Eq.Cpsi}), 
$$\underline{\boldsymbol{C}} (L)  {\bfpsi}^m_{\vsbf } - \underline{\boldsymbol{F}} (L)  {\bfpsi}^n_{\vsbf } = \underline{\boldsymbol{R}} (L)  \boldsymbol{x}_{n\vsbf } + (\underline{\boldsymbol{D}} (L)  - \underline{\boldsymbol{F}} (L)) {\bfpsi}^n_{\vsbf },$$
where two terms on the right-hand side are orthogonal at any lead and lag. Denoting by $\boldsymbol{S}_1(\thbf)$ and $\boldsymbol{S}_2(\thbf)$ the spectral density matrices of these terms respectively, then, for any $\thbf \in \boldsymbol{M}$, we have 
\begin{align*}
\boldsymbol{S}(\thbf) &= \boldsymbol{S}_1 (\thbf) + \boldsymbol{S}_2(\thbf) \\
&= \boldsymbol{R}(\thbf) \boldsymbol{\Sigma}_n^x (\thbf) {\boldsymbol{R}}^\dag (\thbf) + 
\boldsymbol{D}(\thbf){\boldsymbol{D}}^\dag (\thbf)  + \boldsymbol{F}(\thbf) {\boldsymbol{F}}^\dag (\thbf) - \boldsymbol{D}(\thbf) {\boldsymbol{F}}^\dag (\thbf) - \boldsymbol{F}(\thbf) {\boldsymbol{D}}^\dag (\thbf) \\
& = 2 \mbf{I}_q - \boldsymbol{D}(\thbf) {\boldsymbol{F}}^\dag (\thbf) - \boldsymbol{F}(\thbf) {\boldsymbol{D}}^\dag (\thbf) \\
& = 2 \mbf{I}_q - 2 \boldsymbol{H}(\thbf) \boldsymbol{\Delta}^{1/2}(\thbf) {\boldsymbol{H}}^\dag(\thbf) \\
& =  2 \boldsymbol{H}(\thbf) (\mbf{I}_q - \boldsymbol{\Delta}^{1/2}(\thbf)) {\boldsymbol{H}}^\dag (\thbf),
\end{align*}
where the third equality is due to (\ref{Eq.IDR}). Therefore, the largest eigenvalue of $\boldsymbol{S}(\thbf)$ is $2 - 2\sqrt{\Delta_q(\thbf)}$, which is
{(recall \eqref{inEq.Delta0})} less than or equal to $2 - 2 \Delta_q(\thbf) \leq 2 W/\alpha_m$. Since $\tau<2$,  the result follows from taking a positive integer $m_\tau$ such that 
\begin{equation}\label{mtau}
2 W/\alpha_{m_\tau} < \tau.
\end{equation}
This completes the proof.
\end{proof}

Denote by $\mathcal{S}(x_{i\vsbf }, x_{j\vsbf }; \thbf)$ the cross-spectrum between $x_{i\vsbf}, x_{j\vsbf}$ for $i, j \in \mathbb{N}$. The following intermediate lemmas will be used for proving further results.

\begin{lemma}\label{Mylemma9}
Suppose that Assumptions \ref{Ass.A} and \ref{Ass.B} hold. Consider the scalar sequences $\{A_{n\vsbf}, n\in\mathbb N\}$ and $\{B_{n\vsbf}, n\in\mathbb N\}$ such that 
$\lim_{n\to\infty} A_{n\vsbf } = A_{\vsbf }$ and $\lim_{n\to\infty} B_{n\vsbf } = B_{\vsbf }$, for any $\vsbf \in\mathbb Z^3$, with $A_{\vsbf},B_{\vsbf}\in \bm{\mathcal X}$. Suppose also that the rf $A_n,B_n, A,$ and $B$ are co-homostationary with $x_\ell$, $\ell\in\mathbb N$. 
Then, there exists a sequence $\{s_i, i\in\mathbb N,\, s_i<s_{i+1}\}$, such that $\{A_{s_i\vsbf}, i\in\mathbb N\}$ and $\{B_{s_i\vsbf}, i\in\mathbb N\}$ satisfy
$\underset{i\to\infty}{\lim} \ \mathcal{S}(A_{s_i \vsbf },  B_{s_i \vsbf }; \thbf) =  \mathcal{S}(A_{\vsbf },  B_{\vsbf }; \thbf)$, $\mathcal L$-a.e. in $\Thbf$.

\end{lemma} 
\begin{proof}
First, we have $\langle A_{n\vsbf}, B_{n\vsbf} \rangle = (8 \pi^3)^{-1} \int_{\Thbf} \mathcal{S}(A_{n\vsbf},  B_{n\vsbf}; \thbf){\mathrm d} \thbf$. Then, due to continuity of the inner product and convergence of $A_{n\vsbf}$ and $B_{n\vsbf}$, we have
$$
\lim_{n\to\infty}\frac{1}{8 \pi^3} \int_{\Thbf} |\mathcal{S}(A_{n\vsbf},  B_{n\vsbf}; \thbf)  -  \mathcal{S}(A_{\vsbf},  B_{\vsbf}; \thbf)| {\mathrm d} \thbf = 0.
$$
The desired result follows from 
 \citet[p.145]{royden1988real}.
\end{proof}



\begin{lemma}\label{MyLemma10}
Suppose that (i) and (ii) of Theorem \ref{Th. q-DFS} and  Assumptions \ref{Ass.A}-\ref{Ass.C} hold. Then there exists a $q$-dimensional orthonormal rf $\boldsymbol{v}$ such that
\begin{compactenum}
\item[(a)]  $v_{j \vsbf}\in\mathcal G(\bm x)$ for all $j = 1, \ldots, q$;
\item[(b)]  the spectral density matrix of $\boldsymbol{v}$ is $\mbf {I}_q$ $\mathcal L$-a.e. in $\bm M$, whilst it is ${\mbf 0}$ for $\thbf  \notin \boldsymbol{M}$. 
\end{compactenum}
\end{lemma}



\begin{proof}
The proof goes through by repeatedly applying Lemma \ref{Mylemma8} to construct a Cauchy sequence whose limit satisfies the desired properties. 

Let $\boldsymbol{F}_1(\thbf)$ be an element of $\boldsymbol{K}_M$. Let also $\tau = 1/2^2$, $n_1 = m_\tau$, where $m_\tau$ satisfies (\ref{mtau}), $\boldsymbol{G}_1(\thbf) = \boldsymbol{F}_1(\thbf) \boldsymbol{\Lambda}_{n_1}^{-1/2}(\thbf) \boldsymbol{P}_{n_1}(\thbf)$, and  $\boldsymbol{v}_{\vsbf}^{(1)} = \underline{\boldsymbol{G}}_1(L) \boldsymbol{x}_{n\vsbf}$. One can easily check that the spectral density matrix of $\boldsymbol{v}^{(1)}$ equals $\mbf{I}_q$   for $\thbf  \in \boldsymbol{M}$, ${\mbf 0}_q$ for $\thbf  \notin \boldsymbol{M}$.

In the similar way, set $\tau = 1/2^4$ and $n_2 = m_\tau$, where $m_\tau$ satisfies (\ref{mtau}) and $m_\tau \geq n_1$. Set $\boldsymbol{D}(\thbf)$ in (\ref{defDR}) by replacing $\boldsymbol{C}(\thbf)$, $n$, $m$ with $\boldsymbol{F}_1(\thbf)$, $n_2$, $n_1$, respectively, and set $\boldsymbol{F}_2(\thbf)$ as in (\ref{defF}). Then set $\boldsymbol{G}_2(\thbf) = \boldsymbol{F}_2(\thbf) \boldsymbol{\Lambda}_{n_2}^{-1/2}(\thbf) \boldsymbol{P}_{n_2}(\thbf)$ and $\boldsymbol{v}_{\vsbf}^{(2)} = \underline{\boldsymbol{G}}_2(L) \boldsymbol{x}_{n\vsbf}$. The spectral density matrix of $\boldsymbol{v}^{(2)}$ 
equals $\mbf {I}_q$ for $\thbf  \in \boldsymbol{M}$, ${\mbf 0}_q$ for $\thbf  \notin \boldsymbol{M}$. Denote by $\mathcal{A}_1(\thbf)$ the largest eigenvalue of the spectral matrix of $\boldsymbol{v}^{(1)} - \boldsymbol{v}^{(2)}$. According to the definition of $n_1$ and Lemma \ref{Mylemma8}, $\mathcal{A}_1(\thbf) < 1/2^2$ for any $\thbf \in \boldsymbol{\Pi}$, which entails $\Vert v_{j \vsbf}^{(1)} - v_{j \vsbf}^{(2)}\Vert < 1/2$ for all $j = 1, \ldots, q$ and any $\vsbf\in\mathbb Z^3$.

By recursion, set $\tau = 1/2^{2k}$ and $n_k = m_\tau$, where $m_\tau$ satisfies (\ref{mtau}) and $m_\tau \geq n_{k-1}$. Set $\boldsymbol{D}(\thbf)$ in (\ref{defDR}) by replacing $\boldsymbol{C}(\thbf)$, $n$, $m$ with $\boldsymbol{F}_{k-1}(\thbf)$, $n_k$, $n_{k-1}$, respectively, and set $\boldsymbol{F}_k(\thbf)$ as in (\ref{defF}). Then set $\boldsymbol{G}_k(\thbf) = \boldsymbol{F}_k(\thbf) \boldsymbol{\Lambda}_{n_k}^{-1/2}(\thbf) \boldsymbol{P}_{n_k}(\thbf)$ and $\boldsymbol{v}_{\vsbf}^{(k)} = \underline{\boldsymbol{G}}_k(L) \boldsymbol{x}_{n\vsbf}$. The spectral density matrix of $\boldsymbol{v}^{(k)}$ equals $\mbf{I}_q$ for $\thbf  \in \boldsymbol{M}$, ${\mbf 0}_q$ for $\thbf  \notin \boldsymbol{M}$. Denote by $\mathcal{A}_{k-1}(\thbf)$ the largest eigenvalue of the spectral matrix of $\boldsymbol{v}^{(k-1)} - \boldsymbol{v}^{(k)}$. According to the definition of $n_{k-1}$ and Lemma \ref{Mylemma8}, $\mathcal{A}_{k-1}(\thbf) < 1/2^{2(k-1)}$ for any $\thbf \in \boldsymbol{\Pi}$, which entails $\Vert v_{j \vsbf}^{(k-1)} - v_{j \vsbf}^{(k)}\Vert < 1/2^{k-1}$ for all $j = 1, \ldots, q$ and any $\vsbf\in\mathbb Z^3$.

Hence, for all $j = 1, \ldots, q$ and any $\vsbf\in\mathbb Z^3$, we have
$$
\Vert v_{j \vsbf}^{(k)} - v_{j \vsbf}^{(k+h)}\Vert \leq\Vert v_{j \vsbf}^{(k)} - v_{j \vsbf}^{(k+1)}\Vert + \cdots + \Vert v_{j \vsbf}^{(k+h-1)} - v_{j \vsbf}^{(k+h)}\Vert <\sum_{j=k}^{k+h-1}\frac 1{2^j} <\frac 1{2^{k-1}},
$$
which implies that for all $j=1,\ldots, q$, $\{{v}_{j\vsbf}^{(k)}, k \in \mathbb{N}\}$ is a Cauchy sequence. Denote by $\boldsymbol{v}_{\vsbf}=\lim_{k\to\infty} \boldsymbol{v}_{\vsbf}^{(k)}$. Then (b) follows from Lemma \ref{Mylemma9} and the fact that the spectral density matrix of $\boldsymbol{v}^{(k)}$  equals $\mbf {I}_q$ for $\thbf \in \boldsymbol{M}$, ${\mbf 0}_q$ for $\thbf  \notin \boldsymbol{M}$. 

Now it remains to prove (a), for which it suffices to show that each row of $\{\boldsymbol{G}_k, k \in \mathbb{N}\}$ is a STDAS (see Definition \ref{def:STDAS}).  Notice that $\boldsymbol{G}_k(\thbf) \boldsymbol{G}^\dag_k(\thbf) = \boldsymbol{F}_k(\thbf)  \boldsymbol{\Lambda}_{n_k}^{-1}(\thbf)  \boldsymbol{F}^\dag_k(\thbf) $, whose  diagonal entries are equal or less than $1/\lambda_{n_kq}^x(\thbf)$  since $\boldsymbol{F}_k(\thbf)  \in \boldsymbol{K}_M$ because $1/\lambda_{n_kq}^x(\thbf)$ converges to zero $\mathcal L\text{-a.e.}$ in $\Thbf$ by (ii) of Theorem \ref{Th. q-DFS}. Moreover, without loss of generality we can always restrict Assumption \ref{Ass.C} to assume $\lambda_{n_kq}^x(\thbf)>1$ for all $\thbf\in\bm\Theta$ (see the arguments in \citealp[Section 4.2]{FL01}). Then, $1/\lambda_{n_kq}^x(\thbf)<1$ and by Lebesgue's dominated convergence theorem, its integral over $\Thbf$ converges to zero. This concludes the proof of (a).
\end{proof}

%

Now, we apply the results in Lemma \ref{MyLemma10} to define a $q$-dimensional white noise rf over all $\bm\Theta$.

\begin{proposition}\label{MyLemma11}
Suppose that (i) and (ii) of Theorem \ref{Th. q-DFS} and Assumptions \ref{Ass.A}-\ref{Ass.C}  hold. There exists a $q$-dimensional orthonormal white noise rf $\boldsymbol{z}$
such that, for all $j=1,\ldots,q$, $z_{j \vsbf }\in\mathcal G(\bm x)$. 
\end{proposition}

\begin{proof}
The proof goes through by choosing a set $\boldsymbol{N} \in \Thbf$ with Lebesgue measure $\mathcal L(\boldsymbol{N}) = 8\pi^3 = \mathcal L (\Thbf)$ and by obtaining a sequence of $q$-dimensional vector rf, which satisfy (a) in Lemma~\ref{MyLemma10} and have spectral density matrix equal to $\mbf {I}_q$ for $\mathcal L\text{-a.e.}$ in a partition of $\boldsymbol{N}$, and is ${\mbf 0}$ otherwise.

Now, define $\boldsymbol{M}_0^{(1)} = \boldsymbol{\Pi}$. Then, by recursion, define $\nu_a$, $a \in \mathbb{N}$, as the smallest among the integer $m$ such that
$$\mathcal L(\{\thbf \in \boldsymbol{M}_{a-1}^{(1)}, \lambda_{mq}^x(\thbf) > a\}) > 4 \pi^3$$
and define
$\boldsymbol{M}_{a}^{(1)} = \{\thbf \in \boldsymbol{M}_{a-1}^{(1)},  \lambda_{\nu_a q}^x(\thbf) > a\}.$
Clearly, the Lebesgue measure of the set
$$\boldsymbol{N}_1 = \boldsymbol{M}_1^{(1)} \cap \boldsymbol{M}_2^{(1)} \cap \cdots \cap \boldsymbol{M}_a^{(1)} \cap \cdots$$
is not less than $4 \pi^3$. In the similar fashion, define $\boldsymbol{N}_2$ starting with $\boldsymbol{M}_0^{(2)} =  \boldsymbol{\Pi} \setminus \boldsymbol{N}_1$ instead of $\boldsymbol{\Pi}$, and using $\mathcal L ( \boldsymbol{\Pi} \setminus \boldsymbol{N}_1)/2$ instead of $4 \pi^3$. Also, for $b > 2$, define $\boldsymbol{N}_b$ starting with $\boldsymbol{M}_0^{(b)} =  \boldsymbol{\Pi} \setminus \boldsymbol{N}_1 \setminus \boldsymbol{N}_2 \setminus \cdots \setminus \boldsymbol{N}_{b-1}$, and using $\mathcal L (\boldsymbol{\Pi} \setminus \boldsymbol{N}_1 \setminus \boldsymbol{N}_2 \setminus \cdots \setminus \boldsymbol{N}_{b-1})/2$. Letting $\boldsymbol{N} = \boldsymbol{N}_1 \cup \boldsymbol{N}_2 \cup \cdots$, we have
$$\mathcal L (\boldsymbol{N}) = \mathcal L (\boldsymbol{N}_1) + \mathcal L (\boldsymbol{N}_2) + \cdots + \mathcal L (\boldsymbol{N}_b) + \cdots = 8\pi^3,$$
since by construction $\boldsymbol{N}_i \cap \boldsymbol{N}_j = \emptyset $, for $i\neq j$ and $i,j \in \mathbb{N}$.

Lemma~\ref{MyLemma10} can be applied to the subset $\boldsymbol{N}_b$, with the sequence $\alpha_n$ defined as $\alpha_n = a$, where $a$ is the only integer such that $\nu_a \leq n < \nu_{a+1}$. Hence, we obtain a $q$-dimensional vector rf $\{\boldsymbol{v}_{\vsbf}^b = (v_{1\vsbf}^b \ v_{2\vsbf}^b \ \cdots \  v_{q\vsbf}^b)^\top,\, \vsbf \in\mathbb Z^3\}$ such that (i) $v_{j\vsbf}^b\in\mathcal G(\bm x)$ for all $j=1,\ldots, q$; (ii) its spectral density matrix equals $\mbf{I}_q$ for $\mathcal L\text{-a.e.}$ in $\boldsymbol{N}_b$, and is ${\mbf 0}_q$ for $\thbf \notin \boldsymbol{N}_b$.

Finally, set $\boldsymbol{z}_{\vsbf} = \sum_{b=1}^\infty \boldsymbol{v}_{\vsbf}^b$. It is easy to see that $z_{j\vsbf}\in\mathcal G(\bm x)$ for all $j = 1,\ldots, q$ and the spectral density matrix of $\boldsymbol{z}$ equals $\mbf {I}_q$  $\mathcal L\text{-a.e.}$ in $\Thbf$. Therefore, $\boldsymbol{z}$ is a  $q$-dimensional orthonormal white noise rf. 
\end{proof}

Considering the $q$-dimensional orthonormal white noise rf $\boldsymbol{z}$ in Proposition \ref{MyLemma11}, we have the following

\begin{proposition} \label{MySection4.5}
Suppose that (i) and (ii) of Theorem \ref{Th. q-DFS} and Assumptions \ref{Ass.A}-\ref{Ass.C} hold. Then $\text{\rm $\spn$}(\boldsymbol{z}) =  \mathcal{G}(\boldsymbol{x})$.
\end{proposition}

\begin{proof}
Consider a scalar rf $y_{\vsbf}\in\mathcal G(\bm x)$ and consider the projection
$$y_{\vsbf} = \text{proj}(y_{\vsbf}|\spn(\boldsymbol{z})) + r_{\vsbf}.$$
It suffices to prove that $r_{\vsbf}=0$. Let $\boldsymbol{W}(\thbf)$ denote the spectral density matrix of the $(q+1)$-dimensional rf $\{(\boldsymbol{z}_{\vsbf} \ r_{\vsbf})^\top,\, \vsbf \in\mathbb Z^3\}$. According to the proof of Proposition \ref{MyLemma11}, $\boldsymbol{W}(\thbf)$ is diagonal with $\mbf {I}_q$ in the $q\times q$ upper-left submatrix and $\text{det}(\boldsymbol{W}(\thbf)) = \mathcal{S}(r_{\vsbf}, r_{\vsbf}; \thbf).$ Since both $\boldsymbol{z}$ and $r$ belong to $\mathcal{G}(\boldsymbol{x})$, there exist STDASs $\{\boldsymbol{a}_{nj}, n \in \mathbb{N} \}$, for $j = 1, \ldots, q+1$ such that, 
\begin{align}
&\underset{n\to\infty}{\lim} \ \boldsymbol{a}_{nj}(L) \boldsymbol{x}_{n \vsbf} = z_{j\vsbf}, \quad j = 1,\ldots, q,\nn\\
&\underset{n\to\infty}{\lim} \ \boldsymbol{a}_{n,q+1}(L) \boldsymbol{x}_{n \vsbf} = r_{\vsbf}.\nn
\end{align} 
Now, for all $j = 1,\ldots, q+1$, by Definition \ref{def:STDAS} of STDAS, we must have $\lim_{n\to\infty}\int_{\Thbf} \boldsymbol{a}_{nj}(\thbf)\boldsymbol{a}_{nj}^\dag(\thbf) {\mathrm d}\thbf=0$, it follows that 
\beq\label{sciamenna}
\lim_{n\to\infty }\boldsymbol{a}_{nj}(\thbf)\boldsymbol{a}_{nj}^\dag(\thbf)=\lim_{n\to\infty} |\boldsymbol{a}_{nj}(\thbf)|^2= 0,\quad \mathcal L \; \text {a.e. in} \; \Thbf.
\eeq
Therefore (see  \citet[p.145]{royden1988real})
there exists a sub-sequence $\{s_k,\, k\in\mathbb N,s_k<s_{k+1}\}$, defining a corresponding sub-set of $s_k$ elements of $\boldsymbol{a}_{nj}(\thbf)$ collected into the $s_k$-dimensional row vector $\boldsymbol{a}_{s_kj} (\thbf)$, which is such that
\beq\label{sciamenna3}
\lim_{k \to\infty} |\boldsymbol{a}_{s_kj}(\thbf)|^2=0,\quad \mathcal L \; \text {a.e. in} \; \Thbf.
\eeq
Define also the rf $\boldsymbol{x}_{s_k\vsbf }$ obtained from $\boldsymbol{x}_{n\vsbf }$ by setting to zero all entries with the exception of the $s_k$ elements corresponding to the sub-sequence $\{s_k\}$, and 
let $\boldsymbol{Z}_n(\thbf)$ denote the spectral density matrix of $\{(\boldsymbol{a}_{n1}(L) \boldsymbol{x}_{s_k \vsbf}  \cdots  \boldsymbol{a}_{n,q+1}(L) \boldsymbol{x}_{s_k \vsbf})^\top,\,\vsbf \in \mathbb Z^3\}$.
 In view of Lemma \ref{Mylemma9}, there exists a sub-sequence of $\boldsymbol{Z}_n(\bm\theta)$ converging to $\boldsymbol{W}(\bm\theta)$, $\mathcal L$  a.e. in $\Thbf$. Therefore, without loss of generality, we can assume that  $\boldsymbol{Z}_n(\thbf)$ converges to $\boldsymbol{W}(\thbf)$, $\mathcal L$  a.e. in $\Thbf$.

Let $\boldsymbol{f}_{nj}(\thbf) = \boldsymbol{a}_{nj}(\thbf) \boldsymbol{P}_{s_k}^\dag(\thbf)$ and  $\boldsymbol{g}_{nj}(\thbf) = \boldsymbol{a}_{nj}(\thbf) - \boldsymbol{f}_{nj}(\thbf) \boldsymbol{P}_{s_k}(\thbf)$ for $j = 1, \ldots, q+1$ and $\thbf \in\Thbf$. Hence, $\boldsymbol{a}_{nj}(\thbf) = \boldsymbol{f}_{nj}(\thbf) \boldsymbol{P}_{s_k}(\thbf) + \boldsymbol{g}_{nj}(\thbf)$ and, for all $\thbf\in\Thbf$, 
\beq\label{herbert}
|\boldsymbol{a}_{nj}(\thbf)|^2 = |\boldsymbol{f}_{nj}(\thbf)|^2 + |\boldsymbol{g}_{nj}(\thbf)|^2.
\eeq
Indeed, by definition,
$$
\boldsymbol{f}_{nj}(\thbf)\boldsymbol{P}_{s_k}(\thbf)\boldsymbol{g}_{nj}^\dag(\thbf) = \boldsymbol{a}_{nj}(\thbf) \boldsymbol{P}_{s_k}^\dag(\thbf)\boldsymbol{P}_{s_k}(\thbf)\left(\boldsymbol{a}_{nj}(\thbf)-
 \boldsymbol{a}_{nj}(\thbf) \boldsymbol{P}_{s_k}^\dag(\thbf)\boldsymbol{P}_{s_k}(\thbf)
\right)^\dag=0.
$$

Now, \eqref{sciamenna} implies
\beq\label{sciamenna2}
\lim_{n\to\infty} |\boldsymbol{g}_{nj}(\thbf)|^2= 0\quad \mathcal L \; \text {a.e. in} \; \Thbf,
\eeq
and we have also that
\beq\label{sciamenna3}
\lim_{n\to\infty} |\boldsymbol{f}_{nj}(\thbf)|^2= \lim_{n\to\infty} \boldsymbol{a}_{nj}(\thbf) \boldsymbol{P}_{s_k}^\dag(\thbf)\boldsymbol{P}_{s_k}(\thbf)\boldsymbol{a}^\dag_{nj}(\thbf)
=\lim_{k \to\infty} |\boldsymbol{a}_{s_kj}(\thbf)|^2=0,\quad \mathcal L \; \text {a.e. in} \; \Thbf.
\eeq
It follows that the following orthogonal decomposition holds:
\beq\label{bollitomisto}
\underline{\boldsymbol{a}}_{nj}(L) \boldsymbol{x}_{s_k \vsbf} = \underline{\boldsymbol{f}}_{nj}(L)\star \underline{\boldsymbol{P}}_{s_k}(L)  \boldsymbol{x}_{s_k \vsbf} + \underline{\boldsymbol{g}}_{nj}(L) \boldsymbol{x}_{s_k \vsbf}, \quad j=1,\ldots, q+1.
\eeq
Denote by $\boldsymbol{Z}_n^1(\thbf)$ and $\boldsymbol{Z}_n^2(\thbf)$ the spectral density matrices of the rf
$$
\left\{\left(\underline{\boldsymbol{f}}_{n1}(L)\star \underline{\boldsymbol{P}}_{s_k}(L)  \boldsymbol{x}_{s_k \vsbf}  \cdots  \underline{\boldsymbol{f}}_{n,q+1}(L)\star \underline{\boldsymbol{P}}_{s_k}(L)  \boldsymbol{x}_{s_k \vsbf}\right)^\dag,\,\vsbf\in\mathbb Z^3\right\}
$$
and
$$
\left\{\left((\underline{\boldsymbol{g}}_{n1}(L) \boldsymbol{x}_{s_k \vsbf} \cdots  \underline{\boldsymbol{g}}_{n,q+1}(L) \boldsymbol{x}_{s_k \vsbf}\right)^\top,\,\vsbf\in\mathbb Z^3\r\},
$$
respectively. Because of \eqref{herbert} and \eqref{bollitomisto}, we then have 
$$
\boldsymbol{Z}_n(\thbf) =\boldsymbol{Z}_n^1(\thbf) + \boldsymbol{Z}_n^2(\thbf).
$$ 
Notice that $\boldsymbol{Z}_n^1(\thbf)$ is singular for all $\thbf\in\Thbf$, as $k\to\infty$, because  $\boldsymbol{P}_{s_k}(\thbf)$ is $q+1\times s_k$. Hence, 
\beq\label{maccio}
\lim_{n\to\infty} \text{det}(\boldsymbol{Z}_n^1(\thbf))=0, \quad \text{for all }\thbf \in\Thbf.
\eeq
Since $\boldsymbol{g}_{nj}(\thbf)$ is orthogonal to ${\boldsymbol{p}}_{s_k i}^x(\thbf)$ for $i = 1,\ldots, q$, we have
$$
\boldsymbol{Z}_n^2(\thbf))=\boldsymbol{g}_{nj}(\thbf) \boldsymbol{\Sigma}_{s_k}^x(\thbf)  \boldsymbol{g}_{nj}^\dag (\thbf) \leq \lambda_{s_k q+1}^x(\thbf) |\boldsymbol{g}_{nj}(\thbf)|^2
$$
\citep[Exercise 1, p.~287]{LT85}. Now, because of (i) in Theorem \ref{Th. q-DFS} and by \eqref{sciamenna2} we have that $\boldsymbol{Z}_n^2(\thbf)$ converges to zero $\mathcal L\text{-a.e.}$ in $\Thbf$ as $n\to\infty$. 
Therefore, by \eqref{maccio},
\beq
\lim_{n\to\infty}\text{det}(\boldsymbol{Z}_n(\thbf))=0,\quad \mathcal L \text{ a.e. in } \Thbf,
\eeq
which entails that $\text{det}(\boldsymbol{W}(\thbf))= \mathcal{S}(r_{\vsbf}, r_{\vsbf}; \thbf)=0$, $\mathcal L\text{-a.e.}$ in $\Thbf$ and, thus, $r_{\vsbf} = 0$.
\end{proof}

\subsubsection{Canonical decomposition into common and idiosyncratic component}

Consider the canonical decomposition 
\begin{align}\label{eq:app:CD}
x_{\ell \vsbf} &= \text{proj}(x_{\ell \vsbf}|\mathcal{G}(\boldsymbol{x})) + \delta_{\ell \vsbf}= \gamma_{\ell \vsbf}+ \delta_{\ell \vsbf}, \; \text { say.}
\end{align}
So far by means of Propositions \ref{MyLemma11} and \ref{MySection4.5}, we have shown that if (i) and (ii) of Theorem \ref{Th. q-DFS} hold, then, there exists a $q$-dimensional orthonormal white noise rf $\bm z$ such that 
\begin{align}\label{eq:app:CD2}
\gamma_{\ell \vsbf} &=  \underline{\boldsymbol{c}}_\ell(L)\boldsymbol{z}_{\vsbf},
\end{align}
with ${\boldsymbol{c}}_\ell \in L_2^q(\Thbf, \mathbb{C}).$

Now, consider a generic $n$-dimensional rf $\bm w_n$ satisfying Assumptions \ref{Ass.A}-\ref{Ass.B} with dynamic spatio-temporal eigenvalues $\lambda_{nj}^w(\thbf)$, $j=1,\ldots,n$, $\thbf\in\Thbf$.
Then, for any $\thbf\in\Thbf$ let $\lambda_{j}^w(\thbf)=\sup_{n\in\mathbb N} \lambda_{nj}^w(\thbf)$ and recall that, since $\lambda_{nj}^w(\thbf)$ is an increasing sequence in $n$ then $\sup_{n\in\mathbb N} \lambda_{nj}^w(\thbf)=\lim_{n\to\infty} \lambda_{nj}^w(\thbf)$. So $\lambda_{j}^w(\thbf)$ is the $j$-th largest dynamic spatio-temporal eigenvalue of the infinite dimensional spectral density matrix $\bm\Sigma^w(\thbf)$ of the infinite dimensional rf $\bm w$.
The proof of the sufficient condition is concluded by means of the next two results.

\begin{proposition} \label{Theorem 1}
Under Assumptions \ref{Ass.A}-\ref{Ass.B}, the following statements are equivalent:
\begin{compactenum}
\item[(a)] $ \boldsymbol{w}=\{(w_{1\vsbf}\ w_{2\vsbf}\cdots w_{\ell\vsbf}\cdots )^\top, \vsbf\in \mathbb Z^3\}$ is idiosyncratic;
\item[(b)] the function $\lambda^w_1:\Thbf\to \mathbb R^+$ is essentially bounded, i.e., $\text{\upshape ess} \sup(\lambda_{1}^w)<\infty$ where $\text{\upshape ess} \sup(\lambda_{1}^w)=\inf\{M:\mathcal L[\thbf:\lambda_{1}^w(\thbf)>M]=0\}$;
\item[(c)] Define $\boldsymbol\Upsilon: \boldsymbol{\Psi} \to L^{\infty}_2 (\Thbf,\mathbb{C},\bm\Sigma^w_n)$ as $\boldsymbol\Upsilon(\boldsymbol f) = \boldsymbol f$, the mapping $\boldsymbol\Upsilon$ is continuous.
\end{compactenum}
\end{proposition}

\begin{proof}
We first show that (a) and (c) are equivalent, $(a) \Leftrightarrow (c)$. To this end, notice that
\begin{equation}
\Vert \underline{\boldsymbol {a}}_{n}(L)\boldsymbol{w}_{\sbf t} \Vert = \Vert \boldsymbol{a}_{n}  \Vert_{\Sigmabf^w} = \Vert \boldsymbol{\Upsilon}({\boldsymbol a}_n) \Vert_{\Sigmabf^w}, \label{Eq. chain}
\end{equation}
where the first equality follows from the application of the isometric isomorphism $\mathcal{J}^{-1}$ (see Remarks \ref{rem:Jinv} and \ref{rem:spfilt}), while the second from the definition of the mapping $\boldsymbol{\Upsilon}$. Then, by definition of idiosyncratic process $\lim_{n\to\infty} \Vert \underline{\boldsymbol {a}}_{n}(L)\boldsymbol{w}_{\sbf t} \Vert =0$, thus by (\ref{Eq. chain}) it follows that $\lim_{n\to\infty} \Vert \boldsymbol{\Upsilon}({\boldsymbol a}_n) \Vert_{\Sigmabf^w}=0$, which implies that the linear mapping $\boldsymbol{\Upsilon}$ is continuous
at zero. From \citet[Proposition 1.1, p.26]{C85}, it follows that $\boldsymbol{\Upsilon}$ is continuous everywhere. This proves $(a) \Leftrightarrow (c)$.

To prove that (b) and (c) are equivalent, $(b) \Leftrightarrow (c)$, first notice that continuity and boundedness are equivalent for linear maps between normed vector spaces \citep[Theorem 1, p. 257]{royden1988real}. Then, consider the definition of operator norm:
$$
\Vert \boldsymbol{\Upsilon} \Vert = \sup_{\boldsymbol f \in \boldsymbol \Psi, \Vert \boldsymbol f \Vert =1 } \left \Vert \boldsymbol{\Upsilon}(\boldsymbol f) \right \Vert_{\Sigmabf^w}
$$
and notice that boundedness of $\boldsymbol{\Upsilon}$ means that $\Vert \boldsymbol{\Upsilon} \Vert \leq c < \infty$, for $c \in \mathbb{R}^+$. We now show that $\Vert \boldsymbol{\Upsilon} \Vert=\sqrt{\text{ess} \sup (\lambda_{1}^w)}$.
This would imply that $(b) \Leftrightarrow (c)$.

Let us define,
$\boldsymbol f^{[n]}$ as the infinite dimensional vector with $f_j^{[n]}=f_j$ for $j\le n$ and $f_j^{[n]}=0$ for $j>n$, $\boldsymbol f^{\{n\}}$ as the $n$-dimensional sub-vector made of the first $n$ entries of $\boldsymbol f^{[n]}$,
and
$$
\psi_n = \sup_{\boldsymbol f \in \boldsymbol \Psi, \Vert \boldsymbol f \Vert =1 } \left \Vert \boldsymbol{\Upsilon}(\boldsymbol f^{[n]}) 
\right \Vert_{\Sigmabf^w},
$$
so, for $\boldsymbol \Psi_n = L^{\infty}_2(\Thbf,\mathbb{C}) \cap  L^{\infty}_2(\Thbf,\mathbb{C}, \lambda^w_{1n}) $, we have 
\begin{eqnarray*}
\psi^2_n &=& \sup_{\boldsymbol f \in \boldsymbol \Psi, \Vert \boldsymbol f \Vert =1 } \frac{1}{8\pi^3} \int_{\Thbf} \boldsymbol{f}^{\{n\}}(\thbf)  \bm\Sigma^w_n(\thbf) 
{\boldsymbol{f}}^{\{n\}\dag}(\thbf)  \mathrm d\thbf \\
& = & \sup_{h \in \boldsymbol \Psi_n, \Vert \boldsymbol h \Vert =1 } \frac{1}{8\pi^3} \int_{\Thbf} \Vert h(\thbf) \Vert^2 \lambda^w_{1n}(\thbf) \mathrm d\thbf,
\end{eqnarray*}
where the last equality follows form \citet[Theorem 4, p.285]{LT85}. Moreover, \citet[Theorem 1.5, p.28]{C85} implies that $\psi_n^2 = \text{ess} \sup \lambda_{1n}^w(\thbf)$. Finally we notice that $\Vert \boldsymbol\Upsilon \Vert^2 = \lim_{n\to\infty} \psi^2_n = \lim_{n\to\infty}  {\text{ess} \sup \lambda^w_{1n}(\thbf)}$
$=  \text{ess} \sup \lim_{n\to\infty} \lambda^w_{1n}(\thbf) = \text{ess} \sup \lambda^w_{1}(\thbf)$.  
\end{proof}


\begin{proposition}\label{deltaidio}
Under Assumptions \ref{Ass.A}-\ref{Ass.B}, 
$\boldsymbol{\delta}=\{(\delta_{1\vsbf}\ \delta_{2\vsbf}\cdots \delta_{\ell\vsbf}\cdots )^\top, \vsbf\in \mathbb Z^3\}$ is idiosyncratic, where $\delta_{\ell\vsbf}$, $\ell\in\mathbb N$, are defined in \eqref{eq:app:CD}.
\end{proposition}

\begin{proof}
Start by considering again \eqref{Eq.xPQ} and let, $\underline {\bm\pi}_{n\ell}(L)$ and $\underline {\bm q}_{n\ell}(L)$ be the $\ell$-th $q$-dimensional and $(n-q)$-dimensional rows of $ \underline{{\boldsymbol{P}}}^\dag_n (L)$ and $ \underline{{\boldsymbol{Q}}}^\dag_n (L)$, respectively, for any given $\ell\le n$. Then,
\beq\label{eq:app:gd}
x_{\ell\vsbf}= \underline {\bm\pi}_{n\ell}(L) \star\underline{{\boldsymbol{P}}} (L)\bm x_{n\vsbf}+ \underline {\bm q}_{n\ell}(L)\star\underline{{\boldsymbol{Q}}} (L)\bm x_{n\vsbf}=
\underline {\bm\pi}_{n\ell}(L)\star\underline{\bm\Lambda}_n^{1/2}(L) \bm \psi_{\vsbf}^n + \underline {\bm q}_{n\ell}(L)\star\underline{{\boldsymbol{Q}}} (L)\bm x_{n\vsbf}= \gamma_{\ell\vsbf}^n+\delta_{\ell\vsbf}^n.
\eeq
For any given $m\in\mathbb N$, let $\bm\Sigma_m^\delta(\bm\theta)$ be the spectral density matrix of $\bm\delta_m=\{(\delta_{1\vsbf}\cdots \delta_{m\vsbf} )^\top, \vsbf\in \mathbb Z^3\}$ and for $n>m$ 
let $\bm\Sigma_m^{\delta^n}(\bm\theta)$ be the spectral density matrix of $\bm\delta_m^n=\{(\delta_{1\vsbf}^n\cdots \delta_{m\vsbf}^n )^\top, \vsbf\in \mathbb Z^3\}$, where $\delta_{\ell\vsbf}^n = x_{\ell\vsbf}-\gamma_{\ell\vsbf}^n$ (see \eqref{eq:app:gd}). Then, from Theorem \ref{Mythm5}, we have that 
$$
\lim_{n\to\infty}\gamma_{\ell\vsbf}^n=  \gamma_{\ell\vsbf}, 
$$
in mean-square, with $\gamma_{\ell\vsbf}$, $\ell\in\mathbb N$, defined in \eqref{eq:app:CD}. Thus,
$$
\lim_{n\to\infty}\delta_{\ell\vsbf}^n=\delta_{\ell\vsbf}
$$
in mean-square for any $\ell\le m$. Notice that, although Theorem \ref{Mythm5} is proved in the next section, its proof only requires (i) and (ii) in Theorem \ref{Th. q-DFS} to hold as in this proof, so there is no feedback loop between the two theorems.

By Lemma \ref{Mylemma9}, a sub-sequence of $\bm\Sigma_m^{\delta^n}(\bm\theta)$ converges to $\bm\Sigma_m^\delta(\bm\theta)$, $\mathcal L\text{-a.e. in } \bm\Theta$. 
\[
\lim_{n\to\infty} \left\Vert\bm\Sigma_m^{\delta^n}(\bm\theta)-\bm\Sigma_m^{\delta}(\bm\theta) \right\Vert_F = 0, \quad \mathcal L\text{-a.e. in } \bm\Theta.
\]
where $\Vert\cdot\Vert_F$ denotes the Frobenius norm.
Then, by definition of the spectral norm, denoted as $\Vert\cdot\Vert$,
\beq\label{my27}
\lim_{n\to\infty} \left\vert\lambda_{m1}^{\delta^n}(\bm\theta)- \lambda_{m1}^{\delta}(\bm\theta) \right\vert=\lim_{n\to\infty} \left\Vert\bm\Sigma_m^{\delta^n}(\bm\theta)-\bm\Sigma_m^{\delta}(\bm\theta) \right\Vert\le \lim_{n\to\infty} \left\Vert\bm\Sigma_m^{\delta^n}(\bm\theta)-\bm\Sigma_m^{\delta}(\bm\theta) \right\Vert_F=0.
\eeq

Moreover, since $\bm\Sigma_m^{\delta^n}(\bm\theta)$ is the upper-left $m\times m$ submatrix of $\bm\Sigma_n^{\delta^n}(\bm\theta)$, we have by Weyl's inequality and definition of $\delta_n^n$ in \eqref{eq:app:gd},
\[
\lambda_{m1}^{\delta^n}(\bm\theta)\le \lambda_{n1}^{\delta^n}(\bm\theta)= \lambda_{n,q+1}^{x}(\bm\theta),
\]
for any $m\le n$, any $n\in\mathbb N$, and any $\bm\theta\in\bm\Theta$. Therefore, letting $\lambda_{q+1}^{x}(\bm\theta)= \lim_{n\to\infty} \lambda_{n,q+1}^{x}(\bm\theta)$,
from \eqref{my27}
\[
\lambda_{m1}^{\delta}(\bm\theta) \le \lambda_{q+1}^{x}(\bm\theta), \quad \mathcal L\text{-a.e. in } \bm\Theta,
\]
and since this is true for any $m\in\mathbb N$, then, letting $\lambda_{1}^{\delta}(\bm\theta) =\lim_{m\to\infty} \lambda_{m1}^{\delta}(\bm\theta)$,
\beq\label{DUE}
\lambda_{1}^{\delta}(\bm\theta) \le \lambda_{q+1}^{x}(\bm\theta), \quad \mathcal L\text{-a.e. in } \bm\Theta. 
\eeq
So by (i) in Theorem \ref{Th. q-DFS}, $\lambda_{1}^{\delta}$ is essentially bounded, and by Proposition \ref{Theorem 1}, $\bm\delta$ is idiosyncratic.
\end{proof}

To conclude, by Weyl's inequality and Propositions \ref{Theorem 1} and \ref{deltaidio}, we have 
$$
\lambda_{nq}^{\gamma}(\thbf) \geq \lambda_{nq}^{x}(\thbf) - \lambda_{n1}^{\delta}(\thbf)\ge  \lambda_{nq}^{x}(\thbf),
$$ 
for any given $n\in\mathbb N$ and $\bm\theta\in\bm\Theta$. Therefore, given (ii) of Theorem \ref{Th. q-DFS}, 
\beq\label{UNO}
 \lambda_{q}^{\gamma}(\thbf) =\lim_{n\to\infty} \lambda_{nq}^{\gamma}(\thbf) =\infty, \quad \mathcal L\text{-a.e. in } \bm\Theta.
\eeq

By \eqref{DUE} and \eqref{UNO}, we showed that if (i) and (ii) in Theorem \ref{Th. q-DFS} hold then (iv) and (v) in Definition \ref{Def. q_DFS} hold and we can write decomposition \eqref{Eq. x-decomp} with 
idiosyncratic component $\xi_{\ell\vsbf} = \delta_{\ell\vsbf}$ and common component $\chi_{\ell\vsbf}= \gamma_{\ell\vsbf}=\underline{\bm c}_\ell(L) \bm z_{\vsbf}$ with 
$\bm c_\ell\in L_2^q(\bm\Theta,\mathbb C)$. Because of Remark \ref{Rem_Indet} we can always find a transformation such that we can also write $\chi_{\ell\vsbf}=\underline{\bm b}_\ell(L) \bm u_{\vsbf}$ as in \eqref{Eq. x-decomp2} this proves part (i) and (ii) in Definition \ref{Def. q_DFS}. Finally, part (iii) in Definition \ref{Def. q_DFS} follows from orthogonality of $\gamma_{\ell\vsbf}$ and $\delta_{\ell\vsbf}$ in the canonical decomposition \eqref{eq:app:CD}. This completes the proof of the sufficient condition.


\subsection{Proof of Theorem \ref{Th. q-DFS} - necessary condition}\label{app:necFL}

From (iii) in Definition \ref{Def. q_DFS}
\[
\Sgmbf(\bm\theta)=\bm\Sigma_n^\chi(\bm\theta)+\bm\Sigma_n^\xi(\bm\theta).
\]
Then, for all $\bm\theta\in\bm\Theta$ and any $n\in\mathbb N$, by Weyl's inequality and since, by definition, $\bm\Sigma_n^\xi(\bm\theta)$ is positive semi-definite for all $n\in\mathbb N$,
\[
\lambda_{nq}^x(\bm\theta) \ge \lambda_{nq}^\chi(\bm\theta)+\lambda_{nn}^\xi(\bm\theta)\ge \lambda_{nq}^\chi(\bm\theta).
\]
By taking the limit for $n\to\infty$ and because of (v) in Definition \ref{Def. q_DFS}, we prove (ii) in Theorem \ref{Th. q-DFS}. 
Again by Weyl's inequality, for all $\bm\theta\in\bm\Theta$ and any $n\in\mathbb N$, we also have
\[
\lambda_{n,q+1}^x(\bm\theta)\le \lambda_{n,q+1}^\chi(\bm\theta)+\lambda_{n1}^\xi(\bm\theta)=\lambda_{n1}^\xi(\bm\theta)
\]
By taking the limit for $n\to\infty$ and because of (iv) in Definition \ref{Def. q_DFS}, we prove (i) in Theorem \ref{Th. q-DFS}. This completes the proof of the necessary condition.

\subsection{Proof of Corollary \ref{Mythm3}}\label{app:spanu}

Supposing that $\boldsymbol{x}$ is a $q$-GSTFM with representation (\ref{Eq. x-decomp})-(\ref{Eq. x-decomp2}), as we have shown,  $\boldsymbol{x}$ also has the canonical representation
$$x_{\ell \vsbf} = \gamma_{\ell \vsbf} + \delta_{\ell \vsbf},$$ 
where $\gamma_{\ell \vsbf} = \text{proj}(x_{\ell \vsbf}|\mathcal{G}(\boldsymbol{x})) = \underline{\boldsymbol{c}}_\ell(L)\boldsymbol{z}_{\vsbf}$ and $\boldsymbol{z}_{\vsbf}$ is a $q$-dimensional orthonormal white noise rf and $\spn(\boldsymbol{z}) = \mathcal{G}(\boldsymbol{x})$. Since $\boldsymbol{\xi}$ is idiosyncratic, we have $\mathcal{G}(\boldsymbol{x}) \subseteq \spn(\boldsymbol{\chi})$, which, by noting that $\spn(\boldsymbol{\chi}) \subseteq \spn(\boldsymbol{u})$, entails $\spn(\boldsymbol{z}) \subseteq \spn(\boldsymbol{u})$. On the other hand, since both $\boldsymbol{z}$ and $\boldsymbol{u}$ are $q$-dimensional white noise rf, we have $\spn(\boldsymbol{z}) = \spn(\boldsymbol{u})$. Therefore,
$\mathcal{G}(\boldsymbol{x}) = \spn(\boldsymbol{\chi}) = \spn(\boldsymbol{u}).$
This implies that $\chi_{\ell \vsbf} \in \mathcal{G}(\boldsymbol{x})$ and $\xi_{\ell \vsbf} \perp \mathcal{G}(\boldsymbol{x})$, so that $\chi_{\ell \vsbf}=\gamma_{\ell\vsbf} = \text{\rm proj}(x_{\ell \vsbf}|\mathcal{G}(\boldsymbol{x}))$ and $\xi_{\ell \vsbf} = \delta_{\ell \vsbf}$. Uniqueness follows from uniqueness of the canonical representation.
This completes the proof.


\section{Proof of results of Section 5}\label{appC}
\subsection{Proof of Theorem \ref{Mythm5}}\label{app:DPCA}

 The proof requires the following definition and preliminary lemmas.

\begin{definition} [Cauchy sequence of spaces]\label{defCauchy}
For any $ n\in \mathbb{N}$, let $\boldsymbol{v}_n = \{\boldsymbol{v}_{n\vsbf}, \vsbf \in \mathbb{Z}^3\}$ be a $q$-dimensional orthonormal white noise rf such that $\bm v_{n\vsbf}\in\bm{\mathcal X}$ and is co-homostationary with $x_\ell$, $\ell\in\mathbb N$, so that $\boldsymbol{v}_n$ and $\boldsymbol{v}_m$ are co-homostationary for any $n$ and $m$. 
Denote by $\boldsymbol{A}^{mn}(\thbf)$ the $q\times q$ matrix whose $(h, k)$ entry is the cross spectrum $\mathcal{S}(v_{mh\vsbf}, v_{nk\vsbf}; \thbf)$, for $\thbf \in\Thbf$. The orthogonal projection, element by element, of $\boldsymbol{v}_{m\vsbf}$ on the process $\boldsymbol{v}_n$ is $\underline{\boldsymbol{A}}^{mn}(L) \boldsymbol{v}_{n\vsbf}$.
Consider the orthogonal decomposition
\begin{equation}\label{EqCauchy}
\boldsymbol{v}_{m\vsbf} = \underline{\boldsymbol{A}}^{mn}(L) \boldsymbol{v}_{n\vsbf} + \boldsymbol{\rho}_{\vsbf}^{mn},
\end{equation}
and let $\boldsymbol{\varrho}^{mn}(\thbf)$ denote the spectral density matrix of $\boldsymbol{\rho}_{\vsbf}^{mn}$. The sequence $\{\boldsymbol{v}_n, n \in \mathbb{N}\}$ generates a Cauchy sequence of spaces if, for a given $\epsilon >0$ and $\mathcal L$-a.e.~in $\Thbf$, there exists an integer $m_\epsilon (\thbf)$ such that for $n, m > m_\epsilon (\thbf)$, $\text{\rm trace}(\boldsymbol{\varrho}^{mn}(\thbf)) < \epsilon$.

\end{definition}

\begin{lemma}\label{Mylemma12}
Assume that $\{\boldsymbol{v}_{n}, n \in \mathbb{N}\}$ fulfills Definition \ref{defCauchy} and $y = \{y_{\vsbf }, \vsbf \in \mathbb{Z}^3\}$ is such that $y_{\vsbf}\in \bm{\mathcal X}$ and is co-homostationary with  $x_\ell$, $\ell\in\mathbb N$. Let $Y_{n\vsbf }$ be the orthogonal projection of $y_{\vsbf }$ on the process $\boldsymbol{v}_{n}$, i.e., $Y_{n\vsbf } = \text{\rm proj}(y_{\vsbf}|\text{\rm $\spn$}(\boldsymbol{v}_{n}))$. Then,  $Y_{n\vsbf }$ converges in $\bm{\mathcal X}$ in mean-square, as $n\to\infty$. 
\end{lemma}

\begin{proof}
Considering projections
$$y_{\vsbf} = Y_{n\vsbf} + r_{n\vsbf} = \underline{\boldsymbol{b}}_n(L) \boldsymbol{v}_{n\vsbf} + r_{n\vsbf},$$
$$y_{\vsbf} = Y_{m\vsbf} + r_{m\vsbf} = \underline{\boldsymbol{b}}_m(L) \boldsymbol{v}_{m\vsbf} + r_{m\vsbf},$$
where  ${\boldsymbol{b}}_n, {\boldsymbol{b}}_m \in L_2^q(\Thbf, \mathbb{C})$, this yields
$$\underline{\boldsymbol{b}}_n(L) \boldsymbol{v}_{n\vsbf}  - \underline{\boldsymbol{b}}_m(L) \boldsymbol{v}_{m\vsbf} = r_{m\vsbf} - r_{n\vsbf}.$$
Now we show the spectral density of the rf on the left-hand side converges to zero $\mathcal L$-a.e. in $\Thbf$. Note that the spectral density of the rf on  left-hand side is the cross spectrum between the left- and right- hand sides, which, due to the definition of  $r_{n\vsbf}$ and $r_{m\vsbf}$, is the sum of two cross spectra: $\mathcal{S}(r_{n\vsbf}, \underline{\boldsymbol{b}}_m(L) \boldsymbol{v}_{m\vsbf}; \thbf) + \mathcal{S}(r_{m\vsbf}, \underline{\boldsymbol{b}}_n(L) \boldsymbol{v}_{n\vsbf}; \thbf).$ In view of (\ref{EqCauchy}), we have
$$\mathcal{S}(r_{n\vsbf}, \underline{\boldsymbol{b}}_m(L) \boldsymbol{v}_{m\vsbf}; \thbf) 
= \mathcal{S}(r_{n\vsbf}, \underline{\boldsymbol{b}}_m(L)\star \underline{\boldsymbol{A}}^{mn}(L) \boldsymbol{v}_{n\vsbf} + \underline{\boldsymbol{b}}_m(L)  \boldsymbol{\rho}_{\vsbf}^{mn}; \thbf) 
= \mathcal{S}(r_{n\vsbf},  \underline{\boldsymbol{b}}_m(L)  \boldsymbol{\rho}_{\vsbf}^{mn}; \thbf).$$
Note that both the spectral density of $r_{n\vsbf}$ and the squared entries of ${\boldsymbol{b}}_m$ are bounded in modulus by the spectral density of $y_{\vsbf}$. Hence,  the fact that $\{\boldsymbol{v}_{n}, n \in \mathbb{N}\}$ generates a Cauchy sequence of spaces implies that $\mathcal{S}(r_{n\vsbf},  \underline{\boldsymbol{b}}_m(L)  \boldsymbol{\rho}_{\vsbf}^{mn}; \thbf)$ converges to zero $\mathcal L$-a.e. in $\Thbf$ as $m, n \rightarrow \infty$. Similar argument holds also for $\mathcal{S}(r_{m\vsbf}, \underline{\boldsymbol{b}}_n(L) \boldsymbol{v}_{n\vsbf}; \thbf)$. Therefore, the spectral density of $Y_{n\vsbf} - Y_{m\vsbf}$ converges to zero $\mathcal L$-a.e. in $\Thbf$ as $m, n \rightarrow \infty$. Since the spectral densities of $Y_{n\vsbf}$ and $Y_{m\vsbf}$ are dominated by that of $y_{\vsbf}$, by Lebesgue's dominated convergence theorem, the integral of 
the spectral density of $Y_{n\vsbf} - Y_{m\vsbf}$ converges to zero as $m, n \rightarrow \infty$, which implies that $Y_{n\vsbf}$ is a Cauchy sequence and thus converges in $\bm{\mathcal X}$, as $n\to\infty$.
\end{proof}

\begin{lemma}\label{Mylemma13}
Suppose that (i) and (ii) of Theorem \ref{Th. q-DFS} and Assumptions \ref{Ass.A}-\ref{Ass.C} hold. Then, $\{{\bfpsi}^n, n\in \mathbb{N}\}$, as defined in \eqref{B10}, generates a Cauchy sequence of spaces.
\end{lemma}

\begin{proof}
For $n >m$, in (\ref{Eq.Cpsi}), letting $\boldsymbol{C} = \boldsymbol{I}_q$ yields
\begin{equation}\label{Eq.psim}
{\bfpsi}^m_{\vsbf} = \underline{\boldsymbol{D}} (L)  {\bfpsi}^n_{\vsbf} + \boldsymbol{\rho}_{\vsbf}^{mn}.
\end{equation}
where 
$
{\boldsymbol{D}}(\thbf) =  {\boldsymbol{\Lambda}}^{-1/2}_m(\thbf) {\boldsymbol{P}}_m(\thbf) {\boldsymbol{P}}^\dag_n(\thbf) {\boldsymbol{\Lambda}}^{1/2}_n(\thbf). 
$
Let $\boldsymbol{\varrho}^{mn}(\thbf)$ denote the spectral density matrix of $\boldsymbol{\rho}_{\vsbf}^{mn}$. Lemma \ref{MyLemma7} and (ii) of Theorem \ref{Th. q-DFS} imply that $\text{trace}(\boldsymbol{\varrho}^{mn}(\thbf))$ converges to zero $\mathcal L$-a.e. in $\Thbf$. 

On the other hand,
\begin{equation}\label{Eq.psin}
{\bfpsi}^n_{\vsbf} = \underline{\boldsymbol{D}}^\dag (L)  {\bfpsi}^m_{\vsbf} + \boldsymbol{\rho}_{\vsbf}^{nm}.
\end{equation}
Using (\ref{Eq.psim}) and (\ref{Eq.psin}), we have
$$\boldsymbol{I}_q = {\boldsymbol{D}}(\thbf) {\boldsymbol{D}}^\dag(\thbf) + \boldsymbol{\varrho}^{mn}(\thbf)  
= {\boldsymbol{D}}^\dag(\thbf) {\boldsymbol{D}}(\thbf) + \boldsymbol{\varrho}^{nm}(\thbf)$$
$\mathcal L$-a.e. in $\Thbf$. Taking the trace on both sides and noting that the trace of ${\boldsymbol{D}}(\thbf) {\boldsymbol{D}}^\dag(\thbf)$ is equal to that of ${\boldsymbol{D}}^\dag(\thbf) {\boldsymbol{D}}(\thbf)$, we have $\text{trace}(\boldsymbol{\varrho}^{mn}(\thbf)) = \text{trace}(\boldsymbol{\varrho}^{nm}(\thbf))$ $\mathcal L$-a.e. in $\Thbf$. Finally, $\text{trace}(\boldsymbol{\varrho}^{mm}(\thbf)) = 0$. Therefore, $\text{trace}(\boldsymbol{\varrho}^{mn}(\thbf))$ converges to zero $\mathcal L$-a.e. in $\Thbf$ for any $m, n \rightarrow \infty$.
\end{proof}

Finally, let us consider again the orthogonal decomposition in (\ref{eq:app:gd}) for any $\ell\le n$, i.e.,
$$x_{\ell \vsbf} = \underline{{\boldsymbol{\pi}}}_{n\ell} (L) \star\underline{\boldsymbol{\Lambda}}^{1/2}_n (L) {\bfpsi}^n_{\vsbf} +  \underline{{\boldsymbol{q}}}_{n\ell}(L)\star \underline{{\boldsymbol{Q}}}_n(L) \boldsymbol{x}_{n\vsbf}= \gamma_{\ell\vsbf}^n+\delta_{\ell\vsbf}^n.
$$
Due to Lemma \ref{Mylemma12} and Lemma \ref{Mylemma13}, 
$$
\lim_{n\to\infty}\gamma_{\ell\vsbf}^n=  \gamma_{\ell\vsbf}
$$
in mean-square and $\gamma_{\ell\vsbf}\in \bm{\mathcal X}$. This entails that 
$$
\lim_{n\to\infty}\delta_{\ell\vsbf}^n =x_{\ell\vsbf} - \gamma_{\ell\vsbf}=\delta_{\ell\vsbf}
$$
in mean-square, with $\delta_{\ell\vsbf}\in\bm{\mathcal X}$.

Moreover, $\gamma_{\ell\vsbf}$ is an aggregate, i.e., it belongs to $\mathcal G(\bm x)$, since ${\boldsymbol{\pi}}_{n\ell}{\boldsymbol{P}}_{n}$ is a STDAS. To see this, note that the spectral density matrix of $\gamma_{\ell \vsbf}^n$, which is ${\boldsymbol{\pi}}_{n\ell}(\thbf){\boldsymbol{\Lambda}}_{n}(\thbf) {\boldsymbol{\pi}}^\dag_{n\ell}(\thbf)$, is not smaller than $\lambda_{nq}^x(\thbf) {\boldsymbol{\pi}}_{n\ell}(\thbf) {\boldsymbol{\pi}}^\dag_{n\ell}(\thbf)$ and it is bounded above by the spectral density of $x_{\ell\vsbf}$, call it $\sigma^2_{\ell}(\thbf)$. Therefore, we have ${\boldsymbol{\pi}}_{n\ell}(\thbf) {\boldsymbol{\pi}}^\dag_{n\ell}(\thbf) \leq \sigma^2_{\ell}(\thbf)/\lambda_{nq}^x(\thbf)$. 
Assuming again without loss of generality that $\lambda_{nq}^x(\thbf)>1$ for all $\thbf\in\bm\Theta$ (see the arguments in \citealp[Section 4.2]{FL01}), it follows that $\sigma_{\ell}^2(\thbf)/\lambda_{nq}^x(\thbf)$ is bounded above by $\sigma^2_{\ell}(\thbf)$ and it converges to zero $\mathcal L$-a.e. in $\Thbf$, because of (ii) in Theorem \ref{Th. q-DFS}. Hence, by Lebesgue's dominated convergence theorem, $\int_{\Thbf}\sigma_{\ell}(\thbf)/\lambda_{nq}^x(\thbf)\mathrm d\thbf$ converges to zero.

By construction, $\delta_{\ell\vsbf}^n$ is orthogonal to ${\bfpsi}^n_{\vsbf-\hbf}$ for any $\hbf \in \mathbb{Z}^3$. Since $\mathcal{G}(\boldsymbol{x}) = \spn(\boldsymbol{z})$ by Proposition \ref{MySection4.5} and the process $\boldsymbol{z}$ has been obtained by taking limits of linear combinations of the elements of $\bm\psi^n$, continuity of the inner product implies that $\delta_{\ell \vsbf} \perp \mathcal{G}(\boldsymbol{x})$.
 Then, by uniqueness of the canonical decomposition  we have that $\gamma_{\ell\vsbf}=\chi_{\ell\vsbf}$ for all $\vsbf\in\mathbb Z^3$. This completes the proof.

 \section{Some novel results on spectral density matrix estimation for  spatio-temporal random fields}\label{Spec.Y}
In this section we derive some novel results about spectral density estimation of spatio-temporal rf. Our theory builds on and extends the one
already available in \citet{DPW17}, which is available for spatial processes only, and the one of \citet{wu2018asymptotic} and \citet{wu2005nonlinear}, which study the time series case. 

Consider a generic $n$-dimensional rf:
$\Y_{n} =\{\bm Y_{n\vsbf}= (Y_{1\vsbf} \cdots Y_{n\vsbf})^\top,\vsbf\in\mathbb Z^3\}$ where, for any $\ell=1,\ldots, n$
\beq
Y_{\ell\vsbf} = F_\ell(\Epsbf_{\vsbf - \k}; \k \in \mathbb{Z}^2 \times \mathbb N_0), \label{Eq_Fl}
\eeq
for some  measurable function $F_\ell(\cdot)$ a such that $Y_{\ell\vsbf}$ exists and a, possibly infinite dimensional, zero-mean and i.i.d. rf $\{\Epsbf_{\vsbf} = (\epsilon_{1\vsbf} \ \epsilon_{2\vsbf} \cdots)^\top; \vsbf \in \mathbb{Z}^3\}$. 
Let $Y_{\ell\vsbf} \in L_p$, for some $p \geq 1$, and denote as $Y^*_{\ell\vsbf} = F_\ell(\Epsbf^*_{\vsbf - \k}; \k \in \mathbb{Z}^2 \times \mathbb N_0)$, with $
\Epsbf^*_{\vsbf}=  \Epsbf_{\vsbf}$ if $\vsbf \neq \bm 0$ and $\Epsbf^*_{\0}= \widetilde{\Epsbf}_{\0}$, where $\widetilde{\Epsbf}_{\vsbf_1}, \Epsbf_{\vsbf_2}$ are 
i.i.d. for $\vsbf_1, \vsbf_2 \in \mathbb{Z}^3$.  Recall that we can arbitrarily set the location of the origin $\boldsymbol 0$ because of the Assumption \ref{Ass.A} 
(homostationarity).  Then, we consider the following {\it functional dependence measure} 
\begin{equation}\label{Def.delta}
\delta_{\vsbf, p}^{[\ell]} = \left( {\rm E} \vert Y_{\ell\vsbf} - Y^*_{\ell\vsbf}\vert^p \right)^{1/p}.
\end{equation}
Let also
$\m = (m_1 \ m_2 \ m_3)^\top$ and define
$$
\varphi_{\m, p}^{[\ell]} = \sum_{|s_1| > m_1} \sum_{|s_2| > m_2} \sum_{t > m_3} \delta_{\vsbf, p}^{[\ell]}.
$$
The definition of $\varphi_{\m, p}^{[\ell]}$ is natural in our spatio-temporal rf setting: it measures two-sided dependence over space and one-sided dependence over time for $Y_{\ell\vsbf}$. Therefore, $\varphi_{\m, p}^{[\ell]}$ considers a form of dependence which acts in each direction of $\mathbb{Z}^3$: this is different from the approach of \cite[Definition 2.1 and see also their discussion on p.4315]{DPW17}, where the dependence can be only in one spatial dimension.



Given the sample $\{Y_{\ell\vsbf}=x_{\ell(s_1\ s_2\ t)},\ \ell=1,\ldots, n,\ s_1=1,\ldots, S_1,\ s_2=1,\ldots, S_2,\ t=1,\ldots, T\}$, we consider the estimator of the spectral spectral density matrix  $\widehat{\boldsymbol \Sigma}_n^y(\thbf)$ of $\bm Y_n$ with $(i,j)$-th, $i,j=1.\ldots, n$, generic entry 
\begin{equation}\label{hatsigma1}
\widehat{\sigma}_{ij}^y(\thbf) =\frac{1}{S_1 S_2 T} \sum_{\vsbf_1,\vsbf_2 = (1\ 1\ 1)^\top}^{ (S_1 \ S_2 \ T)^\top} \!\!\!\!   Y_{i\vsbf_1} Y_{j\vsbf_2} K_1\left(\frac{s_{11}-s_{21}}{B_{S_1}} \right) K_2\left(\frac{s_{12}-s_{22}}{B_{S_2}} \right) K_3\left(\frac{t_{1}-t_{2}}{B_{T}} \right) e^{-i \left\langle \vsbf_1 - \vsbf_2,  \thbf\right\rangle}
\end{equation}
with $\vsbf_1 = (s_{11} \ s_{12} \ t_1), \vsbf_2 = (s_{21} \ s_{22} \ t_2)$, $K_1(\cdot), K_2(\cdot), K_3(\cdot)$ being kernel functions and $B_{S_1}, B_{S_2}, B_{T}$ being bandwidths satisfying Assumption \ref{AssSpec2}(i) and (ii).

The following conditions are imposed 

\begin{cond}\label{condA}
For all $\ell \in \mathbb{N}$, ${\rm E} (Y_{\ell\0}) = 0$, $Y_{\ell\0} \in L_p$, for some $p \geq 4$.
\end{cond}

\begin{cond}\label{condB}
For all $\ell \in \mathbb{N}$, $\delta_{\vsbf, p}^{[\ell]} \leq A_1 \rho_1^{|s_1|} \rho_2^{|s_2|} \rho_3^{|t|}$, for some finite $\rho_1, \rho_2, \rho_3 \in (0, 1)$ and $A_1 >0$, independent of $\ell$.
\end{cond}

\begin{cond}\label{condC}
For all $i, j \in \mathbb{N}$, $\vert{\rm E} (Y_{i\vsbf} Y_{j\0})\vert \le A_2 \rho_4^{|s_1|} \rho_5^{|s_2|} \rho_6^{|t|}$, for some finite $\rho_4, \rho_5, \rho_6 \in (0, 1)$ and $A_2 \in (0, \infty)$, independent of $i$ and $j$.
\end{cond}

Then, we have the following result

\begin{proposition}\label{ThmSpecDen1}
Let Conditions \ref{condA} and \ref{condB}, and Assumption \ref{AssSpec2}  hold. Then, there exists a finite $C>0$ depending on $p$ such that
\begin{equation}
\max_{1\leq i, j \leq n} \sup_{\thbf \in \Thbf} \left[ {\rm E} \left| \widehat{\sigma}_{ij}^y(\thbf)  - {\rm E} (\widehat{\sigma}_{ij}^y(\thbf)) \right|^{p/2} \right]^{2/p} \leq C \left[ \frac{(\log B_{S_1}  \log B_{S_2}  \log B_{T} )^2 B_{S_1} B_{S_2} B_{T}}{S_1 S_2 T} \right]^{1/2}.\nn
\end{equation}
\end{proposition}


For $p = 4$, Proposition~\ref{ThmSpecDen1} bounds the variance of $\widehat{\sigma}_{ij}^y(\thbf)$. The following result bounds the mean square error 

%

\begin{proposition}\label{ThmSpecDen2}
Let Conditions \ref{condA}, \ref{condB}, and \ref{condC}, and Assumption \ref{AssSpec2} hold. Then, there exists a finite $C^*>0$ such that 
\begin{align}
&\quad \max_{1\leq i, j \leq n} \sup_{\thbf \in \Thbf} {\rm E} \left| \widehat{\sigma}_{ij}^y(\thbf)  - \sigma_{ij}^y(\thbf) \right|^2 \leq  C^* \max\left\lbrace \frac{(\log B_{S_1} \log B_{S_2}  \log B_{T} )^2 B_{S_1} B_{S_2} B_{T}}{S_1 S_2 T},  \frac 1{B_{S_1}^{2\vartheta_1}},  \frac 1{B_{S_2}^{2\vartheta_2}},  \frac 1{B_T^{2\vartheta_3}} \right\rbrace. \nn
\end{align}
\end{proposition}

This implies mean-square consistency of $\widehat{\sigma}_{ij}^y(\thbf)$, as $S_1,S_2,T\to\infty$.

\subsection{Proof of Proposition \ref{ThmSpecDen1}}\label{App.pfSigmaYX}

The following preliminary result is needed to prove Proposition \ref{ThmSpecDen1}.

\begin{lemma}\label{lemYYtilde}
For any $i=1,\ldots, n$, let $Y_{i\vsbf} \in L_p$ for some $p > 1$, ${\rm E}Y_{i\vsbf} = 0$, $\alpha_{\vsbf} \in \mathbb{C}$ for $\vsbf \in \mathbb{Z}^3$,  $A_{\bar{\vsbf}} = (\sum_{\vsbf = \1}^{\bar{\vsbf}} |\alpha_{\vsbf}|^{2})^{1/2}$ and $C_p = 18p^{3/2}(p-1)^{-1/2}$. Then, for  $\bar{\vsbf}=(S_1\ S_2\ T)^\top$, we have 
$$\left({\rm E} \left| \sum_{\vsbf = \1}^{\bar{\vsbf}} \alpha_{\vsbf} Y_{i\vsbf} \right|^p \right)^{1/p} \leq C_p A_{\bar{\vsbf}} \varphi_{\0, p}^{[i]}, \quad
\left({\rm E} \left| \sum_{\vsbf = \1}^{\bar{\vsbf}} \alpha_{\vsbf} \widetilde{Y}_{i\vsbf} \right|^p \right)^{1/p}  \leq C_p A_{\bar{\vsbf}} \varphi_{\0, p}^{[i]}$$
and
$$\left({\rm E} \left| \sum_{\vsbf = \1}^{\bar{\vsbf}} \alpha_{\vsbf} (Y_{i\vsbf} - \widetilde{Y}_{i\vsbf}) \right|^p \right)^{1/p} \leq C_p A_{\bar{\vsbf}} \varphi_{\m+\1, p}^{[i]}.$$
\end{lemma}

\begin{proof}
Let $\taubf: \mathbb{Z} \rightarrow \mathbb{Z}^3$ be a bijection, which is non-decreasing in the third coordinate.
 Define for $l \in \mathbb{Z}$ and $(\taubf(l))_3 \leq t$, where $(\taubf(l))_3$ means the third coordinate of $\taubf(l)$,
 the projection
$$P_l(Y_{i\vsbf}) = {\rm E}(Y_{i\vsbf} | \mathcal{F}_l) - {\rm E}(Y_{i\vsbf} | \mathcal{F}_{l-1}),$$ where 
$\mathcal{F}_l = \sigma(\Epsbf_{\taubf(h)}; h \leq l).$ Define for $\k \in \mathbb{Z}^2\times \mathbb{N}_0$, the shift  $\mathcal{T}^{\k} \mathcal{F}_l = \sigma(\Epsbf_{\taubf(h) - \k}; h \leq l).$ By (\ref{Eq_Fl}), we then have
\begin{align*}
&\quad \left({\rm E} \left|  P_l(Y_{i\vsbf})  \right|^p \right)^{1/p} \\
 &= \left({\rm E} \left| {\rm E}(Y_{i\vsbf} | \mathcal{F}_l) - {\rm E}(Y_{i\vsbf} | \mathcal{F}_{l-1})  \right|^p \right)^{1/p} \\
&= \left({\rm E} \left| {\rm E}(Y_{i\0} | \mathcal{T}^{\vsbf} \mathcal{F}_l) - {\rm E}(Y_{i\0} | \mathcal{T}^{\vsbf} \mathcal{F}_{l-1}) \right|^p \right)^{1/p} \\
&= \left\lbrace {\rm E} \left|  {\rm E} [F_i(\Epsbf_{-\k}; \k \in \mathbb{Z}\times \mathbb{Z} \times \mathbb{N}_0) | \mathcal{T}^{\vsbf} \mathcal{F}_l] - 
{\rm E} [F_i(\Epsbf_{-\k}, \widetilde{\Epsbf}_{\taubf(l) - \vsbf}; \k \in \mathbb{Z}\times \mathbb{Z} \times \mathbb{N}_0 \setminus   \{\vsbf - \taubf(l)\}) | \mathcal{T}^{\vsbf} \mathcal{F}_l] \right|^p \right\rbrace^{1/p} \\
&\leq \left({\rm E} \left|  F_i(\Epsbf_{-\k}; \k \in \mathbb{Z}\times \mathbb{Z} \times \mathbb{N}_0) - 
F_i(\Epsbf_{-\k}, \widetilde{\Epsbf}_{\taubf(l) - \vsbf}; \k \in \mathbb{Z}\times \mathbb{Z} \times \mathbb{N}_0 \setminus   \{\vsbf - \taubf(l)\}) \right|^p \right)^{1/p} \\
&= \left({\rm E} \left|  F_i(\Epsbf_{\vsbf - \taubf(l) -\k}; \k \in \mathbb{Z}\times \mathbb{Z} \times \mathbb{N}_0) - 
F_i(\Epsbf_{\vsbf - \taubf(l) -\k}, \widetilde{\Epsbf}_{\0}; \k \in \mathbb{Z}\times \mathbb{Z} \times \mathbb{N}_0 \setminus   \{\vsbf - \taubf(l)\}) \right|^p \right)^{1/p} \\
&=  \left({\rm E} \left|  Y_{i \vsbf - \taubf(l)} - Y^*_{i \vsbf - \taubf(l)} \right|^p \right)^{1/p} \\
&= \delta_{\vsbf - \taubf(l), p}^{[i]}.
\end{align*}

Let $\mathbb{E}_t =\{l \in \mathbb{Z}: (\taubf(l))_3 \leq t\}$. Noting that $\sum_{\vsbf = \1}^{\bar{\vsbf}} \alpha_{\vsbf} Y_{i\vsbf} = \sum_{l \in \mathbb{E}_T} P_l(\sum_{\vsbf = \1}^{\bar{\vsbf}} \alpha_{\vsbf} Y_{i\vsbf})$, we have
$$\left({\rm E} \left| \sum_{\vsbf = \1}^{\bar{\vsbf}} \alpha_{\vsbf} Y_{i\vsbf}  \right|^p \right)^{1/p}
= \left({\rm E} \left| \sum_{l \in \mathbb{E}_T} P_l\left(\sum_{\vsbf = \1}^{\bar{\vsbf}} \alpha_{\vsbf} Y_{i\vsbf}\right) \right|^p \right)^{1/p},$$
where $\{P_l\left(\sum_{\vsbf = \1}^{\bar{\vsbf}} \alpha_{\vsbf} Y_{i\vsbf}\right); l \in \mathbb{E}_T \}$ forms a martingale difference sequence in time, that is, 
$${\rm E}\left[P_l\left(\sum_{\vsbf = \1}^{\bar{\vsbf}} \alpha_{\vsbf} Y_{i\vsbf}\right)\vert \mathcal{F}_{l-1}\right] = 0.$$ Hence, by Minkowski's, Cauchy-Schwarz's and Burkholder's inequalities (see \citealt[Lemma 1]{WS07}),  
\begin{align*}
\left({\rm E} \left| \sum_{\vsbf = \1}^{\bar{\vsbf}} \alpha_{\vsbf} Y_{i\vsbf} \right|^p \right)^{2/p} 
&\leq C_p^{2}  \sum_{l \in \mathbb{E}_T} \left({\rm E} \left|  P_l\left(\sum_{\vsbf = \1}^{\bar{\vsbf}} \alpha_{\vsbf} Y_{i\vsbf}\right) \right|^p \right)^{2/p} \\
&\leq C_p^{2} \sum_{l \in \mathbb{E}_T} \left(\sum_{\vsbf = \1}^{\bar{\vsbf}} \left|\alpha_{\vsbf}\right|  \delta_{\vsbf - \taubf(l), p}^{[i]} \right)^2 \\
&\leq C_p^{2} \sum_{l \in \mathbb{E}_T} \left(  \left(\sum_{\vsbf = \1}^{\bar{\vsbf}} \left|\alpha_{\vsbf}\right|^2  \delta_{\vsbf - \taubf(l), p}^{[i]} \right) \left(\sum_{\vsbf = \1}^{\bar{\vsbf}}   \delta_{\vsbf - \taubf(l), p}^{[i]} \right) \right) \\
&\leq C_p^{2} \sum_{l \in \mathbb{E}_T} \left( \left(\sum_{\vsbf = \1}^{\bar{\vsbf}} \left|\alpha_{\vsbf}\right|^2  \delta_{\vsbf - \taubf(l), p}^{[i]} \right)  \varphi_{\0, p}^{[i]} \right) \\
  &\leq C_p^{2} \sum_{\vsbf = \1}^{\bar{\vsbf}} \alpha_{\vsbf}^2
   \left(\varphi_{\0, p}^{[i]}\right)^2,
\end{align*}
which implies $\left({\rm E} \left| \sum_{\vsbf = \1}^{\bar{\vsbf}} \alpha_{\vsbf} Y_{i\vsbf} \right|^p \right)^{1/p} \leq C_p A_{\bar{\vsbf}} \varphi_{\0, p}^{[i]}.$

Repeating the similar arguments for $\widetilde{Y}_{i\vsbf}$ yields $\left({\rm E} \left| \sum_{\vsbf = \1}^{\bar{\vsbf}} \alpha_{\vsbf} \widetilde{Y}_{i\vsbf} \right|^p \right)^{1/p} \leq C_p A_{\bar{\vsbf}} \varphi_{\0, p}^{[i]}$. Noting that 
$$\sum_{\vsbf = \1}^{\bar{\vsbf}} \alpha_{\vsbf}(Y_{i\vsbf} - \widetilde{Y}_{i\vsbf})= \sum_{l \in \mathbb{F}_T} P_l\left(\sum_{\vsbf = \1}^{\bar{\vsbf}} \alpha_{\vsbf}(Y_{i\vsbf} - \widetilde{Y}_{i\vsbf})\right),$$ where $\mathbb{F}_T = \{l \in \mathbb{Z}: |\taubf(l))_1|> m_1, |\taubf(l))_2|> m_2, T \geq (\taubf(l))_3> m_3\}$,  the last inequality follows from repeating the similar arguments as above.
\end{proof}

In the sequel, $A_1, A_2, ...$ will denote generic positive constants.
Letting $$a_{\bar{\vsbf}, \rbf} = K_1 (r_1/B_{S_1}) K_2 (r_2/B_{S_2}) K_3 (r_3/B_{T}) e^{-i\left\langle \rbf, \thbf \right\rangle},$$ where $\rbf = (r_1 \ r_2 \ r_3)^\top \in \mathbb{Z}^3$, then (\ref{hatsigma1}) can be written as
\begin{equation*}
\widehat{\sigma}_{ij}^y(\thbf) = \frac{1}{\mathcal{R}} \sum_{\vsbf_1 = \1}^{\bar{\vsbf}} \sum_{\vsbf_2 = \1}^{\bar{\vsbf}}  Y_{i\vsbf_1} Y_{j\vsbf_2} a_{\bar{\vsbf}, \vsbf_1 - \vsbf_2}.
\end{equation*}
Assumption~\ref{AssSpec2} implies that $\sum_{\rbf \in \mathbb{Z}^3} |a_{\bar{\vsbf}, \rbf}|^2/(B_{S_1} B_{S_2} B_T) < \infty$. Let $\hbf = (h_1 \ h_2 \ h_3)^\top \in \mathbb{Z}^3$ and $\mathcal R=S_1S_2T$. Then, we consider an approximation of $\widehat{\sigma}_{ij}^y(\thbf)$ by
\begin{align}
\widetilde{\sigma}_{ij}^y(\thbf) &= \frac{1}{\mathcal{R}} \sum_{\vsbf_1 = \1}^{\bar{\vsbf}} \sum_{\vsbf_2 = \1}^{\bar{\vsbf}}  \widetilde{Y}_{i\vsbf_1} \widetilde{Y}_{j\vsbf_2} a_{\bar{\vsbf}, \vsbf_1 - \vsbf_2}  \nonumber \\
&= \frac{1}{\mathcal{R}} \sum_{\vsbf_1 = \1}^{\bar{\vsbf}}  \widetilde{Y}_{i\vsbf_1} \widetilde{Y}_{j\vsbf_1}
+ \frac{2}{\mathcal{R}}  \sum_{\vsbf_2 = \2}^{\bar{\vsbf}} \widetilde{Y}_{j\vsbf_2} \sum_{s_{11}= \max (1, s_{21} - h_1)}^{s_{21}-1} \sum_{s_{12}= \max (1, s_{22} - h_2)}^{s_{22}-1} \sum_{t_1 = \max (1, t_2 - h_3)}^{t_2-1} a_{\bar{\vsbf}, \vsbf_1 - \vsbf_2} \widetilde{Y}_{i\vsbf_1} \nonumber \\
&\quad +\frac{2}{\mathcal{R}} \sum_{\vsbf_2 = \h+\1}^{\bar{\vsbf}} \widetilde{Z}_{ij, \vsbf_2}, \label{sigmatilde}
\end{align}
where $\widetilde{Z}_{ij, \vsbf_2} = \widetilde{Y}_{j\vsbf_2} \sum_{\vsbf_1 = \1}^{\vsbf_2 - \h} a_{\bar{\vsbf}, \vsbf_1 - \vsbf_2} \widetilde{Y}_{i\vsbf_1}$.

Now let us consider the $\sigma$-field 
$$\mathcal{F}_{\m, \vsbf} = \sigma(\Epsbf_{\vsbf - \k}; |k_1| \leq m_1,  |k_2| \leq m_2,  0 \leq k_3 \leq m_3)
$$ 
with $\k = (k_1 \ k_2 \ k_3)^\top \in \mathbb{Z}^2 \times \mathbb{N}_0$ and $\m = (m_1 \ m_2 \ m_3)^\top \in \mathbb{N}_0^3$. 
The $\sigma$-field $\mathcal{F}_{\m, \vsbf} $ is based on a hyperrectangle in $\mathbb{Z}^3$, where for each coordinate defining the spatio-temporal point $\vsbf- \k$,  we consider its distance form $\vsbf$.  We refer to \cite{ElMV03} p. 328 for a discussion on other options which can be considered to define this $\sigma$-field. For the sake of completeness of our argument, here we emphasize that our choice for $\mathcal{F}_{\m, \vsbf}$ is the natural one to consider in our spatio-temporal setting, where we impose homostationarity (see Assumption \ref{Ass.A}) and we are interested in controlling the rf functional dependence in each direction of $\mathbb{Z}^3$. 

Note that $\{\widetilde{Z}_{ij, \vsbf_2+l \hbf}; l \in \mathbb{N}_0\}$ are martingale differences when $|h_1| > m_1, |h_2| > m_2, h_3 > m_3$, that is, ${\rm E}[\widetilde{Z}_{ij, \vsbf_2+l \hbf} \vert \mathcal{F}_{\m, \vsbf_2+(l-1) \hbf}] = 0$. Taking  $|h_1| = 2m_1, |h_2| = 2m_2, h_3 = 2m_3$ and using Lemma~\ref{lemYYtilde}, there exists $A_3>0$ such that
\begin{align*}
\frac{ \max_{1\leq i, j \leq n} \left({\rm E} \left|	\widetilde{Z}_{ij, \vsbf_2} \right|^p \right)^{1/p}}{(B_{S_1} B_{S_2} B_T )^{1/2}}
&=  \max_{1\leq i, j \leq n} \left({\rm E} \left| \widetilde{Y}_{j\vsbf_2} \right|^p \right)^{1/p} \frac{\left({\rm E} \left| \sum_{\vsbf_1 = \1}^{\vsbf_2 - \h} a_{\bar{\vsbf}, \vsbf_1 - \vsbf_2} \widetilde{Y}_{i\vsbf_1} \right|^p \right)^{1/p}}{(B_{S_1} B_{S_2} B_T )^{1/2}}
 \\
&\leq \max_{1\leq i, j \leq n} \left({\rm E} \left| \widetilde{Y}_{j\0} \right|^p \right)^{1/p} C_p \varphi_{\0, p}^{[i]} \left(\frac{\sum_{\vsbf_1 = \1}^{\vsbf_2 - \h} |a_{\bar{\vsbf}, \vsbf_1 - \vsbf_2}|^2}{B_{S_1} B_{S_2} B_T} \right)^{1/2} \leq A_3\\
\end{align*}
due to Conditions \ref{condA} and \ref{condB}.
For the last term in (\ref{sigmatilde}), we consider splitting the interval $[h_1 + 1, S_1] \times [h_2 + 1, S_2] \times [h_3 + 1, T]$ into consecutive blocks, each of which has same size $2m_1 \times 2m_2 \times m_3$ (here, without loss of generality, we assume $(S_1-h_1-1)/(2m_1), (S_2-h_2-1)/(2m_2), (T-h_3-1)/m_3 \in \mathbb{N}$). There are $C^* =(S_1-h_1-1)(S_2-h_2-1)(T-h_3-1)/(4m_1 m_2 m_3) = O(\mathcal{R}/(4m_1 m_2 m_3))$ number of such blocks. In this manner, for $\vsbf_2$ in the non-consecutive blocks, $\{\widetilde{Z}_{ij, \vsbf_2}; \h+\1 \leq \vsbf_2 \leq \bar{\vsbf}\}$ forms a martingale difference sequence. Therefore, we have
\begin{align}
&\quad \frac{\max_{1\leq i, j \leq n} \left({\rm E} \left| \frac{2}{\mathcal{R}} \sum_{\vsbf_2 = \h+\1}^{\bar{\vsbf}} \widetilde{Z}_{ij, \vsbf_2}  \right|^p \right)^{1/p}}{\left( m_1 m_2 m_3 B_{S_1} B_{S_2} B_T /\mathcal{R} \right)^{1/2}} \nonumber \\
&\leq \max_{1\leq i, j \leq n} \frac{4}{4m_1 m_2 m_3} \sum_{a_1 = 1}^{2m_1} \sum_{a_2 = 1}^{2m_2} \sum_{a_3 = 1}^{m_3} \frac{\left({\rm E} \left| \left(\frac{\mathcal{R}}{4m_1 m_2 m_3}\right)^{-1/2} \sum_{l = 1}^{C^*} \widetilde{Z}_{ij, \vsbf_2+l\hbf}  \right|^p \right)^{1/p}}{(B_{S_1} B_{S_2} B_T )^{1/2}}  \leq A_4. \label{Ztilde}
\end{align}

Now, Let $\widetilde{\Gamma}_{ij, \vsbf_1} = {\rm E} (\widetilde{Y}_{i\0} \widetilde{Y}_{j\vsbf_1})$ and consider the first two terms in (\ref{sigmatilde}) together in view of the fact that for any $\vsbf_1 \in \mathbb{Z}^3$, 
$$\underset{\mathcal{R}\rightarrow \infty}{\lim}\frac{\max_{1\leq i, j \leq n} \left({\rm E} \left|  \mathcal{R}^{-1} \sum_{\vsbf=\1}^{\bar{\vsbf}} \widetilde{Y}_{i\vsbf} \widetilde{Y}_{j\vsbf+\vsbf_1} - \widetilde{\Gamma}_{ij, \vsbf_1} \right|^{p/2} \right)^{2/p}}{(m_1 m_2 m_3/\mathcal{R})^{1/2}} < \infty$$ due to the central limit theorem for $m$-dependent process (see, e.g., \citealt[Theorem 2.8.1]{Lehmann98}) and Conditions \ref{condA} and \ref{condB}. Note that in the first two terms, we are combining all the terms of the form $\widetilde{Y}_{i\vsbf_1} \widetilde{Y}_{j\vsbf_1 + \k}$ where $|k_1| \leq h_1, |k_2| \leq h_2, 
0 \leq k_3 \leq h_3$. Denoting by $\widetilde{W}_{ij, \mathcal{R}}/\mathcal{R}$ the last term in (\ref{sigmatilde}) and considering that $K_1, K_2, K_3$ are bounded with support $[-1, 1]$ because of Assumption \ref{AssSpec2}, we then have
\begin{align*}
&\quad  \frac{\max_{1\leq i, j \leq n} \left({\rm E} \left| \widetilde{\sigma}_{ij}^y(\thbf) - \widetilde{W}_{ij, \mathcal{R}}/\mathcal{R} - {\rm E} [\widetilde{\sigma}_{ij}^y(\thbf) - \widetilde{W}_{ij, \mathcal{R}}/\mathcal{R}] \right|^{p/2} \right)^{2/p}}{\left( m_1^2 m_2^2 m_3^2 B_{S_1} B_{S_2} B_T/\mathcal{R} \right)^{1/2}} \\
&\leq \frac{1}{(m_1 m_2 m_3 B_{S_1} B_{S_2} B_T)^{1/2}} \sum_{k_1 = 1}^{2m_1} \sum_{k_2 = 1}^{2m_2} \sum_{k_3 = 1}^{2m_3} |a_{\bar{\vsbf}, \k}| A_5\\
&\leq \frac{A_6 \min(8 m_1 m_2 m_3, B_{S_1} B_{S_2} B_T)}{(m_1 m_2 m_3 B_{S_1} B_{S_2} B_T)^{1/2}} A_5 \leq A_7.
\end{align*}
Now, due to (\ref{Ztilde}) and the above inequality, there exists $A_8>0$ such that
\begin{equation*}
 \frac{\max_{1\leq i, j \leq n} \left({\rm E} \left| \widetilde{\sigma}_{ij}^y(\thbf)  - {\rm E} [\widetilde{\sigma}_{ij}^y(\thbf)] \right|^{p/2} \right)^{2/p}}{\left( m_1^2 m_2^2 m_3^2 B_{S_1} B_{S_2} B_T/\mathcal{R} \right)^{1/2}} \leq A_8.
\end{equation*}

Define $U_{\bar{\vsbf}}^{[i]}(\thbf) = \sum_{\vsbf=\1}^{\bar{\vsbf}} Y_{i\vsbf} e^{-i\left\langle \vsbf, \thbf\right\rangle}$ and $\widetilde{U}_{\bar{\vsbf}}^{[i]}(\thbf)  = \sum_{\vsbf=\1}^{\bar{\vsbf}} \widetilde{Y}_{i\vsbf} e^{-i\left\langle \vsbf, \thbf\right\rangle}$. Lemma~\ref{lemYYtilde} then implies that 
\begin{equation}\nn
\left({\rm E} \left| U_{\bar{\vsbf}}^{[i]}(\thbf)   - \widetilde{U}_{\bar{\vsbf}}^{[i]}(\thbf)  \right|^p \right)^{1/p} \leq C_p \mathcal{R}^{1/2}  \varphi_{\m+\1, p}^{[i]}
 \ \text{and} \ 
\left({\rm E} \left| U_{\bar{\vsbf}}^{[i]}(\thbf)  \right|^p \right)^{1/p}
 +  \left({\rm E} \left| \widetilde{U}_{\bar{\vsbf}}^{[i]}(\thbf)  \right|^p \right)^{1/p} \leq 2 C_p \mathcal{R}^{1/2}  \varphi_{\0, p}^{[i]}.
\end{equation}

Denote by $\widehat{K}_1, \widehat{K}_2, \widehat{K}_3$ the Fourier transforms of $K_1, K_2, K_3$, respectively, we then have 
$$ \widehat{\sigma}_{ij}^y(\thbf) = \frac{1}{\mathcal{R}} \int_{\mathbb{R}^3} \widehat{K}_1(u_1)  \widehat{K}_2(u_2) \widehat{K}_3(u_3)     U_{\bar{\vsbf}}^{[i]}(\thbf^*_{\u})   \bar{U}_{\bar{\vsbf}}^{[j]}(\thbf^*_{\u}) {\rm d}\u,$$ 
where $\thbf^*_{\u} = (B_{S_1}^{-1} u_1 \ B_{S_2}^{-1} u_2 \ B_{T}^{-1} u_3) + \thbf$, $\u = (u_1 \ u_2 \ u_3)$ and $\bar{U}_{\bar{\vsbf}}^{[j]}$ denotes the conjugate of  $U_{\bar{\vsbf}}^{[j]}$. Hence, we have
\begin{align*}
&\quad \left({\rm E} \left| \widehat{\sigma}_{ij}^y(\thbf) - \widetilde{\sigma}_{ij}^y(\thbf) \right|^{p/2} \right)^{2/p}  \\
&\leq  \frac{1}{\mathcal{R}} \int_{\mathbb{R}^3} |\widehat{K}_1(u_1)  \widehat{K}_2(u_2) \widehat{K}_3(u_3)| \left({\rm E} \left| \left(U_{\bar{\vsbf}}^{[i]}(\thbf^*_{\u}) - \widetilde{U}_{\bar{\vsbf}}^{[i]}(\thbf^*_{\u})\right)  \bar{U}_{\bar{\vsbf}}^{[j]}(\thbf^*_{\u}) \right|^{p/2} \right)^{2/p} {\rm d}\u \\
&\quad +  \frac{1}{\mathcal{R}} \int_{\mathbb{R}^3} |\widehat{K}_1(u_1)  \widehat{K}_2(u_2) \widehat{K}_3(u_3)| \left({\rm E} \left| \widetilde{U}_{\bar{\vsbf}}^{[i]}(\thbf^*_{\u}) \left(\bar{U}_{\bar{\vsbf}}^{[j]}(\thbf^*_{\u}) - \bar{\widetilde{U}}_{\bar{\vsbf}}^{[j]}(\thbf^*_{\u})\right) \right|^{p/2} \right)^{2/p} {\rm d}\u \\
&\leq  \int_{\mathbb{R}^3} |\widehat{K}_1(u_1)  \widehat{K}_2(u_2) \widehat{K}_3(u_3)| C_p^2  \left(\varphi_{\m+\1, p}^{[i]} \varphi_{\0, p}^{[j]} + \varphi_{\0, p}^{[i]} \varphi_{\m+\1, p}^{[j]} \right) {\rm d}\u, \\
&\leq A_9  \left(\varphi_{\m+\1, p}^{[i]} \varphi_{\0, p}^{[j]} + \varphi_{\0, p}^{[i]} \varphi_{\m+\1, p}^{[j]} \right).
\end{align*}


Now, combining all the results from above, we have 
\begin{align*}
&\quad \max_{1\leq i, j \leq n}  \sup_{\thbf \in \Thbf} \left({\rm E} \left| \widehat{\sigma}_{ij}^y(\thbf)  - {\rm E} \widehat{\sigma}_{ij}^y(\thbf) \right|^{p/2} \right)^{2/p} \\
&\leq \max_{1\leq i, j \leq n}  \sup_{\thbf \in \Thbf} \left[ \left({\rm E} \left| \widehat{\sigma}_{ij}^y(\thbf)  - \widetilde{\sigma}_{ij}^y(\thbf) \right|^{p/2} \right)^{2/p} 
+ |{\rm E}\left( \widehat{\sigma}_{ij}^y(\thbf)  - \widetilde{\sigma}_{ij}^y(\thbf) \right)| 
+ \left({\rm E} \left| \widetilde{\sigma}_{ij}^y(\thbf)  - {\rm E} \widetilde{\sigma}_{ij}^y(\thbf) \right|^{p/2} \right)^{2/p} \right] \\
&\leq A_{10} \max_{1\leq i \leq n} \varphi_{\m+\1, p}^{[i]} + A_8 \left( m_1^2 m_2^2 m_3^2 B_{S_1} B_{S_2} B_T/\mathcal{R} \right)^{1/2}. 
\end{align*}
 Under Conditions \ref{condA} and \ref{condB}, there exists $A_{11} \in (0, \infty)$ such that $\varphi_{\m+\1, p}^{[i]} < A_{11} \rho_1^{m_1} \rho_2^{m_2} \rho_3^{m_3}$. Take $m_1 = -\log_{\rho_1}(B_{S_1}^{a_1}), m_2 = -\log_{\rho_2}(B_{S_2}^{a_2}), m_3 = -\log_{\rho_3}(B_T^{a_3})$ such that $a_1, a_2, a_3>0$ and $a_1 > (1/b_1 -1)/2, a_2 > (1/b_1^* -1)/2, a_3 > (1/b_1^{**} -1)/2$. This yields
$$\varphi_{\m+\1, p}^{[i]} < A_{11} B_{S_1}^{-a_1} B_{S_2}^{-a_2} B_T^{-a_3} = o(B_{S_1}^{(1-1/b_1)/2} B_{S_2}^{(1-1/b_1^*)/2} B_T^{(1-1/b_1^{**})/2}) = o((B_{S_1} B_{S_2} B_T/\mathcal{R})^{1/2}),$$ and $m_1 m_2 m_3 = A_{12} \log(B_{S_1}) \log(B_{S_2})  \log(B_T)$ for a positive constant $A_{12}$.
Hence, there exists $C>0$ such that
\begin{align*}
& \frac{\max_{1\leq i, j \leq n}  \sup_{\thbf \in \Thbf} \left({\rm E} \left| \widehat{\sigma}_{ij}^y(\thbf)  - {\rm E} \widehat{\sigma}_{ij}^y(\thbf) \right|^{p/2} \right)^{2/p}}{\left(\log(B_{S_1}) \log(B_{S_2})  \log(B_T))^2 B_{S_1} B_{S_2} B_T/\mathcal{R} \right)^{1/2}} \\&\leq 
\frac{A_{10} \max_{1\leq i \leq n} \varphi_{\m+\1, p}^{[i]} + A_8 \left( m_1^2 m_2^2 m_3^2 B_{S_1} B_{S_2} B_T/\mathcal{R} \right)^{1/2}}{A_{12}^{-1} \left( m_1^2 m_2^2 m_3^2 B_{S_1} B_{S_2} B_T/\mathcal{R} \right)^{1/2}} <C.
\end{align*}
This completes the proof.


\subsection{Proof of Proposition \ref{ThmSpecDen2}}

Noting that
\begin{equation*}
{\rm E}|\widehat{\sigma}_{ij}^y(\thbf)  - \sigma_{ij}^y(\thbf)|^2 \leq
 2\left[ {\rm E}|\widehat{\sigma}_{ij}^y(\thbf)  - {\rm E} \widehat{\sigma}_{ij}^y(\thbf)|^2  
+ {\rm E}|{\rm E} \widehat{\sigma}_{ij}^y(\thbf) - \sigma_{ij}^y(\thbf) |^2 \right],
\end{equation*}
where, because of Proposition \ref{ThmSpecDen1} when $p=4$, the first term on the right-hand side satisfies
\begin{equation}\label{VarSpec}
\max_{1\leq i, j \leq n} \sup_{\thbf \in \Thbf} {\rm E}|\widehat{\sigma}_{ij}^y(\thbf)  - {\rm E} \widehat{\sigma}_{ij}^y(\thbf)|^2  \leq  C^2 (\log(B_{S_1}) \log(B_{S_2}) \log(B_{T}))^2 B_{S_1} B_{S_2} B_{T}/\mathcal{R},
\end{equation}
it only remains to deal with the second term.

Since 
\begin{align*}
&  {\rm E} \widehat{\sigma}_{ij}^y(\thbf) \\
&= \frac{1}{\mathcal{R}} \sum_{\vsbf_1 = \1}^{\bar{\vsbf}} \sum_{\vsbf_2 = \1}^{\bar{\vsbf}}  \Gamma_{ij, \vsbf_1 - \vsbf_2}^y  K_1\left(\frac{s_{11}-s_{21}}{B_{S_1}} \right) K_2\left(\frac{s_{12}-s_{22}}{B_{S_2}} \right) K_3\left(\frac{t_{1}-t_{2}}{B_{T}} \right) e^{-i \left\langle \vsbf_1 - \vsbf_2,  \thbf\right\rangle} \\
&=  \frac{1}{\mathcal{R}} \sum_{s_1 = 1-S_1}^{S_1-1} \sum_{s_2 = 1-S_2}^{S_2-1}  \sum_{t = 1-T}^{T-1}  (S_1 - |s_1|)  (S_2 - |s_2|)  (T - |t|) \Gamma_{ij, \vsbf} 
K_1\left(\frac{s_1}{B_{S_1}} \right) 
K_2\left(\frac{s_2}{B_{S_2}} \right) K_3\left(\frac{t}{B_{T}} \right) e^{-i \left\langle \vsbf,  \thbf\right\rangle}  \\ 
&= \sum_{s_1 = 1-S_1}^{S_1-1} \sum_{s_2 = 1-S_2}^{S_2-1}  \sum_{t = 1-T}^{T-1} \left(1-\frac{|s_1|}{S_1}\right) \left(1-\frac{|s_2|}{S_2}\right) \left(1-\frac{|t|}{T}\right) \Gamma_{ij, \vsbf} 
K_1\left(\frac{s_1}{B_{S_1}} \right) K_2\left(\frac{s_2}{B_{S_2}} \right) K_3\left(\frac{t}{B_{T}} \right) e^{-i \left\langle \vsbf,  \thbf\right\rangle},
\end{align*}
we have
\begin{align*}
& \left|  \left({\rm E} \widehat{\sigma}_{ij}^y(\thbf) - \sigma_{ij}^y(\thbf) \right) \right|  \\
&\leq  \left| \sum_{s_1 = 1-S_1}^{S_1-1} \sum_{s_2 = 1-S_2}^{S_2-1}  \sum_{t = 1-T}^{T-1} \left( K_1\left(\frac{s_1}{B_{S_1}} \right) K_2\left(\frac{s_2}{B_{S_2}} \right) K_3\left(\frac{t}{B_{T}} \right) - 1\right) \Gamma_{ij, \vsbf} e^{-i \left\langle \vsbf,  \thbf\right\rangle} \right| \\
& + \left| \sum_{s_1 = 1-S_1}^{S_1-1} \sum_{s_2 = 1-S_2}^{S_2-1}  \sum_{t = 1-T}^{T-1} \left[  \left(1-\frac{|s_1|}{S_1}\right) \left(1-\frac{|s_2|}{S_2}\right) \left(1-\frac{|t|}{T}\right) - 1\right] \Gamma_{ij, \vsbf} \right. \\
&\quad \times \left. K_1\left(\frac{s_1}{B_{S_1}} \right) K_2\left(\frac{s_2}{B_{S_2}} \right) K_3\left(\frac{t}{B_{T}} \right) e^{-i \left\langle \vsbf,  \thbf\right\rangle} \right| \\
& + \left| \sum_{|s_1| \geq S_1} \sum_{|s_2| \geq S_2} \sum_{|t| \geq T} \Gamma_{ij, \vsbf}  e^{-i \left\langle \vsbf,  \thbf\right\rangle} \right| \\
& =: \mathcal{C}_{ij, \bar{\vsbf}}^{[1]}(\thbf) +  \mathcal{C}_{ij, \bar{\vsbf}}^{[2]}(\thbf) + \mathcal{C}_{ij, \bar{\vsbf}}^{[3]}(\thbf).
\end{align*}
Under Assumptions~\ref{AssSpec2} and Condition \ref{condC}, there exists some $A_3, A_4 >0$ such that, for all $i, j$ and $\thbf$,
\begin{align*}
\mathcal{C}_{ij, \bar{\vsbf}}^{[1]}(\thbf)
&\leq A_3 \sum_{s_1 = -\infty}^{\infty} \sum_{s_2 = -\infty}^{\infty} \sum_{t = -\infty}^{\infty} \rho_4^{|s_1|} \rho_5^{|s_2|} \rho_6^{|t|} 
\l\{
\left(\frac{|s_1|}{B_{S_1}}\right)^{\vartheta_1} +\left(\frac{|s_2|}{B_{S_2}}\right)^{\vartheta_2} +\left(\frac{|t|}{B_{T}}\right)^{\vartheta_3}\r.\nn\\
&\l.
+\left(\frac{|s_1|}{B_{S_1}}\right)^{\vartheta_1} \left(\frac{|s_2|}{B_{S_2}}\right)^{\vartheta_2} +
\left(\frac{|s_1|}{B_{S_1}}\right)^{\vartheta_1} \left(\frac{|t|}{B_{T}}\right)^{\vartheta_3}+
\left(\frac{|s_2|}{B_{S_2}}\right)^{\vartheta_2} \left(\frac{|t|}{B_{T}}\right)^{\vartheta_3}+
\left(\frac{|s_1|}{B_{S_1}}\right)^{\vartheta_1}\left(\frac{|s_2|}{B_{S_2}}\right)^{\vartheta_2} \left(\frac{|t|}{B_{T}}\right)^{\vartheta_3}
\r\} \\
&\leq A_4 \l\{B_{S_1}^{-\vartheta_1} +B_{S_2}^{-\vartheta_2}+ B_{T}^{-\vartheta_3}\r\}. 
\end{align*}
Moreover, there exist $A_5, ..., A_{13} >0$ such that, for all $i, j$ and $\thbf$,
\begin{align*}
\mathcal{C}_{ij, \bar{\vsbf}}^{[2]}(\thbf)
&\leq A_5 \sum_{s_1 = -\infty}^{\infty} \sum_{s_2 = -\infty}^{\infty} \sum_{t = -\infty}^{\infty} \left( \frac{|s_1|}{S_1} + \frac{|s_2|}{S_2} + \frac{|t|}{T} +\frac{|s_1 s_2|}{S_1 S_2} +  \frac{|s_1 t|}{S_1 T} + \frac{|s_2 t|}{S_2 T} + \frac{|s_1 s_2 t|}{S_1 S_2 T}\right)  \rho_4^{|s_1|} \rho_5^{|s_2|} \rho_6^{|t|}  \\
&\leq  A_6 S_1^{-1} + A_7 S_2^{-1} + A_8 T^{-1} + A_9 (S_1 S_2)^{-1} + A_{10} (S_1 T)^{-1} + A_{11} (S_2 T)^{-1} + A_{12} (S_1 S_2 T)^{-1} \\
&\leq A_{13} \max(S_1^{-1}, S_2^{-1}, T^{-1})
\end{align*}
and
\begin{align*}
\mathcal{C}_{ij, \bar{\vsbf}}^{[3]}(\thbf) \leq A_2 \sum_{|s_1| \geq S_1} \sum_{|s_2| \geq S_2} \sum_{|t| \geq T} \rho_4^{|s_1|} \rho_5^{|s_2|} \rho_6^{|t|} = o(\max(S_1^{-1}, S_2^{-1}, T^{-1})).
\end{align*}
Therefore, 
$$\max_{1\leq i, j \leq n} \sup_{\thbf \in \Thbf} \left| {\rm E} \widehat{\sigma}_{ij}^y(\thbf) - \sigma_{ij}^y(\thbf)  \right|^2  = O\left\lbrace \max\left[ B_{S_1} B_{S_2} B_T/\mathcal{R}, B_{S_1}^{-2\vartheta_1}, B_{S_2}^{-2\vartheta_2}, B_T^{-2\vartheta_3}\right] \right\rbrace,$$
which, together with (\ref{VarSpec}), completes the proof.

\section{Proof of results of Section 6}\label{APPE} 

\subsection{Proof of Theorem~\ref{Thmsigmax}}

We first prove the following three propositions.
 
\begin{proposition}\label{lukagu}
Let Assumption \ref{Ass.A}, \ref{Ass.B}, \ref{Ass.MA}, and \ref{Ass.coef1} hold. Then, there exists a finite $C>0$ such that $\sup_{\thbf \in \Thbf}\sup_{n\in\mathbb N} \lambda_{n1}^\xi(\bth) \le C$.
\end{proposition}

\begin{proof}
Let $\sigma_{ij}^{\xi}(\thbf)$ denote the cross-spectral density of $\xi_i$ and $\xi_j$. Then, by Assumptions~\ref{Ass.MA} and \ref{Ass.coef1}, for all $j=1,\ldots, n$ and all $\thbf\in\Thbf$,
\begin{align}
\sum_{i=1}^\infty \vert\sigma_{ij}^{\xi}(\thbf)\vert&\le \sum_{i=1}^\infty\sum_{\ell=1}^\infty \l\vert\beta_{i\ell}(\thbf)\bar{\beta}_{j\ell}(\thbf)
\r\vert\le 
\sum_{i=1}^\infty\sum_{\ell=1}^\infty
\l(\sum_{\kappa_1,\kappa_2\in\mathbb Z}\sum_{\kappa_3=0}^\infty \vert \beta_{i\ell,\kbf}\vert \r)\l(\sum_{\kappa_1,\kappa_2\in\mathbb Z}\sum_{\kappa_3=0}^\infty \vert \bar{\beta}_{j\ell,\kbf}\vert \r)\nn\\
&\le\sum_{i=1}^\infty\sum_{\ell=1}^\infty
\l(\sum_{\kappa_1,\kappa_2\in\mathbb Z}\sum_{\kappa_3=0}^\infty A^\xi_{i\ell} {\rho_1^{\xi|\kappa_1|}} {\rho_2^{\xi|\kappa_2|}} {\rho_3^{\xi\kappa_3}}\r)
\l(\sum_{\kappa_1,\kappa_2\in\mathbb Z}\sum_{\kappa_3=0}^\infty A^\xi_{j\ell} {\rho_1^{\xi|\kappa_1|}} {\rho_2^{\xi|\kappa_2|}} {\rho_3^{\xi\kappa_3}}\r)\nn\\
&\le (A^\xi)^2C_{{\rho_1^\xi}{\rho_2^\xi}{\rho_3^\xi}},\nn
\end{align}
for some finite $C_{{\rho_1^\xi}{\rho_2^\xi}{\rho_3^\xi}}>0$ depending only on ${\rho_1^\xi}$, ${\rho_2^\xi}$, and ${\rho_3^\xi}$. This yields that
$$
\sup_{\thbf\in\Thbf}\max_{j=1, \ldots, n} \sum_{i=1}^{n} \vert \sigma_{ij}^{\xi}(\thbf) \vert  \leq \sup_{\thbf\in\Thbf} \max_{j=1, \ldots, n} \sum_{i=1}^{\infty} \vert  \sigma_{ij}^{\xi}(\thbf) \vert \leq(A^\xi)^2C_{{\rho_1^\xi}{\rho_2^\xi}{\rho_3^\xi}}.
$$
Let $\bm\Sigma_n^\xi(\thbf)$ be the spectral density matrix of $\bm\xi_n$. By H\"older inequality, for all $n\in\mathbb N$ and all $\thbf\in\Thbf$,
$$
[\lambda_{n1}^\xi(\bth)]^2= \Vert \bm\Sigma_n^\xi(\thbf)\Vert^2\le\Vert \bm\Sigma_n^\xi(\thbf)\Vert^2_F\le \Vert \bm\Sigma_n^\xi(\thbf)\Vert_1\Vert \bm\Sigma_n^\xi(\thbf)\Vert_\infty =\Vert \bm\Sigma_n^\xi(\thbf)\Vert_1^2=\left[ \max_{j=1, \ldots, n} \sum_{i=1}^{n} \vert \sigma_{ij}^{\xi}(\thbf) \vert \right]^2,
$$
where $\Vert\cdot\Vert$ denotes the spectral norm, $\Vert\cdot\Vert_F$ denotes the Frobenius norm, and $\Vert\cdot\Vert_1$ and $\Vert\cdot\Vert_\infty$ denote the column-wise and row-wise norms, respectively.
It follows that,
$$
\sup_{\thbf\in\Thbf}\sup_{n\in\mathbb N}\lambda_{n1}^\xi(\bth) \leq \sup_{\thbf\in\Thbf}\sup_{n\in\mathbb N}  \max_{j=1, \ldots, n} \sum_{i=1}^{n} \vert \sigma_{ij}^{\xi}(\thbf) \vert\le 
(A^\xi)^2C_{{\rho_1^\xi}{\rho_2^\xi}{\rho_3^\xi}}.
$$
This completes the proof.
\end{proof}

\begin{proposition} \label{Our_FHLZ17}
Let Assumptions~\ref{Ass.A}, \ref{Ass.MA}, \ref{Ass.coef1}, and \ref{Ass.moments} hold. Then, for all $\ell \in \mathbb{N}$, there exist a finite $p>4$, and $\widetilde{\rho}_1, \widetilde{\rho}_2, \widetilde{\rho}_3 \in (0, 1)$, and $\widetilde{A}_1,\widetilde{A}_2 >0$, such that
\begin{enumerate}
\item [(i)] ${\rm E}\left(|x_{\ell \vsbf}|^p\right) \leq \widetilde{A}_1$;
\item [(ii)] $\delta^{[\ell]}_{\vsbf, p}(\x) \leq \widetilde{A}_2 \widetilde{\rho}_1^{|s_1|} \widetilde{\rho}_2^{|s_2|} \widetilde{\rho}_3^{|t|}.$
\end{enumerate}
\end{proposition}

\begin{proof}
The Minkowski inequality yields
$$\left\lbrace {\rm E}\left(\vert x_{\ell \vsbf}\vert^p\right)\right\rbrace^{1/p} \leq \left\lbrace {\rm E}\left(\vert \chi_{\ell \vsbf}\vert^p\right)\right\rbrace^{1/p} + \left\lbrace {\rm E}\left(\vert \xi_{\ell \vsbf}\vert^p\right)\right\rbrace^{1/p},$$
so it suffices to bound the two terms on the RHS. Using again the Minkowski inequality and Assumptions~\ref{Ass.MA}, \ref{Ass.coef1} and \ref{Ass.moments}, there is a constant $\widetilde{A}_1 >0$ such that
\begin{align*}
\left\lbrace {\rm E}\left(\vert \xi_{\ell \vsbf}\vert^p\right)\right\rbrace^{1/p} &= \left\lbrace {\rm E}\left(\left\vert \sum_{\kappa_1,\kappa_2\in\mathbb Z}  \sum_{\kappa_3 =0}^\infty  \sum_{j =1}^\infty \beta_{\ell j, \kbf} \varepsilon_{j, \vsbf-\kbf} \right\vert^p\right)\right\rbrace^{1/p} \\
&\leq  \sum_{\kappa_1,\kappa_2\in\mathbb Z}  \sum_{\kappa_3 =0}^\infty  \sum_{j =1}^\infty \left\lbrace {\rm E}\left(\left\vert  \beta_{\ell j, \kbf} \varepsilon_{j, \vsbf-\kbf} \right\vert^p\right)\right\rbrace^{1/p} \\
&\leq  \sum_{\kappa_1,\kappa_2\in\mathbb Z}  \sum_{\kappa_3 =0}^\infty  \sum_{j =1}^\infty \vert \beta_{\ell j, \kbf} \vert \left\lbrace {\rm E}\left(\left\vert   \varepsilon_{j, \vsbf-\kbf} \right\vert^p\right)\right\rbrace^{1/p}\\
& \leq A^\xi  \widetilde A_{{\rho_1^\xi}{\rho_2^\xi}{\rho_3^\xi}},\nn
\end{align*}
for some finite $\widetilde A_{{\rho_1^\xi}{\rho_2^\xi}{\rho_3^\xi}}>0$ depending only on ${\rho_1^\xi}$, ${\rho_2^\xi}$, and ${\rho_3^\xi}$.
Similarly, we have $\left\lbrace {\rm E}\left(\vert \chi_{\ell \vsbf}\vert^p\right)\right\rbrace^{1/p} \leq {A}^\chi\widetilde A_{{\rho_1^\chi}{\rho_2^\chi}{\rho_3^\chi}}$ 
for some finite $\widetilde A_{{\rho_1^\chi}{\rho_2^\chi}{\rho_3^\chi}}>0$ depending only on ${\rho_1^\chi}$, ${\rho_2^\chi}$, and ${\rho_3^\chi}$.
Letting $\widetilde{A}_1 = (A^\xi  \widetilde A_{{\rho_1^\xi}{\rho_2^\xi}{\rho_3^\xi}} + {A}^\chi\widetilde A_{{\rho_1^\chi}{\rho_2^\chi}{\rho_3^\chi}})^p$ yields ${\rm E}\left(|x_{\ell \vsbf}|^p\right) \leq \widetilde{A}_1$. This proves part (i).

Turning to part (ii), for $t<0$, the inequality holds trivially since $\delta^{[\ell]}_{\vsbf, p}(\x) = 0$. For $t \geq 0$,
notice that
\begin{align*}
\delta^{[\ell]}_{\vsbf, p}(\x) &= \{ {\rm E} \vert x_{\ell\vsbf} - x^*_{\ell\vsbf}\vert^p \}^{1/p} \\
&\leq \{ {\rm E} \vert \chi_{\ell\vsbf} - \chi^*_{\ell\vsbf}\vert^p \}^{1/p} + \{ {\rm E} \vert\xi_{\ell\vsbf} - \xi^*_{\ell\vsbf}\vert^p \}^{1/p} \\
&= \delta^{[\ell]}_{\vsbf, p}(\bchi) + \delta^{[\ell]}_{\vsbf, p}(\bfxi).
\end{align*}
Assumptions~\ref{Ass.MA}, \ref{Ass.coef1} and \ref{Ass.moments} entail that
\begin{align*}
\delta^{[\ell]}_{\vsbf, p}(\bfxi) &= \left\lbrace {\rm E} \left\vert \sum_{j =1}^\infty \beta_{\ell j, \vsbf} (\varepsilon_{j, \0} -\varepsilon_{j, \0}^*) \right\vert^p \right\rbrace^{1/p}  \\
&\leq  \sum_{j =1}^\infty \vert \beta_{\ell j, \vsbf}\vert  \left\lbrace {\rm E} \left\vert(\varepsilon_{j, \0} -\varepsilon_{j, \0}^*) \right\vert^p \right\rbrace^{1/p} \leq  
\bar A A^\xi{\rho_1^\xi}^{|s_1|} {\rho_2^\xi}^{|s_2|}  {\rho_3^\xi}^{t}.
\end{align*}
Similarly, we have that $\delta^{[\ell]}_{\vsbf, p}(\bchi) \leq \bar AA^\chi {\rho_1^\chi}^{|s_1|} {\rho_2^\chi}^{|s_2|}  {\rho_3^\chi}^{t}$. The result then follows by letting 
$\widetilde A_2=\bar A\max(A^\chi,A^\xi)$
and $\widetilde{\rho}_{h} = \max({\rho_{h}^{\xi}}, {\rho_{h}^{\chi}})$, $h = 1, 2, 3$.
\end{proof}

\begin{proposition} \label{Our_FHLZ17bis}
Let Assumptions~\ref{Ass.A}, \ref{Ass.MA}, \ref{Ass.coef1}, and \ref{Ass.moments} hold. Then, letting $\Gamma_{ij, \vsbf}^x = {\rm E} (x_{i\vsbf} x_{j\0})$,  for all $i, j \in  \mathbb{N}$, there exist finite $\widetilde{\rho}_4, \widetilde{\rho}_5, \widetilde{\rho}_6 \in (0, 1)$ and $\widetilde{A}^*>0$, such that $\vert \Gamma_{ij, \vsbf}^x \vert \leq \widetilde{A}^* \widetilde{\rho}_4^{|s_1|} \widetilde{\rho}_5^{|s_2|} \widetilde{\rho}_6^{|t|}$.
\end{proposition}

\begin{proof}
We have $\vert \Gamma_{ij, \vsbf}^x \vert = \vert \Gamma_{ij, \vsbf}^{\chi} +  \Gamma_{ij, \vsbf}^{\xi} \vert \leq \vert \Gamma_{ij, \vsbf}^{\chi} \vert + \vert \Gamma_{ij, \vsbf}^{\xi} \vert$. Since $\{{\boldsymbol \varepsilon}_{\vsbf} = (\varepsilon_{1\vsbf} \ \varepsilon_{2\vsbf} \ \cdots)^\top, \vsbf \in \mathbb{Z}^3\}$  is an i.i.d.~infinite dimensional zero-mean orthonormal rf, by Assumptions~\ref{Ass.MA} and \ref{Ass.coef1} we have
\begin{align*}
\vert \Gamma_{ij, \vsbf}^{\xi} \vert &= \left\vert {\rm E} \left\lbrace \left(\sum_{\kappa_1,\kappa_2\in\mathbb Z}  \sum_{\kappa_3 =0}^\infty  \sum_{h =1}^\infty \beta_{i h, \kbf} \varepsilon_{h, \vsbf-\kbf}\right)  \left(\sum_{\tilde{\kappa}_1,\tilde{\kappa}_2\in\mathbb Z}  \sum_{\tilde{\kappa}_3 =0}^\infty  \sum_{\tilde{h} =1}^\infty \beta_{j \tilde{h}, \tilde{\kbf}} \varepsilon_{\tilde{h}, -\tilde{\kbf}}\right) \right\rbrace \right\vert \\
&= \left\vert {\rm E} \left\lbrace \left(\sum_{\kappa_1,\kappa_2\in\mathbb Z}  \sum_{\kappa_3 =0}^\infty  \sum_{h =1}^\infty \beta_{i h, \kbf} \beta_{j h, {\kbf}-\vsbf} \varepsilon_{h, \vsbf-\kbf}^2 \right)  \right\rbrace \right\vert \\
&\leq \sum_{\kappa_1,\kappa_2\in\mathbb Z}  \sum_{\kappa_3 =0}^\infty   \sum_{h =1}^\infty \left\vert \beta_{i h, \kbf} \beta_{j h, {\kbf}-\vsbf} \right\vert  {\rm E}(\varepsilon_{h, \vsbf-\kbf}^2) \\
&\leq \sum_{\kappa_1,\kappa_2\in\mathbb Z}  \sum_{\kappa_3 =0}^\infty   \sum_{h =1}^\infty A_{ih}^\xi A_{jh}^\xi \rho_1^{\xi,|\kappa_1|+|\kappa_1-s_1|} 
\rho_2^{\xi,|\kappa_2|+|\kappa_2-s_2|} \rho_3^{\xi, \kappa_3+\kappa_3-t}
\\
&\leq \l\{ \sum_{h =1}^\infty A_{ih}^\xi A_{jh}^\xi \sum_{\kappa_1,\kappa_2\in\mathbb Z}  \sum_{\kappa_3 =0}^\infty 
\rho_1^{\xi,2|\kappa_1|} 
\rho_2^{\xi,2|\kappa_2|} \rho_3^{\xi, 2\kappa_3} \r\} {\rho^{\xi s_1|}_1} {\rho_2^{\xi|s_2|}} {\rho^\xi_3}^{|t|}\\
&\leq (A^\xi)^2 \widetilde{A}^*_{\rho_1^\xi,\rho_2^\xi,\rho_3^\xi}
 {\rho^{\xi s_1|}_1} {\rho_2^{\xi|s_2|}} {\rho^\xi_3}^{|t|}
\end{align*}
for some finite $\widetilde{A}^*_{\rho_1^\xi,\rho_2^\xi,\rho_3^\xi}>0$ depending only on ${\rho_1^\xi}$, ${\rho_2^\xi}$, and ${\rho_3^\xi}$.

Similarly, there exists a finite $\widetilde{A}^*_{\rho_1^\chi,\rho_2^\chi,\rho_3^\chi}>0$ depending only on ${\rho_1^\chi}$, ${\rho_2^\chi}$, and ${\rho_3^\chi}$, such that
$\vert \Gamma_{ij, \vsbf}^{\chi} \vert \leq \widetilde{A}^*_{\rho_1^\chi,\rho_2^\chi,\rho_3^\chi} {\rho^\chi_1}^{|s_1|} {\rho^\chi_2}^{|s_2|} {\rho^\chi_3}^{|t|}$. Letting 
$\widetilde A^* = \max (\widetilde{A}^*_{\rho_1^\chi,\rho_2^\chi,\rho_3^\chi}, \widetilde{A}^*_{\rho_1^\xi,\rho_2^\xi,\rho_3^\xi})$ and
$\widetilde{\rho}_4 = \max({\rho_1^\xi}, {\rho^\chi_1})$, $\widetilde{\rho}_5 = \max({\rho^\xi_2}, {\rho^\chi_2})$, and $\widetilde{\rho}_6 = \max({\rho^\xi_3}, {\rho^\chi_3})$ yields the desired result.
\end{proof}

It immediately follows from Assumption \ref{Ass.A} and Proposition \ref{Our_FHLZ17}, that for any $n\in\mathbb N$ each component $x_\ell$, $\ell=1,\ldots, n$ of the rf $\bm x_n$ satisfies Conditions \ref{condA} and \ref{condB} in Appendix \ref{Spec.Y}. Moreover, by Proposition \ref{Our_FHLZ17bis}, we have that also Condition \ref{condC} is satisfied by all couples $x_i$ and $x_j$, $i,j=1,\ldots, n$.
Therefore, we can apply Propositions \ref{ThmSpecDen1}
 and
 \ref{ThmSpecDen2} to the entries of the estimator of the spectral density matrix of $\bm x_n$, defined in \eqref{hatSimgax}. This yields the desired result.


\subsection{Proof of Theorem~\ref{Thm.hatChin}}\label{App.Pfrate}

Hereafter, let
\[
\alpha_{\bar{\vsbf}}=\max\left\lbrace \frac{(\log B_{S_1} \log B_{S_2}  \log B_{T} )^2 B_{S_1} B_{S_2} B_{T}}{S_1 S_2 T},  \frac 1{B_{S_1}^{2\vartheta_1}},  \frac 1{B_{S_2}^{2\vartheta_2}},  \frac 1{B_T^{2\vartheta_3}} \right\rbrace.
\]
Moreover, we let $C_1, C_2, \ldots$ and $C_1^*, C_2^*, \ldots$  denote generic finite positive constants independent of $n$. We let also 
$\e_{ni}$, $i= 1, \ldots, n,$ denote the $i$-th canonical basis of $\mathbb{R}^n$. Throughout, we make use of the spectral norm, denoted as $\Vert\cdot\Vert$, and the Frobenius norm, denoted as $\Vert\cdot\Vert_F$.
The proof relies on some lemmas introduced below.

\begin{lemma}\label{lem.LamHat}
Let Assumptions \ref{Ass.A}, \ref{Ass.B}, \ref{Ass.MA},  \ref{Ass.coef1}, \ref{Ass.moments}, and~\ref{AssSpec2}  hold. Then, for all $n\in\mathbb N$, 
\begin{enumerate}[(i)]
\item $\sup_{\thbf \in \Thbf} n^{-2} {\rm E} \left\Vert \widehat{\boldsymbol \Sigma}_n^x(\thbf) - {\boldsymbol \Sigma}_n^x(\thbf) \right\Vert^2  \leq  C \alpha_{\bar{\vsbf}}$, where $C$ is the same constant as in Theorem~\ref{Thmsigmax};

\item $\max_{1\leq i \leq n} \sup_{\thbf \in \Thbf} n^{-1} {\rm E} \left\Vert \e_{ni}^\top \left(\widehat{\boldsymbol \Sigma}_n^x(\thbf) - {\boldsymbol \Sigma}_n^x(\thbf)\right) \right\Vert^2  \leq  C \alpha_{\bar{\vsbf}}$, where $C$ is the same constant as in Theorem~\ref{Thmsigmax};

\item $\sup_{\thbf \in \Thbf} n^{-1}  \left\Vert {\boldsymbol \Sigma}_n^x(\thbf) - {\boldsymbol \Sigma}_n^{\chi}(\thbf) \right\Vert  \leq  C_1^* n^{-1}$;

\item $\sup_{\thbf \in \Thbf} n^{-2} {\rm E} \left\Vert \widehat{\boldsymbol \Sigma}_n^x(\thbf) - {\boldsymbol \Sigma}_n^{\chi}(\thbf) \right\Vert^2  \leq  C_1 \max(n^{-2}, \alpha_{\bar{\vsbf}})$;

\item $\max_{1\leq i \leq n} \sup_{\thbf \in \Thbf} n^{-1} {\rm E} \left\Vert \e_{ni}^\top \left(\widehat{\boldsymbol \Sigma}_n^x(\thbf) - {\boldsymbol \Sigma}_n^{\chi}(\thbf)\right) \right\Vert^2  \leq  C_2 \max(n^{-1}, \alpha_{\bar{\vsbf}})$.

\end{enumerate}
\end{lemma}

\begin{proof}
Part (i). For any $\thbf\in\Thbf$ and any $n\in\mathbb N$, we have
\begin{align}
\left\Vert \widehat{\boldsymbol \Sigma}_n^x(\thbf) - {\boldsymbol \Sigma}_n^x(\thbf) \right\Vert^2 &\le  
\left\Vert \widehat{\boldsymbol \Sigma}_n^x(\thbf) - {\boldsymbol \Sigma}_n^x(\thbf) \right\Vert^2_F\nn\\
& = {\rm trace}\left[ \left(\widehat{\boldsymbol \Sigma}_n^x(\thbf) - {\boldsymbol \Sigma}_n^x(\thbf)\right)^\dag \left(\widehat{\boldsymbol \Sigma}_n^x(\thbf) - {\boldsymbol \Sigma}_n^x(\thbf)\right) \right] 
= \sum_{i=1}^n \sum_{j=1}^n \left| \widehat{\sigma}_{ij}^x(\thbf)  - \sigma_{ij}^x(\thbf) \right|^2. \nn
\end{align}
The result then follows from Theorem~\ref{Thmsigmax}.  


Part (ii). Noticing that 
$$n^{-1} {\rm E} \left\Vert \e_{ni}^\top \left(\widehat{\boldsymbol \Sigma}_n^x(\thbf) - {\boldsymbol \Sigma}_n^x(\thbf)\right) \right\Vert^2 = n^{-1} \sum_{j=1}^n {\rm E} \left| \widehat{\sigma}_{ij}^x(\thbf)  - \sigma_{ij}^x(\thbf) \right|^2,$$
the result then follows from Theorem~\ref{Thmsigmax}.

Part (iii). Since
$\left\Vert {\boldsymbol \Sigma}_n^x(\thbf) - {\boldsymbol \Sigma}_n^{\chi}(\thbf) \right\Vert = \left\Vert {\boldsymbol \Sigma}_n^{\xi}(\thbf) \right\Vert$, the result follows immediately from Proposition~\ref{lukagu}.


Part (iv). Denote by ${\boldsymbol \Sigma}_n^{\xi}(\thbf)$ the $n\times n$ spectral density matrix of ${\boldsymbol \xi}_{n}$. We have
\begin{align*}
\left\Vert \widehat{\boldsymbol \Sigma}_n^x(\thbf) - {\boldsymbol \Sigma}_n^{\chi}(\thbf) \right\Vert 
&= \left\Vert \widehat{\boldsymbol \Sigma}_n^x(\thbf) - \left({\boldsymbol \Sigma}_n^x(\thbf)- {\boldsymbol \Sigma}_n^{\xi}(\thbf) \right)\right\Vert \leq \left\Vert \widehat{\boldsymbol \Sigma}_n^x(\thbf) - {\boldsymbol \Sigma}_n^x(\thbf) \right\Vert + \left\Vert {\boldsymbol \Sigma}_n^{\xi}(\thbf) \right\Vert. 
\end{align*}
The result then follows immediately from part (i) and Proposition~\ref{lukagu}.

Part (v). Note that
$$\left\Vert \e_{ni}^\top \left(\widehat{\boldsymbol \Sigma}_n^x(\thbf) - {\boldsymbol \Sigma}_n^{\chi}(\thbf)\right) \right\Vert 
\leq \left\Vert \e_{ni}^\top \left(\widehat{\boldsymbol \Sigma}_n^x(\thbf) - {\boldsymbol \Sigma}_n^x(\thbf)\right) \right\Vert + \left\Vert \e_{ni}^\top 
 {\boldsymbol \Sigma}_n^{\xi}(\thbf) \right\Vert.
$$
The result then follows immediately from part (ii) and Proposition~\ref{lukagu}.
\end{proof}
%
%

\begin{lemma}\label{lem.lamChi2}
Let Assumptions \ref{Ass.A}, \ref{Ass.B}, \ref{Ass.MA},  \ref{Ass.coef1}, \ref{Ass.moments}, and~\ref{AssSpec2}  hold. Then, for all $n\in\mathbb N$ and  all $j = 1, \ldots, q$,
\begin{enumerate}[(i)]
\item $\sup_{\thbf \in \Thbf} n^{-2} {\rm E} \vert \widehat{\lambda}_{nj}^x(\bth) - \lambda_{nj}^x (\bth)\vert^2 \leq C \alpha_{\bar{\vsbf}}$, where $C$ is the same constant as in Lemma~\ref{lem.LamHat}(i);
\item $\sup_{\thbf \in \Thbf} n^{-1}  \vert {\lambda}_{nj}^x(\bth) - \lambda_{nj}^{\chi}(\bth)\vert \leq C_1^* n^{-1}$, where $C_1^*$ is the same constant as in Lemma~\ref{lem.LamHat}(iii);
\item 
$
\sup_{\thbf \in \Thbf} n^{-2} {\rm E} \vert \widehat{\lambda}_{nj}^x(\bth) - \lambda_{nj}^{\chi}(\bth)\vert^2 \leq C_1 \max(n^{-2}, \alpha_{\bar{\vsbf}}),
$ where $C_1$ is the same constant as in Lemma~\ref{lem.LamHat}(iv);

%
\end{enumerate}
\end{lemma}

\begin{proof}
For any two $n\times n$ matrices $\A_1$ and $\A_2$, Wely's inequality (see Appendix  \ref{App:Weyl}) implies that
\begin{equation}\label{eq.Wely}
\vert\nu_{\ell}(\A_1+\A_2) - \nu_{\ell}(\A_1)\vert \leq \Vert \A_2 \Vert, \ell = 1, \ldots, n.
\end{equation}

Part (i). Letting $\A_1 = {\boldsymbol \Sigma}_n^{x}(\thbf)$, $\A_2 = \widehat{\boldsymbol \Sigma}_n^x(\thbf) - {\boldsymbol \Sigma}_n^{x}(\thbf)$,
 the result then follows from \eqref{eq.Wely}, which implies $\vert \widehat{\lambda}_{nj}^x(\bth) - \lambda_{nj}^{x}(\bth)\vert \leq \Vert \widehat{\boldsymbol \Sigma}_n^x(\thbf) - {\boldsymbol \Sigma}_n^{x}(\thbf) \Vert$
and Lemma~\ref{lem.LamHat}(i).

Part (ii). Letting 
 $\A_1 = {\boldsymbol \Sigma}_n^{\chi}(\thbf)$, $\A_2 = {\boldsymbol \Sigma}_n^x(\thbf) - {\boldsymbol \Sigma}_n^{\chi}(\thbf)$ yields, because of \eqref{eq.Wely},
 $\vert {\lambda}_{nj}^x(\bth) - \lambda_{nj}^{\chi}(\bth)\vert \leq \left\Vert {\boldsymbol \Sigma}_n^x(\thbf) - {\boldsymbol \Sigma}_n^{\chi}(\thbf) \right\Vert.$ The result hence follows from Lemma \ref{lem.LamHat}(iii).

Part (iii).  Letting $\A_1 = {\boldsymbol \Sigma}_n^{\chi}(\thbf)$, $\A_2 = \widehat{\boldsymbol \Sigma}_n^x(\thbf) - {\boldsymbol \Sigma}_n^{\chi}(\thbf)$,
 the result then follows from \eqref{eq.Wely}, which implies
$\vert \widehat{\lambda}_{nj}^x(\bth) - \lambda_{nj}^{\chi}(\bth)\vert \leq \Vert \widehat{\boldsymbol \Sigma}_n^x(\thbf) - {\boldsymbol \Sigma}_n^{\chi}(\thbf) \Vert$
and Lemma~\ref{lem.LamHat}(iv).
%
%
\end{proof}

\begin{lemma}\label{lem.PLamW}
Let Assumptions \ref{Ass.A}, \ref{Ass.B}, 
\ref{Ass.coef1}, and \ref{Ass.lamChi}
 hold. Then, for all $n\in\mathbb N$,
\begin{equation}
\max_{1\leq i \leq n} \sup_{\thbf \in \Thbf} \left\Vert \e_{ni}^\top  n^{1/2}({{\Pbf}^{\chi}_n}(\bth))^\dag \right\Vert \leq C_{2}.\nn
\end{equation}
\end{lemma}
\begin{proof}
Let ${p}_{n,jk}^{\chi}(\thbf)$, $j=1, \ldots, q$, $k= 1, \ldots, n,$ denote $(j,k)$-th entry of ${\Pbf}_{n}^{\chi}(\thbf)$.
Then, 
$$
\max_{1\leq k \leq n} \sup_{\thbf \in \Thbf} \sigma^{\chi}_{k k}(\bth) = \max_{1\leq k \leq n} \sup_{\thbf \in \Thbf} \sum_{j=1}^q \lambda_{nj}^{\chi}(\bth) \vert {p}_{n,jk}^{\chi}(\bth)\vert^2 < \infty.
$$
Indeed, by Assumptions~\ref{Ass.MA} and \ref{Ass.coef1}
\[
\sigma^{\chi}_{k k}(\bth) = \sum_{\hbf} e^{-i\langle \hbf,\thbf\rangle} \text E[\chi_{k\vsbf}\chi_{k\vsbf-\hbf}] \le \sum_{\hbf} \vert\text E[\chi_{k\vsbf}\chi_{k\vsbf-\hbf}] \vert\le\sum_{h_1}\sum_{h_2}\sum_{h_3} \widetilde{A}^*_{\rho_1^\chi,\rho_2^\chi,\rho_3^\chi} {\rho^\chi_1}^{|h_1|} {\rho^\chi_2}^{|h_2|} {\rho^\chi_3}^{|h_3|}<\infty,
\]
where  $\widetilde{A}^*_{\rho_1^\chi,\rho_2^\chi,\rho_3^\chi}$, $\rho^\chi_1, \rho^\chi_2$, and $\rho^\chi_3$ are defined in the proof of Proposition \ref{Our_FHLZ17bis}. Moreover, $\sigma^{\chi}_{k k}(\bth)\ge 0$ \citep[Corollary 4.3.2, p.120]{BD06}. 

Now, since, by Assumption \ref{Ass.lamChi}, for all $j=1,\ldots, q$,
$$
\sup_{\thbf\in\Thbf} \frac {\lambda_{nj}^{\chi}(\bth)}n > \sup_{\thbf\in\Thbf} \undertilde{\omega}(\thbf)_{j}>0,
$$ 
we must have $\sup_{\thbf\in\Thbf} n \vert {p}_{n,jk}^{\chi}(\bth)\vert^2 < C_{2}$, where $C_2$ is independent of $j$ and $k$. This completes the proof.
\end{proof}
%
%
%

\begin{lemma}\label{lem.Lamxxhat}
Let Assumptions \ref{Ass.A}, \ref{Ass.B}, 
\ref{Ass.coef1}, and \ref{Ass.lamChi}
 hold. Then, for all $n\in\mathbb N$,
%
%

For all $j=1,\ldots, q-1$ there exist continuous functions 
$\bth \mapsto \widetilde{\omega}_j^*(\bth)$ and $\bth \mapsto \undertilde{\omega}_j^*(\bth)$
such that for all $\thbf \in\Thbf$
$$
0<
\undertilde{\omega}_{j+1}^*(\bth)\leq
\lim_{n\to\infty}\frac{\lambda_{n,j+1}^{x}(\bth)}n\leq \widetilde{\omega}_{j+1}^*(\bth)<\undertilde{\omega}_j^*(\bth) \leq \lim_{n\to\infty}\frac{ \lambda_{nj}^{x}(\bth)}n \leq \widetilde{\omega}_j^*(\bth)<\infty. 
$$
%
%
%
%
\end{lemma}

\begin{proof}
%
The result follows immediately from Lemma \ref{lem.lamChi2}(ii)  and Assumption~\ref{Ass.lamChi}. 
\end{proof}

\begin{lemma}\label{Lem.KhatK}
Let Assumptions \ref{Ass.A}, \ref{Ass.B}, \ref{Ass.MA},  \ref{Ass.coef1}, \ref{Ass.lamChi}, \ref{Ass.moments}, \ref{AssSpec2}, and \ref{Ass.ExpK}  hold.  Then, for all $n\in\mathbb N$,
$$\max_{1\leq \ell \leq n} \sup_{\thbf \in \Thbf} {\rm E} \left\Vert \e_{n\ell}^\top  n^{1/2} \left[{\K}^{\chi}_{n}(\thbf) - \widehat{\K}^x_{n}(\thbf)\right] \right\Vert^2  \leq C_3^* \max(n^{-1}, \alpha_{\bar{\vsbf}}).$$
\end{lemma}

\begin{proof}
Theorem~2 in \cite{yu2015useful} implies that there exists a $q\times q$ orthogonal matrix ${\boldsymbol O}_q$ such that
\begin{align}\label{Ineq.Yu2015}
\left\Vert {\Pbf}_n^{\chi}(\thbf) - {\boldsymbol O}_q \widehat{\Pbf}_n^{x}(\thbf) \right\Vert &\leq \left\Vert {\Pbf}_n^{\chi}(\thbf) - {\boldsymbol O}_q \widehat{\Pbf}_n^{x}(\thbf) \right\Vert_{\rm F}\nn \\
&\leq C_4^* \frac{\left\Vert \widehat{\boldsymbol \Sigma}_n^x(\thbf) - {\boldsymbol \Sigma}_n^{\chi}(\thbf) \right\Vert}{{\lambda}_{nq}^{\chi}(\bth) - {\lambda}_{n,q+1}^{\chi}(\bth)} \leq C_5^* n^{-1} \left\Vert \widehat{\boldsymbol \Sigma}_n^x(\thbf) - {\boldsymbol \Sigma}_n^{\chi}(\thbf) \right\Vert,
\end{align}
where we used the fact that ${\lambda}_{n,q+1}^{\chi}(\bth)=0$ and the last inequality is due to Lemma~\ref{lem.Lamxxhat}. Note that
\begin{align*}
& \e_{n\ell}^\top  n^{1/2} \left[{\K}^{\chi}_{n}(\thbf) - \widehat{\K}^x_{n}(\thbf)\right] \\
=&\, \e_{n\ell}^\top  n^{1/2} {\Pbf}_n^{\chi\dag}(\thbf) \left[ {\Pbf}_n^{\chi}(\thbf) - {\boldsymbol O}_q \widehat{\Pbf}_n^{x}(\thbf) \right] + \e_{n\ell}^\top  n^{1/2} \left[ {\Pbf}_n^{\chi\dag}(\thbf) {\boldsymbol O}_q - \widehat{\Pbf}_n^{x\dag}(\thbf) \right] \widehat{\Pbf}_n^{x}(\thbf)\\
=&\,  {\boldsymbol d}_{n\ell}^{(1)}(\thbf) + {\boldsymbol d}_{n\ell}^{(2)}(\thbf), \qquad \text{say}.
\end{align*}
For ${\boldsymbol d}_{n\ell}^{(1)}(\thbf)$,
in view of  \eqref{Ineq.Yu2015} and Lemmas  \ref{lem.LamHat}(iv) and \ref{lem.PLamW}, we have
$$\max_{1\leq \ell \leq n} \sup_{\thbf \in \Thbf} {\rm E} \left\Vert {\boldsymbol d}_{n\ell}^{(1)}(\thbf) \right\Vert^2 \leq C_6^* \max(n^{-2}, \alpha_{\bar{\vsbf}}).$$
As for ${\boldsymbol d}_{n\ell}^{(2)}(\thbf)$, note that
\begin{align*}
\sum_{\ell = 1}^n  \left\Vert n^{1/2} \e_{n\ell}^\top \left[{\Pbf}_n^{\chi}(\thbf) - {\boldsymbol O}_q \widehat{\Pbf}_n^{x}(\thbf)\right] \right\Vert^2 &= \left\Vert n^{1/2} \left[{\Pbf}_n^{\chi}(\thbf) - {\boldsymbol O}_q \widehat{\Pbf}_n^{x}(\thbf)\right] \right\Vert_{\rm F}^2\\
&\leq (C_5^*)^2 n^{-2} \left\Vert n^{1/2} \left[\widehat{\boldsymbol \Sigma}_n^x(\thbf) - {\boldsymbol \Sigma}_n^{\chi}(\thbf)\right] \right\Vert_{\rm F}^2 \\
&= (C_5^*)^2 n^{-1} \sum_{\ell = 1}^n \left\Vert  \e_{n\ell}^\top \left[\widehat{\boldsymbol \Sigma}_n^x(\thbf) - {\boldsymbol \Sigma}_n^{\chi}(\thbf)\right] \right\Vert^2.
\end{align*}
Hence, for all $\ell = 1, \ldots, n$, ${\rm E}\left\Vert n^{1/2} \e_{n\ell}^\top \left[{\Pbf}_n^{\chi}(\thbf) - {\boldsymbol O}_q \widehat{\Pbf}_n^{x}(\thbf)\right] \right\Vert^2$ is of order no greater than $$\max_{1\leq i\leq n} n^{-1} {\rm E} \left\Vert  \e_{ni}^\top \left[\widehat{\boldsymbol \Sigma}_n^x(\thbf) - {\boldsymbol \Sigma}_n^{\chi}(\thbf)\right] \right\Vert^2,$$
which, in view of Lemma~\ref{lem.LamHat}(v), entails
$$\max_{1\leq \ell \leq n} \sup_{\thbf \in \Thbf} {\rm E} \left\Vert {\boldsymbol d}_{n\ell}^{(2)}(\thbf) \right\Vert^2 \leq C_7^* \max(n^{-1}, \alpha_{\bar{\vsbf}}).$$
The result then follows.
\end{proof}

We can then prove the theorem as follows. 
Let $\bchi_{n \vsbf} = (\chi_{1\vsbf}, \ldots, \chi_{n\vsbf})^\top$ and $\bfxi_{n \vsbf} = (\xi_{1\vsbf}, \ldots, \xi_{n\vsbf})^\top$.
Recall from \eqref{chatbot2} in Remark \ref{rem:freddo} that
\begin{align*}
{\underline{\K}}^{\chi}_{n}(L) \x_{n\vsbf} = {\underline{\K}}^{\chi}_{n}(L) \bchi_{n \vsbf} + {\underline{\K}}^{\chi}_{n}(L) \bfxi_{n \vsbf} 
= {\underline{\K}}^{\chi}_{n}(L) \bchi_{n \vsbf} = \bchi_{n \vsbf}.
\end{align*}

Letting $$\mathcal{D}_{{\vsbf}} = \{(s_1, s_2, t): \underline{\kappa}_1(s_1) \leq s_1 \leq \overline{\kappa}_1(s_1), \underline{\kappa}_2(s_2) \leq s_2 \leq \overline{\kappa}_2(s_2), \underline{\kappa}_3(t)   \leq t \leq \overline{\kappa}_3(t)\},$$
we then have 
\begin{align*}
  \widehat{\chi}^{(n)}_{\ell\vsbf} - \chi^n_{\ell\vsbf}  
&=  \sum_{\vsbf\pr \in \mathcal{D}_{{\vsbf}}} \e_{n\ell}^\top \left( \widehat{\underline{\K}}^x_{n\vsbf\pr} - \underline{\K}^{\chi}_{n\vsbf\pr} \right)  \x_{n\vsbf - \vsbf\pr} + \sum_{\vsbf\pr \in \mathbb{Z}^3 \setminus \mathcal{D}_{{\vsbf}}} \e_{n\ell}^\top \underline{\K}^{\chi}_{n\vsbf\pr} \x_{n\vsbf - \vsbf\pr} = {a}_{n\ell\vsbf} +  \widetilde{a}_{n\ell\vsbf}, \; \text{say}.
\end{align*}
We first derive convergence rate of ${a}_{n\ell\vsbf}$. 
Letting $\bm C_{n\kbf}$ be the $n\times q$ matrix with entries  $c_{\ell j, \kbf}$, $\ell=1,\ldots,n$, $j=1,\ldots, q$, $\kbf\in\mathbb Z^2\times \mathbb N_0$, from Assumption \ref{Ass.MA}, we have
\[
\bm\chi_{n\vsbf-\vsbf\pr}= \sum_{\kbf} \bm C_{n\kbf} \bm v_{\vsbf-\vsbf\pr-\kbf}.
\]
Therefore,
\begin{align*}
\vert  {a}_{n\ell\vsbf} \vert
&\leq \left\vert \sum_{\vsbf\pr \in \mathcal{D}_{{\vsbf}}} \e_{n\ell}^\top \left( \widehat{\underline{\K}}^x_{n\vsbf\pr} - \underline{\K}^{\chi}_{n\vsbf\pr} \right)  \l( \sum_{\kbf} \bm C_{n\kbf} \bm v_{\vsbf-\vsbf\pr-\kbf}\r)\right\vert + \left\vert \sum_{\vsbf\pr \in \mathcal{D}_{{\vsbf}}} \e_{n\ell}^\top \left( \widehat{\underline{\K}}^x_{n\vsbf\pr} - \underline{\K}^{\chi}_{n\vsbf\pr} \right) \bfxi_{n \vsbf-\vsbf\pr}  \right\vert
=  a_{n\ell\vsbf}^{(1)} + a_{n\ell\vsbf}^{(2)}, \; \text{say}.
\end{align*}

Consider $a_{n\ell\vsbf}^{(1)}$. By Cauchy-Schwarz inequality
\begin{align}
 {\rm E}(a_{n\ell\vsbf}^{(1)}) &\leq  \sum_{\vsbf\pr \in \mathcal{D}_{{\vsbf}}} {\rm E} \left\vert  \e_{n\ell}^\top \left( \widehat{\underline{\K}}^x_{n\vsbf\pr} - \underline{\K}^{\chi}_{n\vsbf\pr} \right) \left( \sum_{\kbf} \bm C_{n\kbf} \bm v_{\vsbf-\vsbf\pr-\kbf} 
 \right) \right\vert\nn \\
&\leq  \sum_{\vsbf\pr \in \mathcal{D}_{{\vsbf}}} \left\lbrace {\rm E} \left\Vert  n^{1/2} \e_{n\ell}^\top \left( \widehat{\underline{\K}}^x_{n\vsbf\pr} - \underline{\K}^{\chi}_{n\vsbf\pr} \right) \right\Vert^2\right\rbrace^{1/2} \left\lbrace n^{-1} {\rm E} \left\Vert 
\sum_{\kbf} \bm C_{n\kbf} \bm v_{\vsbf-\vsbf\pr-\kbf} \right\Vert^2\right\rbrace^{1/2}. \label{CSI}
\end{align}
Now, because of Assumption \ref{Ass.MA}, which implies that $\bm v$ is a white noise rf, and Assumption \ref{Ass.coef1}
\begin{align}
\frac 1n {\rm E}\left\Vert  \sum_{\kbf} \bm C_{n\kbf} \bm v_{\vsbf-\vsbf\pr-\kbf}  \right\Vert^2 
&\le \frac 1n {\rm E}\left\Vert  \sum_{\kbf} \bm C_{n\kbf} \bm v_{\vsbf-\vsbf\pr-\kbf}  \right\Vert^2_F\nn\\ 
&= \frac 1n  \sum_{\kbf} \sum_{\kbf\pr} {\rm trace} \left(\bm C_{n\kbf} {\rm E} (\bm v_{\vsbf-\vsbf\pr-\kbf} \bm v_{\vsbf-\vsbf\pr-\kbf\pr}^\dag)   \bm C_{n\kbf\pr}^\dag \right)\nn \\ 
&=\frac 1n \sum_{\kbf} {\rm trace} \left( \bm C_{n\kbf}  \bm C_{n\kbf}^\dag \right) = \frac 1n \sum_{\kbf} \sum_{j=1}^q \sum_{\ell =1}^n \vert c_{\ell j, \kbf}\vert^2 \nn\\
&\le \frac 1n \sum_{\kbf} \sum_{j=1}^q \sum_{\ell =1}^n\l(A^{\chi}_{\ell j} {\rho_1^{\chi|\kappa_1|}} {\rho_2^{\chi|\kappa_2|}} {\rho_3^{\chi\kappa_3}}\r)^2\le q\max_{1\le \ell \le n}\max_{1\le j\le q} (A^{\chi}_{\ell j})^2 \widetilde A_{{\rho_1^\chi}{\rho_2^\chi}{\rho_3^\chi}}\nn\\
&\le q (A^{\chi})^2\widetilde A_{{\rho_1^\chi}{\rho_2^\chi}{\rho_3^\chi}}= C_8^*, \;\text {say},\label{primopezzo}
\end{align}
for some finite $\widetilde A_{{\rho_1^\chi}{\rho_2^\chi}{\rho_3^\chi}}>0$ depending only on ${\rho_1^\chi}$, ${\rho_2^\chi}$, and ${\rho_3^\chi}$.
Furthermore, noting that 
\begin{align*}
\left\Vert  n^{1/2} \e_{n\ell}^\top \left( \widehat{\underline{\K}}^x_{n\vsbf\pr} - \underline{\K}^{\chi}_{n\vsbf\pr} \right) \right\Vert 
&= \left\Vert\frac{1}{8\pi^3} \int_{\Thbf}   n^{1/2} \e_{n\ell}^\top \left[\widehat{\K}^x_{n}(\thbf) - {\K}^{\chi}_{n}(\thbf)\right]    e^{i\langle\vsbf\pr,\thbf\rangle}{\rm d} \thbf
\right\Vert \nn\\
&\leq
\frac{1}{8\pi^3} \int_{\Thbf} \left\Vert  n^{1/2} \e_{n\ell}^\top \left[\widehat{\K}^x_{n}(\thbf) - {\K}^{\chi}_{n}(\thbf)\right] \right\Vert  {\rm d} \thbf,
\end{align*}
we have
\begin{align*}
\left\Vert  n^{1/2} \e_{n\ell}^\top \left( \widehat{\underline{\K}}^x_{n\vsbf\pr} - \underline{\K}^{\chi}_{n\vsbf\pr} \right) \right\Vert^2 &\leq
\left( \frac{1}{8\pi^3}\right)^2 \left\lbrace \int_{\Thbf} \left\Vert  n^{1/2} \e_{n\ell}^\top \left[\widehat{\K}^x_{n}(\thbf) - {\K}^{\chi}_{n}(\thbf)\right] \right\Vert  {\rm d} \thbf \right\rbrace^2 \\
&\leq   \frac{1}{8\pi^3}  \int_{\Thbf} \left\Vert  n^{1/2} \e_{n\ell}^\top \left[\widehat{\K}^x_{n}(\thbf) - {\K}^{\chi}_{n}(\thbf)\right] \right\Vert^2 {\rm d} \thbf,
\end{align*}
which, because of Lemma~\ref{Lem.KhatK}, entails that
\begin{align}
{\rm E} \left\Vert  n^{1/2} \e_{n\ell}^\top \left( \widehat{\underline{\K}}^x_{n\vsbf\pr} - \underline{\K}^{\chi}_{n\vsbf\pr} \right) \right\Vert^2 &\leq  \frac{1}{8\pi^3}  \int_{\Thbf} {\rm E} \left\Vert  n^{1/2} \e_{n\ell}^\top \left[\widehat{\K}^x_{n}(\thbf) - {\K}^{\chi}_{n}(\thbf)\right] \right\Vert^2 {\rm d} \thbf\nn \\
&\leq  \sup_{\thbf \in \Thbf}  {\rm E} \left\Vert  n^{1/2} \e_{n\ell}^\top \left[\widehat{\K}^x_{n}(\thbf) - {\K}^{\chi}_{n}(\thbf)\right] \right\Vert^2\nn\\
&\leq C_3^* \max(n^{-1}, \alpha_{\bar{\vsbf}}).\label{secondopezzo}
\end{align}
 By using \eqref{primopezzo} and \eqref{secondopezzo}, into \eqref{CSI}, and recalling \eqref{def.underlineK} and \eqref{troncamento}, we have
\begin{align}
\max_{1\leq \ell \leq n} {\rm E}(a_{n\ell\vsbf}^{(1)})  &\leq \sqrt{C_3^*C_8^*} \max(n^{-1/2}, \alpha_{\bar{\vsbf}}^{1/2}) (\overline{\kappa}_1(s_1)-\underline{\kappa}_1(s_1)) (\overline{\kappa}_2(s_2)-\underline{\kappa}_2(s_2)) (\overline{\kappa}_3(t)-\underline{\kappa}_3(t))\nn\\
&\leq \sqrt{C_3^*C_8^*} \max(n^{-1/2}, \alpha_{\bar{\vsbf}}^{1/2}) M_{S_1}M_{S_2} M_T.\nn
\end{align}
Applying the same arguments will yield the same convergence rate of $a_{n\ell\vsbf}^{(2)}$. We hence have 
$$
\max_{1\leq \ell \leq n} {\rm E}(a_{n\ell\vsbf})  \leq C \max(n^{-1/2}, \alpha_{\bar{\vsbf}}^{1/2}) M_{S_1}M_{S_2} M_T,
$$
for some finite $C>0$ independent of $n, S_1, S_2$ and $T$.

Turning to $\widetilde{a}_{n\ell\vsbf}$, we have
\begin{align*}
\vert  \widetilde{a}_{n\ell\vsbf} \vert
&\leq \left\vert \sum_{\vsbf\pr \in \mathbb{Z}^3 \setminus \mathcal{D}_{{\vsbf}}} \e_{n\ell}^\top \underline{\K}^{\chi}_{n\vsbf\pr} \left( \sum_{\kbf} \bm C_{n\kbf} \bm v_{\vsbf-\vsbf\pr-\kbf} 
 \right) \right\vert + \left\vert \sum_{\vsbf\pr \in \mathbb{Z}^3 \setminus \mathcal{D}_{{\vsbf}}} \e_{n\ell}^\top \underline{\K}^{\chi}_{n\vsbf\pr} \bfxi_{n \vsbf-\vsbf\pr}  \right\vert
=  \widetilde{a}_{n\ell\vsbf}^{(1)} + \widetilde{a}_{n\ell\vsbf}^{(2)}, \; \text{say}.
\end{align*}
Consider $\widetilde{a}_{n\ell\vsbf}^{(1)}$. By repeating the same arguments for bounding ${\rm E}(a_{n\ell\vsbf}^{(1)})$, and using Assumption~\ref{Ass.ExpK} and Lemma~\ref{lem.PLamW}, we have
\begin{align}\label{eq.a1tilde}
{\rm E}(\widetilde{a}_{n\ell\vsbf}^{(1)}) &\leq \sum_{\vsbf\pr \in \mathbb{Z}^3 \setminus \mathcal{D}_{{\vsbf}}} \left\lbrace {\rm E} \left\Vert  n^{1/2} \e_{n\ell}^\top \underline{\K}^{\chi}_{n\vsbf\pr} \right\Vert^2\right\rbrace^{1/2} \left\lbrace n^{-1} {\rm E} \left\Vert 
\sum_{\kbf} \bm C_{n\kbf} \bm v_{\vsbf-\vsbf\pr-\kbf} \right\Vert^2\right\rbrace^{1/2}\nn\\
&\le \sqrt{C_8^*} \left\lbrace {\rm E} \left\Vert  n^{1/2} \e_{n\ell}^\top \underline{\K}^{\chi}_{n(0\ 0\ 0)}\right\Vert^2\right\rbrace^{1/2} \sum_{\vsbf\pr \in \mathbb{Z}^3 \setminus \mathcal{D}_{{\vsbf}}} (1+\varepsilon_1)^{-|s_1\pr|}  (1+\varepsilon_2)^{-|s_2\pr|} (1+\varepsilon_3)^{-|t\pr|}\nn \\
&\leq  {C}_9^* \left\lbrace {\rm E} \left\Vert  n^{1/2} \e_{n\ell}^\top \underline{\K}^{\chi}_{n(0\ 0\ 0)} \right\Vert^2\right\rbrace^{1/2} (1+\varepsilon_1)^{-\kappa^*_1(s_1)}  (1+\varepsilon_2)^{-\kappa^*_2(s_2)} (1+\varepsilon_3)^{-\kappa^*_3(t)}\nn \\
&\leq C^* (1+\varepsilon_1)^{-\kappa^*_1(s_1)}  (1+\varepsilon_2)^{-\kappa^*_2(s_2)} (1+\varepsilon_3)^{-\kappa^*_3(t)},\nn
\end{align}
for some finite $C^*>0$ independent of $\ell,n,S_1,S_2$, and $T$. 
Applying the same arguments will yield the same convergence rate of $\widetilde{a}_{n\ell\vsbf}^{(2)}$.

All of the above arguments yield
\begin{align*}
\max_{1\leq \ell \leq n} {\rm E} \left\vert \widehat{\chi}^{(n)}_{\ell\vsbf} - \chi_{\ell\vsbf} \right\vert 
\leq  \ &  C \max(n^{-1/2}, \alpha_{\bar{\vsbf}}^{1/2}) M_{S_1}M_{S_2} M_T + C^* (1+\varepsilon_1)^{-\kappa^*_1(s_1)}  (1+\varepsilon_2)^{-\kappa^*_2(s_2)} (1+\varepsilon_3)^{-\kappa^*_3(t)}.
\end{align*}
This proves part (i). 

Furthermore,  we have that 
\[
(1+\varepsilon_1)^{-\kappa^*_1(s_1)}  (1+\varepsilon_2)^{-\kappa^*_2(s_2)} (1+\varepsilon_3)^{-\kappa^*_3(t)}=o\l\{(1+\varepsilon_1)^{-M_{S_1}}  (1+\varepsilon_2)^{-M_{S_2}} (1+\varepsilon_3)^{-M_T} \r\},
\]
which by Assumption \ref{AssSpec2} is dominated by the first term of part (i). This proves part (ii) and completes the proof.

\section{Proofs of Results of of Section 7 }\label{App.PFqselect} 

In order to prove the theorem, we need the following result

\begin{lemma}\label{Lem: lamdahat}
Let Assumptions \ref{Ass.A}, \ref{Ass.B},  \ref{Ass.MA},  \ref{Ass.coef1}, \ref{Ass.lamChi}, \ref{Ass.moments}, \ref{AssSpec2}, and  \ref{Ass.pnvs} hold. Then, for all $\epsilon >0$ there  exist $\delta_{\epsilon}$, $S_{1\epsilon}$, $S_{2\epsilon}$ and $T_{\epsilon}$ such that for any fixed $q_{\max}$, $n$, $S_1 > S_{1\epsilon}$, $S_2 > S_{2\epsilon}$, and $T > T_{\epsilon}$,
\begin{align*}
& \max_{1\leq k \leq q_{\max}} \sup_{\thbf \in \Thbf} {\rm P}\left\lbrace \min\left[ \frac 1{ \log B_{S_1}  \log B_{S_2}  \log B_{T} }\sqrt{\frac{S_1 S_2 T}{B_{S_1} B_{S_2} B_T}}, B_{S_1}^{\vartheta_1}, B_{S_2}^{\vartheta_2}, B_T^{\vartheta_3}\right]   \frac{\vert \widehat{\lambda}^x_{nk}(\thbf) - {\lambda}^x_{nk}(\thbf) \vert}n >  \delta_{\epsilon} \right\rbrace < \epsilon.
\end{align*}

\end{lemma}

\begin{proof}
The result follows immediately from Lemma~A.1 in \cite{HallinLiska2007} and Lemma~\ref{lem.lamChi2}(i).
\end{proof}


It then suffices to prove that for all $k\neq q$, $0\leq k \leq q_{\max}$, as $n,S_1,S_2,T\to\infty$,
\beq\label{desiderio}
{\rm P}\left[ \wh{\rm IC}^{(n)}(k) - \wh{\rm IC}^{(n)}(q) >0 \right] \rightarrow 1.
\eeq
Start with the case that $k < q$. Letting
$$
\widehat{D}_{nk} = \frac{1}{n} \sum_{j=k+1}^n \frac{1}{8\pi^3} 
\int_{\Thbf} \widehat{\lambda}^x_{nj}(\thbf)\mathrm d\thbf
\quad\text{ and }\quad
{D}_{nk} = \frac{1}{n} \sum_{j=k+1}^n\frac 1{8\pi^3}\int_{\Thbf} {\lambda}^x_{nj}(\thbf)\mathrm d\thbf,
$$
we have
\begin{align}
\widehat{\rm IC}^{(n)}(k) - \widehat{\rm IC}^{(n)}(q)
 = \log\left[\frac{(\widehat{D}_{nk} - \widehat{D}_{nq})}{\widehat{D}_{nq}} + 1\right] +(k-q) p(n, S_1,S_2,T).\label{eq.DDbarhat}
\end{align}

Now Lemma~\ref{lem.Lamxxhat}, implies that $\inf_{\thbf \in \Thbf} n^{-1} \lambda^x_{nq} >0$ and this yields that, as $n\rightarrow \infty$

\begin{equation}\label{eq.DDbar}
\log\left[\frac{({D}_{nk} - {D}_{nq})}{{D}_{nq}} + 1\right] > 0.
\end{equation}
The desired result \eqref{desiderio} follows by using \eqref{eq.DDbar}, Lemma~\ref{Lem: lamdahat}, and Assumption \ref{Ass.pnvs} (specifically, the assumption  $p(n, S_1,S_2,T \rightarrow 0$)) in \eqref{eq.DDbarhat}.

  
Let us then consider the case $k > q$. Note that 
\beq\label{FFFF}
{D}_{nq} - {D}_{nk} = \frac{1}{n} \sum_{j=q+1}^k \frac 1{8\pi^3}\int_{\Thbf} {\lambda}^x_{nj}(\thbf)\mathrm d\thbf\le \frac{(k-q)C}n,
\eeq
for some finite $C>0$ independent of $n$. Indeed, by Weyl's inequality and Proposition \ref{lukagu},
$$
\sup_{\thbf \in \Thbf}\sup_{n\in\mathbb N} \lambda^x_{n,q+1}(\thbf)\le \sup_{\thbf \in \Thbf}\sup_{n\in\mathbb N} \l\{\lambda^{\chi}_{n,q+1}(\thbf)+ \lambda^{\xi}_{n,q+1}(\thbf)\r\}=\sup_{\thbf \in \Thbf}\sup_{n\in\mathbb N} \lambda^{\chi}_{n,q+1}(\thbf)\le C.
$$ 
From \eqref{FFFF} and Lemma~\ref{Lem: lamdahat} we have
$$\widehat{D}_{nq} - \widehat{D}_{nk} = O_{\rm P}\left\lbrace \max\left[\frac 1n, \log B_{S_1}  \log B_{S_2}  \log B_{T}\sqrt{\frac{B_{S_1} B_{S_2} B_T}{S_1 S_2 T}}, \frac 1{B_{S_1}^{\vartheta_1}}, \frac 1{B_{S_2}^{\vartheta_2}},\frac 1{B_{T}^{\vartheta_3}}\right]\right\rbrace,
$$
which implies also
\begin{align*}
&\log\left[\frac{(\widehat{D}_{nq} - \widehat{D}_{nk})}{\widehat{D}_{nk}} +1 \right]= O_{\rm P}\left\lbrace \max\left[\frac 1n, \log B_{S_1}  \log B_{S_2}  \log B_{T}\sqrt{\frac{B_{S_1} B_{S_2} B_T}{S_1 S_2 T}}, \frac 1{B_{S_1}^{\vartheta_1}}, \frac 1{B_{S_2}^{\vartheta_2}},\frac 1{B_{T}^{\vartheta_3}}\right]\right\rbrace.
\end{align*}
Now, given that
\begin{align*}
\widehat{\rm IC}_{n}(q) - \widehat{\rm IC}_{n}(k)
 = \log\left[\frac{(\widehat{D}_{nk} - \widehat{D}_{nq})}{\widehat{D}_{nq}} + 1\right] +(q-k) p(n, S_1,S_2,T),
 \end{align*}
 \eqref{desiderio} follows from Assumption~\ref{Ass.pnvs}, which yields 
$$
{\rm P}\l((k-q) p(n, S_1,S_2,T) >  \log\left[\frac{(\widehat{D}_{nq} - \widehat{D}_{nk})}{\widehat{D}_{nk}}+ 1\right]\r)\to 1 \quad \text{as} \quad n,S_1,S_2,T\to\infty.
$$
This completes the proof.


\section{Selection of the number of factors in practice}\label{sec:HLABC}

Clearly, $c \mapsto \widehat{q}_{c}^{(n)}$ is a non-increasing map: a small (large) value of $c$ corresponds to underpenalization (overpenalization). Therefore, the correct identification of  $q$ should be based on a sequence of $c$, starting from a small value until appropriate penalization is reached. A thorough discussion and numerical analysis on this aspect in the context of GDFM is available in \cite{HallinLiska2007}. Adopting it to our setting of spatio-temporal setting, we propose the following procedure for the choice of $c$.
For a given sample  of dimension $n$ and a fixed $c>0$, consider a sequence of estimator $\widehat{q}_{c}^{(n_j)}, j = 1, \ldots, J$, where $0 < n_1 < \ldots < n_J = n$, and define a measure of variability by
\begin{equation}\label{eq.Sc}
S_c = \frac 1J \sum_{j=1}^J \left( \widehat{q}_{c}^{(n_j)} - \frac 1J \sum_{j=1}^J \widehat{q}_{c}^{(n_j)} \right)^2.
\end{equation}
Notice that for $c$ close to zero, due to underpenalization issue, one always obtains $\widehat{q}_{c}^{(n)} = q_{\max}$, so that this yields the first ``stability interval" of the map $c \mapsto S_c$, where $S_c = 0$ for any $c$ in this interval. On the contrary, for a large $c$, overpenalization leads to a stability interval of the map $c \mapsto S_c$, where $S_c = 0$ and $\widehat{q}_{c}^{(n)} = 0$ for any $c$ in this interval. Numerical studies in \cite{HallinLiska2007} suggest choosing $c$ and the corresponding $\widehat{q}_{c}^{(n)}$ that belongs to the second stability interval of the map $c \mapsto S_c$.  

\begin{algorithm}[H] \label{Algm2}
\SetAlgoLined
\KwIn{data $\{x_{\ell\vsbf},\ \ell =1,\ldots, n, \vsbf=(s_1\ s_2\ t)^\top,  s_1=1,\ldots, S_1,  s_2=1,\ldots, S_2, t=1,\ldots, T\}$;
upper bound $q_{\max}$;\linebreak
a sequence $n_1, \ldots, n_J$ integers for subsample dimensions; \linebreak
a sequence values $c_1, \ldots, c_L$ of reals; \linebreak
kernel functions $K_1(\cdot)$, $K_2(\cdot)$, and $K_3(\cdot)$;\linebreak 
bandwidths integers $B_{S_1}, B_{S_2}$, and $B_T$; \linebreak 
 penalty $p(n, S_1,S_2,T)$.}
\KwOut{$\widehat{q}_{\widehat{c}}^{(n)}$.}

Choose a random permutation of the $n$ cross-sectional items.

\For{$\ell \gets 1$ to $L$}{
\For{$j \gets 1$ to $J$}{
Compute $\widehat{\boldsymbol \Sigma}_{n_j}^x(\thbf_{\hbf})$ as in \eqref{hatSimgax}, with $\thbf_{\hbf}$ as in Remark \ref{rem:freq}.

Compute the $q_{\max}$ largest eigenvalues of $\widehat{\boldsymbol \Sigma}_{n_j}^x(\thbf)$. 

Compute the information criterion $\widehat{\rm IC}^{(n_j)}(k)$ as in \eqref{eq.hatICtext}, with the penalty $k c_{\ell} p(n_j, , S_1,S_2,T)$.

Obtain $\widehat{q}_{c_{\ell}}^{(n_j)}$.
}

Compute $S_{c_{\ell}}$ as in \eqref{eq.Sc}.
}

Plot $c \mapsto S_c$ and choose a $\widehat{c}$ that belongs to the second stability interval of the plot.

\KwRet {$\widehat{q}_{\widehat{c}}^{(n)}$.}

\caption{Algorithm for selection of $q$}
\label{algorithm.Estq}
\end{algorithm}

%
%
%
%
%
%
%


In Figure~\ref{selectq3_AR_MA}, we plot $c \mapsto 4S_c$ (in blue) and $c \mapsto \widehat{q}_{{c}}^{(n)}$ (in red)  for  Model (a) in \eqref{modelAR} and Model (b) in \eqref{modelMA}. The first stability interval, where $S_c=0$ and $c$ is close to 0, corresponds to $\widehat q_c^{(n)} = q_{\max}$. For $c$ in the second stability interval, we have $\widehat{q}_{\widehat{c}}^{(n)} = q$ for both models, as expected. These plots and the results in Table \ref{Tab.selectq_AR_MA} are obtained with $n = 100$, $(S_1, S_2, T) = (25, 25, 25)$, $q=3$, $n_j = n - 5j, j = 1, 2, \ldots, 16$, $c_{\ell} = \ell/2000, \ell = 0, 1, \ldots, 6000$, $q_{\max} = 10$ 
 and penalty 
$$
p(n, S_1,S_2,T) = (n^{-1} + B_{S_1}^{-\vartheta_1} + B_{S_2}^{-\vartheta_2} + B_T^{-\vartheta_3} +  V_{\bar{\vsbf}}^{-1}) \log [\min(n, B_{S_1}^{\vartheta_1}, B_{S_2}^{\vartheta_2}, B_T^{\vartheta_3}, V_{\bar{\vsbf}})],
$$
where
$$
V_{\bar{\vsbf}} = \frac{(S_1 S_2 T)^{1/2}}{(B_{S_1} B_{S_2} B_T)^{1/2} \log(B_{S_1}) \log(B_{S_2}) \log(B_{T})}.
$$

 \begin{figure}[h]
 \begin{center}
 \includegraphics[scale=0.3]{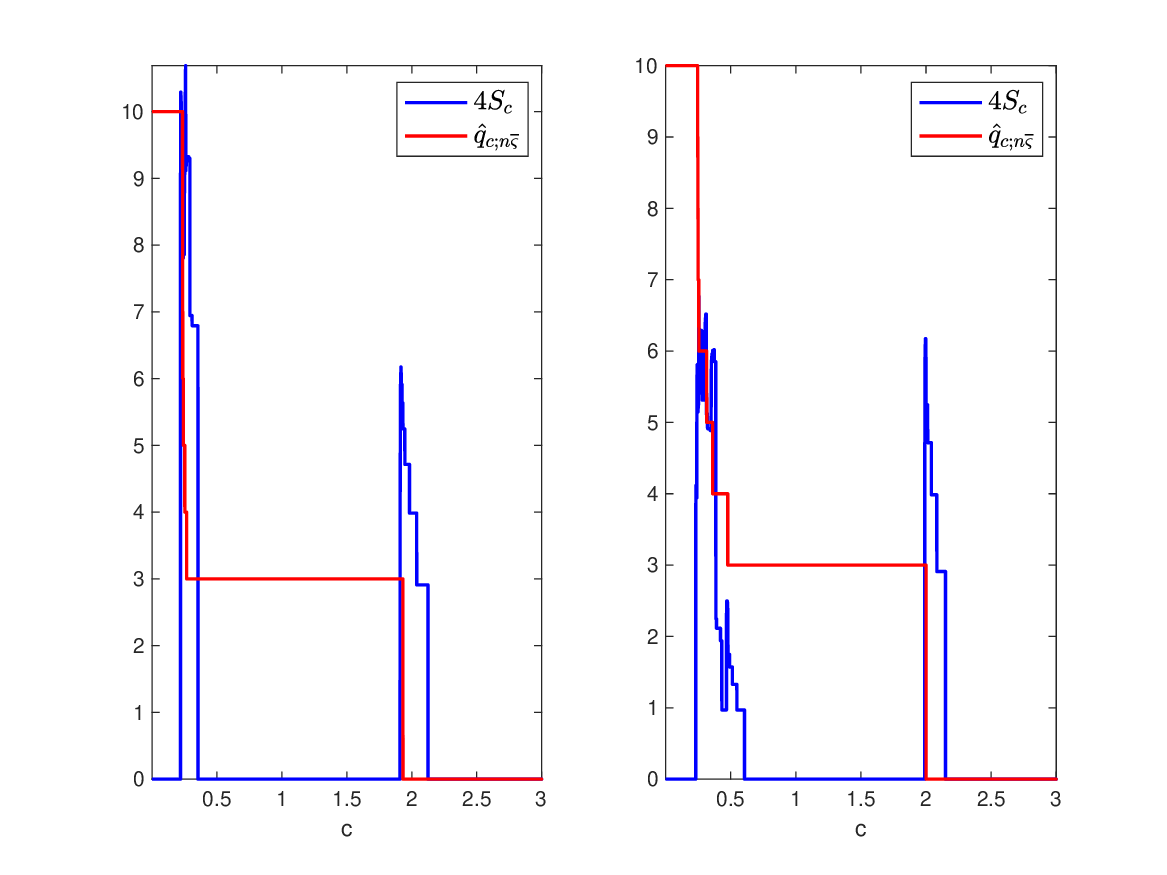} 
  \caption{Plots of $c \mapsto 4S_c$ (blue) and $c \mapsto \widehat{q}_{{c}}^{(n)}$ (red) for model \eqref{modelMA} (left) and model \eqref{modelAR} (right), with $n=100$, $(S_1,S_2,T)=(25,25,25)$, and $q=3$ }\label{selectq3_AR_MA}
 \end{center}
 \end{figure}

\clearpage
\section{Real data example} \label{Sec: real}

\subsection{Data preparation}

We apply the proposed method to model a resting-state cerebral functional magnetic resonance imaging (rs-fMRI) dataset obtained from the Alzhimer’s Disease Neuroimaging Initiative (ADNI) research project (\url{https://adni.loni.usc.edu}). The primary goal of ANDI is to 
measure the progression of mild cognitive impairment and early Alzheimer’s disease. 
The subjects in our rs-fMRI dataset consists of 139 aged 55-90 years old from 59 research centers in the U.S. and Canada. Their first brain scans after enrolled in ANDI are examined in our dataset. Out of the 139 subjects, we have 41 participants that are cognitively normal (CN), and 98 participants with different conditions or diseases, which can be further divided into several sub-groups: 32 Alzheimer’s disease (AD) patients, 40 mild cognitive impairment (MCI) patients
and 26 significant memory concern (SMC) patients. 

Following the protocol pre-processing steps as described in Appendix \ref{subsec:preprocess}, we obtain numerical brain activity measurements for each subject at 130 time points and at 116 spatial locations, which are corresponding to the Anatomical Automatic Labeling brain atlases template \citep{T-M02_Auto}. 
These locations in human brain are characterised by irregular coordinates. To apply our methodology we need to transfer the data into a regular spatial lattice. This can be performed resorting on optimal transportation (OT) theory, which provides a transportation map from irregular data in $\mathbb{R}^d$ to regular data in   $\mathbb{R}^{d^\prime}$, with $d,d^\prime \geq 1$. Here we simply say that thanks to the use of OT, we obtain an optimal coupling that allows to map irregular spatial 3D data  to regular data on a 2D lattice,  while preserving the spatial structure to the maximum extent. Thanks this procedure,  we obtain a  $10 \times 11$ grid over a 2D spatial network. Moreover, as it is customary in the statistical analysis of fMRI data, for each subject we smooth the time available series by applying moving average filters---with window length $(2, 2, 4)$. 

\subsection{Spatio-temporal correlations}
We first analyse the strength of spatial-temporal dependence in each group of patients. To do so, 
consider the estimator 
$\widehat{\bm\Gamma}_n^x(\boldsymbol h)$ and
let $\widehat{\gamma}_{n, ij}^x(\boldsymbol h)$, $i, j=1,\ldots, n,$ denote the $(i,j)$-th entry of $\widehat{\bm\Gamma}_n^x(\boldsymbol h)$, and denote by
$\widehat{\rho}_{ni}^x(\boldsymbol h) = \widehat{\gamma}_{n, ii}^x(\boldsymbol h)/\widehat{\gamma}_{n, ii}^x(\boldsymbol 0)$ the sample spatio-temporal autocorrelation for the rf $x_{i}$.

To investigate the spatial correlations, Figure~\ref{autoCorrfMRI_averagePatient} shows the heatmap of $n^{-1} \sum_{i = 1}^n \widehat{\rho}_{ni}^x(\boldsymbol h)$ (that is, the average, over patients, of $\widehat{\rho}_{ni}^x(\boldsymbol h)$) for each subgroup, where we set $ h_1, h_2=0,\ldots, 4$ and $h_3 = 0$. Clearly, for all subgroups, strong spatial dependence exists along both spatial directions. For the temporal correlation,  in Figure~\ref{autoCorrfMRI_TdependPatient}, we display the heatmap of $n^{-1} \sum_{i = 1}^n \widehat{\rho}_{ni}^x(\boldsymbol h)$ for each subgroup, where we set $h_1, h_2 = 0$ and $h_3=0,\ldots, 4$. Hence, to model these spatio-temporal dynamics,  we can apply a GSTFM.

 \begin{figure}[h]
  \centering
 \includegraphics[scale = 0.5]{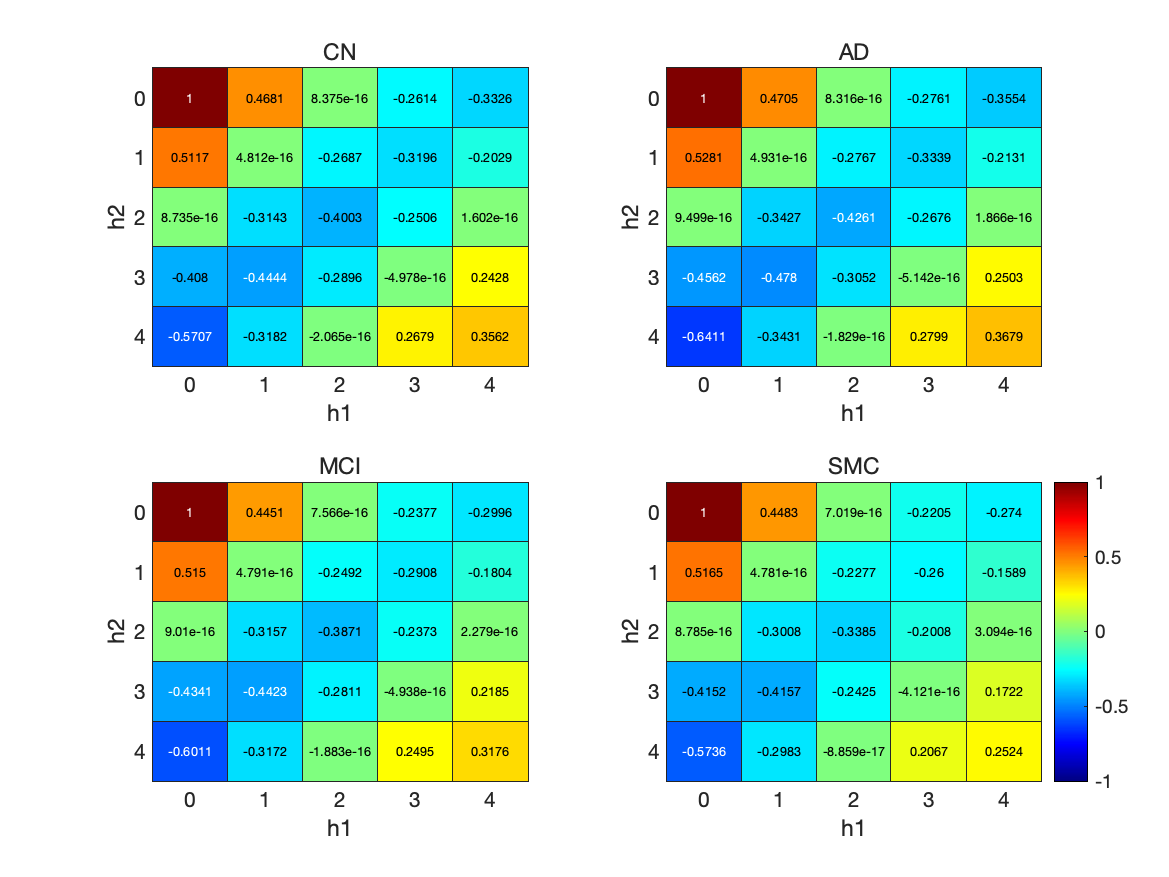} 
 \caption{Spatial correlations. Heatmaps of the $n^{-1} \sum_{i = 1}^n \widehat{\rho}_{ni}^x(\boldsymbol h)$ for the subgroups (CN, AD, MCI, and SMC) ($h_1,h_2=0,\ldots, 4$ and $h_3 = 0$).} \label{autoCorrfMRI_averagePatient}
 \end{figure}
 
 \begin{figure}[h] 
 \centering
 \includegraphics[scale = 0.3]{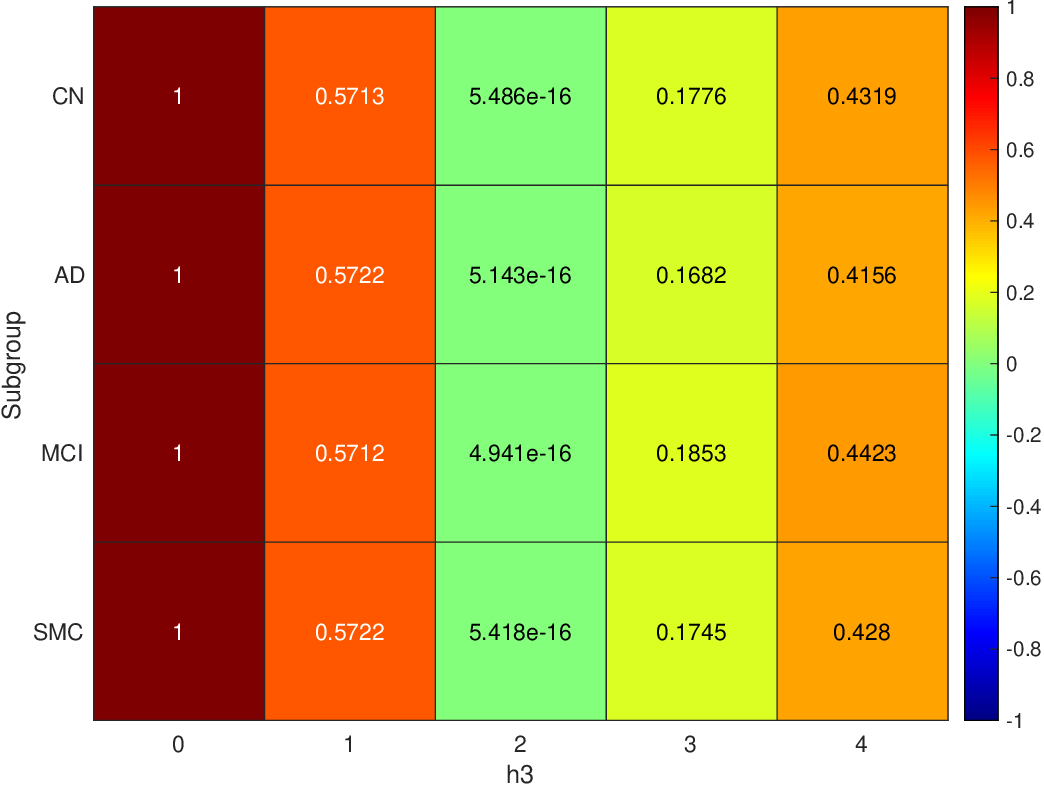} 
 \caption{Temporal correlations. Heatmaps of the $n^{-1} \sum_{i = 1}^n \widehat{\rho}_{ni}^x(\boldsymbol h)$ for the subgroups (CN, AD, MCI, and SMC) ($h_1, h_2 = 0$ and $h_3=0,\ldots, 4$).} \label{autoCorrfMRI_TdependPatient}
 \end{figure} 
 

\subsection{Number of factors} \label{App_Selq}
For all the subgroups, we plot the largest 20 eigenvalues in Figure~\ref{EigenfMRI_piTheta} for frequency $(\pi, \pi, \pi)$ (top panel) and for the averaged values over $8\times 8 \times 8$ frequencies on the regular grid of $[-\pi, \pi]^3$ (bottom panel). A rapid inspection of both figures reveals difference between the CN and the other subgroups: the gaps between the first and second eigenvalues of the CN are significantly larger than the other subgroups. 

 \begin{figure}[h]
 \centering
 \includegraphics[scale = 0.6]{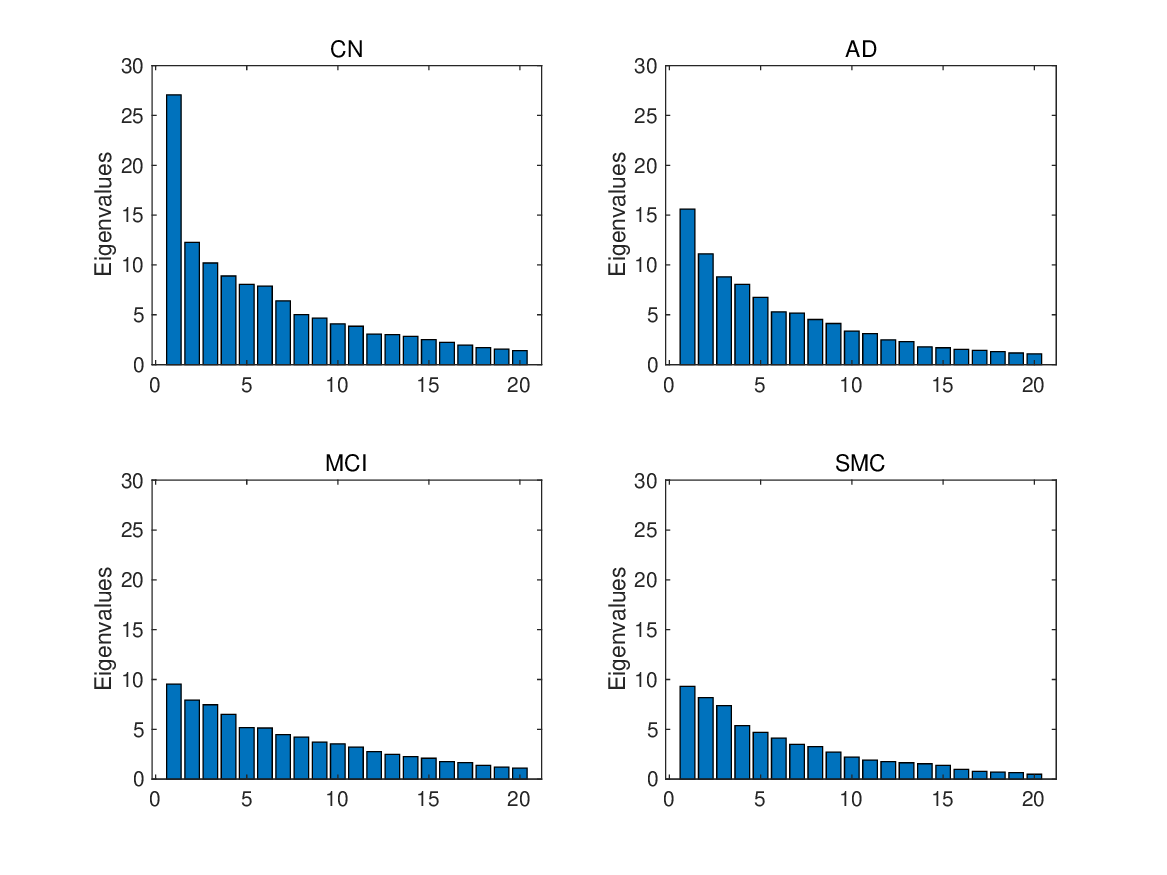}
 \includegraphics[scale = 0.6]{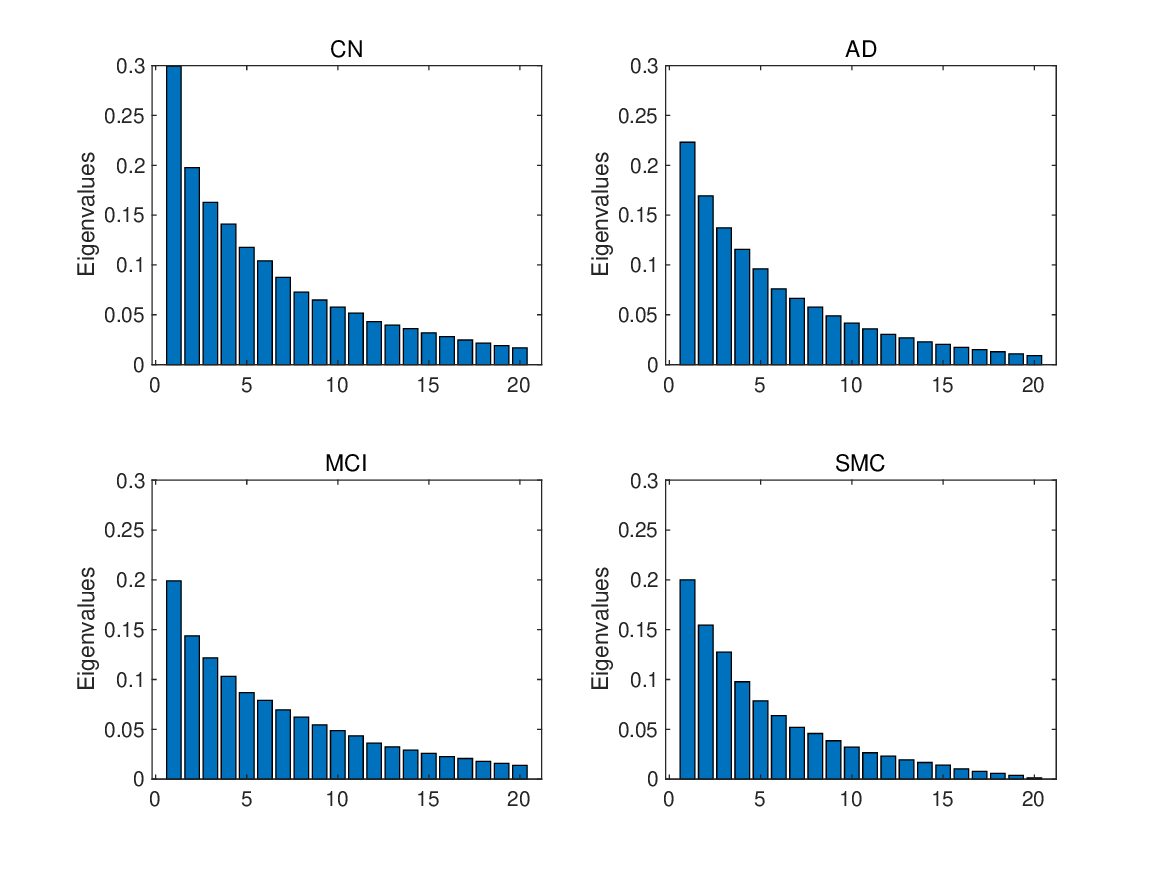} 
 \caption{Plot of the largest 20 dynamic spatio-temporal eigenvalues  at frequency $(\pi, \pi, \pi)$ (top) and averaged over $8\times 8 \times 8$ frequencies on the regular grid of $[-\pi, \pi]^3$ (bottom), for the subgroups (CN, AD, MCI, and SMC).} \label{EigenfMRI_piTheta}
 \end{figure}


 \begin{figure}[h]
 \begin{center}
 \includegraphics[scale = 0.6]{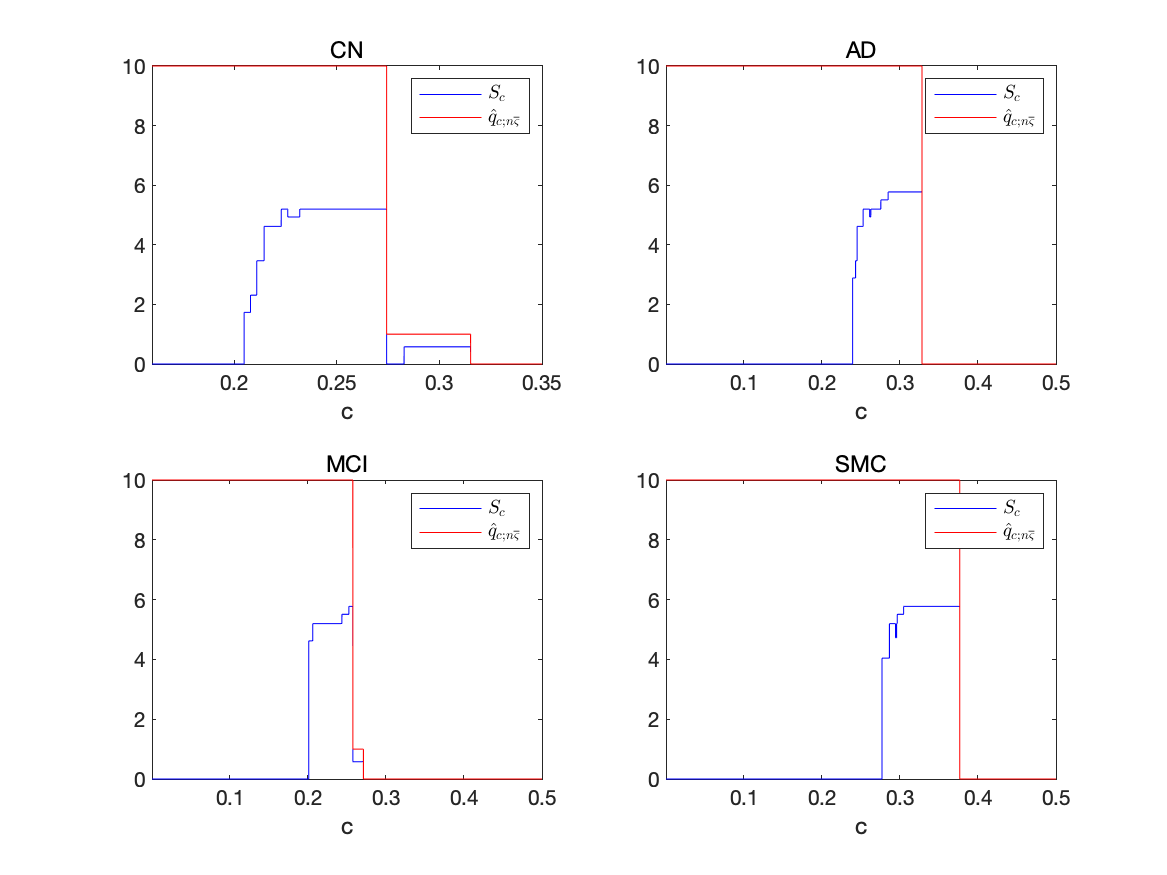} 
  \caption{Plots of $c \mapsto S_c$ (blue) and $c \mapsto \widehat{q}_{\widehat{c}}^{(n)}$ (red) for the subgroups (CN, AD, MCI, and SMC).} \label{fMRI_smooth_selectq}
 \end{center}
 \end{figure}

 To gain further understanding, we then apply Algorithm~\ref{algorithm.Estq} to select the number of latent factors $q$. In Figure~\ref{fMRI_smooth_selectq} we display the plots of $c \mapsto S_c$ (blue line) and $c \mapsto \widehat{q}_{{c}}^{(n)}$ (red line). The plots suggest that $\widehat{q}^{(n)}_{\wh c} = 1$ for the CN group, while $\widehat{q}^{(n)}_{\wh c} = 0$ for the other subgroups. 
Computing all of the averaged dynamic eigenvalues of the CN subgroup, we note that the explained variance  of the common component is $20\%$. 


\subsection{Estimation of the common component} 

Next, we analyse the common component ${\chi}_{\ell\vsbf}$, $\ell = 1, \ldots, n,$ of the CN subgroup based on the estimator $\widehat{\chi}^{(n)}_{\ell\vsbf}$. In order to construct a quantity that represents strength, over time, of the common component for all subjects in the CN subgroup, we define a {temporal coincident indicator} based on the weighted average of the estimated common component (averaged over space in advance). More specifically, the temporal coincident indicator is defined as 
$${\rm CI}_{t}^{(n)} =  \sum_{\ell =1}^n  w_{\ell} \bar{\chi}^{(n)}_{\ell t}, \quad t = 1, \ldots, T,$$ 
where 
$\bar{\chi}^{(n)}_{\ell t} = S_1^{-1} S_2^{-1} \sum_{s_1 =1}^{S_1} \sum_{s_1 =1}^{S_2} \widehat{\chi}^{(n)}_{\ell\vsbf}$
is the average of $\widehat{\chi}^{(n)}_{\ell\vsbf}$ taken over space, and the weight
$$w_{\ell} = \frac{\sum_{t=1}^T (\bar{\chi}^{(n)}_{\ell t})^2}{\sum_{\ell=1}^n \sum_{t=1}^T (\bar{\chi}^{(n)}_{\ell t})^2}$$
is defined according to the level (in time domain) of the common component of subject $\ell$. Heuristically, by this construction of weights, the subjects that are the main drivers of the common factor can gain more weights in ${\rm CI}_{t}^{(n)}$. 

Figure~\ref{Fig.CoIndicator} shows plots of  ${\rm CI}_{t}^{(n)}$, $t=1,\ldots, T$, and of the observed spatio-temporal rf (averaged over space)  for four randomly selected subjects in the CN subgroup. Obvious co-movements of 
the CI and rf can be observed and the CI seems to capture the magnitude of fluctuations of the rf for all subjects. 
To investigate further the behaviour of ${\rm CI}_{t}^{(n)}$ and rf in frequency domain, we plot their periodograms  
in Figure~\ref{Fig.PeriodogramT} (to make a fair comparison, the periodograms, smoothed with Hamming window, of the rf are rescaled through dividing the ratios of their integral over frequency with respect to that of  ${\rm CI}_{t}^{(n)}$).   
{In general, we observe that the periodogram of  ${\rm CI}_{t}^{(n)}$ and the subject specific periodogram have similar shapes: they both have several peaks between $0$ and $0.4\pi$, with the former showing a slight phase shift to higher frequencies compared to the latter.} 

 \begin{figure}[h]
 \begin{center}
 \includegraphics[scale = 0.6]{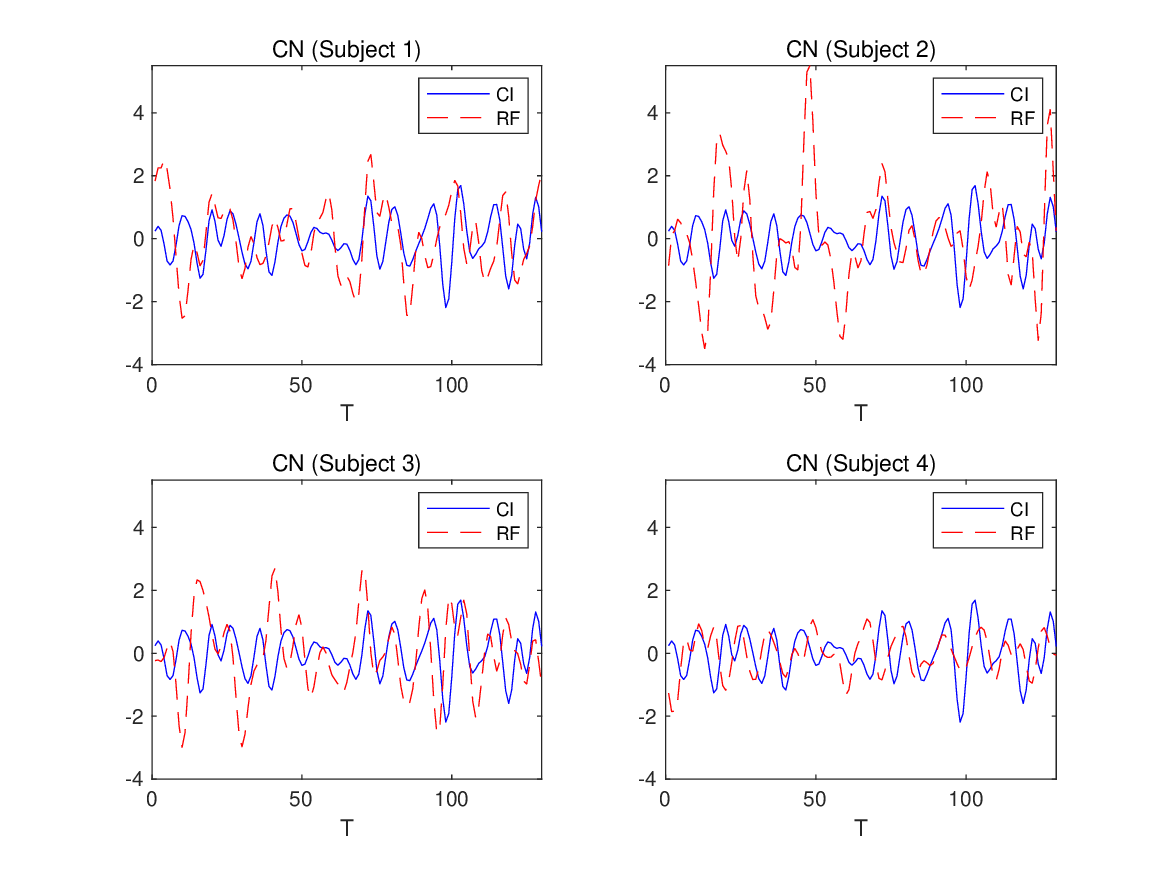}
 \caption{Plots of the temporal coincident indicator $\rm{CI}_t^{(n)}$ and rf  (averaged over space) for four randomly selected subjects in the CN subgroup.} \label{Fig.CoIndicator}
 \end{center}
 \end{figure}

\begin{figure}[h]
 \begin{center}
 \includegraphics[scale = 0.6]{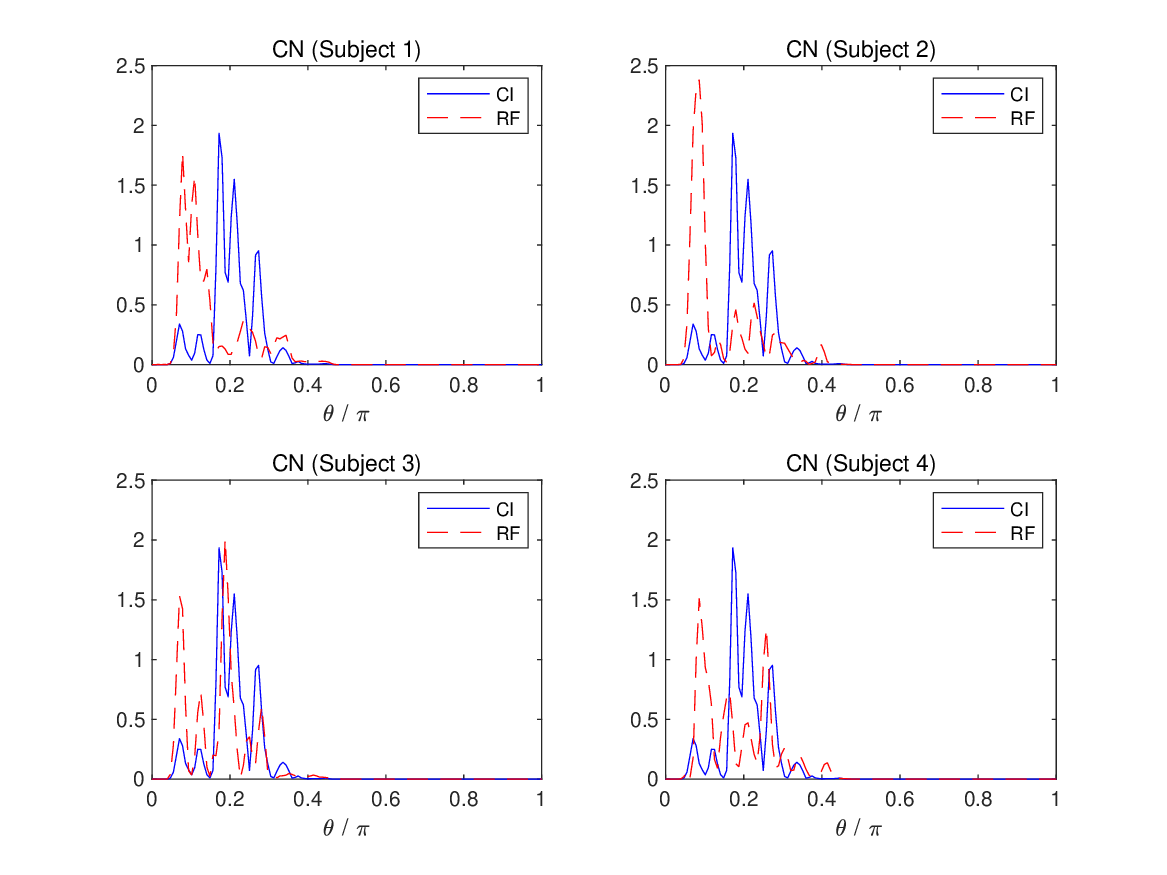}
 \caption{Plots of the periodograms of the temporal coincident indicator and rf  (averaged over space) over frequency $[0, \pi]$ for four randomly selected subjects in the CN subgroup.} \label{Fig.PeriodogramT}
 \end{center}
 \end{figure}

In a similar fashion, we can analyse the common component  in space domain by defining a { spatial coincident indicator} based on the weighted average of the estimated common component (averaged over time). More precisely, the spatial coincident indicator  in space domain is defined as 
$${\rm CI}_{\s}^{(n)} =  \sum_{\ell =1}^n  \widetilde{w}_{\ell} \widetilde{\chi}^{(n)}_{\ell \s}, \quad \s=(s_1\ s_2)^\top, \quad  s_1 = 1, \ldots, S_1, \  s_2 = 1, \ldots, S_2,$$
where 
$\widetilde{\chi}^{(n)}_{\ell \s} = T^{-1} \sum_{t =1}^{T}  \widehat{\chi}^{(n)}_{\ell\vsbf}$
is the average of $\widehat{\chi}^{(n)}_{\ell\vsbf}$ taken over time, and the weight
$$\widetilde{w}_{\ell} = \frac{\sum_{s_1 =1}^{S_1} \sum_{s_1 =1}^{S_2} (\widetilde{\chi}^{(n)}_{\ell t})^2}{\sum_{\ell=1}^n \sum_{s_1 =1}^{S_1} \sum_{s_1 =1}^{S_2} (\widetilde{\chi}^{(n)}_{\ell t})^2}$$
is defined according to the level (in space domain) of the common component of subject $\ell$. 
%
%
%
%
To elaborate more on the ${\rm CI}_{\s}^{(n)}$, we now turn to the analysis of its spatial periodogram. Following \cite[Ch.~4]{MR17}, the spatial periodogram of rf $\{y_{s_1, s_2}; 1\leq s_1 \leq S_1, 1\leq s_2 \leq S_2\}$ is
$$I_{S_1, S_2}(\h) = \frac{1}{(2\pi)^2 S_1S_2} \left\Vert \sum_{s_1 =1}^{S_1} \sum_{s_1 =1}^{S_2} y_{s_1, s_2} e^{-i 2\pi \langle{\bf s} \cdot \h\rangle} \right\Vert^2, \quad \h \in \mathbb{R}^2.$$
  \begin{figure}[h]
 \begin{center}
 \includegraphics[scale = 0.6]{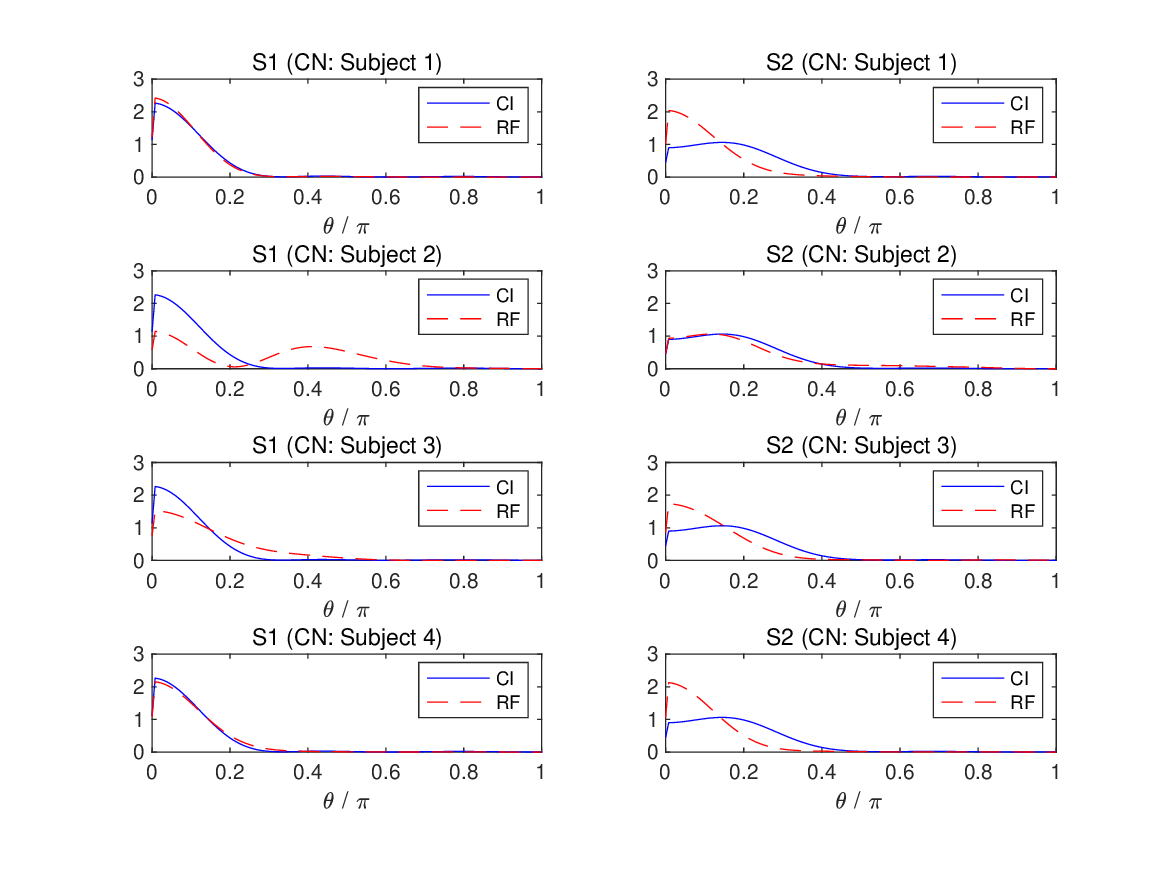}
 \caption{Plots of the periodograms of coincident indicator and rf  (averaged over time) for four randomly selected subjects in the CN subgroup. The spatial CI and rf are averaged over one spatial direction first. 
 } \label{Fig.PeriodogramS1S2}
 \end{center}
 \end{figure}
To evaluate the spatial periodograms over two marginal spatial directions, we process ${\rm CI}_{\s}^{(n)}$ and the rf  (averaged over time) in the following way: (i) take the average over one spatial direction of the data and (ii) compute the periodogram  as in the time domain for the other spatial direction. The periodograms of  ${\rm CI}_{\s}^{(n)}$ and rf along $s_1$ and $s_2$ directions are available in the left and right panels of Figure~\ref{Fig.PeriodogramS1S2}, respectively (coherently with the previous analysis for the time component,  the periodograms of the rf are rescaled). %
%
%
%
%
%
%
  We notice that for the second subject, the periodograms,  along $s_1$ direction, of the rf have two peaks, where the first one has the same frequency as the peak of the periodograms of ${\rm CI}_{\s}^{(n)}$. For all the other cases, the periodograms of  ${\rm CI}_{\s}^{(n)}$ and rf have a single peak that occurs at almost the same frequency. {In general, similarly to what we concluded looking at the time component only, we observe that the periodogram of  ${\rm CI}_{S}^{(n)}$ and the subject specific spatial periodogram have similar shapes.}

%
%
%

\color{black}

\clearpage
\subsection{fMRI data pre-processing} \label{subsec:preprocess}  
In the ADNI research, each subject's brain was scanned by a 3.0T Philips MR Ingenia Elition scanner and produced a 3D map of the voxels measuring the brain activity. The rs-fMRI ouput is firstly preprocessed using the software Data Processing Assistant for Resting-State fMRI (DPARSF) based on Statistical Parametric Mapping 12 (SPM12) on MATLAB. Next, for each participant we discard the first 10 time points to avoid the instability of the initial MRI signals. Moreover, we introduced an extra correction for the acquisition time delay and head motion in the images. The inclusion criteria is below 3 mm translation and  below 3$^\circ$ rotation for the head movements during the fMRI scan. After these corrections, the images were normalized to the standard Montreal Neurological Institute (MNI) template at a 3 mm $\times$ 3 mm $\times$ 3 mm resolution. The final resultant data were filtered through a temporal band-pass (0.01–0.1 Hz) to avoid the interferences of low-frequency drift and physiological noises.


\section{Additional simulation results}\label{App.sim}

We show that the proposed estimator $\widehat{\boldsymbol \Sigma}_n^x(\thbf)$ yields estimated  spatio-temporal dynamic eigenvalues having an eigen-gap under GSTFM (see Theorem~\ref{Th. q-DFS}). To this end, we set   $n = 100$,  $(S_1, S_2, T) = (15, 15, 15)$, and, for each MC run, we estimate the spectral density, as in  \eqref{hatSimgax}, and the spatio-temporal dynamic eigenvalues at 
selected  frequencies.  For the ease-of-computation, we set the frequencies on a $5\times 5\times 5$ equally spaced grid over $[0, \pi]^3$. To see how the eigen-gap changes with the dimension of the rf, we simulate different rf $\x_{m}$, with increasing dimension $m=1,\ldots,n$.  For each frequency over the grid, we treat each estimated eigenvalue as a function of $m$ and obtain $100$  curves (one curve for each MC run). 
To summarize the behavior of these curves,  we consider the largest $q+1$   eigenvalues averaged over the $125$ discrete frequencies and we plot the resulting average curves against $m$. To complete the picture, we repeat this analysis for different number of factors: $q = 2, 3, 4$.

In Figure~\ref{EigenAR_fbplot_q234_freqAve} and Figure \ref{EigenMA_fbplot_q234_freqAve} we display the related  {functional boxplots} (\cite{sun2011functional}). The figures illustrate that, even for small values of $m$, an eigen-gap is clearly detectable: 
 the first $q$ eigenvalues seem to diverge with $m$ almost linearly, while the $(q+1)$-th eigenvalue remains close to zero.  
 

 \begin{figure}[h!]
 \begin{center}
\begin{tabular}{ccc}
 \includegraphics[width=0.3\textwidth, height=0.4\textwidth]{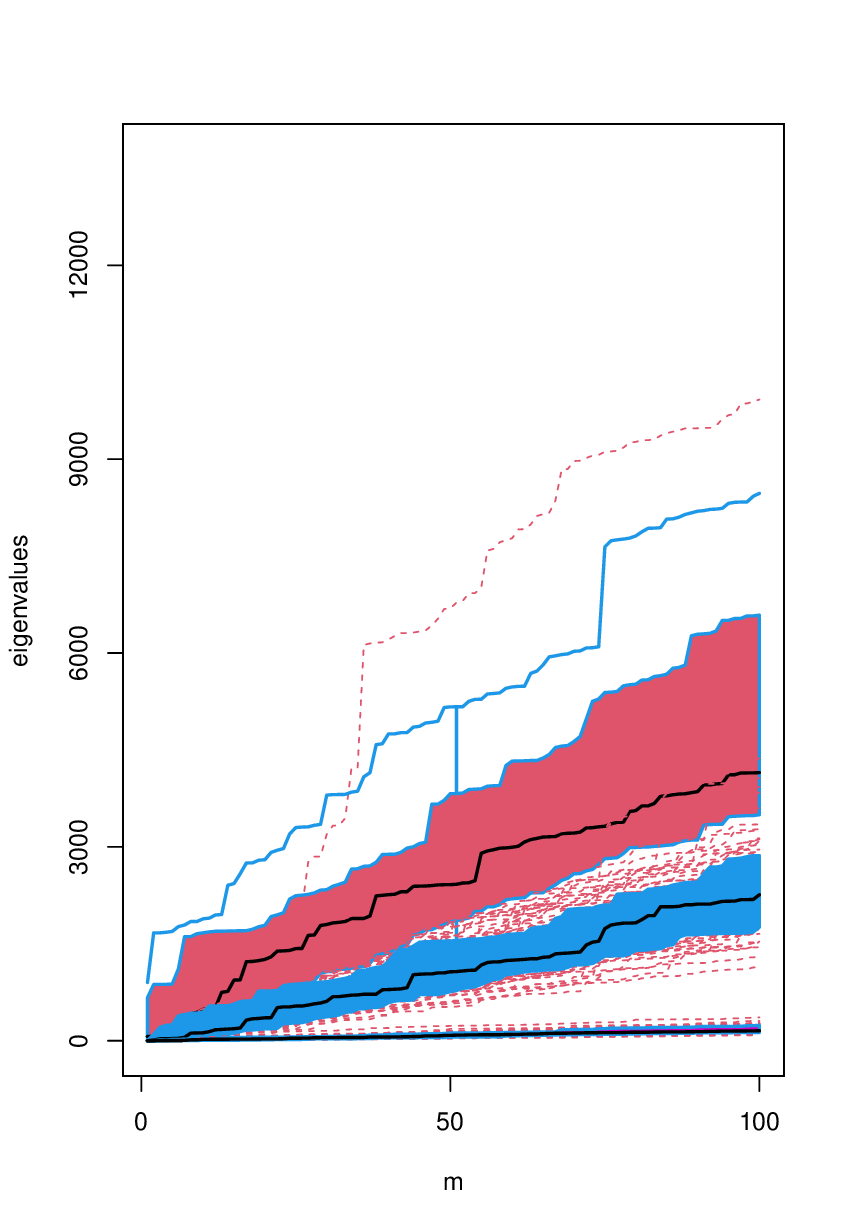}&
  \includegraphics[width=0.3\textwidth, height=0.4\textwidth]{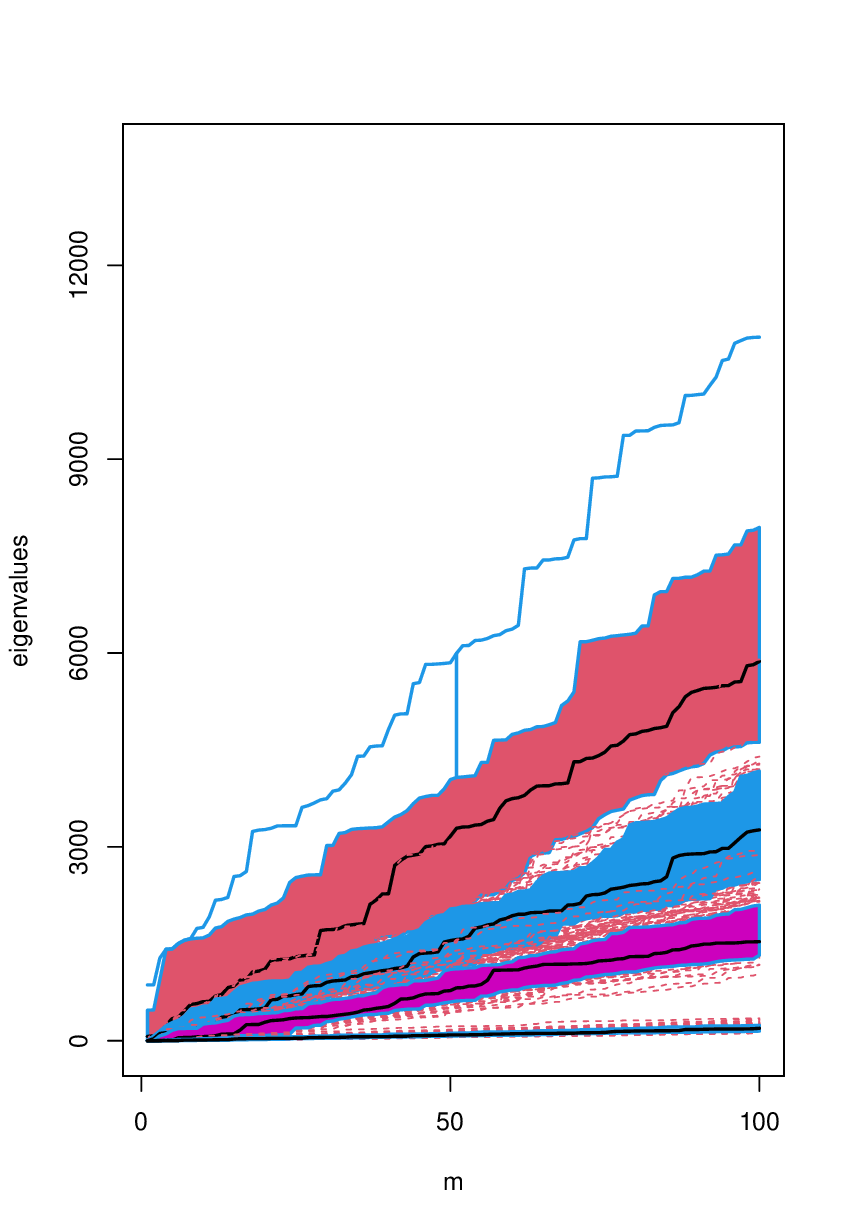}&
   \includegraphics[width=0.3\textwidth, height=0.4\textwidth]{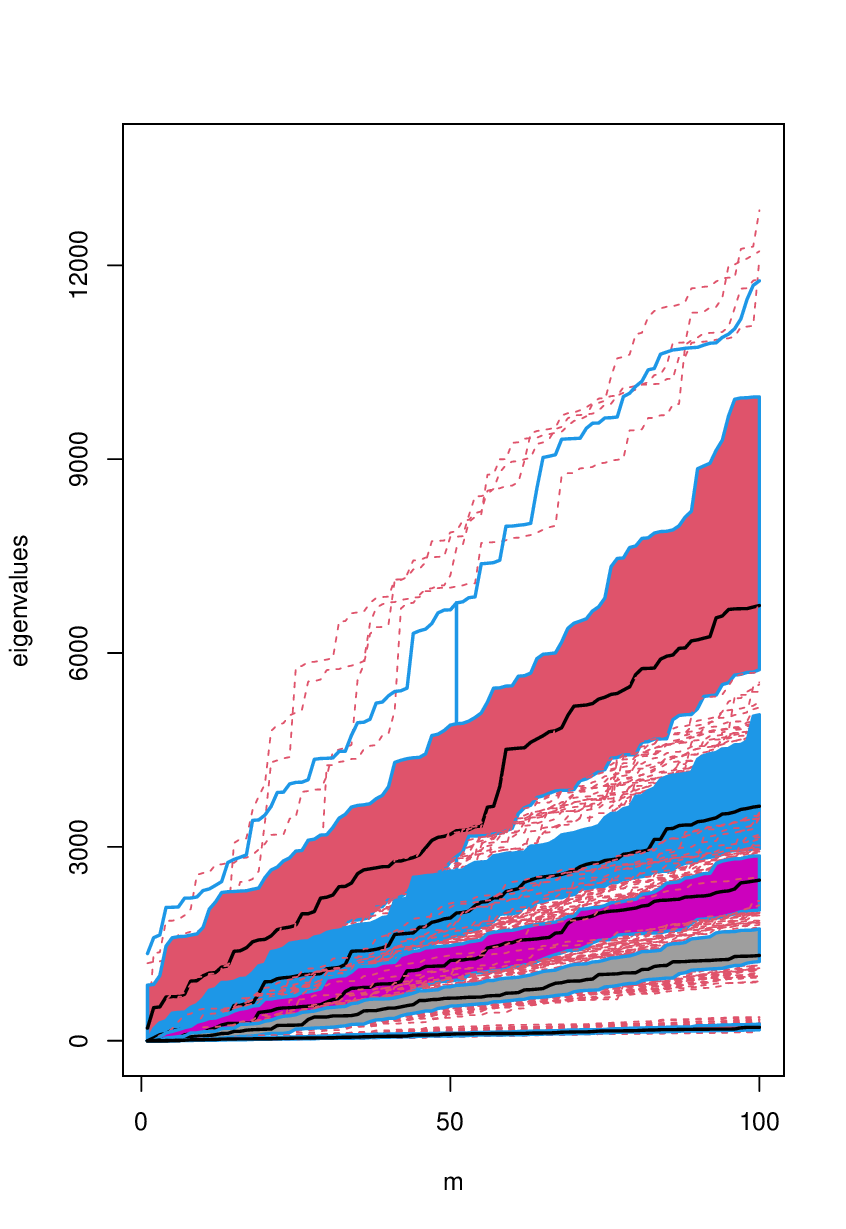}\\
   \end{tabular}
   \caption{Functional boxplots  of the largest $q+1$ spatio-temporal dynamic eigenvalues (averaged over all frequencies) for Model (a) in \eqref{modelAR}, with $n=100$, $(S_1,S_2,T)=(15,15,15)$, and $q = 2$ (left),  $q = 3$ (middle), and $q = 4$ (right). }
   \label{EigenAR_fbplot_q234_freqAve}
%
%
 \begin{tabular}{ccc}
 \includegraphics[width=0.3\textwidth, height=0.4\textwidth]{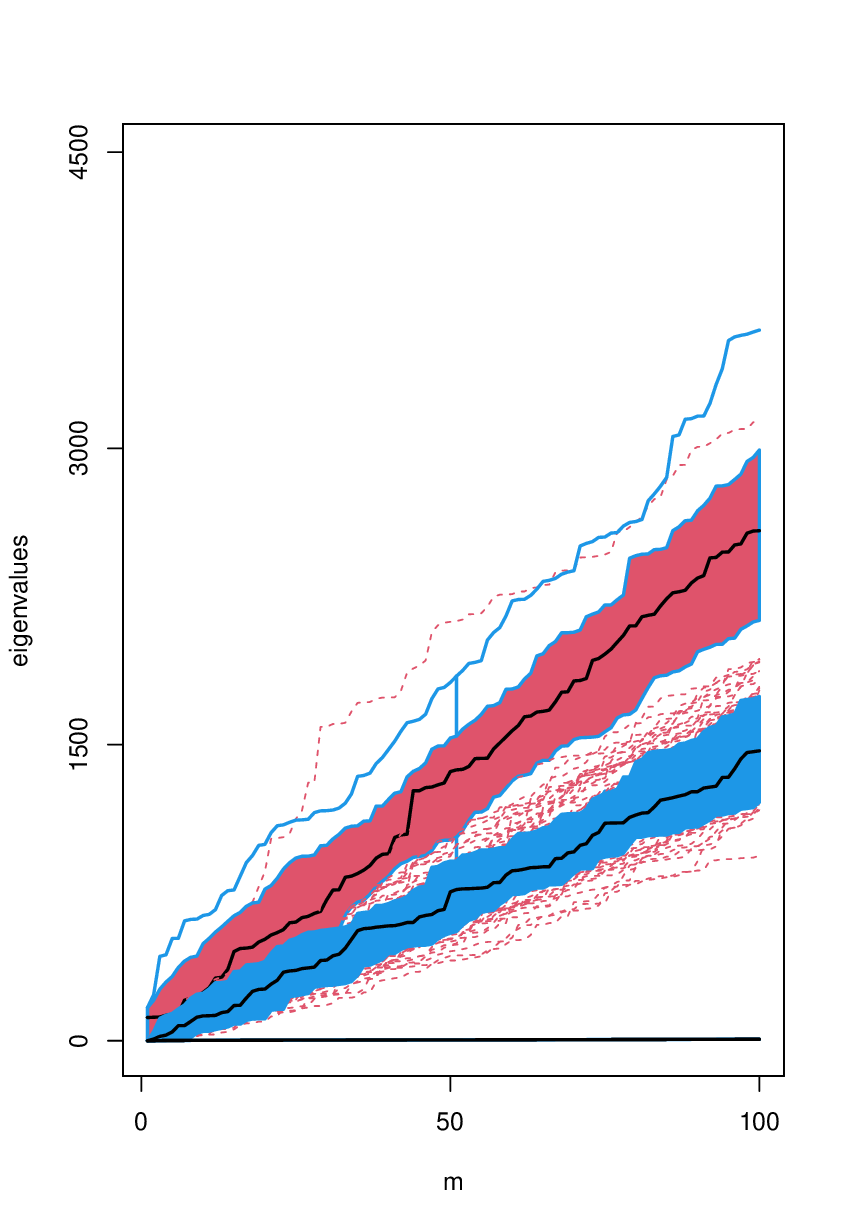}&
  \includegraphics[width=0.3\textwidth, height=0.4\textwidth]{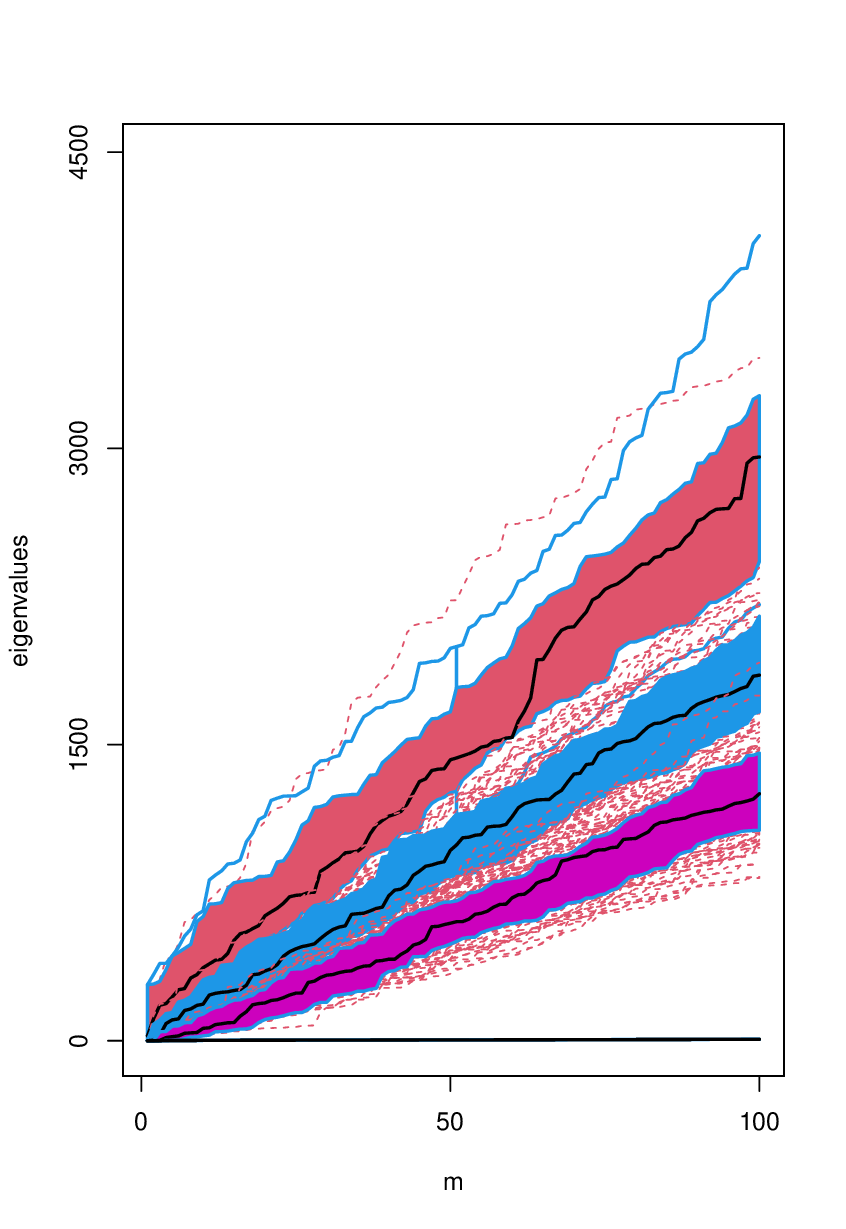}&
   \includegraphics[width=0.3\textwidth, height=0.4\textwidth]{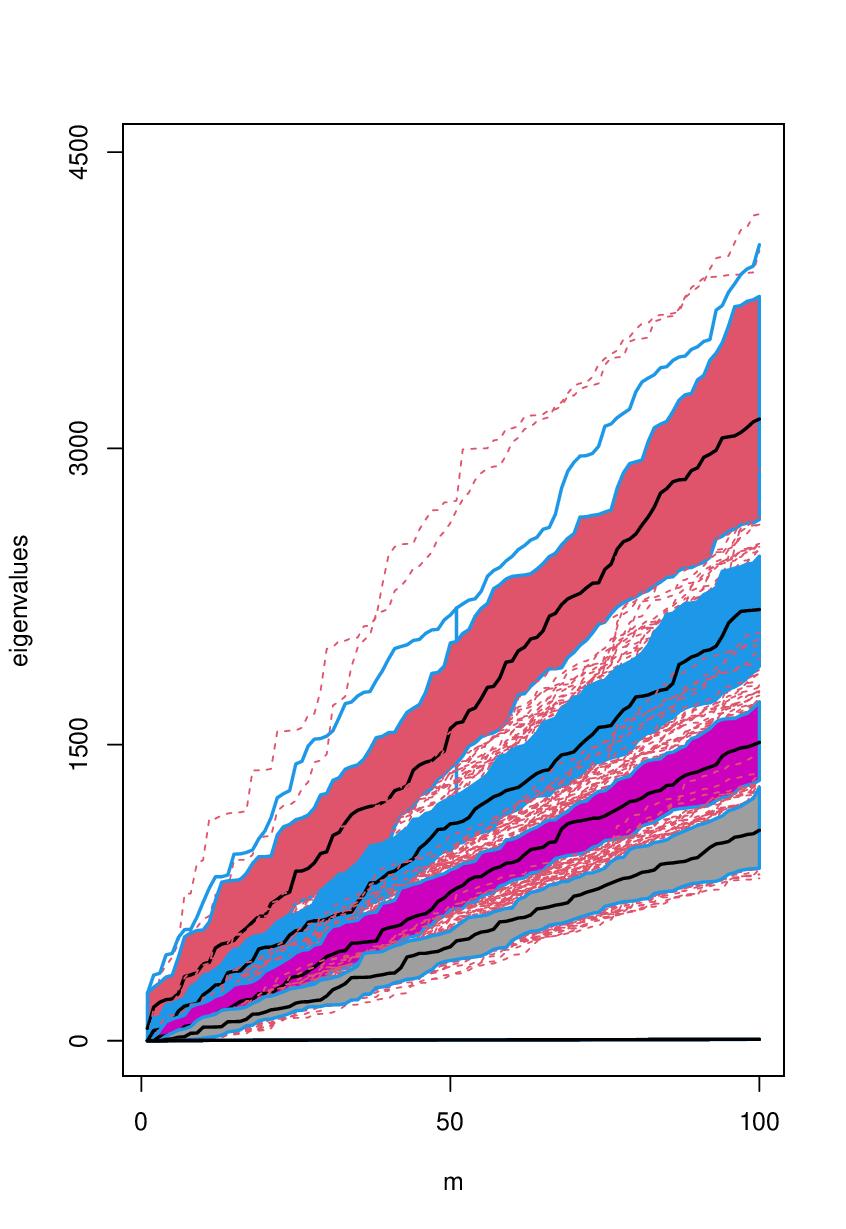}\\
      \end{tabular}
 \caption{Functional boxplots  of the largest $q+1$ spatio-temporal dynamic eigenvalues (averaged over all frequencies) for Model (b) in \eqref{modelMA}, with $n=100$, $(S_1,S_2,T)=(15,15,15)$, and $q = 2$ (left),  $q = 3$ (middle), and $q = 4$ (right).  } \label{EigenMA_fbplot_q234_freqAve} 
 \end{center} 
 \end{figure}


 \begin{figure}[h!]
 \begin{center}
 \begin{tabular}{ccc}
 \includegraphics[width=0.3\textwidth, height=0.4\textwidth]{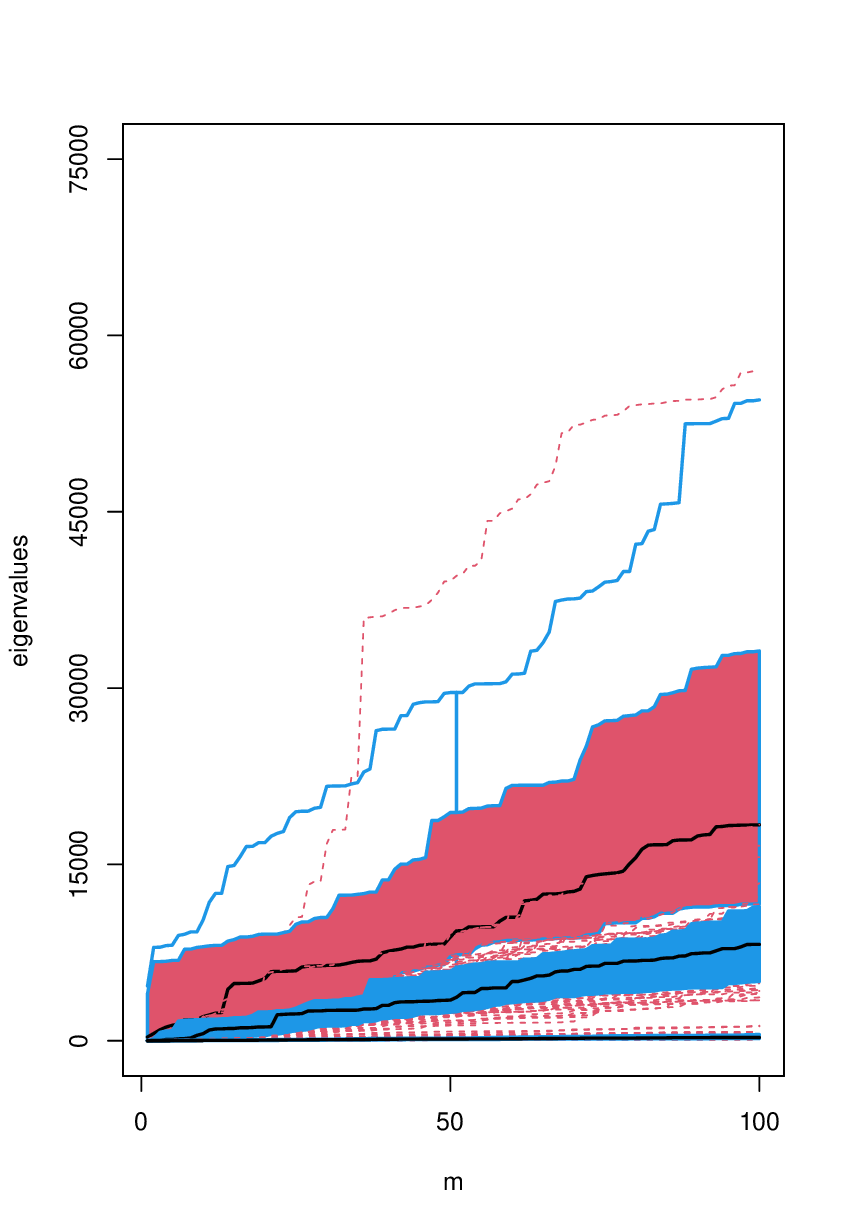}&
  \includegraphics[width=0.3\textwidth, height=0.4\textwidth]{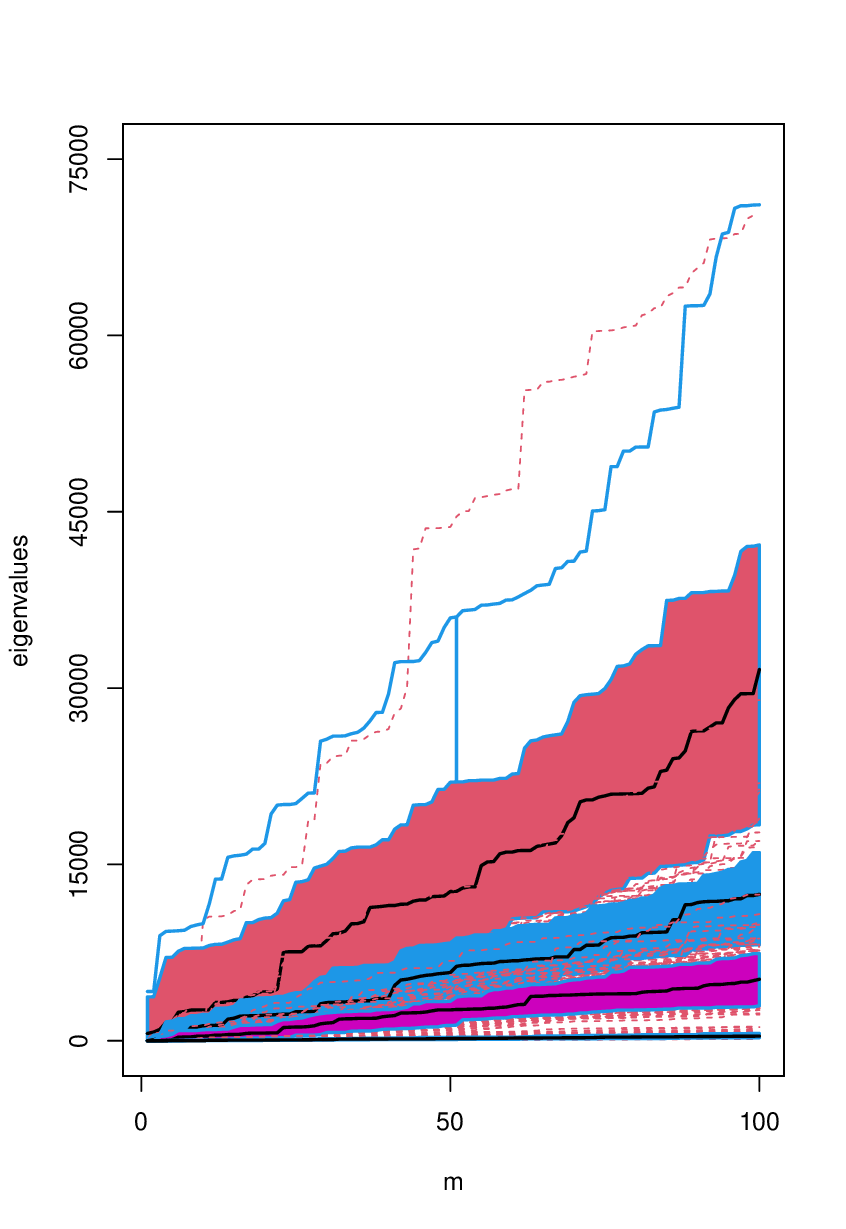}&
   \includegraphics[width=0.3\textwidth, height=0.4\textwidth]{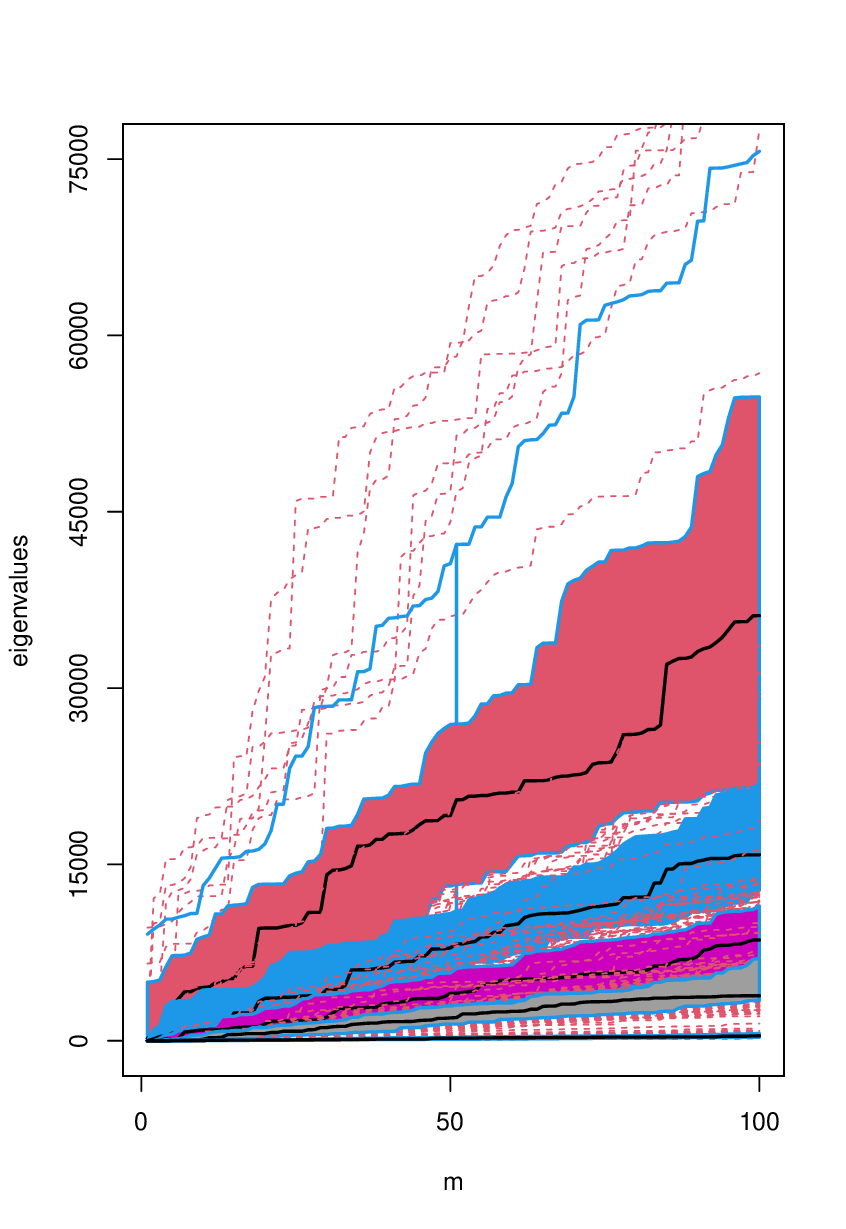}\\
      \end{tabular}
  \caption{Functional boxplots,  over $100$ replications,  of the largest $q+1$ spatio-temporal dynamic eigenvalues at frequency $0$ for Model (a) in \eqref{modelAR}, with $n = 100, (S_1, S_2, T) = (15, 15, 15)$, and $ q = 2$ (left),  $q = 3$ (middle), and $q=4$ (right).} \label{EigenAR_fbplot_q234_freq0} 
%
 \begin{tabular}{ccc}
 \includegraphics[width=0.3\textwidth, height=0.4\textwidth]{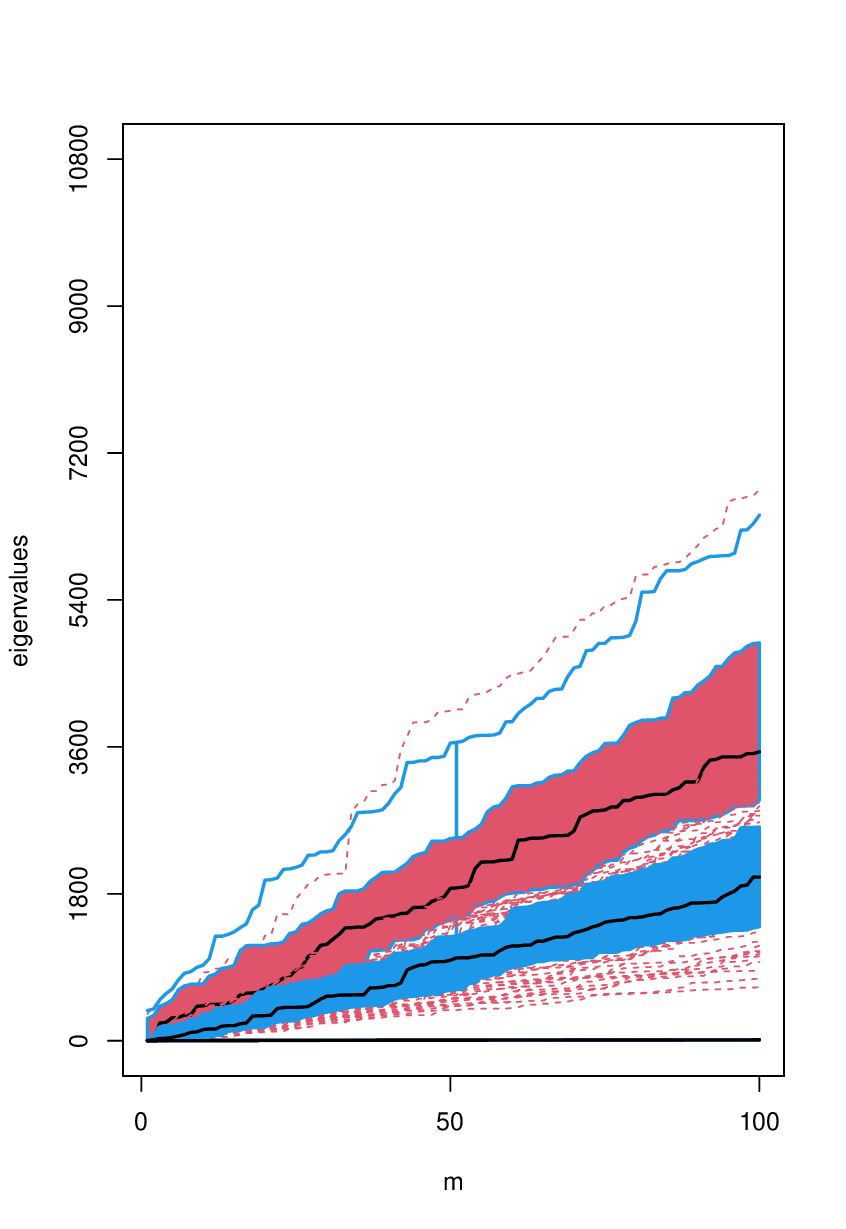}&
  \includegraphics[width=0.3\textwidth, height=0.4\textwidth]{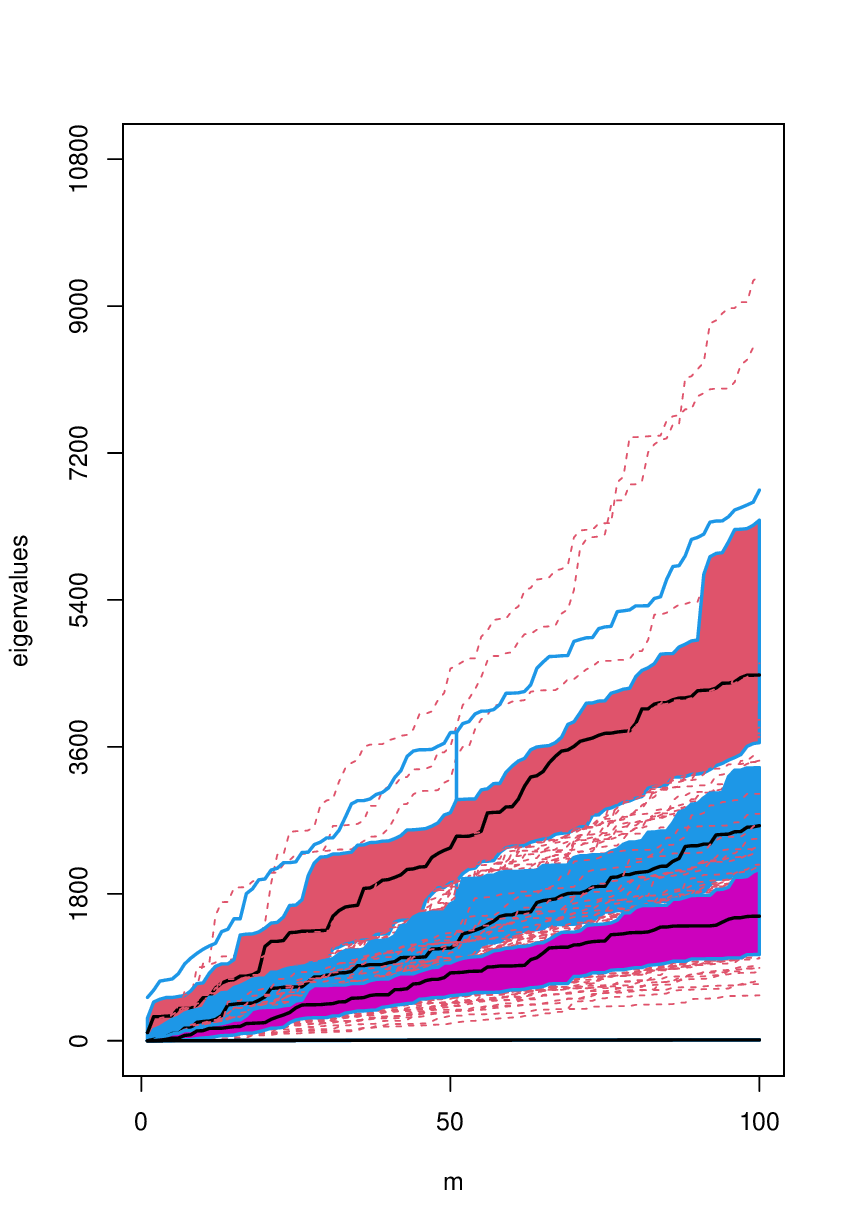}&
   \includegraphics[width=0.3\textwidth, height=0.4\textwidth]{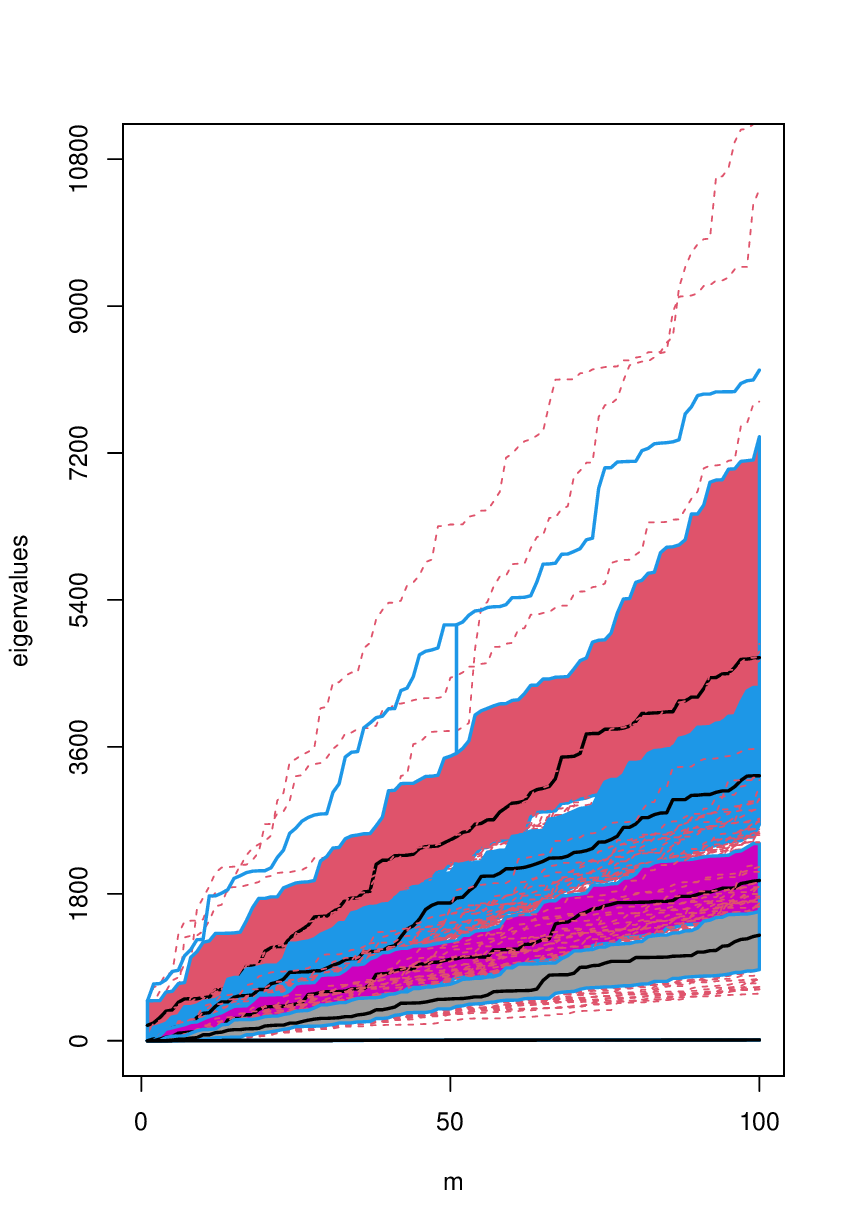}\\
      \end{tabular}
 \caption{Functional boxplots,  over $100$ replications,  of the largest $q+1$ spatio-temporal dynamic eigenvalues at frequency $0$ for Model (b) in  \eqref{modelMA}, with $n = 100, (S_1, S_2, T) = (15, 15, 15)$, and  $q = 2$ (left),  $q = 3$ (middle), and $q=4$ (right).} \label{EigenMA_fbplot_q234_freq0} 
 \end{center}
 \end{figure} 
 Figures~\ref{EigenAR_fbplot_q234_freq0} and \ref{EigenMA_fbplot_q234_freq0} display the functional boxplots of the largest $q+1$  eigenvalues at $0$ frequency for Model (a) in \eqref{modelAR} and Model (b) in \eqref{modelMA}, respectively. The left, middle and right panels in each figure are for $q = 2, 3, 4$, respectively. In each plot, the $50\%$ central regions of different eigenvalues are shown in different colours. The black curves in the central regions represent the sample median functions. The blue curves represent the envelope (i.e., 1.5 times the $50\%$ central region). The red dashed curves are the outliers outside the envelope.

 \begin{figure}[h!]
 \begin{center}
 \begin{tabular}{ccc}
 \includegraphics[width=0.3\textwidth, height=0.4\textwidth]{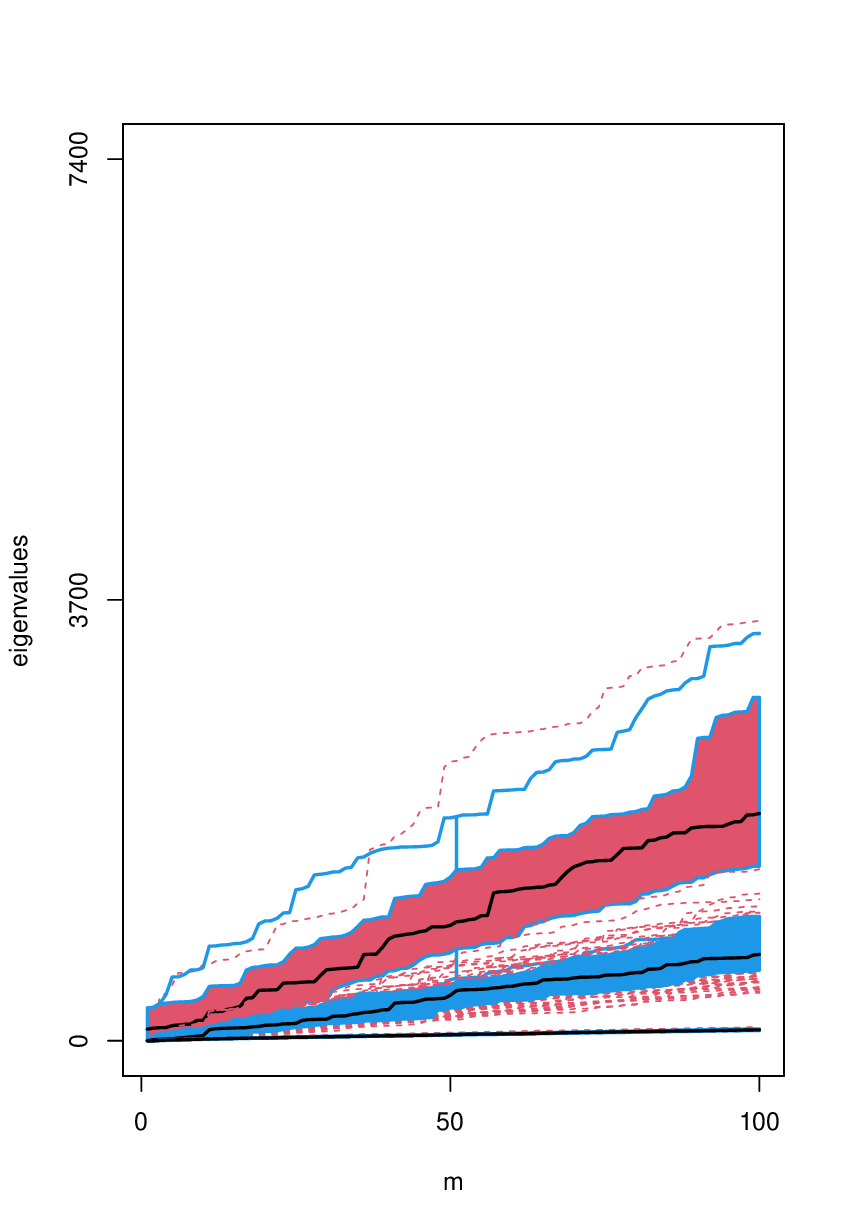} &
\includegraphics[width=0.3\textwidth, height=0.4\textwidth]{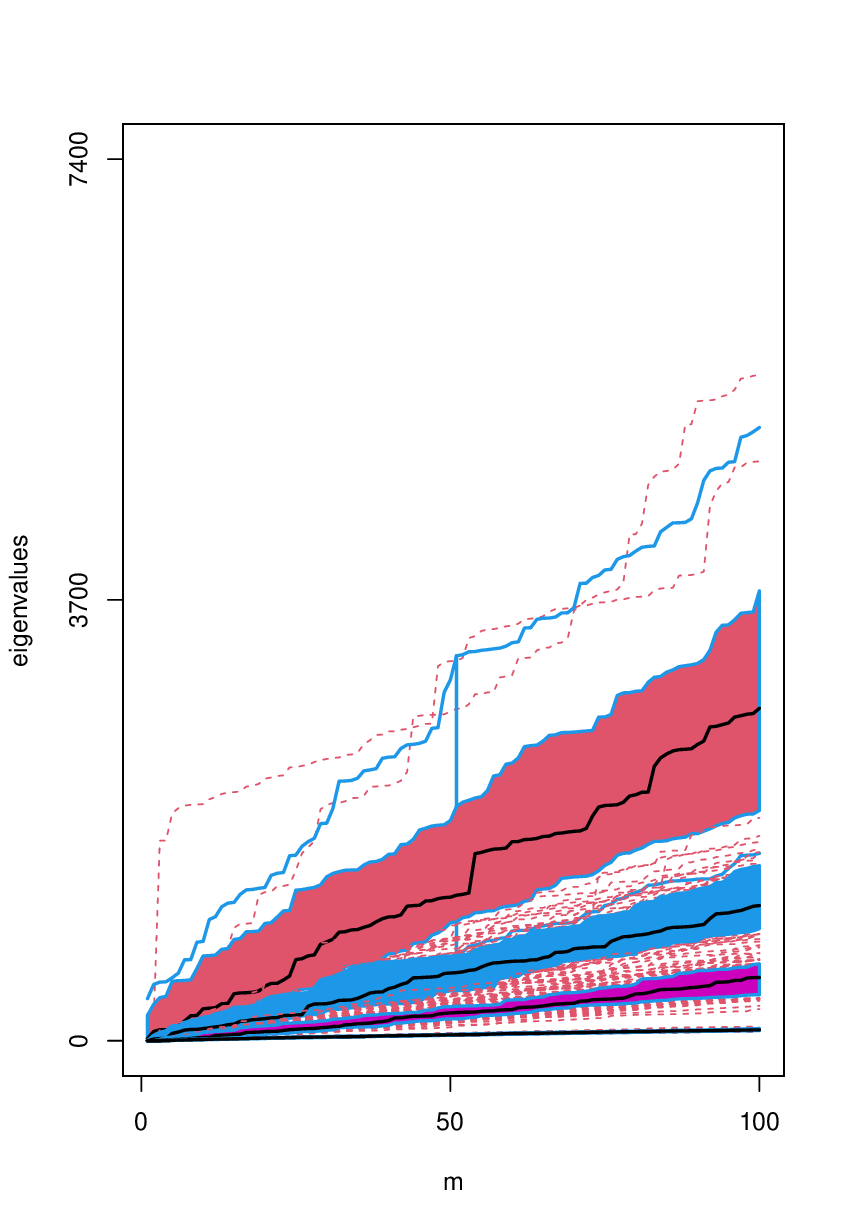} &
\includegraphics[width=0.3\textwidth, height=0.4\textwidth]{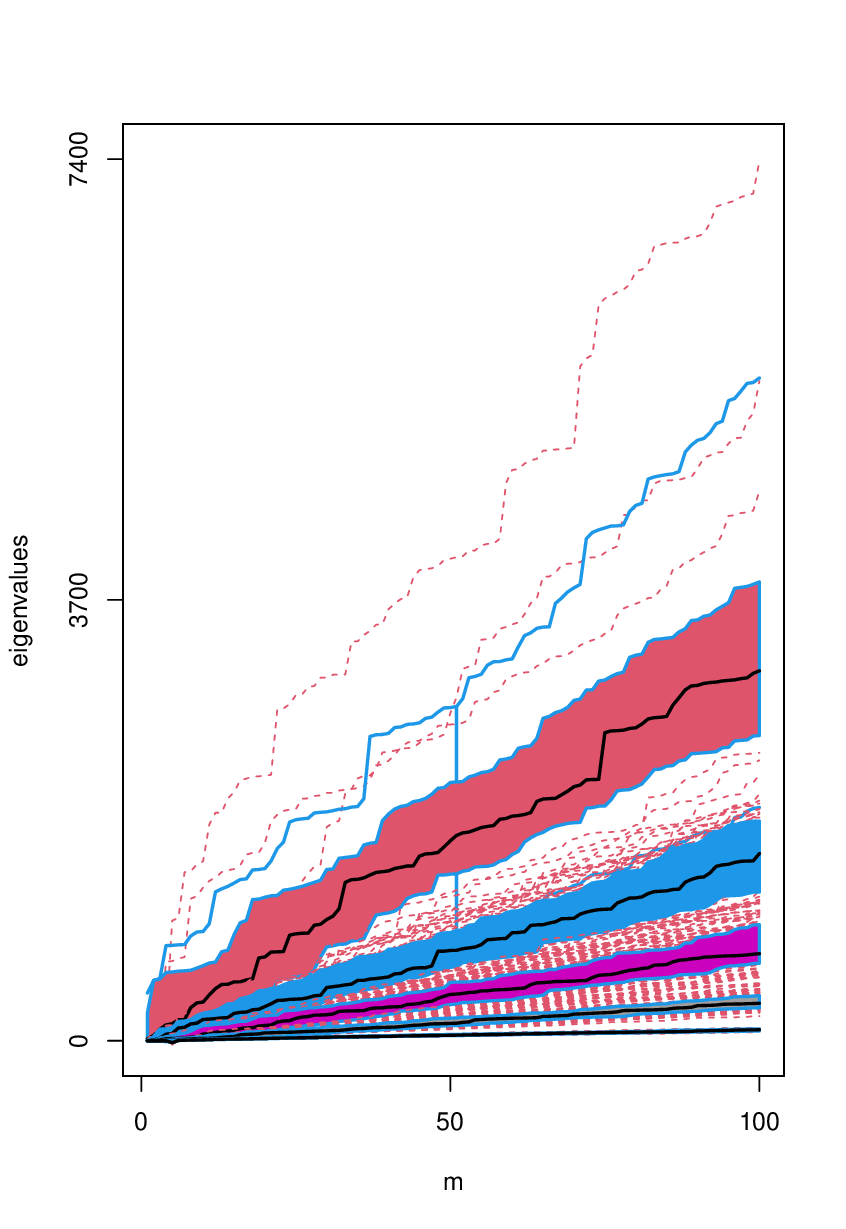}\\
 \end{tabular}
 \caption{Functional boxplots,  over $100$ replications,  of the largest $q+1$ spatio-temporal dynamic eigenvalues at frequency $0$ for Model (a) in \eqref{modelAR} with the idiosyncratic components generated as in \eqref{eq:neve}, and  with $n = 100, (S_1, S_2, T) = (15, 15, 15)$, and $ q = 2$ (left),  $q = 3$ (middle), and $q=4$ (right).} \label{EigenAR_fbplot_q234_freq0}
 \begin{tabular}{ccc}
 \includegraphics[width=0.3\textwidth, height=0.4\textwidth]{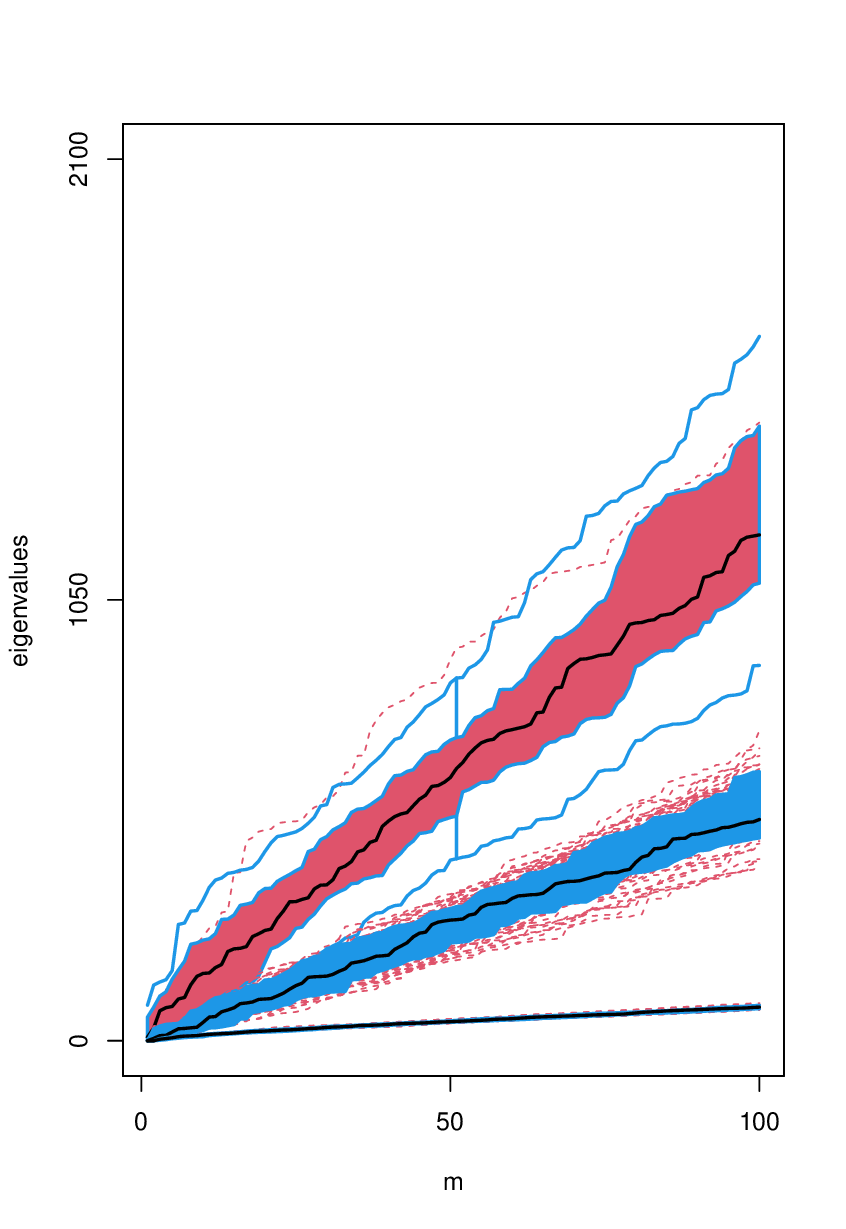} &
\includegraphics[width=0.3\textwidth, height=0.4\textwidth]{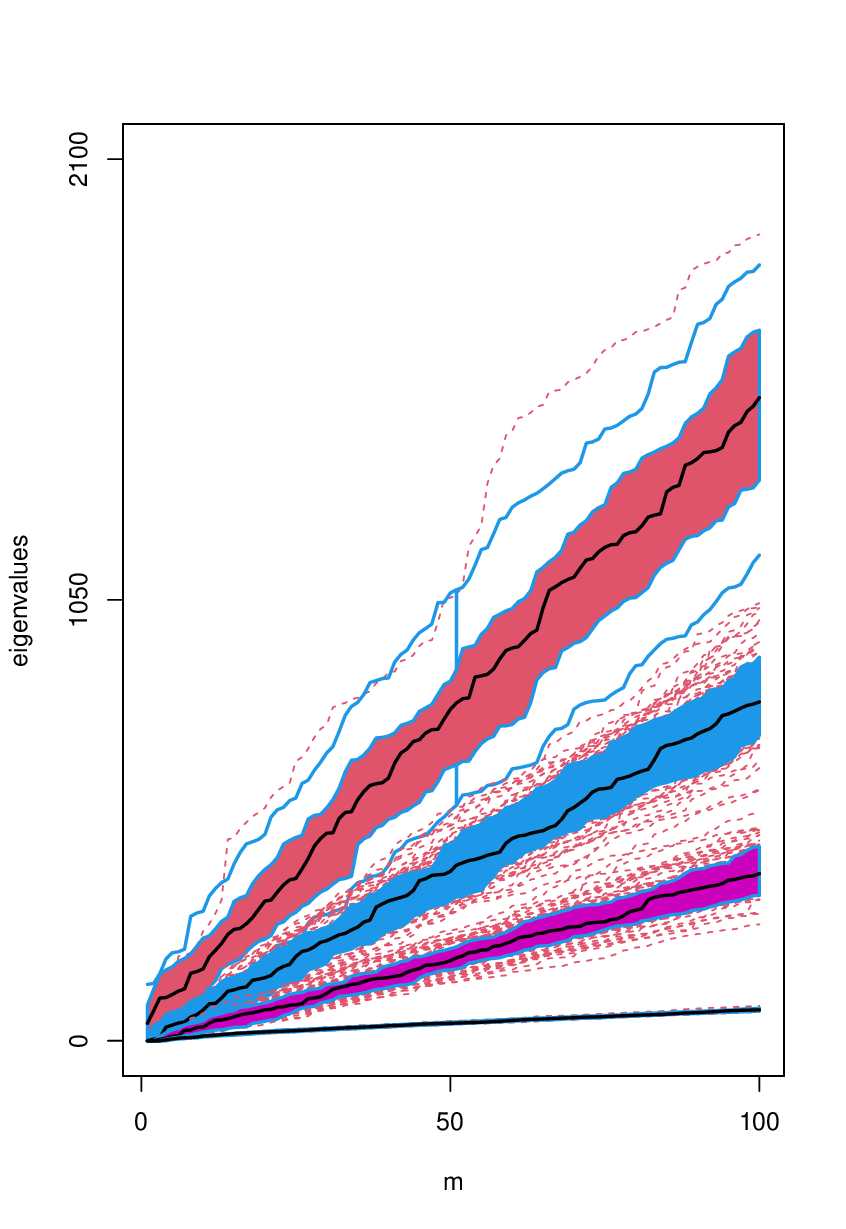} &
\includegraphics[width=0.3\textwidth, height=0.4\textwidth]{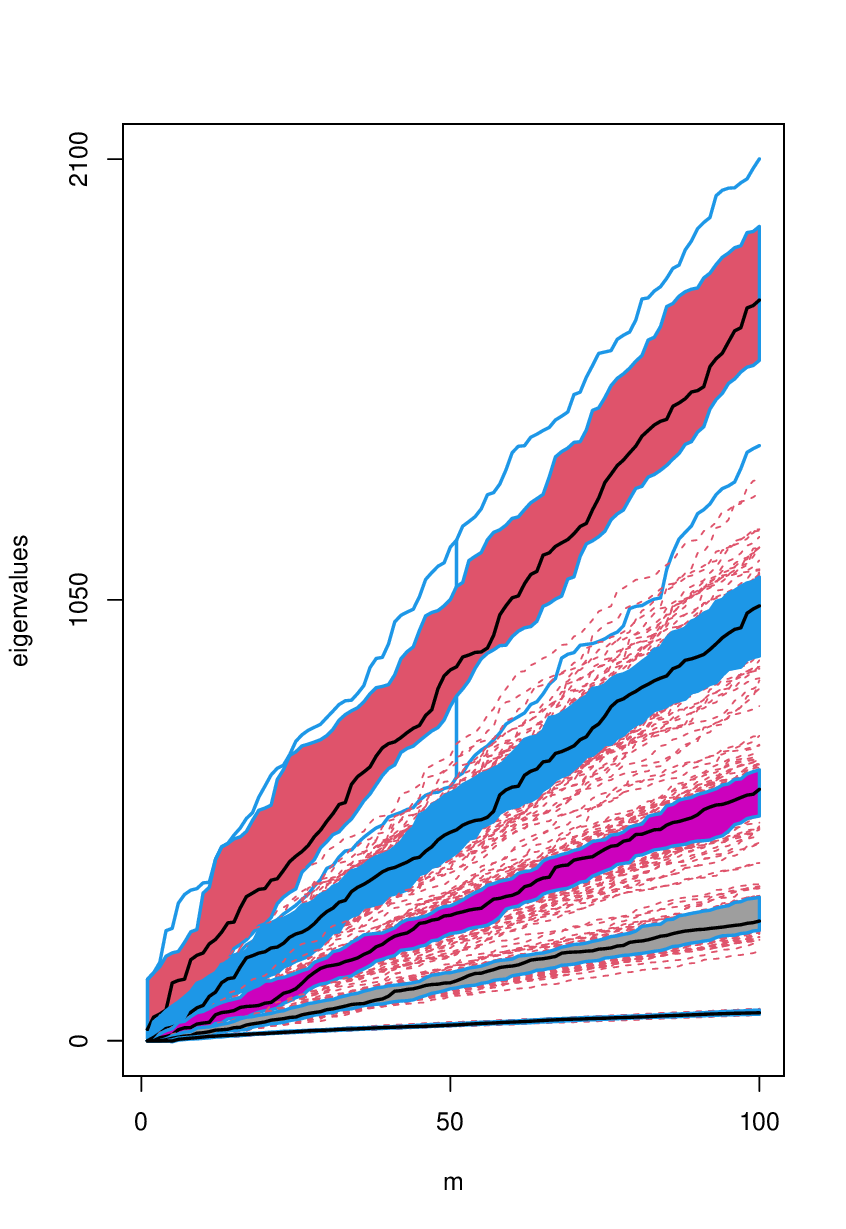}\\
 \end{tabular}
 \caption{Functional boxplots,  over $100$ replications,  of the largest $q+1$ spatio-temporal dynamic eigenvalues at frequency $0$ for Model (b) in  \eqref{modelMA} with the idiosyncratic components generated as in \eqref{eq:neve}, and with $n = 100, (S_1, S_2, T) = (15, 15, 15)$, and  $q = 2$ (left),  $q = 3$ (middle), and $q=4$ (right).} \label{Sim_Eigen_Modelb_GSTFR_idioAC} 
 \end{center} 
 \end{figure}
 Finally, since the GSTFM is able to accommodate also the presence of mildly cross-sectionally correlated idiosyncratic component, we illustrate the presence of the eigen-gap also in a novel  setting. Specifically, we assume that in Model (a) and Model (b),  the idiosyncratic component $\xi_{\ell \vsbf}$ is
\beq\label{eq:neve}
\xi_{\ell \vsbf} = \sum_{\kbf = (-1 \ -1 \ 0)}^{(1 \ 1 \ 1)} \sum_{j=0}^4 0.5^{\vert \kappa_1 \vert + \vert \kappa_2 \vert + \vert \kappa_3  \vert + j} c_{\ell j \kbf} L^{\kbf} {v}_{\ell+j,\vsbf},
\eeq
where ${v}_{\ell,\vsbf}$ from i.i.d. standard normal distribution and $c_{\ell i \kbf}$ from i.i.d. uniform distributions on $[0.5, 0.8]$ are independent. By design, the $\xi_{\ell \vsbf}$'s are autocorrelated in cross-section, space and time. In Figures \ref{EigenAR_fbplot_q234_freq0} and \ref{Sim_Eigen_Modelb_GSTFR_idioAC}, we show that also in this case the eigen-gap is clearly detectable.
$\,$
 \end{supplement}
 
\end{document}